%% file: main.tex
\documentclass{article}
\include{header}

\begin{document}
\input{title}

\input{abstract}
\tableofcontents
\input{document}
\input{acknowledgements}

\bibliographystyle{chicago}
\bibliography{ref}
\input{appendices}
\end{document}

%% file: header.tex
\usepackage{bbm}
\usepackage{dsfont}
\usepackage{epsf}
\usepackage{fancyhdr}
\usepackage{graphics}
\usepackage{graphicx}
\usepackage{psfrag}
\usepackage[dvipsnames]{xcolor}

\usepackage{amsthm}
\usepackage{amsfonts}
\usepackage{amsmath}
\usepackage{amssymb}
\usepackage{bbm}

\usepackage[utf8]{inputenc}
\usepackage[colorlinks=true, allcolors=blue]{hyperref}
\usepackage[authoryear]{natbib}
\usepackage{appendix}

\usepackage{mathtools}
\usepackage{graphicx}
\usepackage{caption}
\usepackage{subcaption}
\newsavebox{\bigimage}
\usepackage{multirow}
\usepackage{float}
\usepackage[nodisplayskipstretch]{setspace} 
\usepackage[capitalise,noabbrev,nameinlink]{cleveref}
\usepackage{bm}
\usepackage{dsfont}
\usepackage{xspace}
\usepackage{pifont}%

\usepackage[flushleft]{threeparttable}
\usepackage{footnote}
\makesavenoteenv{tabular}
\makesavenoteenv{table}

\usepackage{tikz}
\usetikzlibrary{shapes.geometric, arrows}
\usepackage{algorithm}
\usepackage{algorithmic}

\usepackage{listings}  %for code chunk
\newtheorem{definition}{Definition}

\newtheorem{fact}{Fact}
\newtheorem{theorem}{Theorem}

\newtheorem{corollary}{Corollary}

\newtheorem{assumption}{Assumption}
\newtheorem{remark}{Remark}
\newtheorem{example}{Example}

\newtheorem{proposition}{Proposition}

\newcommand{\R}{\mathbb{R}}
\newcommand{\XinfM}{X_{M}^{\mathrm{inf}} }
\newcommand{\XinfMT}{X_{M}^{\mathrm{inf}T}}
\newcommand{\Tau}{\mathcal{T}}

\newcommand{\I}{\mathcal{I}}
\newcommand\independent{\protect\mathpalette{\protect\independenT}{\perp}}
\def\independenT#1#2{\mathrel{\rlap{$#1#2$}\mkern2mu{#1#2}}}

\newcommand{\XinfMstar}{X_{M^{I}}}

\newcommand{\Var}{\mathrm{Var}}

\newcommand{\norm}[1]{\left\lVert#1\right\rVert}

\newcommand{\distras}[1]{%
  \savebox{\mybox}{\hbox{\kern3pt$\scriptstyle#1$\kern3pt}}%
  \savebox{\mysim}{\hbox{$\sim$}}%
  \mathbin{\overset{#1}{\kern\z@\resizebox{\wd\mybox}{\ht\mysim}{$\sim$}}}%
}

\DeclareMathOperator*{\argmin}{argmin}
\DeclareMathOperator*{\argmax}{argmax}

\newlength{\widebarargwidth}
\newlength{\widebarargheight}
\newlength{\widebarargdepth}
\DeclareRobustCommand{\widebar}[1]{%
  \settowidth{\widebarargwidth}{\ensuremath{#1}}%
  \settoheight{\widebarargheight}{\ensuremath{#1}}%
  \settodepth{\widebarargdepth}{\ensuremath{#1}}%
  \addtolength{\widebarargwidth}{-0.3\widebarargheight}%
  \addtolength{\widebarargwidth}{-0.3\widebarargdepth}%
  \makebox[0pt][l]{\hspace{0.3\widebarargheight}%
    \hspace{0.3\widebarargdepth}%
    \addtolength{\widebarargheight}{0.3ex}%
    \rule[\widebarargheight]{0.95\widebarargwidth}{0.1ex}}%
  {#1}}

\makeatletter
\long\def\@makecaption#1#2{
        \vskip 0.8ex
        \setbox\@tempboxa\hbox{\small {\bf #1:} #2}
        \parindent 1.5em  
        \dimen0=\hsize
        \advance\dimen0 by -3em
        \ifdim \wd\@tempboxa >\dimen0
                \hbox to \hsize{
                        \parindent 0em
                        \hfil 
                        \parbox{\dimen0}{\def\baselinestretch{0.96}\small
                                {\bf #1.} #2
                                } 
                        \hfil}
        \else \hbox to \hsize{\hfil \box\@tempboxa \hfil}
        \fi
        }
\makeatother

%% file: title.tex
\begin{center}

{\bf{\LARGE{Data fission: splitting a single data point}}}

\vspace*{.2in}

{\large{
\begin{tabular}{cccc}
James Leiner \textsuperscript{1} & Boyan Duan\textsuperscript{2} &  Larry Wasserman\textsuperscript{1} & Aaditya Ramdas\textsuperscript{1} \\
\end{tabular}
\texttt{\{jleiner,larry,aramdas\}@stat.cmu.edu}\\
\texttt{boyand@google.com}
}}

\vspace*{.2in}

\begin{tabular}{c}
\textsuperscript{1}Department of Statistics and Data Science, Carnegie Mellon University \\
\textsuperscript{2}Google
\end{tabular}

\vspace*{.2in}

\today
\end{center}

%% file: abstract.tex
\begin{abstract}

    Suppose we observe a random vector $X$ from some distribution in a known family with unknown parameters. We ask the following question: when is it possible to split $X$ into two pieces $f(X)$ and $g(X)$ such that neither part is sufficient to reconstruct X by itself, but both together can recover X fully, and their joint distribution is tractable? One common solution to this problem when multiple samples of X are observed is data splitting, but \cite{rasines2021splitting} offers an alternative approach that uses additive Gaussian noise --- this enables post-selection inference in finite samples for Gaussian distributed data and asymptotically when errors are non-Gaussian. In this paper, we offer a more general methodology for achieving such a split in finite samples by borrowing ideas from Bayesian inference to yield a (frequentist) solution that can be viewed as a continuous analog of data splitting. We call our method data fission, as an alternative to data splitting, data carving and p-value masking. We exemplify the method on several prototypical applications, such as post-selection inference for trend filtering and other regression problems, and effect size estimation after interactive multiple testing.
\end{abstract}

%% file: document.tex
\section{Introduction} \label{sec:introduction}

One of the most common practices in applied statistics is data splitting. Given a dataset $X=(X_1,\dots,X_n)$ containing $n$ independent samples, suppose an analyst wishes to divide the data into two smaller independent datasets in order to complete an analysis. The typical method for accomplishing this would be to choose $m$ such that $1\le  m < n$ and then form two new datasets: $f(X)=(X_1,\dots,X_m)$ and $g(X)=(X_{m+1},\dots,X_n)$. However, there are alternative approaches for accomplishing this goal that may be preferable. As a simplified example, consider the setting where we only observe a single data point $X \sim N(0,1)$, and we would like to ``split'' $X$ into two parts such that each part contains some information about $X$, $X$ can be reconstructed from both parts taken together, but neither part is sufficient by itself to reconstruct $X$, and yet the joint distribution of these two parts is known. The constraints mentioned in the previous sentence avoid trivial solutions like outputting $f(X) = X$ and $g(X) = 0$, or $f(X)=X/3$ and $g(X) = 2X/3$, and so on. %In other words, we must ``partition the information'' in $X$ into two complementary pieces. 

Luckily, this example has a simple solution with external randomization. Generate an independent $Z \sim N(0,1)$, and set $f(X)= X+Z$ and $g(X)= X-Z$. One can then reconstruct $X$ by addition (and $Z$ by subtraction, but we care less about $Z$, which has no utility of its own). Just knowing one out of $f(X)$ or $g(X)$ does not allow one to reconstruct $X$, but both parts have nontrivial information about $X$ because their mutual information with $X$ is nonzero. %or simply that neither part is independent of $X$. 
We also know that $f(X)$ and $g(X)$ are actually independent, and their marginal distributions are also Gaussian, so their joint distribution is tractable and known. %Previously, the idea of adding and subtracting Gaussian noise has been employed for various goals in the literature, such as false discovery rate control in regression problems \cite{gaussian_mirrors} and to estimate risk in the normal means problem as discussed in \cite{oliveira2021unbiased}. Here we ask and answer different questions, with different goals in mind, and provide alternatives beyond the Gaussian case.

More generally, we seek to construct a family of pairs of functions $(f_\tau(X), g_\tau(X))_{\tau \in \Tau}$, for some totally ordered set $\Tau$ (typically a subset of the real line), such that we can smoothly trade off the amount of information that each part contains about $X$. When $\tau$ approaches $\tau^+ := \sup\{\tau: \tau \in \Tau\}$, $f(X)$ will approach independence from $X$, while $g(X)$ will essentially equal $X$, but when $\tau$ approaches $\tau^- := \inf\{\tau: \tau \in \Tau\}$, the opposite will happen. To see how to do this, simply choose $Z$ as before, and define $f(X) = X - \tau Z$ and $g(X) = X + \frac{Z}{\tau}$ and let $\Tau := (0,\infty)$. We call this procedure ``data fission'' because we divide $X$ into two parts, each of which provides an independent yet complementary view of the original data.

Data fission is similar in spirit to data splitting. However, data fission manages to achieve the same effect from just a single sample $X$ and not an $n$-dimensional vector. Nevertheless, the connection to data splitting is more than a mere analogy, and the exactly relationship between $\tau$ and $m$ can be quantified such that data fission is viewable as a continuous analog of data splitting in the Gaussian case; we do this in the next section.

Now, can the above ideas be generalized to other distributions? In other words, can we employ external randomization to ``split'' a single data point into two nontrivial parts when the distribution $P$ is not Gaussian? This is the topic of study for the rest of this paper. We provide a positive answer when $P$ is conjugate (in the standard Bayesian sense) to some other distribution $Q$, where the latter will be used (along with $X$) to determine the external randomization. In most cases, $f(X)$ and $g(X)$ will not simply be the sum/difference of $X$ with some $Z$; such a form was achieved only in the Gaussian case. Similarly, $f(X)$ and $g(X)$ will typically not be independent. Nevertheless, they will satisfy the conditions set out in the first paragraph of the paper and can be treated for inferential purposes as single sample variants of data splitting, justifying the title of the paper. 
% Our applications, however, will use these ideas to permit post-selection inference for regression problems with multiple samples.

\subsection{An application: data fission for post-selection inference}

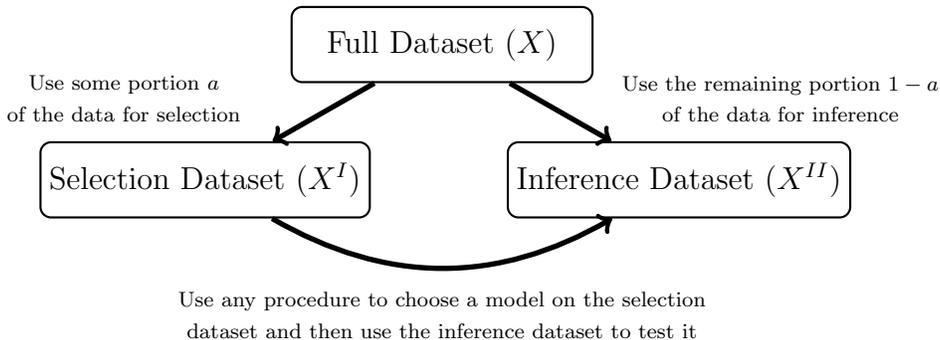
\begin{figure}
\centering
\begin{tikzpicture}[scale=0.9]
\begin{scope}
    \node[rectangle, rounded corners,thick,draw,minimum height=1cm,minimum width=4cm] (A) at (1.5,0) {\large Full Dataset ($X$)};
    \node[rectangle, rounded corners,thick,draw,minimum height=1cm,minimum width=4cm] (B) at (-2,-2) {\large Selection Dataset ($X^{I}$)};
    \node[rectangle, rounded corners,thick,draw,minimum height=1cm,minimum width=4cm] (C) at (5,-2) {\large Inference Dataset ($X^{II}$)};
    \node[align=center] (D) at (-3.2,-0.8)  {\footnotesize Use some portion $a$ \\ \footnotesize of the data for selection};
    \node[align=center] (D) at (6.5,-0.8)  {\footnotesize Use the remaining portion $1 -a$ \\ \footnotesize of the data for inference};
    \node[align=center] (E) at (1.5,-4)  {\footnotesize Use any procedure to choose a model on the selection \\ \footnotesize dataset and then use the inference dataset to test it};
    \draw[->,line width=0.7mm, to path={-- + (\tikztotarget)}]
  (A) edge (B) (A) edge (C) (B) edge[bend right] (C);
\end{scope}
\end{tikzpicture}
\caption{Illustration of typical \textbf{data splitting procedures} for post-selection inference. Splitting the data has the advantage of allowing the user to choose any selection strategy for model selection, but at the cost of decreased power during the inference stage.}
\label{fig:sample_splitting_illustration}
\end{figure}

We primarily focus on demonstrating the applicability of these ideas in the context of (potentially high-dimensional) model selection and post-selection inference. With data splitting, the analyst picks some fraction $a \in \{\tfrac{1}{n}, \ldots, \tfrac{n-1}{n}, 1\}$ of the data to use for model selection and the remaining $1-a$ fraction to use for inference as illustrated in \cref{fig:sample_splitting_illustration}. Data fission is similar in spirit to this idea but instead uses randomization so that part of the information contained in every data point is used for both selection and inference. The procedure broadly works in three stages. 
\begin{enumerate}
    \item Split $X$ into $f(X)$ and $g(X)$ such that $g(X) | f(X)$ is tractable to compute. The parameter $\tau$ controls the proportion of the information to be used for model selection. 
    \item Use $f(X)$ to select a model and/or or hypotheses to test using any procedure available.
    \item Use $g(X) | f(X)$ to test hypotheses and/or perform inference. 
\end{enumerate}

See \cref{fig:data_fission_illustration} for a graphical representation of the above steps. This approach can be viewed as a generalization of the methodologies discussed in \cite{tian2018selective} and \cite{rasines2021splitting}. In the above framework, these approaches amount to letting $f(X) = X + \gamma Z$  and $g(X) = X$ with $Z \sim N(0,\sigma^{2})$ for some fixed constant $\gamma > 0 $. In the case where $X \sim N(\mu, \sigma^{2})$ with known $\sigma^{2}$, the authors show that $g(X) | f(X)$ has a tractable finite sample distribution. In cases where $X$ is non-Gaussian, however, $g(X) | f(X)$ is only described asymptotically. In the next section, we will explore alternative ways  to construct $g(X)$ that result in tractable finite sample distributions for non-Gaussian data.

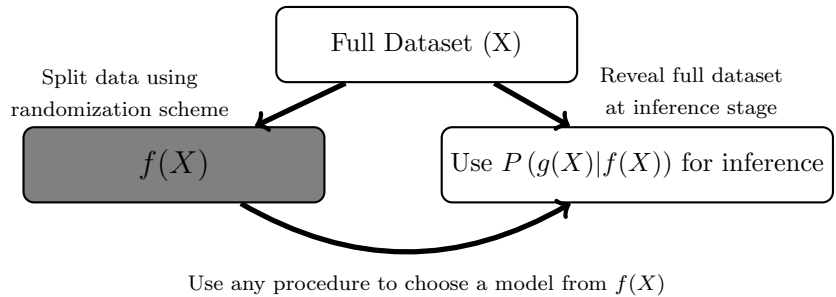
\begin{figure}
\centering
\begin{tikzpicture}[scale=0.8]
	\begin{scope}
		\node[rectangle, rounded corners,thick,draw,minimum height=1cm,minimum width=4cm] (A) at (1.5,0) { Full Dataset (X)};
		\node[rectangle, rounded corners,thick,draw,minimum height=1cm,minimum width=4cm,fill=gray] (B) at (-2.7,-2) {\large $f(X)$};
		\node[rectangle, rounded corners,thick,draw,minimum height=1cm,minimum width=4cm] (C) at (5,-2) { Use $P\left(g(X) | f(X) \right) $ for inference};
		\node[align=center] (D) at (-3.6,-0.8)  {\footnotesize Split data using \\ \footnotesize randomization scheme};
		\node[align=center] (E) at (5.9,-0.8)  {\footnotesize Reveal full dataset \\ \footnotesize at inference stage};
		\node[align=center] (F) at (1.5,-4)  {\footnotesize Use any procedure to choose a model from $f(X)$};
		\draw[->,line width=0.7mm, to path={-- + (\tikztotarget)}]
		(A) edge (B) (A) edge (C) (B) edge[bend right] (C);
	\end{scope}
\end{tikzpicture}
\caption{Illustration of the proposed \textbf{data fission} procedure. Similar to data splitting, it allows for any selection procedure for choosing the model. However, it achieves this through randomization rather than a direct splitting of the data.}
\label{fig:data_fission_illustration}
\end{figure}

In some ways, these methodologies can be seen as a compromise between data splitting and the approach of \emph{data carving} as introduced in \cite{fithian2014optimal}. Data carving, as illustrated in \cref{fig:data_carving_illustration}, improves on data splitting in cases where the conditional distribution of the data given a selection event is known by including the leftover portion of Fisher information that was not used to inform the model choice in the inference procedure. A key limitation of this approach, however, is that it confines the analyst to model selection techniques with tractable post-selective distributions, such as the LASSO as described by \cite{lasso_posi} or more general sequential regression procedures such as those discussed in \cite{sequential_posi}. In many settings, ad-hoc exploratory data analyses such as the plotting of data or removal of outliers are ubiquitous and make data carving intractable.

Although data fission conditions on $f(X)$ rather than the selection event itself, it retains some similarities to data carving insofar as it uses a portion of every data point to inform both selection and inference. This has advantages relative to data splitting in at least two distinct ways. First, certain settings that involve sparse or rarely occurring features may result in a handful of data points having a disproportionate amount of influence --- data fission allows for the analyst to ``hedge their bets'' by including these points in both the selection and inference steps. Second, in settings where the selected model is defined relative to a set of \emph{fixed} covariates, the theoretical justification for data splitting becomes less clear conceptually. (In a fixed-X setup, how can a model that has been selected based on its ability to estimate the conditional distribution $Y \mid X^{\mathrm{first half}}$ also be understood to model the distribution that conditions on the other half of the split dataset  $Y \mid X^{\mathrm{second half}}$?) For example, consider a time series dataset, where splitting a sample may require the analyst to have a selection and inference dataset that span entirely different time periods, or graph data where there is only one observation available for each location on a graph. On the other hand, similar to data splitting, data fission affords the analyst complete flexibility in how they choose their model based on the information revealed in the selection stage. In particular, the procedure can accommodate a model selection process that relies on qualitative or heuristic methods such as visual examination of residual plots or the consultation of domain experts to determine plausibility of the discovered relationships.
%Data carving is able to improve power relative to both data splitting and data fission but at the expense of requiring the analyst to pre-commit to a fixed selection procedure that is tractable either in closed form or through numerical simulations. 

Although we anticipate that these ideas may have other downstream applications beyond selective inference such as data privacy (including differential privacy), creation of fake (synthetic) datasets, and comparing machine learning algorithms, these require new techniques that are out of the scope of the current work.

\begin{figure}
\centering
\begin{tikzpicture}[scale=0.8]
	\begin{scope}
		\node[rectangle, rounded corners,thick,draw,minimum height=1cm,minimum width=4cm] (A) at (1.5,0) {Full Dataset ($X$)};
		\node[rectangle, rounded corners,thick,draw,minimum height=1cm,minimum width=4cm] (B) at (-2,-2) {Selection Event $S(X)$};
		\node[rectangle, rounded corners,thick,draw,minimum height=1cm,minimum width=4cm] (C) at (5.6,-2) { Use $P\left(X|S(X)\right)$ for inference};
		\node[align=center] (D) at (-3,-0.8)  {\footnotesize Choose some fixed \\ \footnotesize selection procedure $S$};
		\node[align=center] (E) at (1.5,-4.2)  {\footnotesize Conditioning on S(X) needs to be tractable \\ \footnotesize (e.g. LASSO, forward selection)};
		\draw[->, line width=0.7mm, to path={-- + (\tikztotarget)}]
		(A) edge (B) (A) edge (C) (B) edge[bend right] (C);
	\end{scope}
\end{tikzpicture}
\caption{Illustration of \textbf{data carving procedure} as discussed in \cite{fithian2014optimal}. Data carving has the advantage of using all unused information for inference, but requires the selection procedure to be \emph{fixed} at the onset of investigation. Moreover, computing the conditional distribution needs to be tractable, either in closed form (e.g. LASSO as described in \cite{lasso_posi}) or through numerical simulation. Thus data carving and fission have complementary benefits and tradeoffs.}
\label{fig:data_carving_illustration}
\end{figure}
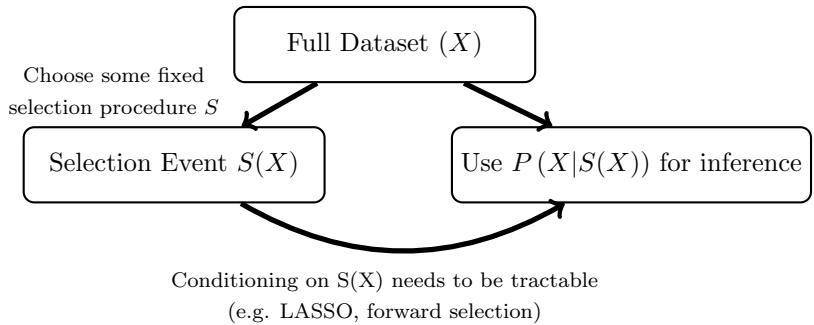

\subsection{Related work on data splitting and carving} 
Our work is influenced by the existing literature on procedures for selective inference after model selection. Although data splitting is perhaps the oldest and most commonly used method for ensuring valid coverage guarantees after model selection, rigorous examination of data splitting has only recently emerged in the literature. \cite{samplesplit_bootstrap} is one such example where the authors examine data splitting in an assumption-lean context. %with weak assumptions on the distribution of the data and no requirements that the model chosen during the selection stage is ``correct".  

Previously, the idea of adding and subtracting Gaussian noise has been employed for various goals in the literature. \cite{tian2018selective} use this method for selective inference by introducing a noise variable that is independent of the original data, and involve it at various stages of the selection procedure (e.g., introducing perturbations to the gradient when using LASSO to select features). An alternative route for selective inference in Gaussian regression that is explored in \cite{gaussian_mirrors} is to leave the response data unperturbed but create randomized versions of the \emph{covariates} (termed ``Gaussian mirrors'') by adding and subtracting Gaussian noise. This allows the authors to create test statistics with symmetric distributions under the null, enabling FDR control for the LASSO in high-dimensional settings. In contrast to both of these approaches, we focus on creating a perturbed version of the response that allows the analyst to conduct hypothesis selection in an arbitrary and heuristic fashion, as opposed to relying on a specific selection procedure such as the LASSO. 

Randomization was also used in \cite{li2021whiteout} to recast the knockoff procedure of \cite{knockoffs} as a selective inference procedure for the linear Gaussian model that adds noise to the OLS estimates ($\widehat{\beta}$) to create a ``whitened'' version of $\widehat{\beta}$ to use for hypothesis selection. The work of \cite{sarkar2021adjusting} explores similar ways of using knockoffs to split $\widehat{\beta}$ into independent pieces for hypothesis selection but uses a deterministic splitting procedure. Although conceptually quite similar to our approach, these methods focus on splitting the coefficient estimates for a \emph{fixed model} into two independent pieces and then adaptively choosing the hypotheses to test during the inference stage. As such, it is most naturally applied in low dimensional settings where model selection is less necessary. Similar randomization schemes for Gaussian distributed data are also explored in \cite{empirical_bayes} but are used for empirical Bayes procedures and not selective inference.

%Of considerable interest to us is making these results applicable in assumption-lean settings where few requirements are placed on the true distribution of the data even when parametric models are ultimately used for inference. \cite{buja2019modelsII} outlines a procedure for generic likelihood-based inference procedures but outside of a selective inference setting. \cite{buja2019modelsI} focuses exclusively on assumption-lean linear regression but does not explore in detail its properties in the context of data splitting. Throughout our work, we try to place minimal assumptions on the conditional mean function $\mu = \mathbb{E}[Y|X]$ of our data following a particular structural form but still require the distribution of the data to be correctly specified. 

% we take the example of application in constructing CIs for selective signals?

\paragraph{Outline of the paper.} 
The general methodology for data fission is introduced in \cref{sec:fission}. We illustrate how the procedure is used in the context of Gaussian data but then generalize this for the broader class of distributions where the data has a known conjugate prior. Examples are given across a variety of distributions commonly used for regression and data analysis such as Gaussian, Poisson, and Binomial. The remainder of the paper explores the use of data fission for four different applications: selective CIs after interactive multiple testing in \cref{sec:interactive_testing}, fixed-design linear regression in \cref{sec:linreg}, fixed-design generalized linear models in \cref{sec:qmle}, and trend filtering in \cref{sec:trendfilter}. Proofs for all theoretical results are omitted from the main body but included in Appendix~\ref{sec:appendix_theory}. 

A key limitation within all of these applications --- shared by much of classical statistics and conditional inference --- is that theoretical guarantees can only be given under correct specification of the distribution of errors due to the need to ensure that the post-selective distribution is known and tractable. However, these guarantees are still assumption-lean in the sense of assuming an unknown form for the conditional mean $\mu$. Situations where the variance is unknown and is estimated before data fission in the Gaussian case are discussed in low dimensional settings, but we leave theoretical guarantees in higher dimensions as an open avenue for future investigation. We provide some concluding remarks in \cref{sec:conclude}.

\section{Techniques to accomplish data fission} \label{sec:fission}
With statistical inference in mind, we explore decompositions of $X$ to $f(X)$ and $g(X)$ such that both parts contain information about a parameter $\theta$ of interest, and there exists some function $h$ such that $X = h(f(X), g(X))$, and with either of the following two properties:
\begin{itemize}
    \item[\textbf{(P1)}.] $f(X)$ and $g(X)$ are independent with known distributions (up to the unknown $\theta$); or
    \item[\textbf{(P2)}.] $f(X)$ has a known marginal distribution and $g(X)$ has a known conditional distribution given $f(X)$ (up to knowledge of $\theta$).
\end{itemize}
The former property implies the latter, but the latter is generally more tractable. We explore alternative formulations that require weaker assumptions in \cref{sec:qmle}. 
% but these rules are less likely to be of general interest than rules which satisfy \textbf{(P1)} or \textbf{(P2)}.
% For now, we think the latter property cannot replace the former for all possible applications. 

% For example, consider the application of data privacy, where a client specify a selection rule to request observations of subjects in his/her interest, and the company give noised observations of the interested subjects. When the independence property is satisfied between $f(X)$ and $g(X)$, the company can use $f(X)$ to filter subjects using the client's selection rule, and provide $g(X)$ as the noised data to the client. In this way, the client can freely perform inference on $g(X)$ without selection bias from the selective request. However, if only the latter property is satisfied, inference on $g(X)$ contains selection bias, and we cannot provide the conditional distribution of $g(X) \mid f(X)$, as it reveals $f(X)$ and breaks the data privacy.

%We consider the parametric setting where we have a sample $X$ from distribution $p_{\theta}(x)$ where $\theta \in \Theta$. In order to make inference on the parameter $\theta$ with multiple steps as we usually do with data splitting approach, we aim at constructing $f(X)$ and $g(X)$ such that both distributions of $f(X)$ and $g(X) \mid f(X)$ are known as a one-to-one map of the parameter $\theta$ of our interest. 

%A simple example is for a sample $X$ from the Gaussian distribution $N(\mu, 1)$ where $\mu \in (-\infty, \infty)$ is our parameter of interest, $f(X) = X$

\subsection{Achieving (P2) using ``conjugate prior reversal''} \label{sec:conjugate}
%The data-masking idea might be applied to other distributions, without generating two independent parts of the data, as long as we know the marginal distribution of one part $g_i$ under the null (to compute $p$-values for identification) and the conditional distribution of $Z_i$ given $g_i$ (to construct CIs). 
% To construct $f(X)$ and $g(X)$ satisfying \textbf{(P2)}, we explore the class of conjugate priors. 
Suppose $X$ follows a distribution that is a conjugate prior distribution of the parameter in some likelihood. We then construct a new random variable $Z$ following that likelihood (with the parameter being $X$). Letting $f(X) = Z$ and $g(X) = X$, the conditional distribution of $g(X) \mid f(X)$ will have the same form as $X$ (with a different parameter depending on the value of $f(X)$). For example, suppose $X \sim \mathrm{Exp}(\theta)$, which is conjugate prior to the Poisson distribution. Thus, we draw $ Z = (Z_1, \ldots, Z_B)$ where each element is i.i.d.\ $Z_i \sim \mathrm{Poi}(X)$ and $B \in \{1, 2, \ldots\}$ is a tuning parameter. Let $f(X) =  Z$ and $g(X) = X$. Then, the conditional distribution of $g(X) \mid f(X)$ is $\mathrm{Gamma}(1 + \sum_{i=1}^B f_i(X), \theta + B)$ %(notice that Exponential distribution is a special case of the Gamma distribution), which can be used for inference on $\theta$.
On the other hand, $f(X) \sim \mathrm{Geo}(\tfrac{\theta}{\theta+B})$. Larger $B$ results in a more informative $f(X)$. More examples of such decompositions, which we term ``conjugate prior reversal'', are in Appendix~\ref{sec:appendix_list_decomp}. One drawback of this approach, however, is that it will often result in a distribution that is no longer straightforwardly related to the original parameter of interest and so extra care needs to be taken when performing inference---we explore this implication in greater detail in \cref{sec:qmle}.

For exponential family distributions, we can construct $f(X)$ and $g(X)$ as follows. 
\begin{theorem} \label{thm:conjugate_reversal}
Suppose that for some $A(\cdot,\cdot),\theta_1,  \theta_2, H(\cdot,\cdot)$, the density of $X$ is given by \begin{align}
    p( x \mid \theta_1,  \theta_2) = H(\theta_1, \theta_2) \exp\{ \theta_1^T  x - \theta_2^T A(x)\}.
\end{align}
Suppose also that we can find   $h(\cdot), T(\cdot)$ and $\theta_3$  such that
\begin{align} \label{eq:ZgivenX}
    p(z \mid x, \theta_{3}) = h( z) \exp\{ x^T  T(z) -  \theta_3^T  A(x)\}
\end{align}
is a well-defined distribution. First, draw $Z \sim p(z|X)$, and let $f(X):= Z$ and $g(X):= X$. Then, $(f(X),g(X))$ satisfy the data fission property (P2). Specifically, note that $f(X)$ has a known marginal distribution $p( z |  \theta_1, \theta_2,\theta_3) = h( z)\frac{H(\theta_1, \theta_2)}{H( \theta_1 +  T( z),  \theta_2 + \theta_3)}$, while $g(X)$ has a known conditional distribution given $f(X)$, which is $p( x \mid  z, \theta_1, \theta_2,  \theta_3) = p( x \mid  \theta_1 +  T(z),  \theta_2 +  \theta_3)$. 
\end{theorem}

Recall that all proofs are in Appendix~\ref{sec:appendix_theory}.
% Typically, $\theta_3$ cannot be zero for $p(z \mid  x)$ to be a well-defined distribution for arbitrary $x$. Also, $p(x \mid \theta_1, \theta_2, \theta_3)$ may not be identifiable: $x$ can be included in $A(x)$, which give us multiple options to construct $p(z \mid x,\theta_1,\theta_2,\theta_3)$ according to the definition of $A (x)$.

    %To decompose $ X$, we draw $ Z$ from~\eqref{eq:ZgivenX}. Then $f( X) =  Z$ has marginal distribution $p( z) = h( z)\frac{H( \theta_1,  \theta_2)}{H( \theta_1 +  T( z),  \theta_2 +  \theta_3)}$; and $g( X) =  X$ has conditional distribution as $p( x \mid  z,  \theta_1,  \theta_2,  \theta_3) = p( x \mid  \theta_1 +  T( z),  \theta_2 +  \theta_3)$. 

\begin{remark}[Trading off information in $f$ and $g$] \label{rmk:tradeoff}
As an extension to the above result, we can draw $B$ samples independently from~\eqref{eq:ZgivenX} denoted as $ Z_i$ for $i \in [B]$, in which case $f(X) = Z \equiv \{Z_1, \ldots, Z_B\}$ has marginal distribution $p(z|\theta_{1},\theta_{2},\theta_{3}) = \left[\prod_{i=1}^B h( z_i)\right]\frac{H(\theta_1, \theta_2)}{H(\theta_1 + \sum_{i=1}^nT(z_i), \theta_2 + B \theta_3)}\text{,}$ and $g(X) =  X$ has conditional distribution $p( x \mid  z, \theta_1, \theta_2, \theta_3) = p(x \mid  \theta_1 + \sum_{i=1}^B T( z_i),  \theta_2 + B\theta_3)\text{.}$ Larger $B$ indicates more information in $f(X)$, and thus less left over in $g(X)$.
\end{remark}

%consider $Z_i \sim \mathrm{Beta}(\theta_i, 1)$, where the null hypothesis is $\theta_i = 1$ and the alternative is $\theta_i > 1$. Given a parameter $n \in \mathbb{N}_+$, we simulate a binomial random variable: $g_i \sim \mathrm{Bin}(n, Z_i)$, and treat it as one part of the data. A larger $n$ indicates more information for $g_i$ (identification) and less information for constructing CIs. To identify non-nulls, we compute (smoothed) $p$-values from $g_i$, which marginally follows a discrete uniform in $\{0, \ldots, n\}$ under the null. To construct CIs, notice that the conditional distribution of $Z_i$ given $g_i$ is $\mathrm{Beta}(\theta_i + g_i, n - g_i + 1)$, whose mean value is $(\theta_i + g_i)/(n + \theta_i + 1)$.
\subsection{Example decompositions} \label{sec:list_decomp}
 Below, we draw attention to techniques used in Sections~\ref{sec:interactive_testing}, \ref{sec:linreg}, \ref{sec:qmle}, and \ref{sec:trendfilter}.  
\begin{itemize}
    \item \textbf{Gaussian.}  Suppose $X \sim N( \mu,  \Sigma)$ is $d$-dimensional ($d \geq 1$). Draw $ Z \sim N(0,  \Sigma)$. Then $f( X) =  X + \tau  Z$, where $\tau \in (0, \infty)$ is a tuning parameter, has  distribution $N( \mu, (1 + \tau^{2})  \Sigma)$; and $g( X) =  X - \tfrac{1}{\tau}  Z$ has distribution $N(\mu,  (1 + \tau^{-2})  \Sigma)$; and $f( X) \independent g( X)$. Larger $\tau$ indicates less informative $f( X)$ (and more informative $g( X) \mid f( X)$). 
    \item \textbf{Bernoulli (P2).} Suppose $X \sim \mathrm{Ber}(\theta)$. Draw $Z \sim \mathrm{Ber}(p)$ where $p \in (0,1)$ is a tuning parameter. Then
     $f(X) = X(1 - Z) + (1 - X)Z$ has marginal distribution $\mathrm{Ber}(\theta + p - 2p\theta)$; and $g(X) = X$ has conditional distribution (given $f(X)$) as $\mathrm{Ber}\left(\frac{\theta}{\theta + (1-\theta) [p/(1-p)]^{2f(X) - 1}}\right)$.  
    %  Note that modeling $\log\left(\frac{\theta + p - 2p\theta}{1 - \theta - p + 2p\theta}\right)$ as $u\beta$ (for covariates $u$) does not imply $\log\left(\frac{\theta}{1 - \theta}\right)$ is $u\beta'$.
     Smaller $p$ indicates more information in $f(X)$.
    \item \textbf{Poisson.} Suppose $X \sim \mathrm{Poi}(\mu)$. Fix parameter $p\in(0,1)$ and draw $Z~\sim~\mathrm{Bin}(X, p)$. Then $f(X) = Z$ has marginal distribution $\mathrm{Poi}(p\mu)$; and $g(X) = X - Z$ is independent of $f(X)$ and has distribution  $\mathrm{Poi}((1 - p)\mu)$. Larger $p$ indicates more informative $f(X)$.
\end{itemize}

A larger (but still incomplete) list of decomposition strategies for specific distributions is available in Appendix~\ref{sec:appendix_list_decomp}. We encourage the reader to consult this list in order to get a full appreciation for the applicability of this approach to disparate problems in statistics.

\subsection{Relationship between data splitting and data fission} \label{sec:split_fission}
%A natural question is under which conditions data fission yields equivalent estimators as data splitting. One way to answer this is to examine when the proportion of Fisher information that is allocated to inference is the same for each approach. 
We explore the conditions under which data fission yields estimators that are comparable to data splitting. Suppose we are given $n$ i.i.d.\ observations $ X =(X_1, \ldots, X_n)$, where $X_i \sim p(\theta)$. The data splitting approach chooses $S$ as a random subset of $[n]$ of size $a$ where $a \in \{\tfrac{1}{n}, \ldots, \tfrac{n-1}{n}, 1\}$ is a tuning parameter. Letting $\mathcal{I}_{X}(\theta)$ denote the Fisher information for the complete sample and  $\mathcal{I}_{1}(\theta)$ denote the Fisher information for a single observation, we then have
\[\mathcal{I}_{X}(\theta) = \underbrace{an\mathcal{I}_{1}(\theta)}_{\text{Used for selection}} + \underbrace{(1-an)\mathcal{I}_{1}(\theta)}_{\text{Used for inference}}.\]
For data fission, note the following identity for smooth parametric models: $
	\nabla^{2} \ell(\theta| X) = \nabla^{2} \ell(\theta ; f(X)) + \nabla^{2} \ell(\theta ; g(X) \mid f(X))$.
Taking expectations, and denoting $\I_{f(X)}$, $\I_{g(X) | f(X)}$ as the Fisher information for the selection and inference datasets yields
\begin{align*}
 \mathcal{I}_{X}(\theta) &=  \I_{f(X)}(\theta)  + \mathbb{E}[-\nabla^{2} \ell(\theta ; g(X) \mid f(X)] \\
 &= \I_{f(X)}(\theta)  + \mathbb E \left[-\mathbb{E}[\nabla^{2} \ell(\theta ; g(X) \mid f(X)) \mid f(X)] \right] = \underbrace{\I_{f(X)}(\theta)}_{\text{for selection}} + \underbrace{\mathbb{E}[\I_{g(X)|f(X)}(\theta)]}_{\text{for inference}} .
\end{align*}
For a fixed parameter $a$ in data splitting, one can choose $\tau$ such that $
\mathbb{E}[\I_{g(X)|f(X)}(\theta)] = (1-an)\mathcal{I}_{1}(\theta)$
to find a comparable information split for data fission. This selected $\tau$ will only guarantee that the inference datasets created by data splitting and data fission approach contain the same information in expectation. For any particular realization of the fission step, $\I_{g(X)|f(X)}$ may be different than $(1-an)\mathcal{I}_{1}(\theta)$. In situations where $g(X)$ and $f(X)$ are independent, this equality can be modified to hold exactly since $\mathbb{E}[\I_{g(X)|f(X)}(\theta)]= \I_{g(X)}$.

\begin{remark}[Trading information between $f$ and $g$, part II] In the list of decompositions and in \cref{rmk:tradeoff} we noted that we could trade off ``information'' between $f(X)$ and $g(X)$ by varying certain hyperparameters. We can now clarify what this means. For a hyperparameter $p \in (a,b)$, we say that larger $p$ corresponds to more informative $f(X)$ and less informative $g(X)|f(X)$ to mean that 
$\lim_{p \rightarrow b} \mathcal{I}_{f(X)}\left(\theta \right) = \mathcal{I}_{X}\left(\theta \right) \text{~  and ~}
\lim_{p\rightarrow b} \mathbb{E}[\mathcal{I}_{g(X)|f(X)}\left(\theta\right)] = 0. $
\end{remark}

We now elaborate on the informal assertion in \cref{sec:introduction} that data fission for Gaussian data can be viewed as a continuous analog of data splitting with two examples. 

\begin{example}[Gaussian Datasets] \label{example:gaussian}
Let $\{X_{i}\}_{i=1}^{n}$ be iid  $N(\theta,\sigma^{2})$. $\mathcal{I}_{1} = \frac{1}{\sigma^{2}}$ and so $\frac{an}{\sigma^{2}}$ is the amount of information used for selection under a data splitting rule. If the data is fissioned using the rule in \cref{sec:list_decomp}, then $\mathcal{I}_{f(X)} = \frac{n}{\sigma^{2}(1 + \tau^2)}$. To compare the approaches, we note that the relation $a = \frac{ 1}{1 + \tau^2}$ results in the same split of information.
\end{example}

\begin{example}[Poisson Datasets] \label{example:poisson}
Let $\{X_{i}\}_{i=1}^{n}$ be i.i.d $\text{Pois}(\mu)$ .  We have that $\mathcal{I}_{1} = \frac{1}{\mu}$ and so $\frac{an}{\mu}$ is the amount of information used for selection under a data splitting rule. If the data is fissioned using the rule described in \cref{sec:list_decomp}, then $\mathcal{I}_{f(X)} = \frac{np}{\mu}$. The relation that equates the amount of information between data splitting and fission is then $a = p$.  
\end{example}

% This makes this methodology intractable for several of the decompositions listed in \cref{sec:list_decomp}.
%Luckily, for Gaussians with known variance, this caveat does not apply. 

 The above examples are simplified because each data point is identically distributed so that each data point contains equivalent amounts of information about the unknown parameter $\theta$. This symmetry allowed us to find a relation between $a$ and $\tau$ that did not depend on the unknown parameters of interest. When the information value of data points varies, the two methods are not as directly comparable 

%TThe above discussion also assumes the parameter of interest $\theta$ is fixed prior to fissioning the data and is therefore the same for both the selection and inference datasets. When the selection dataset is used to \emph{decide} on the parameter to conduct inference on, this calculation is more heuristic. 

\paragraph{Information splitting with non-identically distributed data} \label{sec:non_iid} Consider a setting where data is independent but not identically distributed so each data point has different amounts of information. For instance, consider $Y_{i} \overset{\mathrm{iid}}{\sim} N(\mu, \sigma_{i}^{2})$. Or consider the linear model where $Y_{i} \overset{\mathrm{iid}}{\sim}  N(x_{i}^{T} \beta, \sigma^{2})$ for $x_{i} \in \mathbb{R}^{p}$. In this case, $\mathcal{I}_{Y_{i}}(\beta) = \sigma^{2}x_{i} x_{i}^{T}$, which places greater value on points with higher leverage. More generically, assume that $Y:= (Y_{1},...,Y_{n})$ with $\mathcal{I}_{Y_{i}}(\theta)$ denoting the Fisher information for observation $Y_{i}$ about a parameter $\theta$.  Consider $m$ different ways of allocating the information for a fixed $a$ which chooses $an$ points for the first dataset and the remaining $1-an$ points for the second. Denote these hypothetical splits as $S_{1},...,S_{m}$, where $S_{j} \subseteq [n]$ refers to the indices that have been allocated to the first dataset. %Each split is deterministic, but the $j$th split $S_{j}$ is chosen with probability $p_{j}$. 
Let $S$ denote the random variable that randomly selects $S_{j}$ with probability $p_{j}$. For any $j$, we let $\mathcal{I}_{S_{j}}(\theta) = \sum_{i \in S_{j}} \mathcal{I}_{x_{j}}(\theta)$ denote the information allocated to the first dataset and $\mathcal{I}_{S_{j}^{c}}(\theta) = \sum_{i \in S_{j}^{c}} \mathcal{I}_{Y_{i}}(\theta)$ denote the information allocated to the second dataset.   

In what follows, we let $m = {n \choose an}$ and let $S_{1},...,S_{m}$ denote all possible ways of choosing $an$ data points for the first dataset. We will also make the simplifying assumption that each $p_{j} = \frac{1}{ {n \choose an}}$ so that all possible splits have the same probability of being selected. Following \cref{example:gaussian} and \cref{example:poisson}, suppose we try to create a comparable split of the information using data fission by choosing a fixed $a$ annd then picking $\tau$ such that $\mathcal{I}_{f(X)} = a \sum_{i=1}^{n} \mathcal{I}_{Y_{i}} (\theta)$. In this framework, this will also equate $\mathcal{I}_{f(X)} = \mathbb{E}\left[ \mathcal{I}_{S}(\theta) \right]$ because
$$\mathcal{I}_{f(X)} =  
a \sum_{i=1}^{n} \mathcal{I}_{Y_{i}}(\theta) 
= \frac{1}{ { n \choose an}} \sum_{i=1}^{ {n \choose an}} a\sum_{j=1}^{n} \mathcal{I}_{Y_{j}}(\theta)
=\frac{1}{ { n \choose an}} \sum_{i=1}^{ {n \choose an}} \sum_{j=1}^{n} \mathbbm{1}_{i \in S_{i}} \mathcal{I}_{Y_{j}}(\theta)
= \mathbb{E}\left[ \mathcal{I}_{S}(\theta) \right].$$
Proposition~1 of \cite{rasines2021splitting}, which we restate here with some adjustments to account for our setting, demonstrates that there is a sense in which the information split using fission is more efficient even though we have chosen $\tau$ such that it has the same expected information content as our data splitting procedure. 
\begin{fact}[Proposition 1 of \cite{rasines2021splitting}]
Let $S_{1},...,S_{m}$ be deterministic data splitting rules and $S$ be the random variable which returns a particular split $S_{j}$ with some probability $p_{j}$ such that $\sum_{i=1}^{m} p_{i} = 1$. Suppose that data fission is conducted such that $\mathcal{I}_{f(X)}(\theta) = \mathbb{E}\left[\mathcal{I}_{S} \right]$. Assume that all information matrices $\mathcal{I}_{f(Y)}(\theta)$, $\mathbb{E}[\I_{g(Y)|f(Y)}(\theta)]$, and $\mathcal{I}_{Y_{i}}(\theta)$ are invertible and have dimension $p \times p$. Let $\phi$ be a real-valued, convex, and strictly increasing function function defined on the set of $p \times p$ positive definite matrices. Then 
$\phi\left[\mathcal{I}_{f(Y)}(\theta)^{-1}\right] \le  \mathbb{E} \left[\phi\left( \mathcal{I}_{S}(\theta)^{-1} \right)\right]$. Furthermore, when $g(Y)$ is independent of $f(Y)$, $\phi\left[\mathcal{I}_{g(Y)}(\theta)^{-1}\right] \le  \mathbb{E} \left[\phi\left( \mathcal{I}_{S^{c}}(\theta)^{-1} \right)\right]$.
\end{fact}
%This result is an application of Jensen's inequality combined with the arithmetic mean inequality for Hermitian matrices \citep{bhagwat_subramanian_1978}. 
Intuitively, data fission is more efficient because it splits the information in a determinsitic way, but data splitting introduces randomization into the splitting process which decreases efficiency. Since the width of confidence intervals for MLE parameters is a function of the inverse Fisher information, an immediate consequence of this is that MLE parameters constructed from data fission will have, on average, tighter confidence intervals compared with a data splitting rule with the same expected information split. Note that in the case where the same amount of information is available at each data point (i.e. $\mathcal{I}_{x_{i}}(\theta) = \mathcal{I}_{x_{j}}(\theta)$ for all $i,j$), the above inequality becomes an equality 
and $\phi\left[\mathcal{I}_{f(X)}(\theta)^{-1}\right] =  \mathbb{E} \left[\phi\left( \mathcal{I}_{S}(\theta)^{-1} \right)\right]$. Therefore, in \cref{example:gaussian} and \cref{example:poisson}, there is no apparent advantage to data splitting compared to data fission. When the information at each point varies, however, data fission can offer substantial benefits, which we explore in greater detail in \cref{sec:linreg}.

\section{Application: selective CIs after interactive multiple testing} \label{sec:interactive_testing}

%The rest of the paper will focus on the use of data fission for several selective inference. We explore this idea more extensively in Sections~\ref{sec:linreg} and \ref{sec:trendfilter} but we begin with a simpler example to illustrate the concept. 

%\subsection{Review of some multiple testing methods and data fission proposal}
Suppose we observe $y_{i} \sim N(\mu_i,\sigma^{2})$ for $n$ data points with known $\sigma^{2}$ alongside generic covariates $x_{i} \in \mathcal{X}$. %After observing the data, the analyst has two goals. 
First, the analyst wishes to choose a subset of hypotheses  $\mathcal{R}$ to reject from the set $\{H_{0,i}: \mu_{i} = 0\}$ while controlling the false discovery rate (FDR), which is defined as the expected value of the $\text{false discovery proportion } \mathrm{(FDP)}:= \frac{|x_i \in \mathcal{R}: \mu_{i} = 0|}{\max \{|\mathcal{R}|,1\}}$. After selecting these hypotheses, the analyst then may wish to construct either: 
\begin{itemize}
    \item multiple confidence intervals (CIs) with $1-\alpha$ coverage of $\mu_{i}$ for each $i \in \mathcal{R}$; or 
    \item a single  CI with $1-\alpha$ coverage of $\widebar{\mu} = \frac{1}{|\mathcal{R}|} \sum_{i \in \mathcal{R}} \mu_{i}$.
\end{itemize}
One method for rejecting hypotheses and constructing CIs that achieve coverage for the individual $\mu_{i}$ would be using the BH procedure (\cite{BHprocedure}) to form the rejection set and then construct CIs as described by \cite{BY_CI}. %---which we refer to as BY-corrected CIs in the remainder of this paper. 
In this problem, these would be calculated as $y_{i} \pm z_{\beta/2}$ where $\beta = \frac{|\mathcal{R}|\alpha}{n}$.
% A weakness of this approach is that it affords the analyst no flexibility when forming the rejection set. The procedure forces the analyst to reject hypotheses according to a pre-defined rule (the BH procedure). Moreover, even though we have a CI with valid post-selective coverage for each individual $y_{i}$, 
However, we know of no methodology for aggregating the individual CIs to form a single CI that will cover $\widebar{\mu}$.

An alternative approach for selective inference is to compute a $p$-value for testing each $y_{i}$, and then mask the $p$-value as proposed by the AdaPT~(\cite{adapt_fithian}) and STAR~(\cite{lei2017star}) procedures. These interactive methods allow the data analyst to iteratively build a rejection set in a data adaptive way. 
% Masked versions of $p$-values are given to the analyst for data exploration. The analyst then chooses hypotheses to exclude from rejection one-by-one. After the analyst adds a hypothesis to the rejection set, the full $p$-value is revealed and can then be used to further fine-tune the selection rule for subsequent rejection decisions. 
A drawback of these approaches, however, is that they only are designed to work with a $p$-value. In contrast to the BH procedure, there is no available method to cover either the individual signals $\mu_{i}$ or $\widebar{\mu}$. %Moreover, to the best of our knowledge, there is no existing method to cover $\widebar{\mu}$, even heuristically. Information about the individual $y_{i}$ is discarded when running these analyses, preventing the analyst from constructing CIs from the data. Although constructing CIs by following equation~\ref{eqn:BY_CI_eqn} appears to control the false coverage rate empirically, there is no known theoretical guarantee to ensure that CIs constructed in this way have proper coverage. Moreover, to the best of our knowledge, there is no existing method to cover $\widebar{\mu}$, even heuristically. 
Data fission offers one possible path forward:
\begin{enumerate}
    \item Draw $z_{i} \sim N(0,\sigma^{2})$ and let $f(y_{i}) = y_{i} + \tau z_{i}$ with $g(y_{i}) = y_{i} -\frac{1}{\tau}z_{i}$.
    \item Use $f(y_{i})$ to select a rejection set of suspected non-nulls using any desired error-control procedure (say AdaPT, STAR or BH). 
    % If the analyst believes that the set of rejections should form a convex region, they may wish to use the STAR algorithm. AdaPT or BH procedures may be chosen if no such belief is held. 
    \item After selecting hypotheses, we can form CIs to cover each $\mu_{i} \in \mathcal{R}$ with $1-\alpha$ coverage as $
        g(y_{i}) \pm z_{\alpha/2} \sigma \sqrt{1+\frac{1}{\tau^{2}}}$
    or we can form a $1- \alpha$ CI to cover $\overline{\mu}$ as$
        \frac{\sum_{i\in \mathcal{R}} g(y_{i})}{|\mathcal{R}|} \pm z_{\alpha/2} \sigma \sqrt{\frac{1+\frac{1}{\tau^{2}}}{|\mathcal{R}|}}$.
\end{enumerate}
%The CIs constructed by data fission that cover individual $\mu_{i}$ are valid, their widths do not scale down with the size of the rejection set. If the analyst decides that they are interested in estimating the signal strength over the entire rejection region, however, data fission has good performance compared to existing methods. 

\begin{figure} 
\begin{center}
\begin{subfigure}[t]{1\textwidth}
\centering
\includegraphics[width=0.23\linewidth]{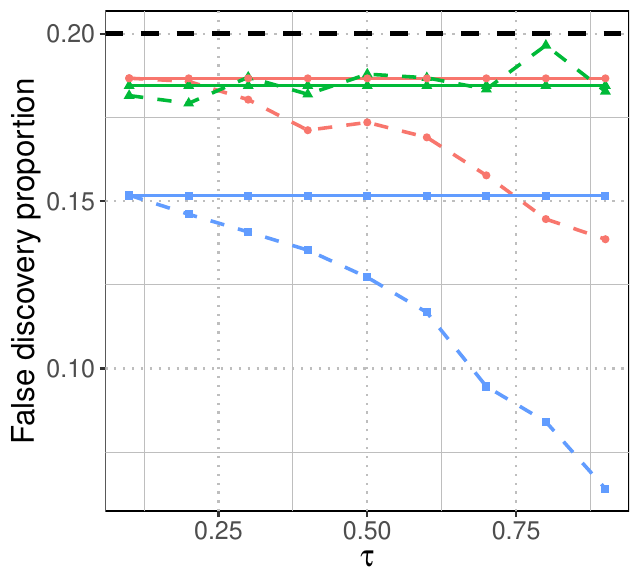}
\includegraphics[width=0.23\linewidth]{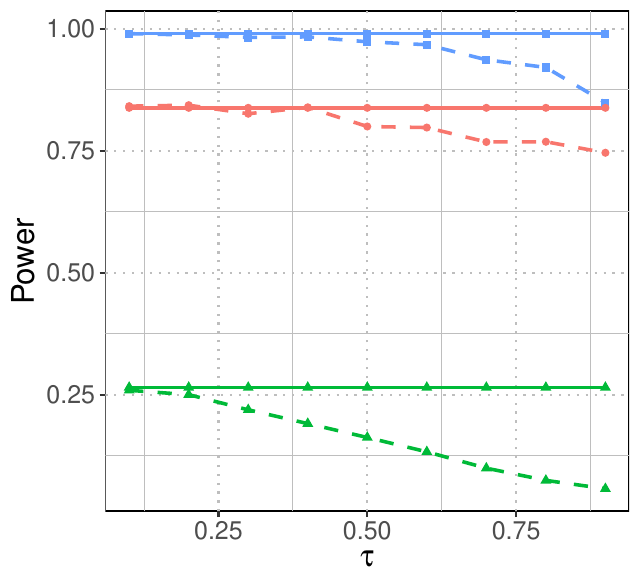}
\includegraphics[width=0.23\linewidth]{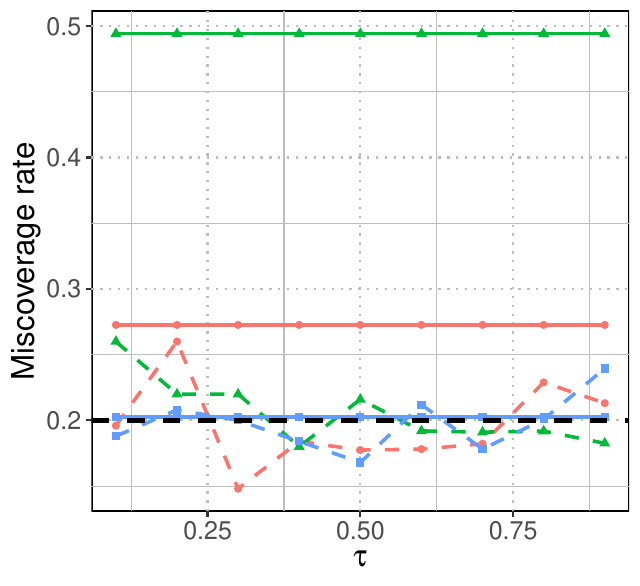}
\includegraphics[width=0.23\linewidth]{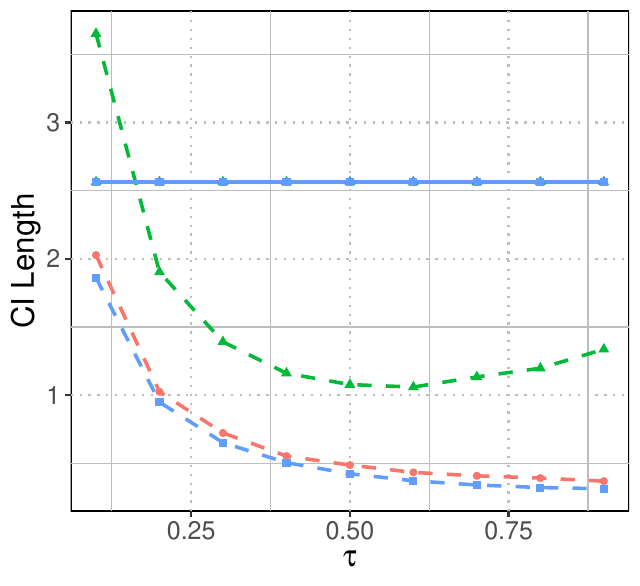}
\end{subfigure}
\vfill
\begin{subfigure}[t]{1\textwidth}
\centering
\includegraphics[width=0.4\linewidth]{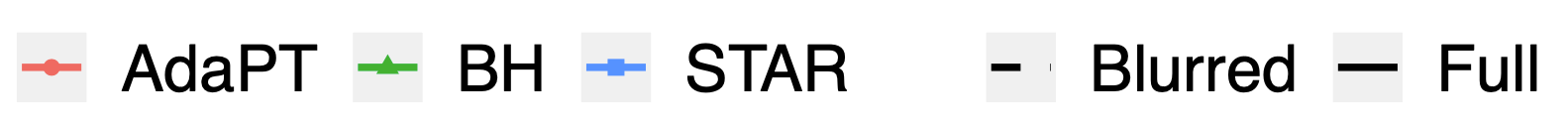}
\end{subfigure}
\caption{Numerical results averaged over $250$ trials for a $50 \times 50$ grid of hypotheses with target FDR level chosen at $0.2$ and $\tau$ varying over $(0,1)$. Solid lines denote metrics for the rejection sets formed using the full dataset and dotted lines denote metrics calculated using the rejection sets formed through data fission. All methods control FDR at the desired level, but ``double dipping'' to form CIs after forming a rejection set %using the full 
data results in invalid coverage. Fissioned CIs have the correct coverage. The fissiones CI lengths decrease as $\tau$ increases because more of the dataset gets reserved for inference.}
\label{fig:STAR_experiments_normal}
\end{center} 
\end{figure}

To illustrate this procedure, we repeat the experiments in Section~4.3 of \cite{lei2017star} but with fissioned data. Specifically, we let the data be arranged on a grid with non-nulls arranged in a circle in the center with $\mu_i = 2$ for each non-null and $\mu_i = 0$ for each null. We then form a rejection set using the BH, AdaPT, and STAR procedures and compute confidence intervals for $\widebar{\mu}$. Results are avaialble in \cref{fig:STAR_experiments_normal}. Fissioning the data allows the analyst to form CIs after forming a rejection set by varying $\tau$, while post-selective CIs for non-fissioned data do not have proper coverage. For a more detailed discussion of the simulation and for additional demonstrations on Poisson data, consult Appendix~\ref{sec:appendix_interactive_hyp_testing}.

\section{Application: selective CIs in fixed-design linear regression} \label{sec:linreg}
We now turn to applying data fission to fixed-design Gaussian linear regression. We expand on the discussion in \cite{rasines2021splitting}, and later build on several results in this section in our treatment of trend filtering in \cref{sec:trendfilter}. We assume that $y_i$ is the dependent variable and $x_i \in \mathbb{R}^p$ is a non-random vector of $p$ features for $i = 1,\ldots, n$ samples.  We denote $X = (x_{1},...,x_{n})^T$ as the model design matrix and  $Y = (y_{1},...,y_{n})^T$ with: 
\[Y = \mu + \epsilon \text{ with } \epsilon = (\epsilon_{1},...,\epsilon_{n})^T \sim N(0, \Sigma), \]
where $\mu = \mathbb{E}[Y|X] \in \R^{n}$ is a fixed unknown quantity and $\epsilon \in \R^{n}$ is a random quantity with a known covariance matrix $\Sigma$ (such as $\sigma^2 I_n$ for known $\sigma$; we discuss unknown $\sigma$ later). 

During the fission phase, we introduce the independent quantities $f(Y)$ and $g(Y)$ created by adding Gaussian noise $Z \sim N(0,\Sigma)$ as described in \cref{sec:list_decomp}, letting
$f(Y) = Y + \tau Z$, and $ g(Y) = Y - \frac{1}{\tau}Z$. 
We use $f(Y)$ to select a model $M \subseteq [p]$ that, in turn, defines a model design matrix $X_{M}$ which is a submatrix of $X$. After selecting $M$, we then use $g(Y)$ for inference by fitting a linear regression on $g(Y)$ against the selected covariates $X_{M}$. 
%However, in an assumption-lean setting where $\mu = \mathbb{E}[Y|X]$ is not guaranteed to be a linear combination of the chosen covariates, it is not clear what the fitted coefficients and corresponding CIs represent. For our purposes, we will use the same problem setup of \cite{buja2019modelsI}, but with one significant modification. The above paper focuses on the \emph{random-X} setting where $(x_{i},y_{i})$ are sampled as random pairs from a common joint density. In the context of data fission, it is only possible to make such an assumption during the fission stage. When conducting inference, $X$ has already been observed once during model selection so we necessarily must condition on the realized values of $X$. This restricts us to a \emph{fixed-X} setting during the inference stage. 

Let $\widehat{\beta}$ (as a function of the chosen model M) be defined in the usual way as
\begin{equation} \label{eqn:beta-hat}
\widehat{\beta}(M) = \argmin_{\widetilde{\beta}} \norm{g(Y)-X_{M}\widetilde{\beta}} ^{2} = (X_{M}^T X_{M})^{-1}X_{M}^{T}g(Y).
\end{equation}
Note that we make no assumptions that $\mu = \mathbb{E}[Y|X]$ is guaranteed to be a linear combination of the chosen covariates. Our target parameter is therefore the best linear approximator of the regression function using the selected model
\begin{equation} \label{eqn:beta-star}
\beta^{*}(M) =\argmin_{\widetilde{\beta}} \mathbb{E}\left[\norm{Y-X_{M}\widetilde{\beta}}^{2} \right]  = (X_{M}^{T}X_{M})^{-1}X_{M}^{T}\mu.
\end{equation}
We are then able to form CIs that guarantee $1- \alpha$ coverage of $\beta^{*}(M)$ as follows. 
\begin{theorem} \label{thm:normal_regression}
For $\widehat{\beta}(M)$ from~\eqref{eqn:beta-hat} and $\beta^{*}(M)$ from~\eqref{eqn:beta-star}, we have
$$\widehat{\beta}(M) \sim N\left(\beta^{*}(M),  (1+\tau^{-2})(X_{M}^{T}X_{M})^{-1}X_{M}^{T} \Sigma X_{M} (X_{M}^{T}X_{M})^{-1} \right).$$
Furthermore, we can form a $1-\alpha$ CI for the $k$th element of $\beta^{*}(M)$ as
$$\widehat{\beta}^{k}(M) \pm z_{\alpha/{2}} \sqrt{(1+\tau^{-2}) \left[(X_{M}^{T}X_{M})^{-1}X_{M}^{T} \Sigma X_{M} (X_{M}^{T}X_{M})^{-1}\right]_{kk}}.$$
\end{theorem}
An assumption of this procedure is that the variance is known in order to do the initial split between $f(Y)$ and $g(Y)$ during the fission phase. In the case of unknown variance, one can use an estimator $\hat{\sigma}$ to create the split, but this forces the analyst to then condition on both
$f(Y)$ and $\hat{\sigma}$, which may not have a tractable distribution in high dimensional settings. In the case where $p$ is fixed and $n \rightarrow \infty$, we explore an extension of this methodology to account for unknown variance in \cref{sec:unknown_variance}. Extending this approach to finite samples and high-dimensional regimes remains an open question for future lines of research.

\begin{figure}[H] 
\centering
\includegraphics[width=0.27\linewidth]{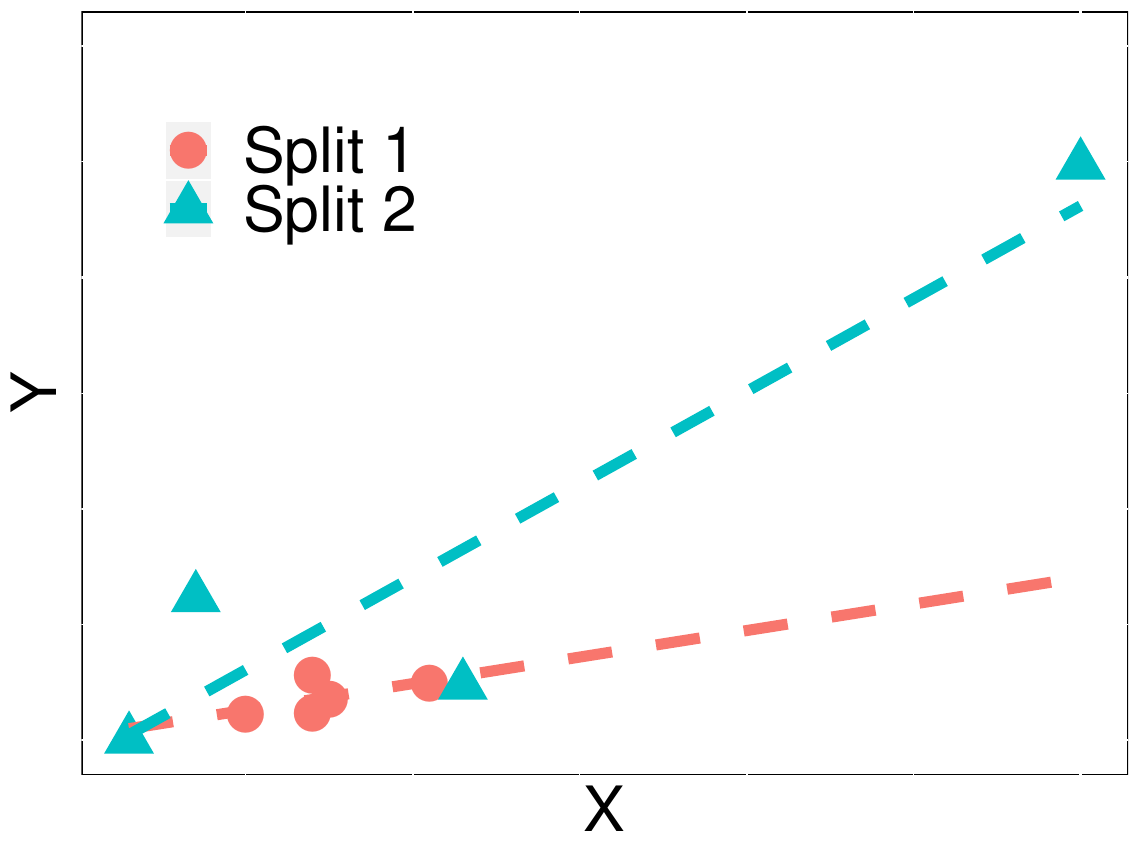}
\hfill
\includegraphics[width=0.27\linewidth]{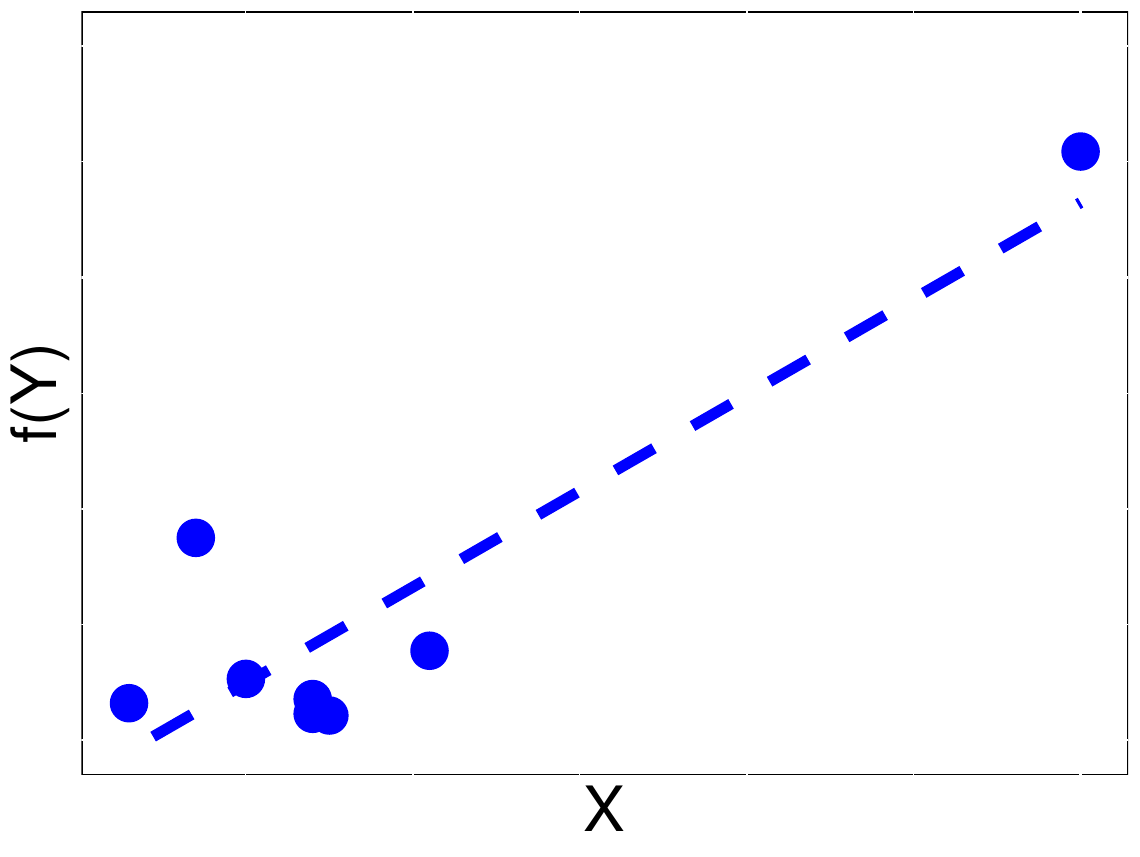}
\hfill
\includegraphics[width=0.27\linewidth]{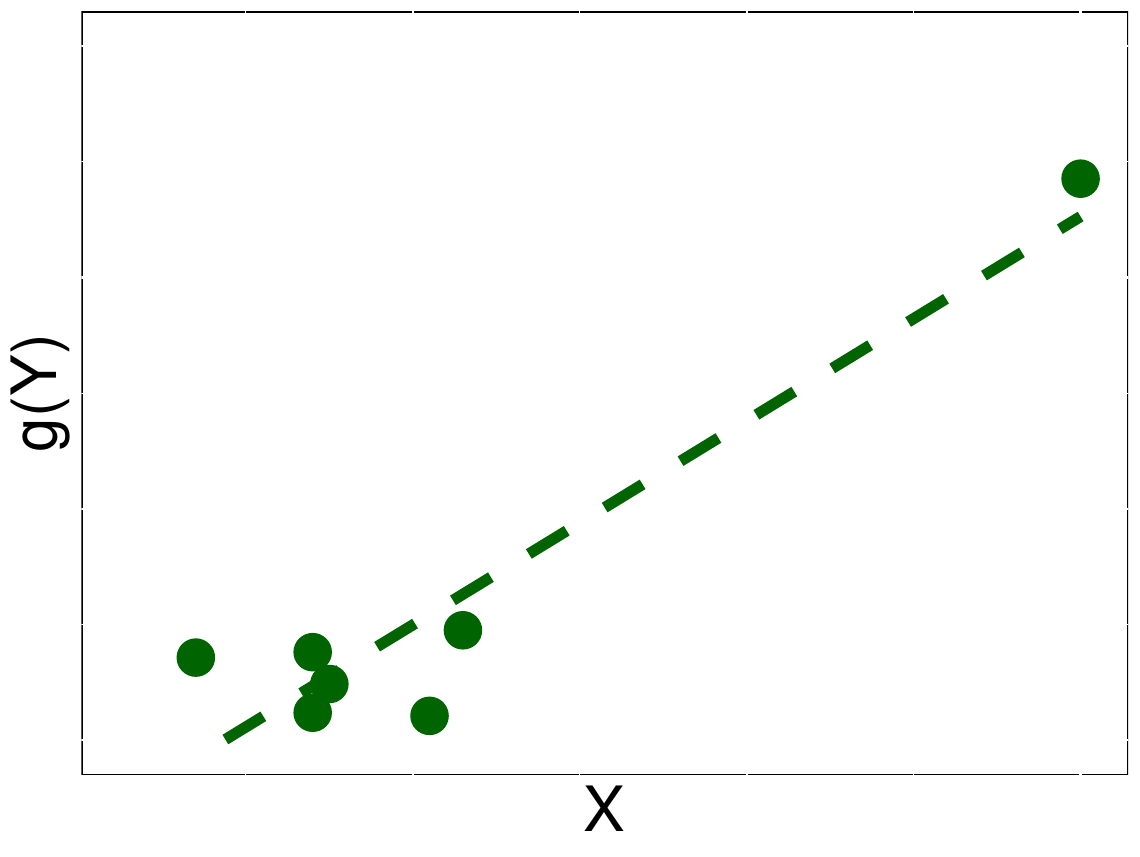}
\caption{Comparison of data splitting (left) and data fission (middle, right) for dataset with one highly influential point. Splitting the data and fitting a regression results in substantially different fitted models because the fitted values are heavily influenced by a single data point. In contrast, data fission keeps the same $X$ location for every data point, but randomly perturbs the response $Y$ with random noise to create new variables $f(Y)$ and $g(Y)$: notice the slight difference in the two figures. This enables the analyst to keep a ``piece'' of every data point in both $f(Y)$ and $g(Y)$, ensuring that leverage points have an impact in both copies of the dataset. }
\label{fig:example_highlev}
\end{figure}

A key advantage that data fission has over data splitting is that it allows the analyst to smoothly trade off information between selection and inference datasets by tuning $\tau$, while data splitting forces the analyst to allocate points discretely to either the selection or inference datasets. When the sample size is large and the distribution of covariates is well behaved, data splitting may be able to tradeoff information relatively smoothly by changing the proportion of data points in each sample. However, data fission will outperform data splitting in settings with small sample size  with a handful of points with high leverage. Data splitting has a disadvantage in this setting because the analyst is forced to choose to allocate each of these leverage points to either the selection or inference dataset. In contrast, data fission enables the analyst to ``hedge their bets" so that a piece of the information contained in each leverage point is allocated to both the selection and inference datasets. \cref{fig:example_highlev} offers an illustration of this tendency on an example dataset and an extended discussion on this topic can be found in Appendix~\ref{sec:appendix_splitting_fission}.

%\subsection{Empirical results} \label{sec:linear_empirical}
We now demonstrate the advantages that data fission has over data splitting through an empirical study. We conduct inference on some vector $Y$ given a set of covariates $X$ and a known covariance matrix $\Sigma = \sigma^{2} I_{n}$  as follows:
\begin{enumerate}
    \item Decompose $y_i$ into $f(y_i) = y_i - Z_i$ and $g(y_i) = y_i + Z_i$ where $Z_i \sim N(0, \sigma^{2})$. 
    \item Fit $f(y_i)$ using LASSO to select features, denoted as $M \subseteq [p]$. \footnote{In our experiments, we ue \texttt{cv.glmnet} in \texttt{R} package \texttt{glmnet} and choosing the tuning parameter $\lambda$ by the 1 standard deviation rule, which can be found in the value of \texttt{lambda.1se})}
    \item Fit $g(y_i)$ by linear regression without regularization using only the selected features %(for example, using \texttt{lm} in \texttt{R} package \texttt{stats}).
    \item Construct CIs for the coefficients trained in step 3, each at level $\alpha$, using \cref{thm:normal_regression}.
\end{enumerate}
%Note that with the above decomposition, $f(y_i)$ and $g(y_i)$ are independent Gaussian with the same mean $\mu_i$ where $\mu_i = \mathbb{E}(y_i \mid  X_i)$.

\paragraph{Simulation setup.} We choose $\sigma^{2} = 1$ and generate $n=16$ data points with $p=20$ covariates. For the first $15$ data points, we have an associated vector of covariates $x_i \in \mathbb{R}^{p}$ generated from independent Gaussians. The last data point, which we denote $x_{\text{lev}}$, is generated in such a way as to ensure it is likely to be more influential than the remaining observations due to having much larger leverage. We define $x_{\text{lev}} = \gamma \left(|X_{1}|_{\infty}, ..., |X_{p}|_{\infty}\right)$ where $X_{k}$ denotes the the $k$-th column vector of the model design matrix $X$ formed from the first $15$ data points and $\gamma$ is a parameter that we will vary within these simulations that reflects the degree to which the last data point has higher leverage than the first set of data points. We then construct $y_i \sim N(\beta^T x_i,\sigma^{2})$. The parameter $\beta$ is nonzero for 4 features: $(\beta_{1}, \beta_{16}, \beta_{17},\beta_{18}) = S_{\Delta}(1,1,-1,1)$ where $S_\Delta$ encodes the signal strength. We use $500$ repetitions and summarize performance as follows. For the selection stage, we compute the power (defined as $\frac{|j \in M: \beta_{j} \ne 0|}{|j \in [p]: \beta_j \ne 0|}$) and precision (defined as $\frac{|j \in M : \beta_j \ne 0|}{|M|}$) of selecting features with a nonzero parameter. For inference, we use the false coverage rate (defined as $\frac{|k \in M: [\beta^{\star}(M)]_k \notin \text{CI}_k|}{\max\{|M|,1\}}$) where $\text{CI}_k$ is the CI for $[\beta^{\star}(M)]_k$. We also track the average CI length within the selected model.

%amongst $\text{CI}_k$ within the selected model . and $\text{CI}_k(1), \text{CI}_k(2)$ are the lower and upper bound of $\text{CI}_k$. 

\iffalse
\begin{align}
    \text{power}^\mathrm{selected} &:= \frac{|j \in M: \beta_{j} \ne 0|}{|j \in [p]: \beta_j \ne 0|}, \text{ and }
    \text{precision}^\mathrm{selected} := \frac{|j \in M : \beta_j \ne 0|}{|M|}.
\end{align}
\begin{align}
    \mathrm{FCR} := \frac{|k \in M: [\beta^{\star}(M)]_k \notin \text{CI}_k|}{\max\{|M|,1\}} , \text{Avg. CI Length} := \frac{1}{|M|}\sum_{k \in M} |\text{CI}_k(2) - \text{CI}_k(1)|,
\end{align}
\fi
\iffalse
Several other metrics we explore include the averaged proportion of falsely reported CIs among those indicating a non-zero parameter (the false sign rate):
\begin{align}
    \mathrm{FSR} := \frac{|k \in M:\{ \beta_k < 0 \text{ and } \text{CI}_k(1) > 0 \} \text{ OR } \{ \beta_k > 0 \text{ and } \text{CI}_k(2) < 0 \}|}{\max\{|k \in M: 0 \notin \text{CI}_k|,1\}},
\end{align}
and the power of correctly reporting CIs that indicates a nonzero parameter with the correct sign:
\begin{align}
    \text{power}^\mathrm{sign} := \frac{|k \in M: \{ \beta_k > 0 \text{ and } \text{CI}_k(1) > 0 \} \text{ OR } \{  \beta_k < 0 \text{ and } \text{CI}_k(2) < 0\} |}{\max\{|j \in [p]: \beta_j \neq 0|,1\}}.
\end{align}

\fi

\begin{figure}[H]
\centering
    \begin{subfigure}[t]{0.32\textwidth}
        \centering
        \includegraphics[trim=5 20 25 60, clip, width=1\linewidth]{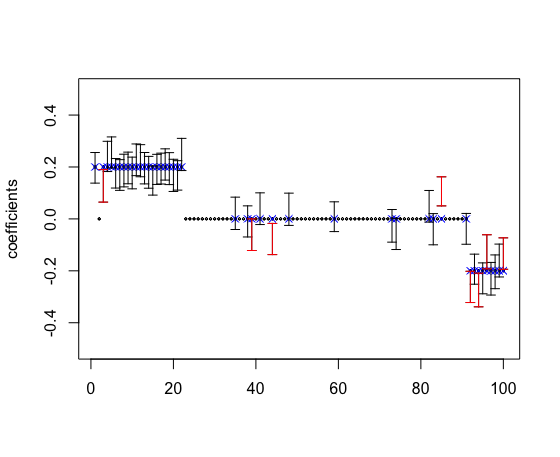}
    \end{subfigure}
    \hfill
    \begin{subfigure}[t]{0.32\textwidth}
        \centering
        \includegraphics[trim=5 20 25 60, clip, width=1\linewidth]{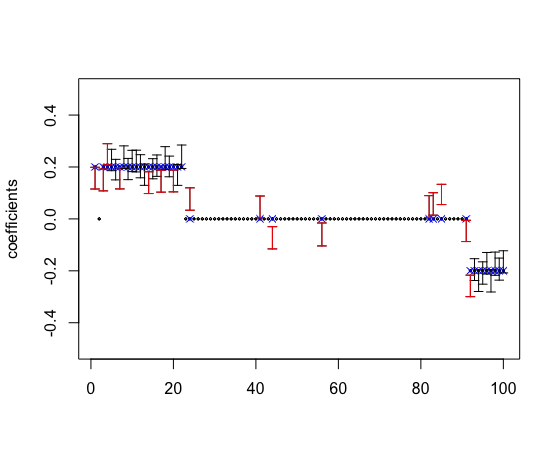}
    \end{subfigure}
    \hfill
    \begin{subfigure}[t]{0.32\textwidth}
        \centering
        \includegraphics[trim=5 20 25 60, clip, width=1\linewidth]{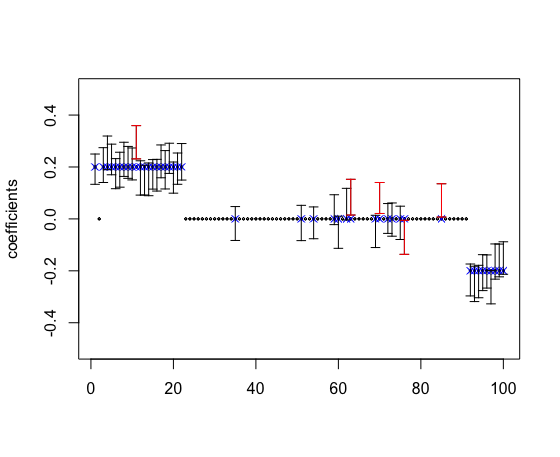}
    \end{subfigure}
    \caption{An instance of the selected feature (blue crosses) and the constructed CIs using fissioned data (left), full data twice (middle), and split data (right) with $S_\Delta = 0.2$ and target FCR set at $0.2$. The selected features are marked by blue crosses, which include all of the nonzero coefficients (corresponding to almost 100\%  power for selection) and also a few zero coefficients (corresponding to around 70\% precision for selection). CIs which do not cover the parameters correctly are marked red.}
    \label{fig:linear_example}
\end{figure}

As an illustration, \cref{fig:linear_example} shows an instance of the selected features and corresponding CIs for an example trial run. As a point of comparison, we compare the CIs constructed using data fission with those constructed using data splitting (when $50\%$ of the dataset is used for selection and the remaining for inference). We also compare these results to the (invalid) procedure where the original dependent variable is used twice to both select features and construct intervals %(replacing $f(y_i)$ and $g(y_i)$ both by $y_i$ in the above algorithm).
The third methodology will not have coverage guarantees but it is still a useful point of comparison for evaluating the performance of the other two (valid) methodologies. \cref{fig:sims_highlev} shows results averaged over $500$ trials. Data splitting and data fission both control the FCR, but data fission dominates data splitting across every other metric---including significantly tighter CIs and higher power and precision. In simulation studies with larger sample sizes and less skewed covariates, data fission and data splitting have comparable performance---see Appendix~\ref{sec:appendix_supplemental_linear} for details. 

\begin{figure}[H] 
\centering
    \begin{subfigure}[t]{1\textwidth}
        \centering
        \includegraphics[width=0.23\linewidth]{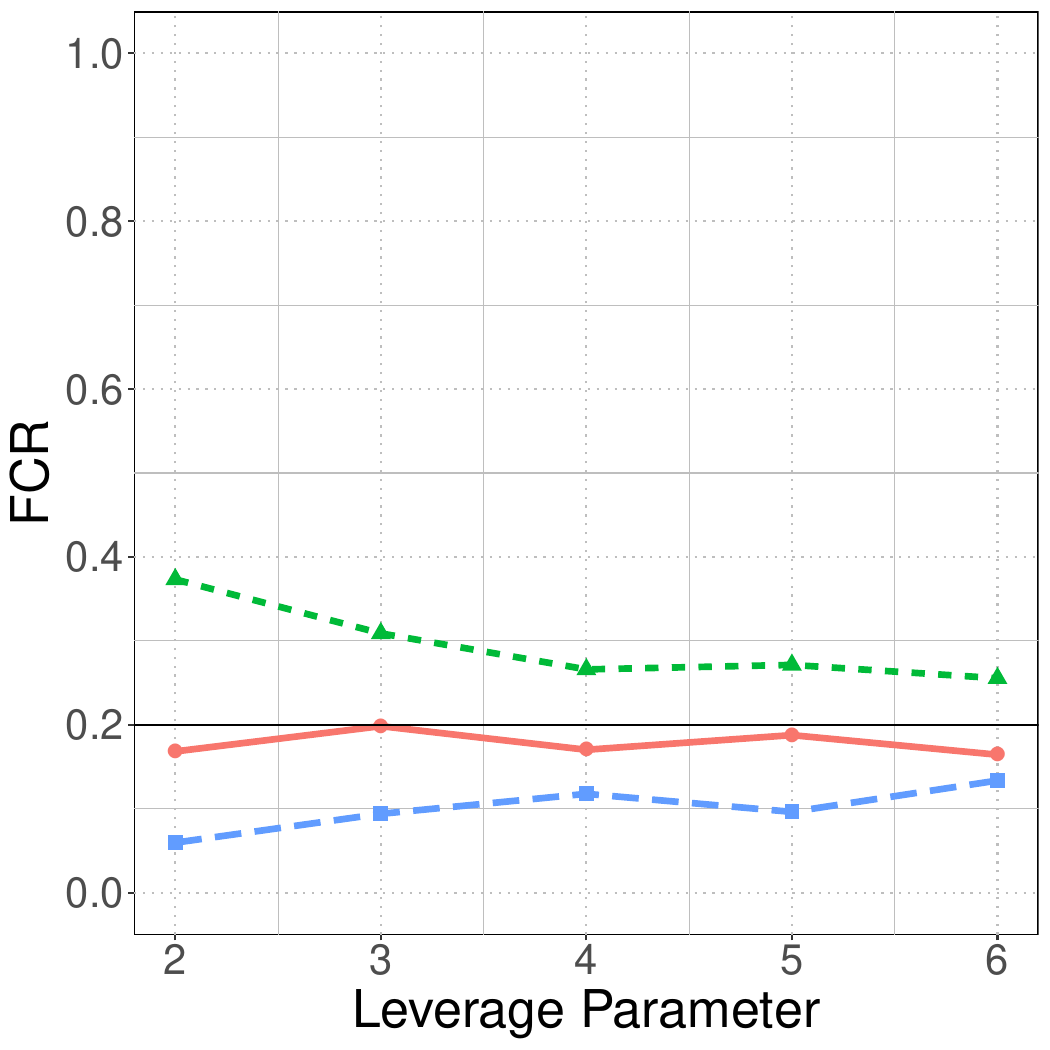}
    \hfill
        \includegraphics[width=0.23\linewidth]{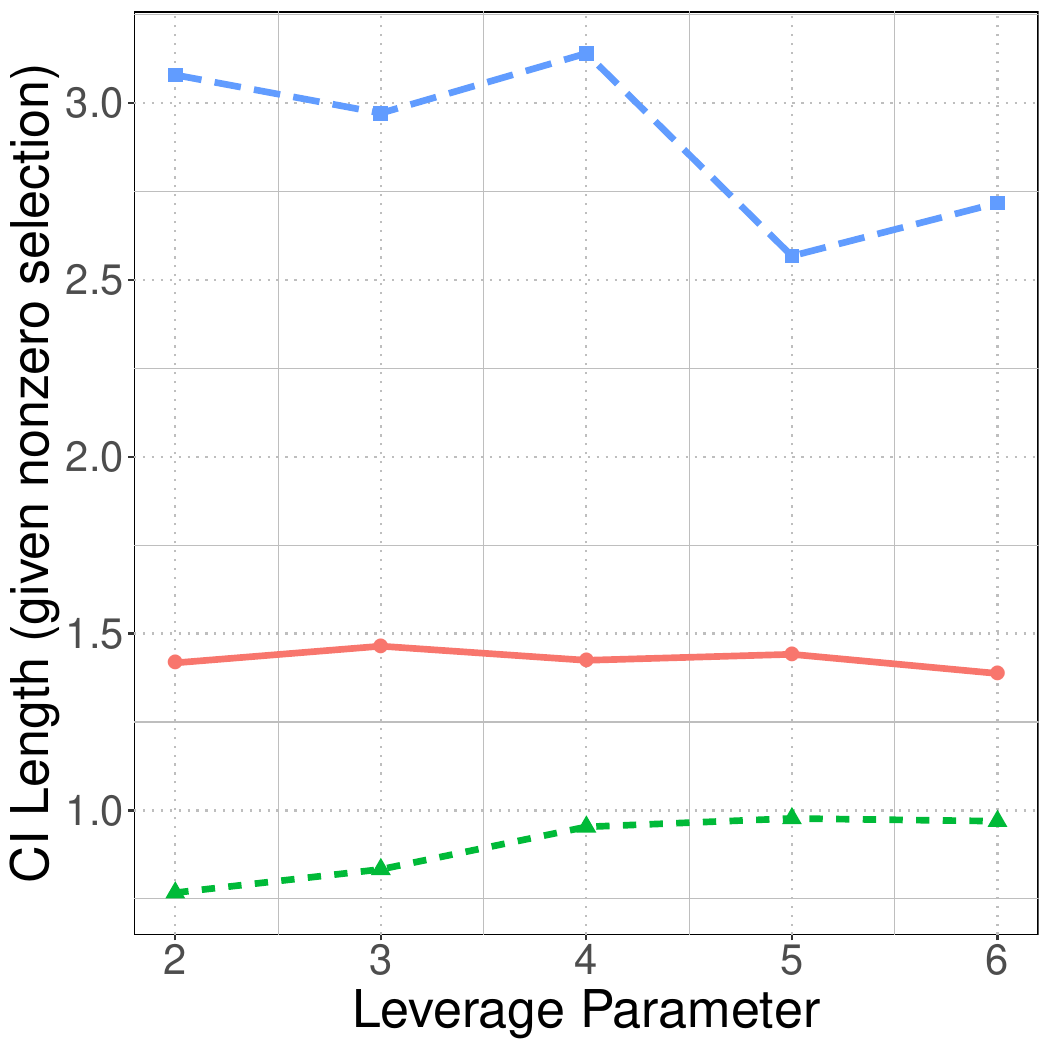}
    \hfill
        \includegraphics[width=0.23\linewidth]{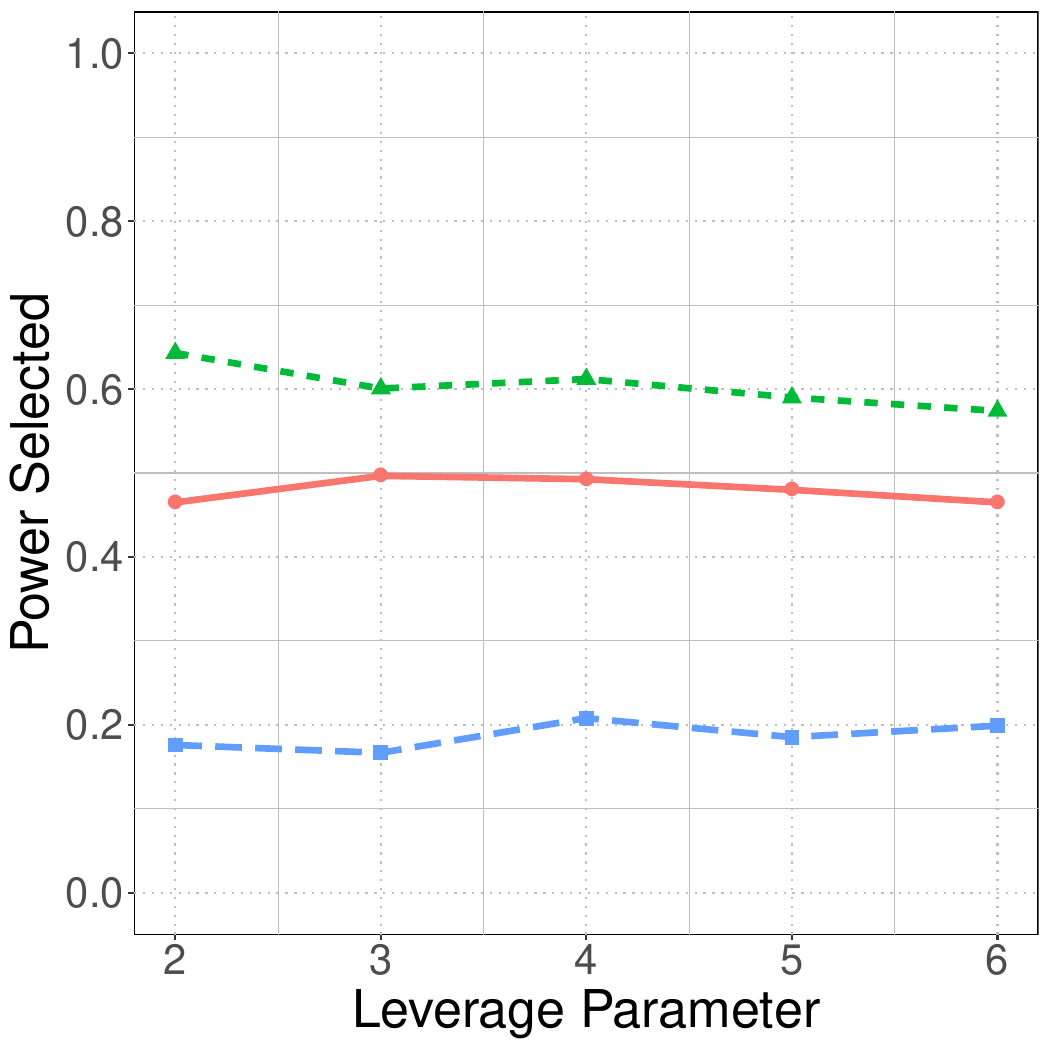}
    \hfill
        \includegraphics[width=0.23\linewidth]{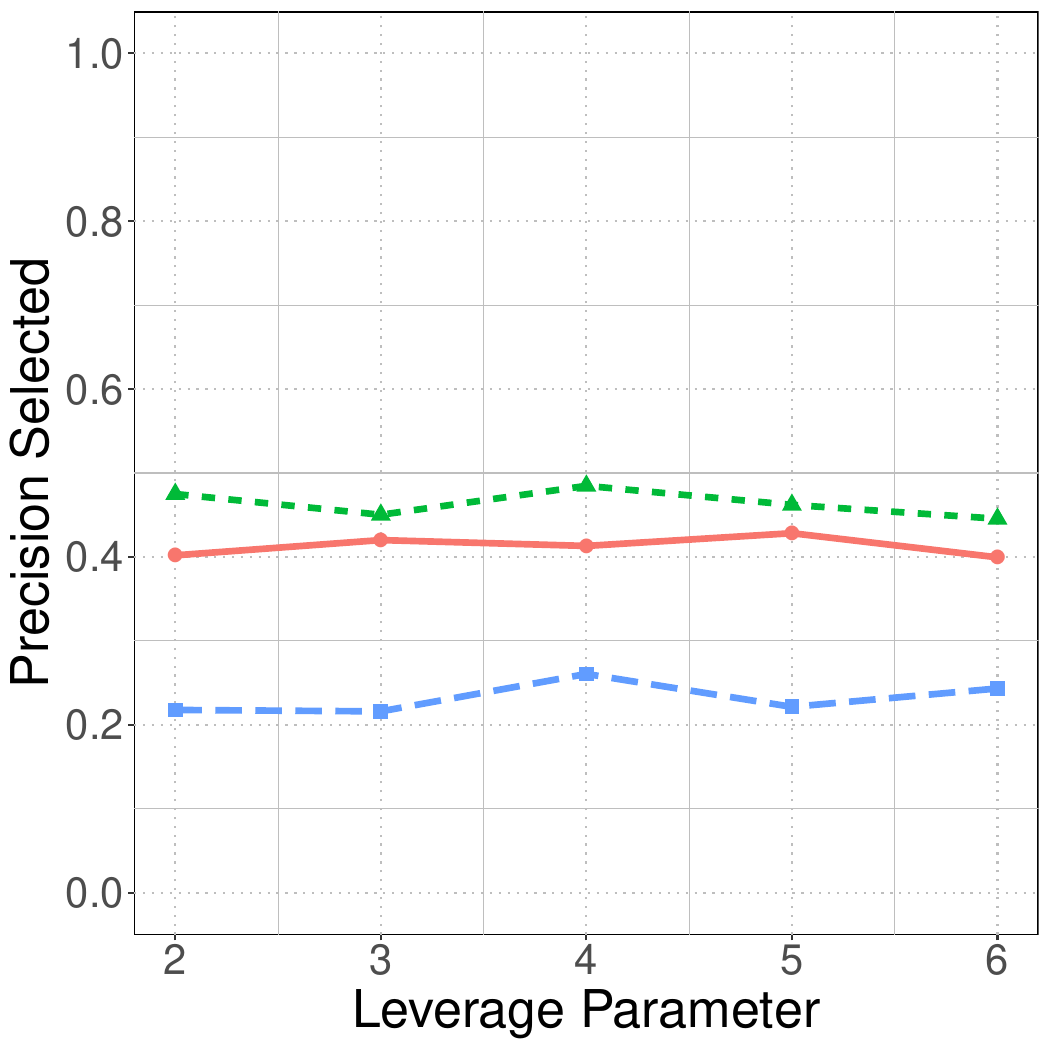}
    \end{subfigure}
\vfill
    \begin{subfigure}[t]{1\textwidth}
        \centering
        \includegraphics[width=0.3\linewidth]{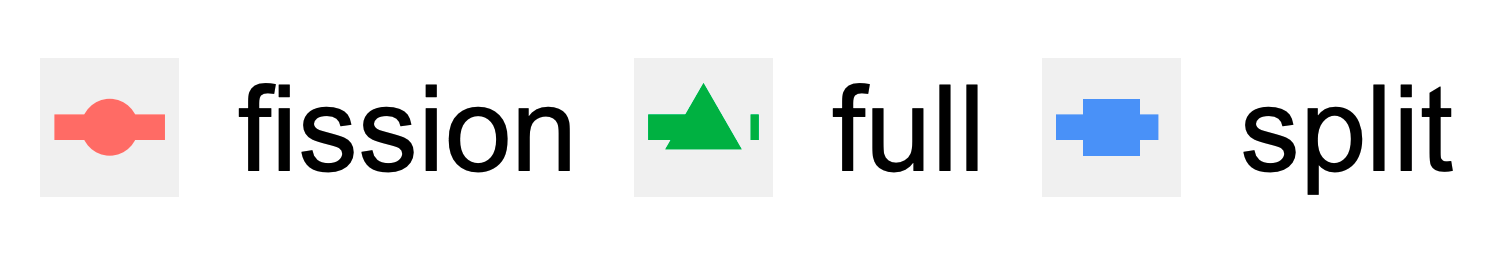}
    \end{subfigure}
    \caption{FCR, average length of the CIs, and power/precision for the selected features, when varying the leverage parameter $\gamma$ in $\{2, 3, 4, 5, 6\}$. The results are averaged over 500 trials. Both data splitting and data fission still control FCR, but data fission now has higher power and precision, as well as tighter CIs than data splitting. }
    \label{fig:sims_highlev}
\end{figure}

\section{Application to selective CIs in fixed-design generalized linear models and other quasi-MLE problems} \label{sec:qmle}
For regression problems involving non-Gaussian distributed data, the same generic principle can be used as in \cref{sec:linreg}. If the underlying distribution of $Y$ is known up to a set of unknown parameters and aligns with one of the distributions described in \cref{sec:appendix_list_decomp}, we can fission the data and then use $f(Y)$ for selection, while reserving $g(Y) | f(Y)$ to conduct inference. Even in cases where we are certain that the distribution of $Y$ has been properly specified and the fission procedure is valid, it is important to construct CIs that are robust to model misspecification for two reasons. First, because it is unlikely in practical settings that we are going to be able to select a set of covariates that corresponds to the ``actual'' conditional mean of the data. Second, the post-selective distribution $f(Y) | g(Y)$ may be ill-behaved and difficult to work with even if $Y$ is easy to model. For example, if $Y \sim \mathrm{Ber}(\theta)$, the post-selective distribution $g(Y) | f(Y)$ constructed from \cref{sec:linreg} is $\mathrm{Ber}\left(\frac{\theta}{\theta + (1-\theta) [p/(1-p)]^{2f(Y) - 1}}\right)$. Maximizing this likelihood function will be challenging since the likelihood is non-convex, so it may be practical to model the likelihood as a probit or logistic model as a working assumption. 

This approach of using maximum likelihood estimation to train a model but to construct guarantees in an assumption-lean setting is termed quasi-maximum likelihood estimation (QMLE). Most work in this setting, such as \cite{buja2019modelsII}, and \cite{samplesplit_bootstrap}, are designed around a background assumption where both the covariates and response are treated as i.i.d.\ random variables. Such theory is inapplicable to data fission since %regardless of whether the covariates are fixed or random at the start of the analysis, 
the covariates are observed during the model selection stage and therefore inference on $g(Y)$ is only valid if both $X$ and $f(Y)$ are conditioned on. \cite{fahrmeir_mle}, however, works in the setting of non-identically distributed but independent data with fixed covariates and therefore can be applied to this setting.  We first recap the relevant theory and then apply it to data fission when used for forming post-selective CIs in GLMs. 

\subsection{Problem setup and a recap of QMLE methods}
Suppose we observe $n$ independent random observations $y_{i} \in \R$ %that are independent but not identically distributed 
alongside fixed covariates $x_{i} \in \R^{p}$ and that we fission each data point such that the below assumption holds.
\begin{assumption}\label{as:1}
Data fission is conducted such that $g(y_{i}) \independent g(y_{k}) | f(Y),X$ for all $i \ne k$.\end{assumption}

 {\color{black} This is a different condition compared to the rules described in \cref{sec:introduction}. Any fissioning rule which satisfies \textbf{(P1)} would also satisfy this (weaker) condition. The majority of decompositions in \cref{sec:appendix_list_decomp} following $\textbf{(P2)}$ also satisfy \cref{as:1} but this is not necessarily the case --- for instance, the third rule for fissioning Gaussian distributed data would result in a post-selective distribution that is tractable but with correlated entries. However, if a fissioning process exists that creates independent entries in $g(Y) | f(Y), X$, the procedures outlined in this section will apply.} 

During model selection, a model $M \subseteq [p]$ is chosen from $f(Y)$ which result in an annealed set of covariates $\widetilde{x}_{i}\in \R^{|M|}$ and model design matrix $X_{M}$. We then find an estimate for $\beta$ using a working model of the density for $g(y_{i}) | f(Y), X$ as $p(g(y_{i}) | \beta, f(Y),X_{M})$, which defines the quasi-likelihood function $L_{n}(\beta) := \sum_{i=1}^{n} \log p(g(y_{i}) | \beta, f(Y), X_{M})$. 
Denote the score function $s_{n}(\beta) = \frac{\partial L_{n} }{ \partial \beta}$ as well as the quantities $H_{n}(\beta) = -\frac{\partial^{2} L_{n} }{ \partial \beta \beta^{T}}$, and $V_{n}(\beta) = \text{Var} \left( s_{n}(\beta) \right)$. If $\mathbb{E}(H_{n}(\beta))$ is positive definite, the target parameter $\beta^{\star}(M)$ is the root of $E(s_{n}(\beta))$ which is also the parameter which minimizes the KL distance between the true distribution (denoted as $q(g(y_{i} | f(Y),X)$) and the working model because
\begin{align*}
    \beta^{\star}_{n}(M) &= \argmin_\beta \mathbb{E}\left( \log q(g(y_{i})|X,f(Y)) - L_{n}(\beta) \right) \\
    &=\argmin_\beta \mathrm{D}_{KL}\left(\prod_{i=1}^n q(g(y_{i})|X,f(Y))|| \prod_{i=1}^n p(g(y_{i}) | \beta, f(Y),X_{M}) \right).
\end{align*}
We define $\hat{\beta}_{n}(M)$ as the  quasi-likelihood maximizer. Under mild regularity conditions, $\hat{\beta}_{n}(M)$ behaves asymptotically like $N\left(\beta^{\star}(M),  H_{n}^{-1} (\beta^{\star}_{n}(M))  V_{n}(\beta^{\star}_{n}(M)) H_{n}^{-1} (\beta^{\star}_{n}(M)) \right)$. Plug-in estimates  $\hat{H}_{n} := H_{n} (\hat{\beta}_{n}(M))$ and $\hat{V}_{n} := s_{n}(\hat{\beta}_{n}(M))s_{n}(\hat{\beta}_{n}(M))^{T}$ will yield asymptotically conservative estimates for the covariance matrix, allowing the user to form confidence intervals. For a detailed technical account, see Appendix~\ref{sec:appendix_QMLE_details}.

\subsection{Simulation results}
We verify the efficacy of this procedure through an empirical simulation. The advantages of data fission are again most apparent in settings with relatively few data points and a handful of leverage points just as we saw in the Gaussian case in \cref{sec:linreg}. 
\paragraph{Simulation setup.} Let $y_i$ be the dependent variable with $y_{i} \sim \text{Pois}\left(\exp\{ \beta^T x_i\}\right)$ where $x_i \in \mathbb{R}^p$ is a vector of $p$ features. We generate $n=16$ data points with $p=20$ covariates, where the parameter $ \beta$ is nonzero for 4 features: $(\beta_1, \beta_{16}, \ldots, \beta_{19}) = S_\Delta (1, 1,-1,1)$ and $S_\Delta$ encodes the signal strength. For the first $15$ observations, $x_i \in \{0,1\}^2 \times \mathbb{R}^{18}$, where the first two follow $\mathrm{Ber}(1/2)$ and the rest follow independent standard Gaussians. The last data point, which we again denote $x_{\text{lev}} = \gamma \left(|X_{1}|_{\infty}, ..., |X_{p}|_{\infty}\right)$ where $X_{k}$ denotes the the $k$-th column vector of the model design matrix $X$ formed from the first $15$ data points and $\gamma$ is the leverage parameter. Our proposed procedure is:

\begin{figure}[H]
\centering
    \begin{subfigure}[t]{1\textwidth}
        \centering
        \includegraphics[width=0.22\linewidth]{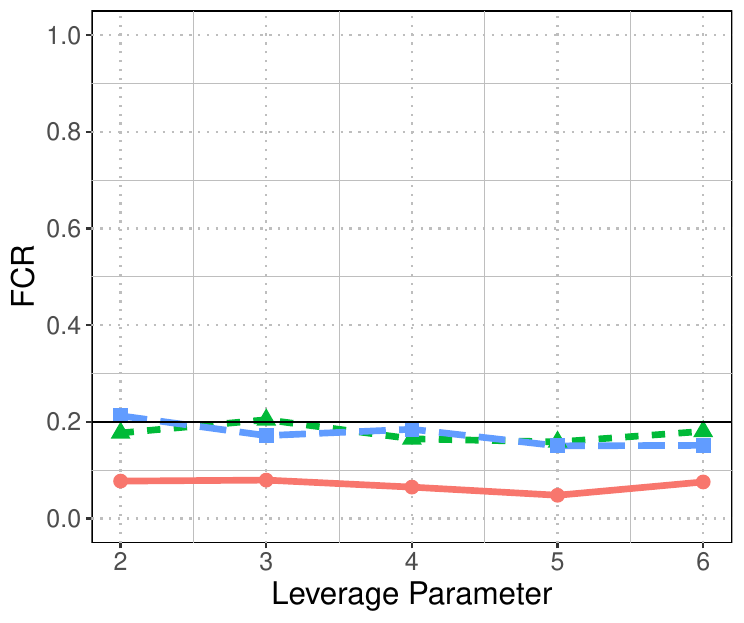}
    \hfill
        \includegraphics[width=0.22\linewidth]{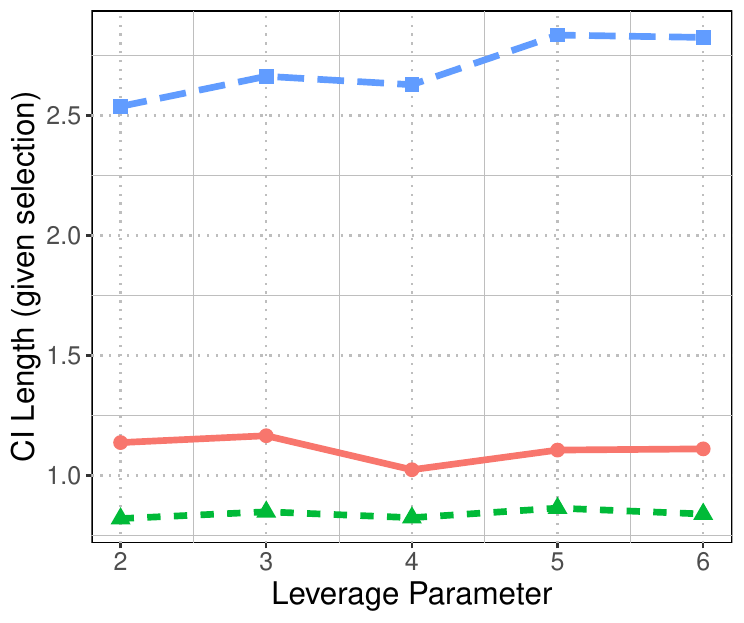}
    \hfill
        \includegraphics[width=0.23\linewidth]{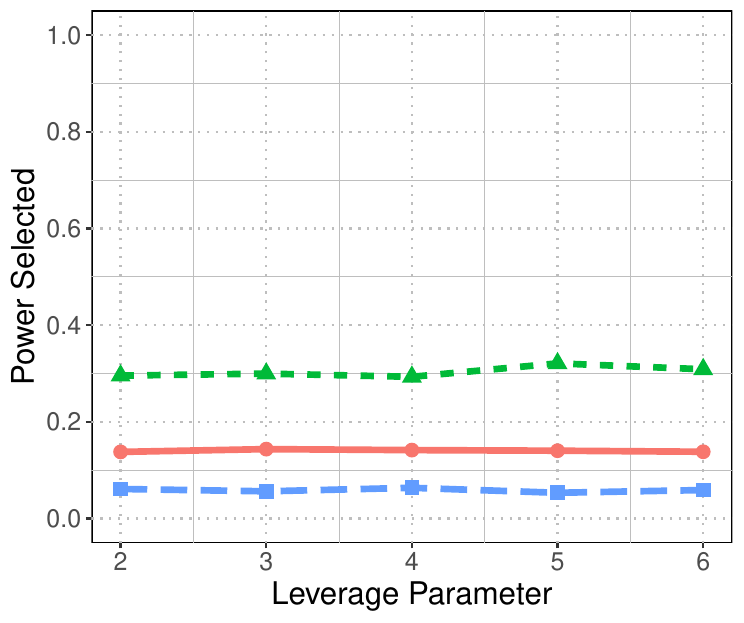}
    \hfill
        \includegraphics[width=0.23\linewidth]{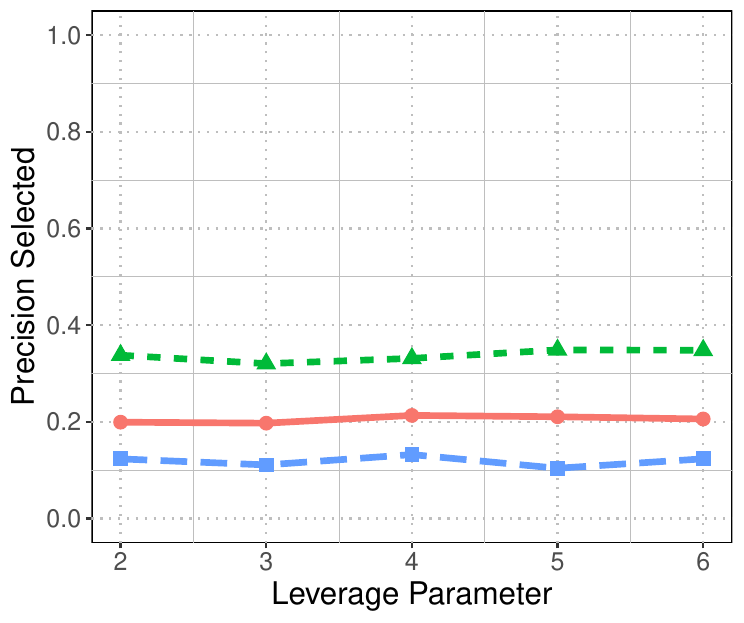}
    \end{subfigure}
\vfill
    \begin{subfigure}[t]{1\textwidth}
        \centering
        \includegraphics[width=0.3\linewidth]{figures/legend_fission_regression.png}
    \end{subfigure}
    \caption{FCR, length of the CIs, FSR, power for the sign of parameters, and power and precision for the selected features, when varying the leverage parameter $\alpha$ in $\{2, 3, 4, 5, 6\}$ for Poisson data over $500$ trials. CIs constructed using data fission are tighter than data splitting, and power during the selection stage is higher..} \label{fig:results_poisson_leverage}
\end{figure}

\begin{enumerate}
    \item Decompose each $y_i$ as $f(y_i) \sim \mathrm{Bin}(y_i, p)$, and $g(y_i) = y_i - f(y_i)$.
    \item Fit $f(y_i)$ by the GLM with log-link and LASSO regularization to select $M \in [p]$. %\footnote{In our examples, we use \texttt{cv.glmnet} in the \texttt{R} package \texttt{glmnet} and choose the tuning parameter $\lambda$ by the 1 standard deviation rule, which can be found in the value of \texttt{lambda.1se}}
    \item Fit $g(y_i)$ by another GLM without regularization with log-link using \textit{only} the selected features %(we use \texttt{glm} in the  \texttt{R} package \texttt{stats}) 
    and an offset of $\log(1-p)$ to account for the effect of the randomization.
    \item Construct CIs for the coefficients trained in the third step, each at level $\alpha$ and with the standard errors estimated as in \cref{thm:var_emp_fahrmeir}. For the experiments shown below, we also apply a finite sample correction as described in \cite{CR2_explanation}.
\end{enumerate}
Since step 2 may not select the correct model, the CIs in step 4 cover $\beta_{n}^{\star}(M)$. %which minimize the KL divergence between the chosen model and the true distribution.

Results are shown in \cref{fig:results_poisson_leverage}. Model performance metrics are the same as those used in \cref{sec:linreg} with $\beta_{n}^{\star}(M)$ replacing $\beta^{\star}(M)$. The CIs for the non-fission approaches (data splitting and full data twice) are also constructed using sandwich estimators of the variance corresponding to \cref{thm:var_emp_fahrmeir}. Every method controls the FCR empirically, but the CI lengths for data fission are significantly tighter than for data splitting. Data fission is also more powerful at the selection stage than data splitting. Results for higher $n$ settings can be found in Appendix~\ref{sec:appendix_poisson} and results when this general framework is applied to logistic regression are contained in Appendix~\ref{sec:appendix_logistic}. Across all of these cases, using the full dataset twice no longer results in FCR control, as expected. %Interestingly, unlike in the Gaussian regression case, data fission results in tighter confidence intervals than data splitting for large $n$ settings as well.  

\section{Application to post-selection inference for trend filtering and other nonparametric regression problems} \label{sec:trendfilter}
Consider a nonparametric setup with $y_{i}$ and covariates $x_{i}$ following
$	y_{i} = f_{0}(x_{i}) + \epsilon_{i}$,
where $f_{0}$ is the underlying function to be estimated and $\epsilon_{i}$ is random noise. We denote $Y= (y_{1},...,y_{n})^{T}$ and $\epsilon = (\epsilon_{1},...,\epsilon_{n})^{T}$. For now, assume that $\epsilon \sim N(0,\Sigma)$ for some known $\Sigma$. We can apply the methodologies of \cref{sec:list_decomp} and \cref{sec:linreg} to this problem as follows:
\begin{enumerate}
	\item Decompose $Y$ into $f(Y) = Y + \tau Z$ and $g(Y) = Y - \frac{1}{\tau}Z$ where $Z \sim N(0, \Sigma)$.
	\item Use $f(X)$ to choose a basis $a_{1},...,a_{p}$ for the series expansion of $x_{i}$ and denote $a(x_{i}) = \left(a_{1}(x_{i}),...,a_{p}(x_{i}) \right)^{T}$. Let $A$ denote the matrix of basis vectors for all $n$ data points. 
	\item Use $g(X) | f(X)$ to construct pointwise or uniform CIs as described below.
\end{enumerate}

\paragraph{Pointwise CI.}
We first note that the fitted line 
\begin{equation} \label{eqn:basis_hat}
	\widehat{\beta}(A) = \argmin_{\beta} \norm{g(Y)-A\beta} ^{2} = (A^{T}A)^{-1}A^{T}g(Y)
\end{equation}
corresponds to $A\widehat{\beta}(A)$ which is just the projection of $g(Y)$ onto the basis chosen during the model selection stage. Meanwhile, we define the \emph{projected regression function} as
\begin{equation}\label{eqn:basis_star}
	\beta^{*}(A) =\argmin_{\beta} \mathbb{E}\left[\norm{Y-A\beta}^{2} \right]  = (A^{T}A)^{-1}A^{T}f_{0}(X).
\end{equation}
We are then interested in using the fitted estimates $\widehat{\mu}_{x_{i}}(A) =a(x_{i})^{T}\widehat{\beta}(A)$ to construct intervals that guarantee coverage for the underlying function $f_{0}(x_{i})$ projected onto the chosen basis which we denote as $\mu^{*}_{x_{i}}(A) = a(x_{i})^{T}\beta^{*}(A)$ and refer to as the \emph{projected mean}.
\begin{corollary} \label{thm:trendfilter_pointwise_method}
% 	Let $\widehat{\beta}(A)$ be defined as in \ref{eqn:basis_hat}, $\beta^{*}(A)$ be defined as in \ref{eqn:basis_star}, and $y_{i}, \epsilon_{i}$ be defined as in \ref{eqn:nonparametric}. 
	Recall the definitions of $y_{i}, \epsilon_{i},\widehat{\beta}(A), \beta^{*}(A)$ from \eqref{eqn:basis_hat} and~\eqref{eqn:basis_star}.
	When $\epsilon$ has known variance $\Sigma$, 
	\(\widehat{\mu}_{x_{i}}(A)\sim N\left(\mu^{*}_{x_{i}}(A),  (1+\tau^{-2})a(x_i)^{T}(A^{T}A)^{-1}A^{T} \Sigma A (A'A)^{-1}a(x_i) \right)\)
	and a $1-\alpha$ CI for $\mu^{*}_{x_{i}}(A)$ is
	\(\widehat{\mu}_{x_i}(A) \pm z_{\alpha/2}\sqrt{(1+\tau^{-2})a(x_i)^{T}(A^{T}A)^{-1}A^{T} \Sigma A (A^{T}A)^{-1}a(x_{i})}. \)
\end{corollary}

\paragraph{Uniformly valid CIs.} 
We can also construct uniformly valid CIs $\{CI_{i}\}_{i=1}^{n}$ to cover the projected mean, controlling the \emph{simultaneous} type I error rate defined as
\(\mathbb{P}(\exists i \in [n]: \mu^{*}_{x_{i}}(A) \notin \text{CI}_i).\)
% We use Eq.~(5.5) in \cite{koenker2011additive} for constructing the uniform CI for additive models. %, which seems to be much tighter than the traditional method using the F distribution. 
Here, we further assume homoscedastic errors so $\Sigma = \sigma^{2} I_{n}$. 
Construct $\text{CI}_i := (\widehat{\mu}_{x_{i}}(A) - w, \widehat{\mu}_{x_{i}}(A) + w)$, where the width $w := c(\alpha) \cdot \sigma\sqrt{(1+\tau^{-2})a(x_i)^T (A^T A)^{-1} a(x_i)}$, where $c(\alpha)$ is a multiplier. By Eq.~(5.5) in \cite{koenker2011additive}, $c(\alpha)$ is the solution to the equation
% \begin{align} \label{eq:multiplier}
\(
	\frac{|\gamma|}{2\pi} e^{-c^{2}/2} + 1 - \Phi(c) = \alpha/2,
\)
% \end{align}
where $|\gamma|=\sum_{i=1}^{n-1} ||\widetilde a(x_{i+1}) - \widetilde a(x_{i})||_2$ is the length of the curve connected by $\widetilde a(x_{i}) := \frac{(A^T A)^{-1/2} a(x_{i})}{||(A^T A)^{-1/2} a(x_{i})||}$. 

\begin{fact}[\cite{koenker2011additive}] \label{trendfilter_uniform_method}
	The above constructed CIs will cover the projected mean uniformly:
	\(
		\mathbb{P}(\exists i \in [n]: \mu^{*}_{x_{i}}(A) \notin \text{CI}_i) \leq \alpha,
	\)
	where $\{\mu^{*}_{x_{i}}(A)\}_{t=1}^n$ denotes the projected mean. 
\end{fact}

% \subsection{Empirical results} \label{sec:trendfilter_empirical}
\subsection{Application to trend filtering} \label{sec:trendfilter_empirical}
We explore how this approach to nonparametric inference performs empirically through the lens of \emph{trend filtering} as proposed by \cite{kim-boyd} and \cite{Steidl06splinesin}. Here, we consider the problem of estimating the underlying smooth trend of a time series $y_t \in \R$ with $t = 1, \ldots, n$. We have the same equation as in the standard nonparametric setup, but with $x_{i} = t$ for all $i$ and thus have that $y_{t} = f_{0}(t) + \epsilon_{t}$. Our goal is to estimate the underlying trend $\left(f_{0}(1),\ldots ,f_{0}(n) \right)$. The approach of (first order) trend filtering is to fit a piecewise linear approximation to this data with adaptively chosen breakpoints or \emph{knots} by solving:
\(
	\widehat{x} = \argmin_{x_{t} \in \R} 1/2\sum_{t=1}^n (y_t - x_t)^2 + \lambda \sum_{t=2}^{n-1} |x_{t-1} - 2x_t + x_{t+1}|.
\)
Although we focus on first order trend filtering, our methodology straightforwardly generalizes to higher order trend filtering. For a detailed overview of trend filtering, including generalization to higher order polynomials, see \cref{sec:background_trendfilter}. 

An issue issue with trend filtering is that uncertainty quantification is not generally tractable %even when strong distributional assumptions are made 
since the knots are chosen adaptively. In some settings, data carving approaches may offer a viable path forward. When the (fused or generalized) lasso is used to select the choice of knots, methods such as \cite{chen2021powerful}, \cite{duy2021parametric}, and \cite{lasso_posi_trendfilter} can be used to construct tests and CIs with valid post-selective distributions. The drawback of these approaches is that they offer the analyst no flexibility during the selection stage. If the analyst wishes to change the degree of the fitted polynomial or decrease the number of chosen knots after seeing preliminary results, inferential guarantees are no longer available. Although such actions are common in applied data analysis, data carving approaches do not easily handle ad-hoc changes to the prespecified selection methodology. Data fission %offers an alternative path forward which 
offers the analyst flexibility to change the amount of regularization (alter the number of knots) or change the degree of the trend filtered estimate after a preliminary view of the selected model performance. 

Similar to \cref{sec:linreg}, an assumption of known variance is required when using data fission in order to ensure the selection and inference datasets are independent. Approaches that estimate variance empirically are explored in Appendix~\ref{sec:trendfilter_estimated_sigma} but we leave theoretical guarantees for future work. %Using decomposition rules for Gaussian datasets from \cref{sec:list_decomp}, 
The proposed procedure is:
\begin{enumerate}
    \item Decompose $Y$ into $f(Y) = Y + \tau Z$, $g(Y) = Y - \frac{1}{\tau}Z$ where $Z \sim N(0, \Sigma)$. %(assuming known $\Sigma$).
    \item Train a $k$-th order trendfilter using $f(Y)$ %(using \texttt{cv.trendfilter} in \texttt{R} package \texttt{genlasso} 
    and select the tuning parameter $\lambda$ by selecting the set of knots $\{\tau_k\}_{k \in 1}^{m}$ with the minimum cross-validation error.\footnote{Although we restrict ourselves to a fixed selection rule for simulation purposes, data fission allows for the use of any arbitrary method of choosing the set of knots. For a discussion of alternative methods for choosing the tuning parameter $\lambda$, see Appendix~\ref{sec:appendix_alternative_knots}.}
    \item Construct a $k$th degree falling factorial basis, as described in \cite{adaptive_piecewise_tibshirani}, with knots at $\{\tau_k\}_{k=1}^m$ and use this to regress against $g(Y)$. Note that when $k=0$ or $k=1$, this reduces to generating the more familiar truncated power basis. 
    %(for example, by using \texttt{bs} in \texttt{R} package \texttt{splines} to generate a basis from the chosen knots for regression)%
    \item Get \emph{pointwise} CIs for each data point $t = 1, \ldots, n$ at level $\alpha$ as described in \cref{thm:trendfilter_pointwise_method}. In cases where $\Sigma=\sigma^{2}I_{n}$, we can also construct \emph{uniform} CIs using \cref{trendfilter_uniform_method}.
\end{enumerate}

\iffalse
The constructed CIs will cover the \textit{projected mean} $\mu^{*}_{x_{i}}(A)$, which is the projection of the function $f_0(t)$ onto the linear spline with knots at $\{\tau_k\}_{k=1}^m$. The pointwise CIs for $t = 1,\ldots, n$ will control FCR at level $\alpha$, while the uniform CIs will control the probability of not covering the projected mean at any time point at level $\alpha$. 
\fi 

%\subsection{Simulation results}
\paragraph{Simulation setup.} 
We verify the efficacy of this approach through an experiment. We simulate data as $y_t = f(t) + z_t$ for $t \in [n]$, where $z_t \sim N(0, \sigma^2)$, and $f_{0}(t)$ is generated by $f_{0}(t+1) = f_{0}(t) + v_t$. %The trend slopes $v_t$ are chosen from a simple Markov process (independent of $\{z_t\}_{t = 1}^n$). 
With probability $1 - p$, we let $v_{t+1} = v_t$ (no slope change in the trend). With probability $p$, we choose $v_{t + 1}$ from a uniform distribution on $[-0.5, 0.5]$. We choose the initial slope $v_1\sim\text{Unif}[-0.5, 0.5]$ and set $f_{0}(0) = 0$.

We use the above procedure to select knots and conduct inference. An example of how this methodology performs when compared to the (invalid) approach of using the full dataset twice is shown in \cref{fig:instance_trendfilter} for \emph{pointwise} CIs. As there is not an obvious way to apply data splitting to trend filtering, this view is excluded. %from our comparison.
Although using the full dataset twice is invalid, it performs worse for datasets that are more volatile, either in terms of the underlying structural trend (i.e. more knots and/or larger slopes) or in terms of the underlying noise level. For a similar demonstration with uniform CIs  %for a single example run, 
see Appendix~\ref{appendix:trendfilter_supplementa_uniforml}. 

\begin{figure}[H]
\centering
    \begin{subfigure}[t]{0.48\textwidth}
        \includegraphics[width=0.45\linewidth]{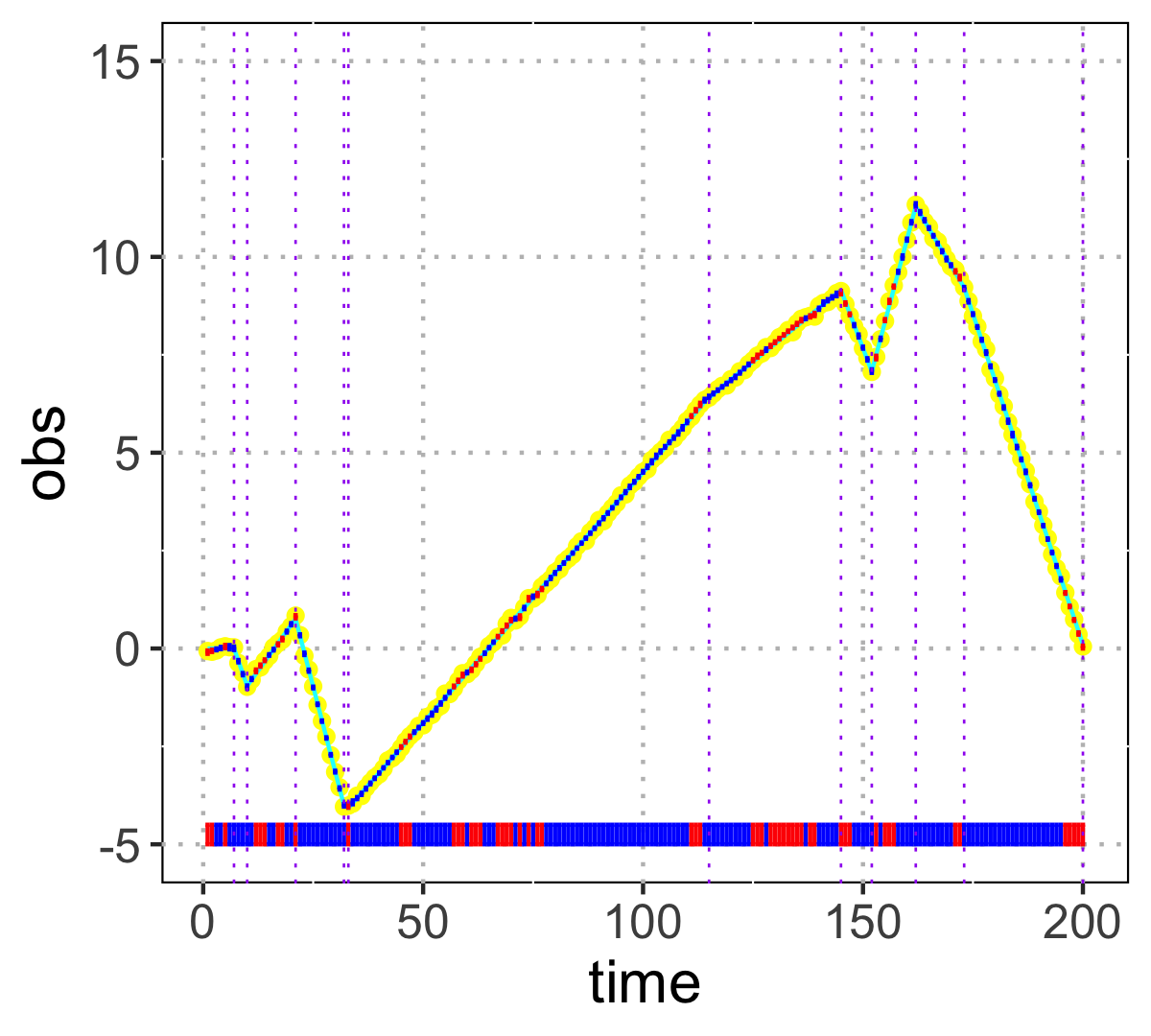}
        \includegraphics[width=0.45\linewidth]{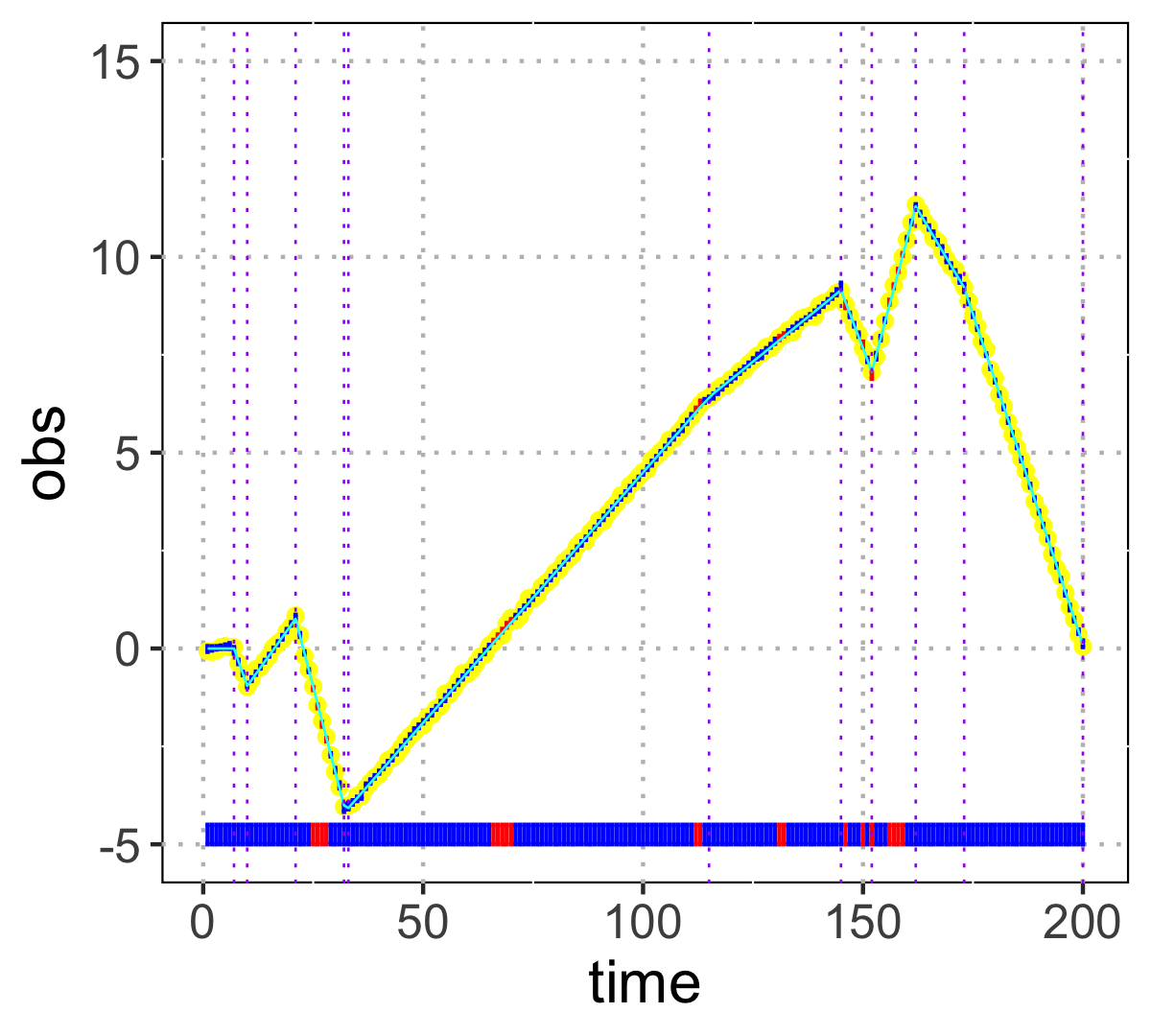}
        \caption{Dataset generated with smaller noise ($\sigma = 0.05$) and low probability of new knots ($p = 0.05$). Target FCR equals $0.2$ with empirical coverage of $0.1$ for data fission and $0.285$ when using the full dataset twice.} 
    \end{subfigure}
    \hspace*{0.25cm}
    \begin{subfigure}[t]{0.48\textwidth}
        \includegraphics[width=0.45\linewidth]{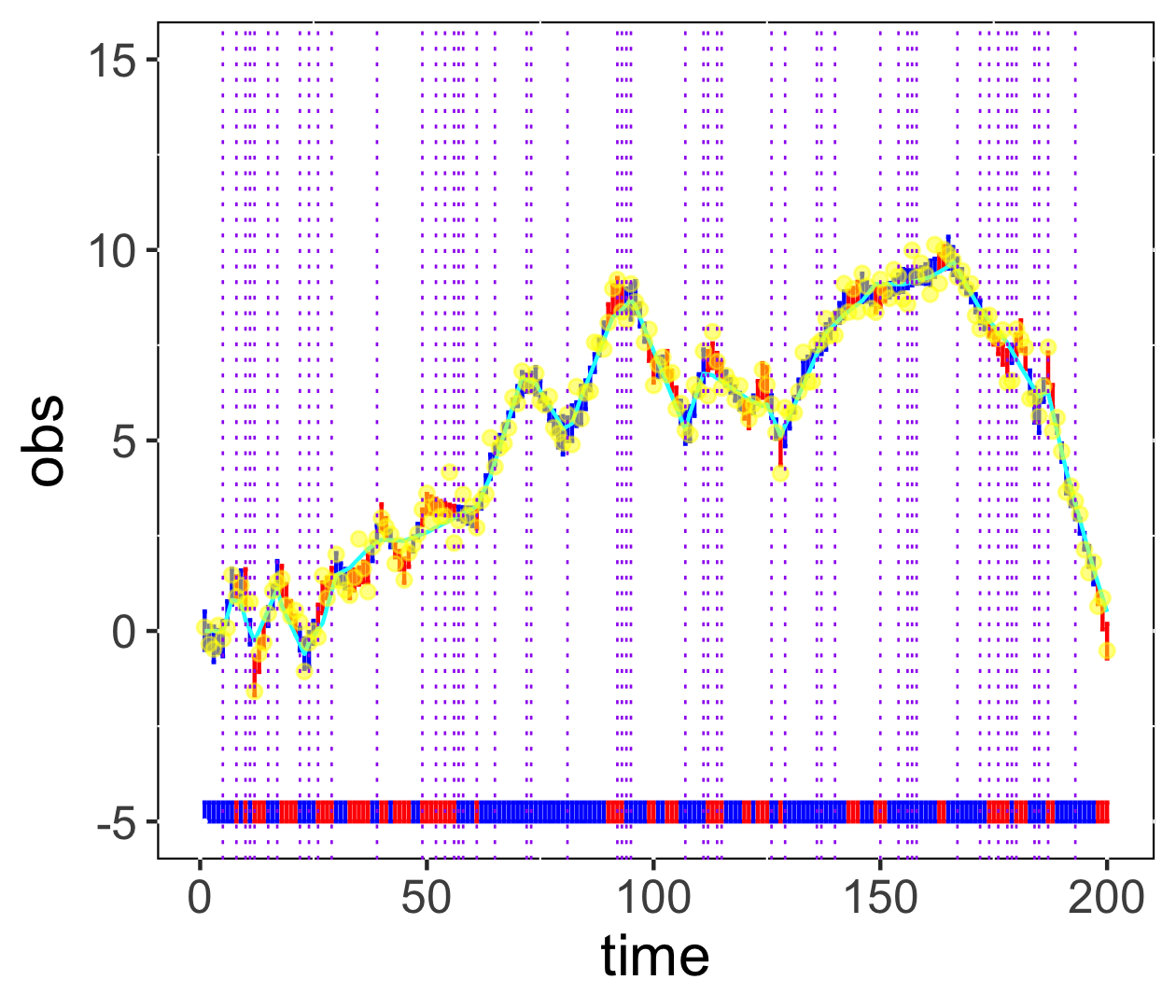}
        \includegraphics[width=0.45\linewidth]{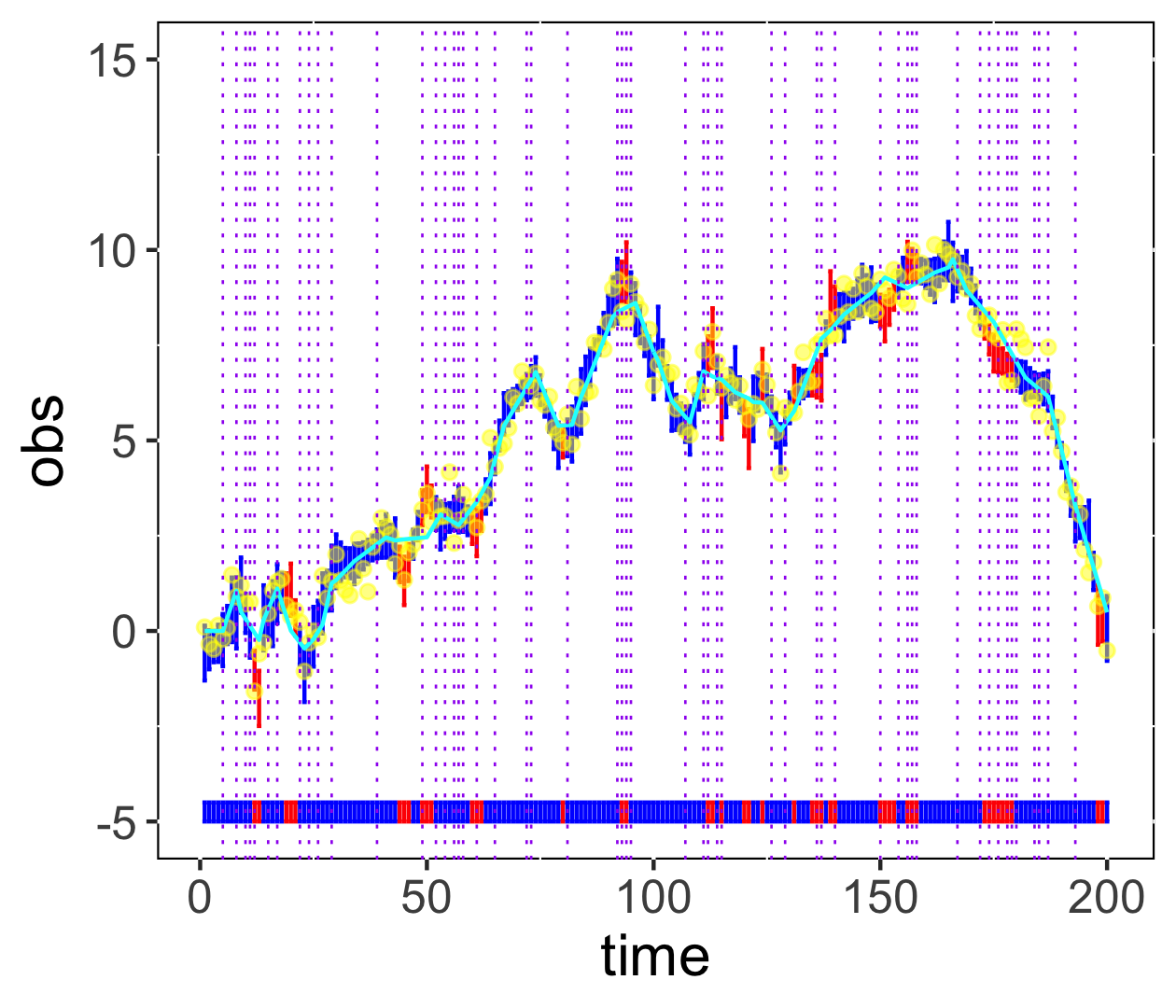}
        \caption{Dataset generated with larger noise ($\sigma = 0.5$) and higher probability of new knots ($p = 0.3$).  Target FCR equals $0.2$ with empirical coverage of $0.225$ for data fission and $0.365$ when using the full dataset twice.} 
    \end{subfigure}
    \caption{Two instances of the observed points (in yellow) and the \textit{pointwise} CIs (in blue if correctly cover the trend, in red if not; the time points with false coverage are also amplified in the bar at the bottom) using two types of methods: full data twice (left), and data fission (right). The underlying projected mean is marked in cyan, which mostly overlaps with the true underlying trend. The true knots are marked by vertical lines. Using data fission results in correct empirical coverage (the $0.225$ above was for just one run, the average is below $0.2$). In contrast, the FCR is not controlled when using the full dataset twice; it worsens as the underlying noise and trend become more volatile.}
    \label{fig:instance_trendfilter}
\end{figure}

\begin{figure}[H]      
\centering
    \begin{subfigure}[t]{0.2\textwidth}
        \centering
        \includegraphics[width=1\linewidth]{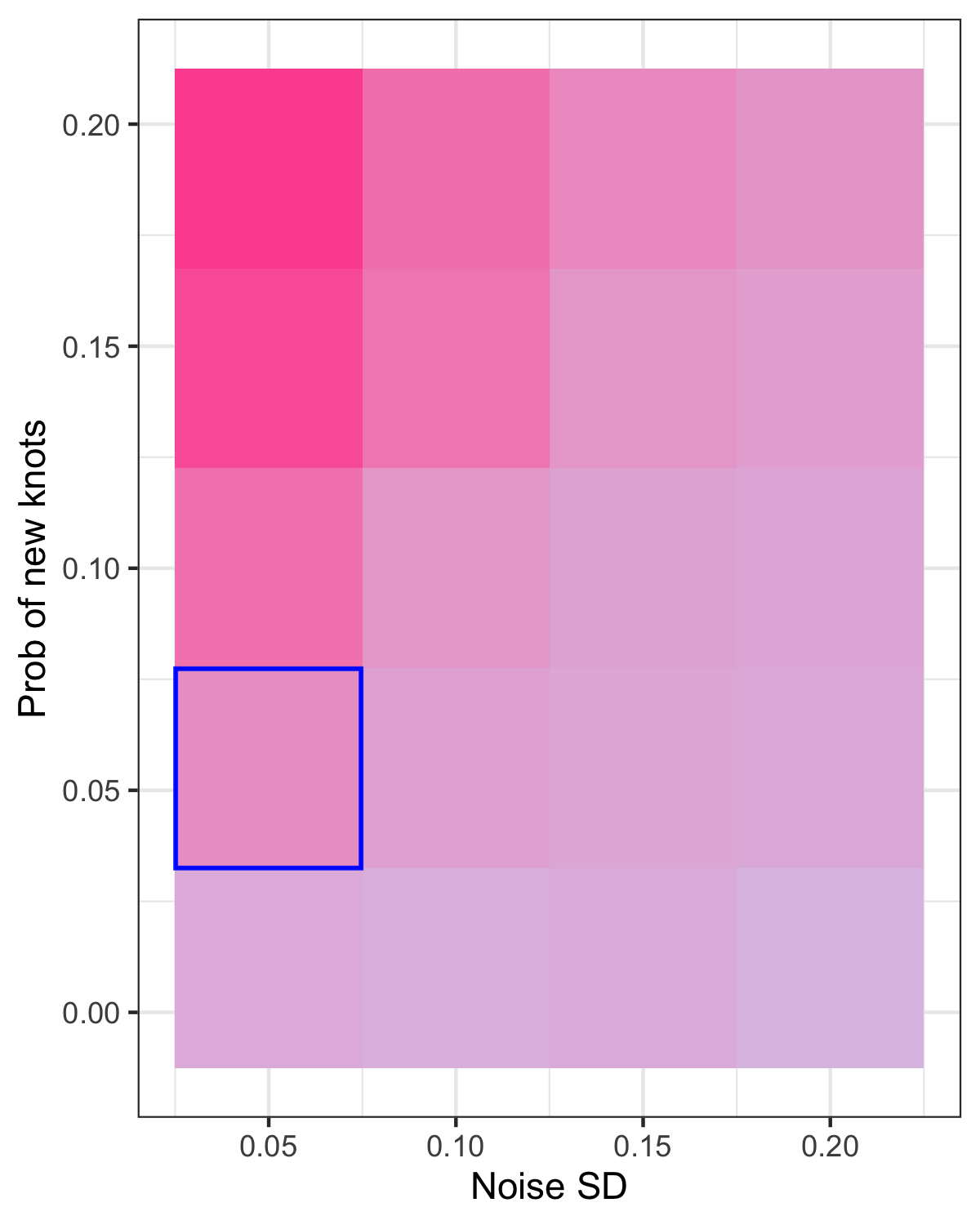}
        \caption{FCR using full data twice}
    \end{subfigure}
\hfill
    \begin{subfigure}[t]{0.2\textwidth}
        \includegraphics[width=1\linewidth]{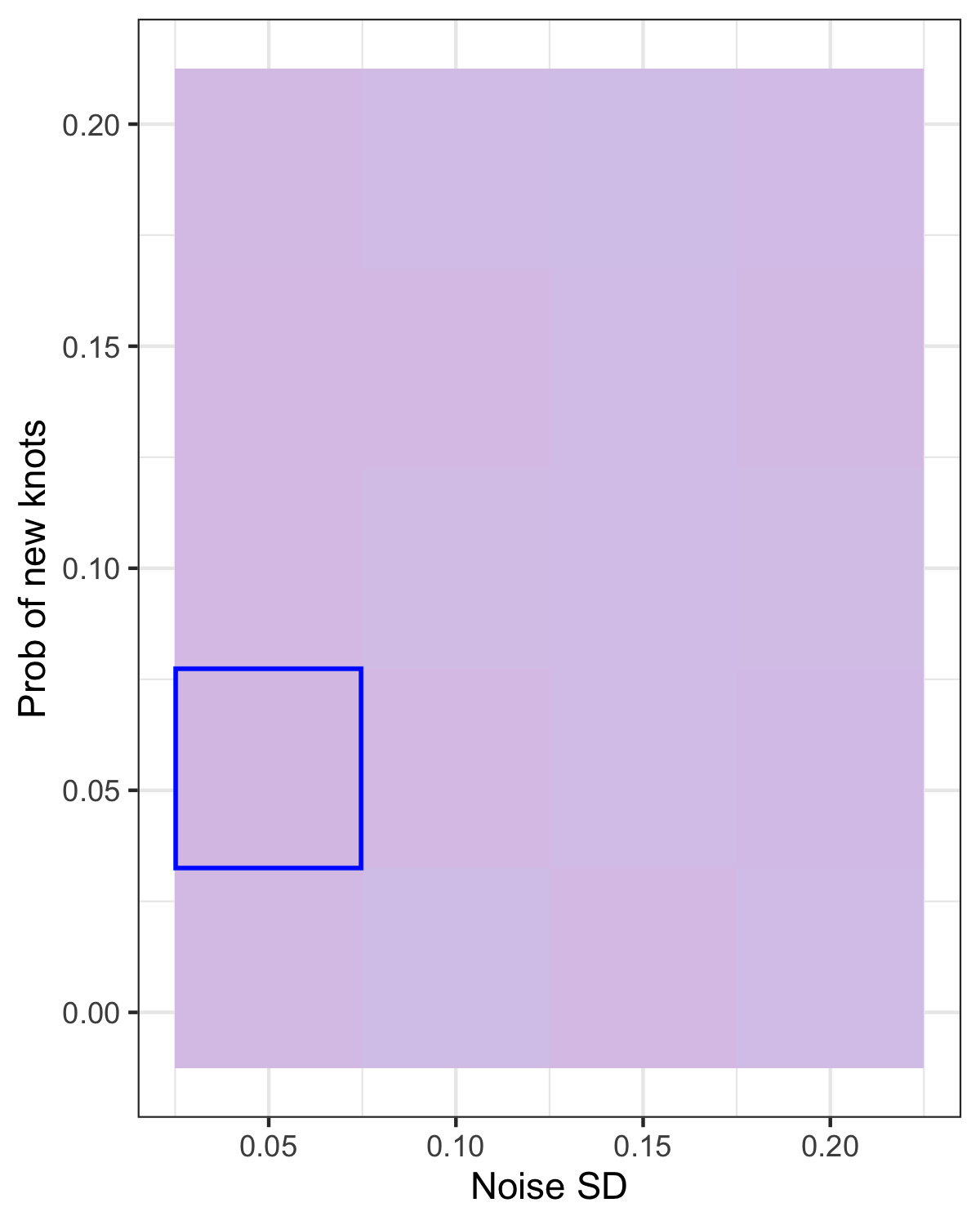}
        \caption{FCR using data fission}
    \end{subfigure}
\hfill
    \begin{subfigure}[t]{0.2\textwidth}
        \centering
        \includegraphics[width=1\linewidth]{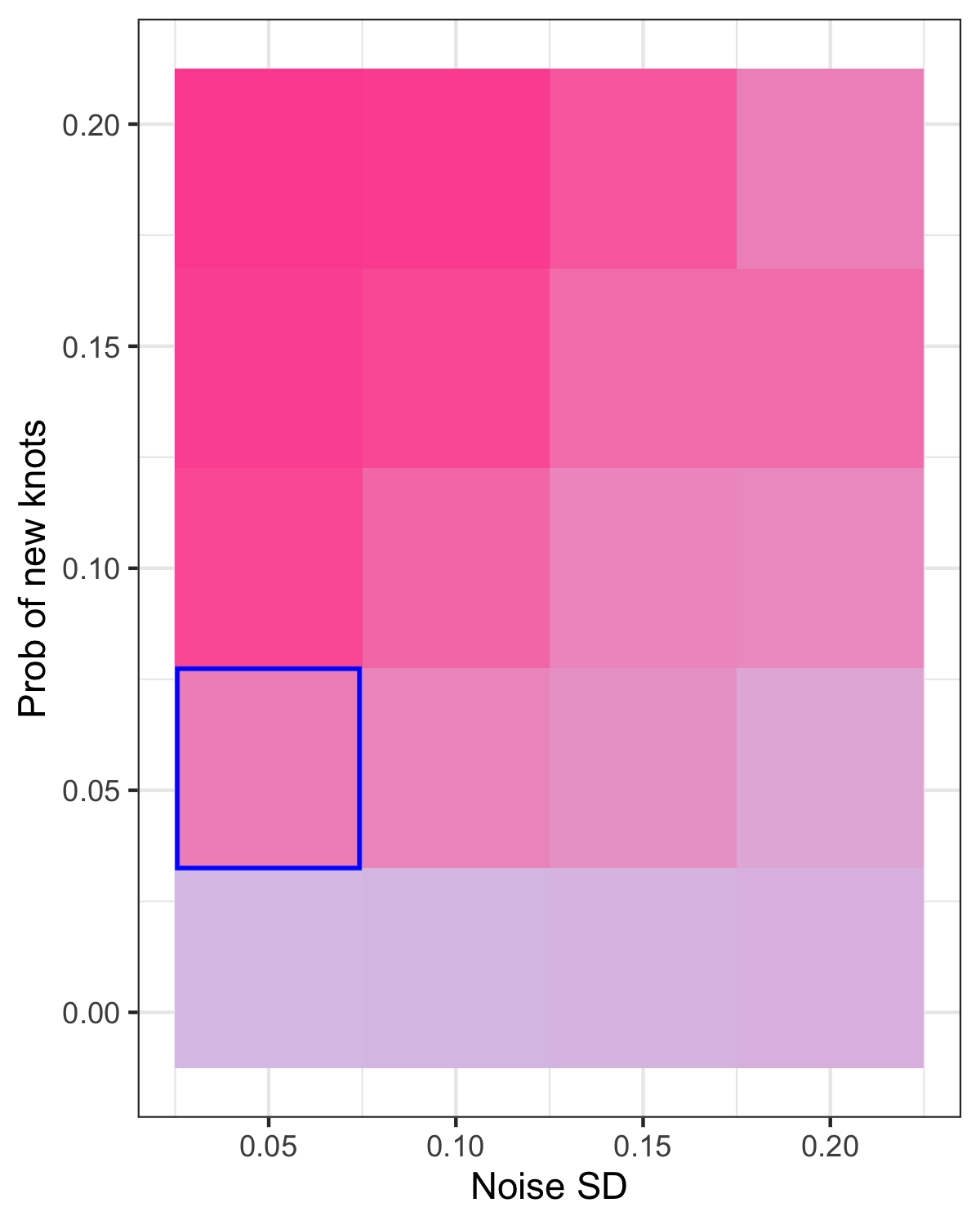}
        \caption{Simult. type I error (double~dip)}
    \end{subfigure}
\hfill
    \begin{subfigure}[t]{0.2\textwidth}
        \includegraphics[width=1\linewidth]{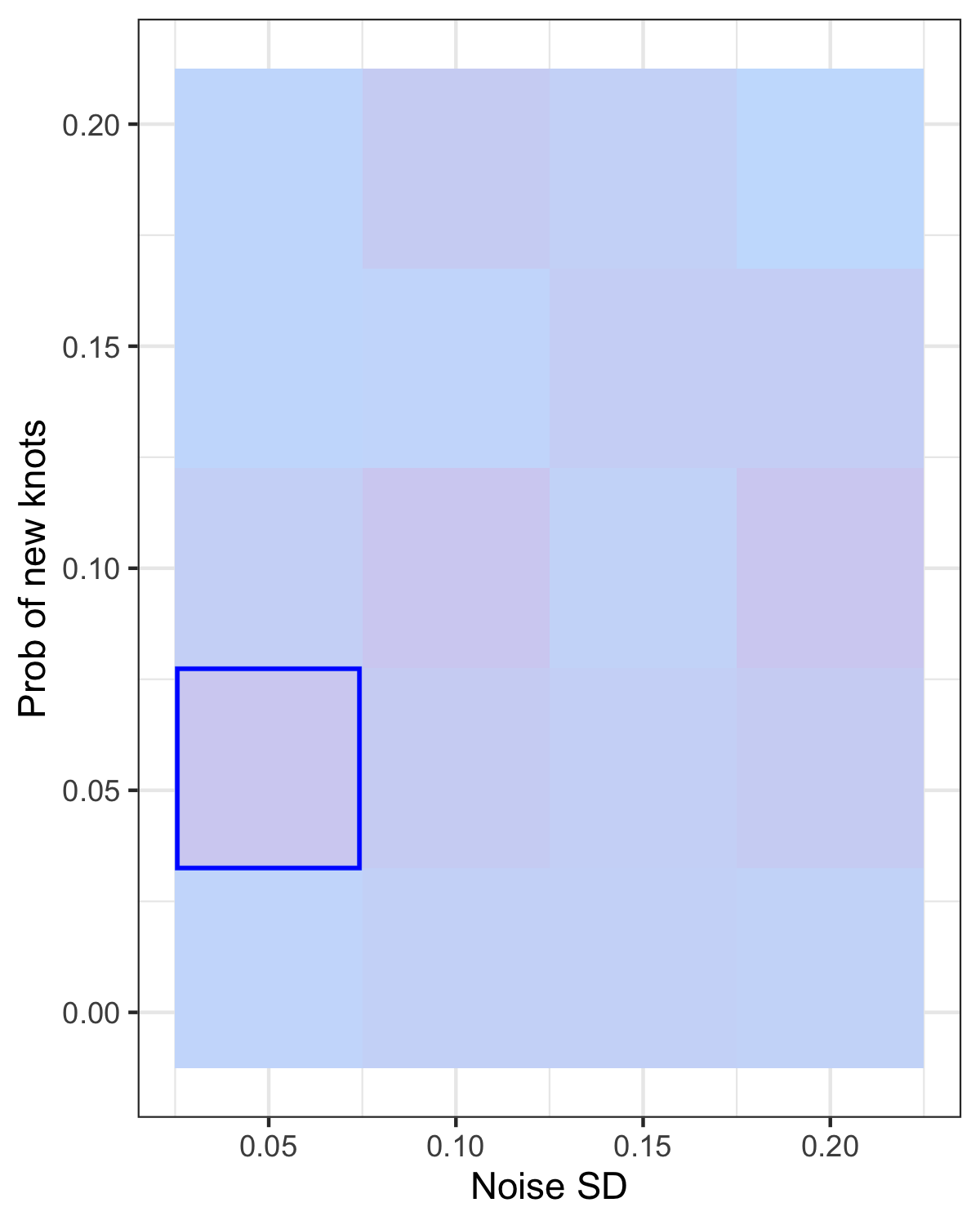}
        \caption{Simult. type I error (fission)}
    \end{subfigure}
\hfill
    \begin{subfigure}[t]{0.1\textwidth}
    \hskip 0pt
        \includegraphics[width=1\linewidth]{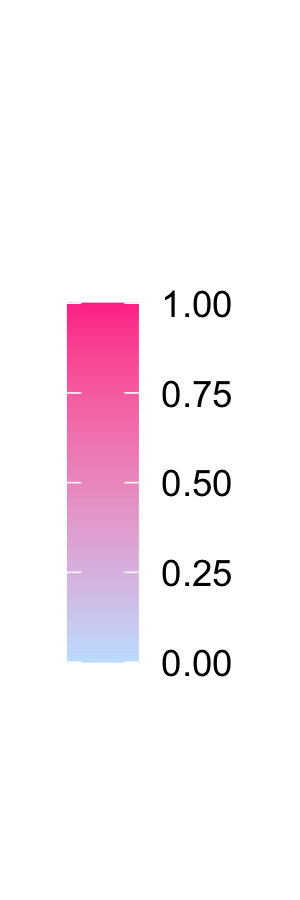}
    \end{subfigure}

    \caption{FCR for the \textbf{pointwise CIs} and simultaneous type I error for \textbf{uniform CIs} when varying the probability of having knots $p$ in $\{0.01, 0.55, 0.1, 0.145, 0.19\}$ and the noise SD in $\{0.05, 0.1, 0.15, 0.2\}$ (the blue circled cell represents the setting for the first shown instance in \cref{fig:instance_trendfilter}). The CIs generated using full data twice do not have valid FCR or simultaneous type I error control, especially when $p$ is large (more knots) and the noise standard deviation is small, but data fission is always valid.}
    \label{fig:trendfilter} 
\end{figure}

We repeat over $500$ trials and report the results in \cref{fig:trendfilter}. Reusing the full dataset performs most comparably to data fission when the noise is small and the probability of new knots is low. For noisier data and more volatile trends, reusing the data results in anticonservative CIs. See Appendix~\ref{appendix:trendfilter_supplemental} for more detailed empirical results.

\subsection{Application to spectroscopy} \label{sec:astro_examples}
%We demonstrate this methodology on an example in observational astronomy. 
In \cite{10.1093/mnras/staa110}, the authors introduce trend filtering as a tool for astronomical data analysis with spectral classification as a motivating example.
%astronomical data analysis because many applications in astronomy contain a step which reduces to one dimensional data compression %One example is the use of trend filtering to aid in spectral classification. 
In spectroscopy, we observe a spectrum consisting of wavelengths ($\lambda$) and corresponding measurements of the coadded flux ($f(\lambda)$). Trend filtering can then used to create a ``spectral template''---a smoothed line of best fit for the observations. This template is combined with emission-line parameter estimates to classify the astronomical object. The authors demonstrate that trend filtering performs well empirically compared to the existing state-of-the-art for creating spectral templates which revolve around low-dimensional principal component analysis. We use this analysis as a way of demonstrating how CIs constructed using data fission appear for real data.

\begin{figure}
\centering
\begin{subfigure}{\textwidth}
\centering
\includegraphics[width=0.6\linewidth]{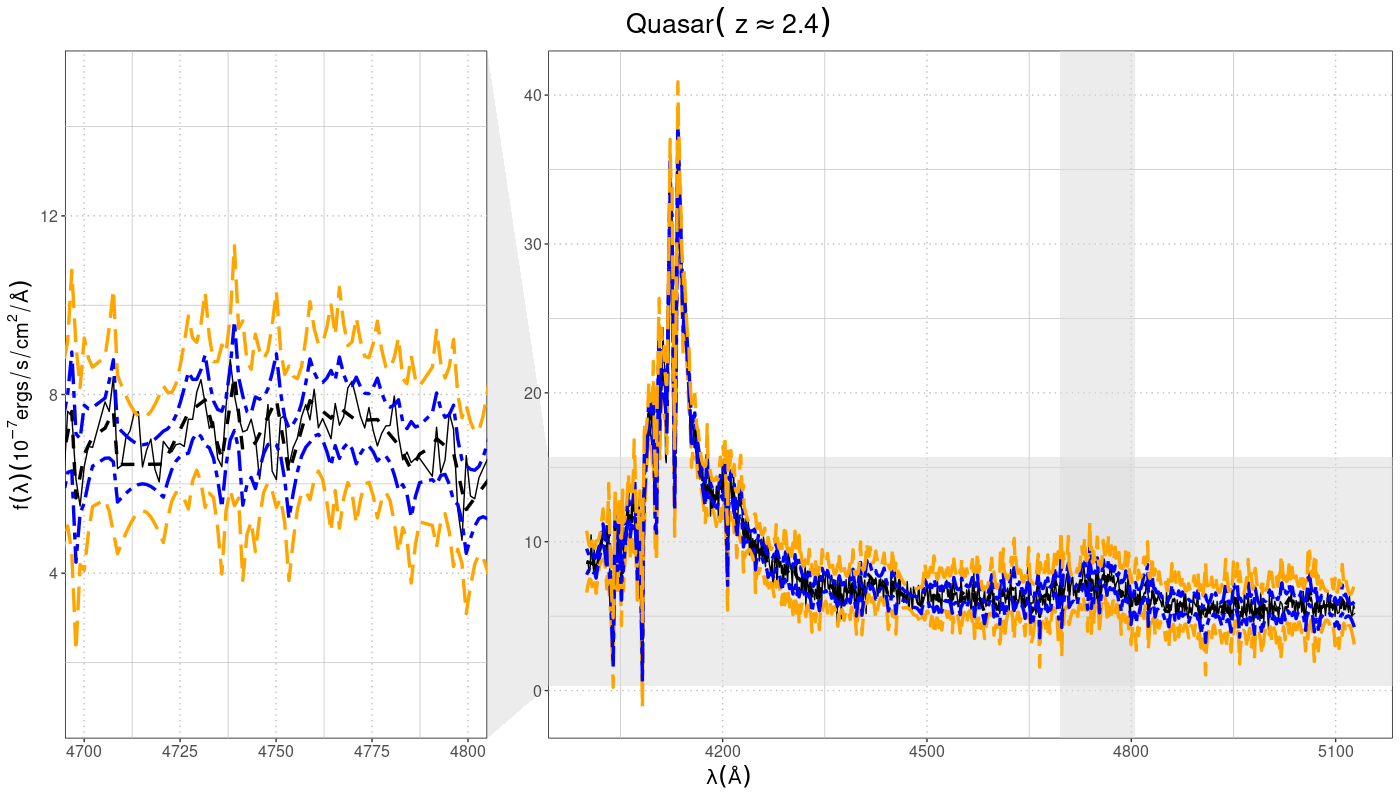}
\end{subfigure}
\begin{subfigure}{\textwidth}
  \centering
  \includegraphics[width=0.5\linewidth]{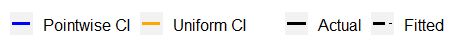}
\end{subfigure}
\caption{Fitted values as well as uniform and pointwise CIs for a quasar object fit using linear trend filtering. The right view shows the trend filter over the entire spectrum, but the left view ``zooms in'' on a smaller subset of the data to aid in visual identification.}
\label{fig:astronomical_examples}
\end{figure}

The dataset used for this analysis is the twelfth data release of the Baryon Oscillation Spectroscopic Survey (\cite{SDSS}). We estimate a spectral template using trend filtering, along with pointwise and uniform confidence bands (using the methods from \cref{sec:trendfilter}) designed to cover the conditional mean $\mu = \mathbb{E}[f(\lambda) | \lambda]$ of the observed spectra of the same three object
 discussed in \cite{10.1093/mnras/staa110}. In this setting, the flux measurement variances are known a priori %as they are constructed from the BOSS spectroscopy pipeline (\cite{BOSS}) 
 and account for the statistical uncertainty introduced by photon noise, CCD read noise, and sky-subtraction error. %One point of departure from the procedure used in \cref{sec:trendfilter} is that we also consider quadratic trend filtering in addition to linear trend filtering to construct a fitted curve. 
The results are shown in \cref{fig:astronomical_examples} for the quasar\footnote{
 DR12, Plate = 7130, MJD =5659, Fiber = 58. Located at (RA,Dec, z)  $\approx (349.737^{\circ}, 33.414^{\circ},2.399)$} 
and in \cref{sec:spectroscopy_supplement} for the remaining objects. Since the confidence bands displayed are covering the conditional mean and \emph{not} the observed data, there is no objective ``ground truth'' to compare the outputs of the model to in order to assess goodness of fit.  Therefore, the results need to be judged holistically.

\section{Conclusion} \label{sec:conclude}
We have proposed a method for selective inference through external randomization that allows flexibility to choose models in an arbitrary and data dependent way, like in data splitting, but also provides a more efficient split of available information in some settings. We demonstrate the efficacy of these methods in four applications: effect size estimation after interactive multiple testing, Gaussian linear regression, generalized linear models, and trend filtering. In the case of linear regression and GLMs, we note tighter CIs and higher power compared to data splitting in settings with leverage points and small sample size. For trend filtering, data fission enables uncertainty quantification while offering flexibility in choosing the number of knots and degree of the polynomial based on heuristic criteria. 

Numerous avenues for follow-up work exist. Although we note promising empirical results that suggest this procedure can be generalized to situations where the variance is unknown and the assumed distribution on the error term is misspecified, additional work needs to be done to establish rigorous guarantees, especially in high dimensional settings. Independently, it is possible to repeat the data fission procedure multiple times in parallel, but aggregating results seems nontrivial. Last, we anticipate applications to contemporary problems such as the creation of fake datasets and enabling differential privacy. %(beyond the Gaussian mechanism). 

%% file: acknowledgements.tex
\paragraph{Acknowledgements.}
We are very grateful to Anna Neufeld and Ameer Dharamshi for discovering errors in the original version of the manuscript and helping to revise our decomposition rules. We also thank Collin Politsch for help with the data for \cref{sec:astro_examples}. This dataset is constructed from SDS-IIII. Funding for SDSS-III has been provided by the Alfred P. Sloan Foundation, the Participating Institutions, the National Science Foundation, and the U.S. Department of Energy Office of Science. The SDSS-III web site is http://www.sdss3.org/.

SDSS-III is managed by the Astrophysical Research Consortium for the Participating Institutions of the SDSS-III Collaboration including the University of Arizona, the Brazilian Participation Group, Brookhaven National Laboratory, Carnegie Mellon University, University of Florida, the French Participation Group, the German Participation Group, Harvard University, the Instituto de Astrofisica de Canarias, the Michigan State/Notre Dame/JINA Participation Group, Johns Hopkins University, Lawrence Berkeley National Laboratory, Max Planck Institute for Astrophysics, Max Planck Institute for Extraterrestrial Physics, New Mexico State University, New York University, Ohio State University, Pennsylvania State University, University of Portsmouth, Princeton University, the Spanish Participation Group, University of Tokyo, University of Utah, Vanderbilt University, University of Virginia, University of Washington, and Yale University.

%% file: appendices.tex
\begin{appendices}\label{sec:appendix}
\crefalias{section}{appendix}

\section{List of decompositions} \label{sec:appendix_list_decomp}
Included in this section is a more extensive list of decomposition rules for several commonly encountered distributions. When discussing each method, we will label them as follows so the reader will be able to relate them easily to the strategies discussed in \cref{sec:fission}:
\begin{itemize}
    \item \textbf{(P1)} will indicate that the decomposition strategy is an instance of the first principle ---  $f(X)$ and $g(X)$ are independent with known distributions. 
    \item \textbf{(P2)} will indicate that the decomposition strategy follows the second principle --- $f(X)$ has a known marginal distribution, while $g(X)$ has a known conditional distribution given $f(X)$. 
    \item \textbf{(P2 CP)} will indicate that the decomposition strategy is an instance of the second principle --- but specifically using ``conjugate-prior reversal''.
\end{itemize}
Proofs of the below claims are deferred to \cref{sec:decomp_proofs}.

%Although we recognize readers may find a close reading of this section tedious, we feel it is important to include this information for two reasons. First, we wish to provide several concrete examples to demonstrate to the reader how the strategies for constructing $f(X)$ and $g(X)$ can be practically instantiated in settings beyond the specific applications found in this paper. Second, we want to highlight the wide applicability of this approach to disparate problems within statistics --- only providing a small number of examples would give the false impression that this methodology is only tractable within a few idealized settings. 
\begin{itemize}
    \item \textbf{Gaussian} 
    \begin{enumerate}
        \item\textbf{(P1)} Suppose $X \sim N( \mu,  \Sigma)$ is $d$-dimensional ($d \geq 1$). Draw $ Z \sim N(0,  \Sigma)$. Then $f( X) =  X + \tau  Z$, where $\tau \in (0, \infty)$ is a tuning parameter, has  distribution $N( \mu, (1 + \tau^{2})  \Sigma)$; and $g( X) =  X - \tfrac{1}{\tau}  Z$ has distribution $N(\mu,  (1 + \tau^{-2})  \Sigma)$; and $f( X) \independent g( X)$. Larger $\tau$ indicates less informative $f( X)$ (and more informative $g( X) \mid f( X)$). 
        \item \textbf{(P2 CP)} Alternatively, draw $ Z \sim N( X, \tau  \Sigma)$, where $\tau \in (0, \infty)$ is a tuning parameter. Then $f( X) =  Z$ has marginal distribution $N( \mu, (1 + \tau)  \Sigma)$; and $g( X) =  X$ has conditional distribution $N(\frac{\tau}{\tau + 1} ( \mu + f( X)/\tau), \frac{\tau}{\tau + 1}  \Sigma)$. Larger $\tau$ indicates less informative $f( X)$.
        \item \textbf{(P2)} More generally, we can add Gaussian noise with an arbitrary covariance matrix. Draw $Z \sim N(0,\Sigma_{0})$ and let $f(X) = X-Z$ with $g(X) = X+Z$ as before. For notational convenience, let $\Sigma_{1} = \Sigma + \Sigma_{0}$ and $\Sigma_{2} = \Sigma - \Sigma_{0}$ Then $f(X) \sim N(\mu, \Sigma + \Sigma_{0})$ and $g(X) | f(X) \sim N\left( \mu + \Sigma_{2}\Sigma_{1}^{-1} \left( f(X) - \mu \right), \Sigma_{1} - \Sigma_{2}\Sigma_{1}^{-1}\Sigma_{2} \right)$. %Note that this result is less elegant than the above formulations, but may be more tractable in cases where $\Sigma$ is unknown at the outset of investigation. 
    \end{enumerate}
    \item \textbf{Gamma} 
    \begin{enumerate}
        \item \textbf{Exponential (P2 CP)} Suppose $X \sim \mathrm{Exp}(\theta)$. Draw $ Z = (Z_1, \ldots, Z_B)$ where  each element is i.i.d.\ $Z_i \sim \mathrm{Poi}(X)$ and $B \in \{1, 2, \ldots\}$ is a tuning parameter. Then $f(X) =  Z$, where each element is independently distributed as $\mathrm{Geo}(\tfrac{\theta}{\theta+1})$. $g(X) = X$ has conditional distribution $\mathrm{Gamma}(1 + \sum_{i=1}^B f_i(X), \theta + B)$. Larger $B$ indicates more informative $f(X)$ (and less informative $g(X) \mid f(X)$).
        
        \textbf{(P2 CP)} Alternatively, we can draw $Z \sim \mathrm{Poi}(\tau X)$, where $\tau \in (0,\infty)$ is a tuning parameter. Then $f(X) = Z$, marginally distributed as $\mathrm{Geo}(\tfrac{\theta}{\theta+\tau})$. $g(X) = X$ has conditional distribution $\mathrm{Gamma}(1 + f(X), \theta + \tau)$. Here, $f(X)$ is most informative when $\tau$ is comparable with $\theta$, and less informative when $\tau$ approaches $0$ or $\infty$. %(although it appears that the left information monotonically decreases in $\tau$).
        
        %\item Suppose $X \sim \mathrm{Gamma}(\theta,1)$. Draw $ Z = (Z_1, \ldots, Z_n)$ where each element is i.i.d.\ $Z_i \sim \mathrm{Poi}(X)$ and $n \in \{1, 2, \ldots\}$ is a tuning parameter. Then $f(X) =  Z$, where each element is independently distributed as $\mathrm{Geo}(1/2)$ when $\theta = 1$, and stochastically larger (smaller) than $\mathrm{Geo}(1/2)$ when $\theta$ is larger (smaller) than one (Indeed, it follows a negative binomial $NB(\theta, 1/2)$). And $g(X) = X$ has conditional distribution $\mathrm{Gamma}(\theta + \sum_{i=1}^n f_i(X), n+1)$. Larger $n$ indicates more informative $f(X)$ (and less informative $g(X) \mid f(X)$).
        
        \item \textbf{(P2 CP)}  Generally, suppose $X \sim \mathrm{Gamma}(\alpha, \beta)$. Draw $ Z = (Z_1, \ldots, Z_B)$ where each element is i.i.d.\ $Z_i \sim \mathrm{Poi}(X)$ and $B \in \{1, 2, \ldots\}$ is a tuning parameter. Then $f(X) =  Z$, where each element is independently distributed as a negative binomial $NB(\alpha, \tfrac{\beta}{\beta+1})$. $g(X) = X$ has conditional distribution $\mathrm{Gamma}(\alpha + \sum_{i=1}^B f_i(X), \beta + B)$. Larger $B$ indicates more informative $f(X)$ (and less informative $g(X) \mid f(X)$).
        
        \item \textbf{(P2 CP)}
        Alternatively, we can draw $Z \sim \mathrm{Poi}(\tau X)$, where $\tau \in (0,\infty)$ is a tuning parameter. Then $f(X) = Z$, marginally distributed as $\mathrm{NB}(\alpha, \tfrac{\beta}{\beta+\tau})$. $g(X) = X$ has conditional distribution $ \mathrm{Gamma}(\alpha + f(X), \beta + \tau)$. $f(X)$ is most informative when $\tau$ is comparable with $\theta$, and less informative when $\tau$ approaches zero or infinity.
        
        \item Note:
        decomposition of the Gamma distribution implies decomposition of the Chi-square distribution~$\chi^2_k$, as it is equivalent to $\mathrm{Gamma}(k/2, 1/2)$.
    \end{enumerate}

    \item \textbf{Beta} 
    \begin{enumerate}
        \item \textbf{(P2 CP)} Suppose $X \sim \mathrm{Beta}(\theta,1)$. Draw $Z \sim \mathrm{Bin}(B, X)$, where $B \in \{1, 2, \ldots\}$ is a tuning parameter. Then $f(X) = Z$ has marginal distribution as a discrete uniform in $\{0, 1, \ldots, B\}$ when $\theta = 1$, and stochastically larger (smaller) than a discrete uniform when $\theta$ is larger (smaller) than one (the PMF of $Z$ is $p_\theta(z) = \frac{\theta \Gamma(z + \theta) B!}{\Gamma(B + 1 + \theta) z!}$). $g(X) = X$ has conditional distribution $\mathrm{Beta}(\theta + f(X), B - f(X) + 1)$. Larger $B$ indicates more informative $f(X)$ (and less informative $g(X) \mid f(X)$).
        
        \item \textbf{(P2 CP)} Similarly, if $X \sim \mathrm{Beta}(1, \theta)$, we can draw $Z \sim \mathrm{Bin}(B, 1 - X)$. Then, and the resulting $g(X) | f(X) \sim \mathrm{Beta}(B - f(X) + 1,\theta + f(X))$ and $f(X)$ has the same marginal distribution as above.
        
        \item \textbf{Multivariate case: Dirichlet (P2 CP)} Suppose $ X \sim \mathrm{Dir}(\theta, 1, \ldots, 1)$, where $(\theta, 1, \ldots, 1)$ is a $d$-dimensional vector with $k \geq 2$. Draw $ Z \sim \mathrm{Multinom}(B,  X)$, where $B \in \{1, 2, \ldots\}$ is a tuning parameter. Then $f( X) =  Z$ has marginal distribution as a discrete uniform in its support $\{z_i \in \{0, \ldots, B\} \text{ for } i \in [d]: \sum_{i=1}^d z_i = B\}$ when $\theta = 1$, and for other $\theta$, the PMF is $p_\theta( z) = \frac{B! \Gamma(z_1 + \theta) \Gamma(d - 1 + \theta)}{z_1! \Gamma(\theta) \Gamma(B + d - 1 + \theta)}$. $g(X) = X$ has conditional distribution $\mathrm{Dir}(\theta + f_1( X), 1 + f_2( X), \ldots, 1+ f_k( X))$. Larger $B$ indicates more informative $f( X)$ (and less informative $g( X) \mid f( X)$).
        
        In the general case, where $ X \sim \mathrm{Dir}(\theta_1, \theta_2, \ldots, \theta_d)$, we can use the same construction. Then $f( X) =  Z$ has marginal distribution 
        \[
        p_{ \theta}( z) = \frac{ \Gamma\left(\sum_{i=1}^d \theta_i \right)B!}{\Gamma\left(B + \sum_{i=1}^d \theta_i\right)}  \prod_{i=1}^d \frac{ \Gamma(\theta_i + z_i)  }{\Gamma(\theta_i)z_{i}!},
        \]
        and $g(X) = X$ has conditional distribution $\mathrm{Dir}(\theta_1 + f_1( X), \ldots, \theta_k + f_d( X))$. Larger $B$ indicates more informative $f( X)$ (and less informative $g( X) \mid f( X)$).
    \end{enumerate}
    \item \textbf{Binomial (P2).} Suppose $X \sim \mathrm{Bin}(n, \theta)$. Draw $Z \sim \mathrm{Bin}(X, p)$ where $p \in (0,1)$ is a tuning parameter. Then $f(X) = Z$ has marginal distribution $\mathrm{Bin}(n, p\theta)$; and $g(X) = X - Z$ has conditional distribution as $$p_{\theta}(g(X) = y | f(X) = z) = \frac{(n-z)!}{y!(n-z-y)!}\left[\frac{(1-p) \theta}{1-\theta}\right]^y\left[\frac{1-\theta}{1-p \theta}\right]^{n-z}.$$ Larger $p$ indicates more informative $f(X)$. 
     %(and less informative $g(X)$). \\
     Note that the decomposition of Binomial is not trivially applicable to Bernoulli distribution since $X = 1$ with probability one if $Z = 1$. 
    \item \textbf{Bernoulli (P2).} Suppose $X \sim \mathrm{Ber}(\theta)$. Draw $Z \sim \mathrm{Ber}(p)$ where $p \in (0,1)$ is a tuning parameter. Then
     $f(X) = X(1 - Z) + (1 - X)Z$ has marginal distribution $\mathrm{Ber}(\theta + p - 2p\theta)$; and $g(X) = X$ has conditional distribution (given $f(X)$) as $\mathrm{Ber}\left(\frac{\theta}{\theta + (1-\theta) [p/(1-p)]^{2f(X) - 1}}\right)$.  
    %  Note that modeling $\log\left(\frac{\theta + p - 2p\theta}{1 - \theta - p + 2p\theta}\right)$ as $u\beta$ (for covariates $u$) does not imply $\log\left(\frac{\theta}{1 - \theta}\right)$ is $u\beta'$.
     Smaller $p$ indicates more information in $f(X)$.
    \item \textbf{Categorical (P2).} Suppose $X \sim \mathrm{Cat}\left(\theta_{1},...,\theta_{d}\right)$. Draw $Z \sim \mathrm{Ber}(p)$ where $p \in (0,1)$ is a tuning parameter. Also draw $D \sim \mathrm{Cat}\left(\frac{1}{d},...,\frac{1}{d}\right)$. Let $f(X) = X(1-Z) + DZ$ and $g(X) = X$. Then $f(X) \sim \mathrm{Cat}\left(\phi_{1},...,\phi_{d}\right)$ with $\phi_{i} = (1-p)\theta_{i} + \frac{p}{d}$. Furthermore $g(X) | f(X)$ has distribution 
    \[p_{\theta} \left(g(X) = s | f(X) = t \right) =   
    \begin{cases}
        \frac{(1-p+ \frac{p}{d})\theta_s}{(1-p)\theta_{s} +p/d } & \text{if $s=t$} \\
         \frac{\theta_s \frac{p}{d}}{(1-p)\theta_{t} +p/d } & \text{if $s \ne t$.} \\
    \end{cases}\]
    Larger $p$ indicates more informative $f(X)$. Note that there is nothing special about choosing $\frac{1}{d}$ above, the method generalizes to any vector of probabilities for $D$. When $d=2$, substituting $1-X$ for $D$ in the above construction recovers the Bernoulli decomposition given above.  
    \item \textbf{Poisson} 
    \begin{enumerate}
    \item \textbf{(P1)}   Suppose $X \sim \mathrm{Poi}(\mu)$. Draw $Z \sim \mathrm{Bin}(X, p)$ where $p \in (0,1)$ is a tuning parameter. Then
     $f(X) = Z$ has marginal distribution $\mathrm{Poi}(p\mu)$; and $g(X) = X - Z$ is independent of $f(X)$ and distributes as $\mathrm{Poi}((1 - p)\mu)$. Larger $p$ indicates more informative $f(X)$.
     \item \textbf{(P2)}   Alternatively, draw $Z \sim \mathrm{Poi}(p)$. We can then exploit the thinning property of the Poisson distribution (\cite{last_penrose_2017}) to construct $f(X) = X+Z \sim \mathrm{Poi}(\mu +p)$. Letting $X=g(X)$ gives $g(X) | f(X) \sim \mathrm{Bin}\left(f(X),\frac{\mu}{\mu + p} \right)$. Larger $p$ corresponds to a less informative $f(X)$ (and more informative $g(X)|f(X)$).
    \end{enumerate}

     %(and less informative $g(X)$). 
    
    \item \textbf{Negative Binomial (P2).} Suppose $X \sim \mathrm{NB}(r, \theta)$. Draw $Z \sim \mathrm{Bin}(X, p)$ where $p \in (0,1)$ is a tuning parameter. Then $f(X) = Z$ has marginal distribution $\mathrm{NB}(r, \frac{\theta}{\theta + p - p\theta})$; and $g(X) = X - Z$ has conditional distribution $\mathrm{NB}(r + Z, \theta+p - p \theta)$. Larger $p$ indicates more informative $f(X)$.
    %(and less informative $g(X)$)
    A special case of Negative Binomial is geometric distribution where $r = 1$.

    %All three distributions have a common property that the sum of two independent $X$ would have the same type of distribution.
    
    %\item \textbf{Alternative decomposition for Exponential.} Suppose $X \sim \mathrm{Exp}(\theta)$. Then $f(X) = \lfloor Z_i \rfloor$ follows $\mathrm{Geo}(1 - e^\theta)$; and $g(X) = X - f(X)$ is independent of $f(X)$ with CDF $\mathbb{P}(g(X) < c) = \frac{1 - e^{-\theta c}}{1 - e^{-c}}$. 
    
    %\item \textbf{Unsolved: sub-Gaussian.} Suppose $X$ follows a sub-Gaussian distribution with mean $\mu$. Draw $Z \sim N(X,1)$, then $f(X) = Z$ also has a sub-Gaussian distribution with mean $\mu$. let $g(X) = X$, and the conditional distribution of $g(X) \mid f(X)$ is sub-Gaussian, but its mean value seems to not be a simple function of $\mu$ and $f(X)$. 
\end{itemize}

\begin{remark}[Relationship to infinite divisibility.] Note that we are able to decompose Poisson, Binomial, and Negative Binomial by drawing a Binomial distribution with size $X$, which is a generic formulation. 
    %They corresponds exactly to the class of $(a,b,o)$ class of distribution. 
    Poisson and Negative Binomial belong to the class of discrete compound Poisson distribution (DCP) with parameters $(\lambda \alpha_1, \lambda \alpha_2, \ldots)$, where: 
    \begin{align}
        \mathbb{E}(t^X) \equiv \sum_{x=0}^\infty \mathbb{P}(X = x)t^x  = \exp\left\{\sum_{k=1}^\infty \alpha_k \lambda (t^k - 1)\right\},
    \end{align}
    for $|t| \leq 1$. It is equivalent to saying that $X$ is infinitely divisible. We note that if $X$ follows a DCP with parameter $(\lambda \alpha_1, \lambda \alpha_2, \ldots)$ and we draw $Z \sim \mathrm{Bin}(X, p)$, then $Z$ marginally follows a DCP with parameter $\left(\sum_{k=1}^\infty \alpha_k \lambda k (1 - p)^{k-1} p, \ldots, \sum_{k=i}^\infty \alpha_k \lambda {k \choose i} (1 - p)^{k-i} p^i, \ldots \right)$. However, it may not be true in general that $X\mid Z$ also follows a (shifted) DCP when $X$ is DCP. These examples are simply three cases where the conditional distribution is tractable.
    \end{remark}

\section{Proofs and additional theoretical results} \label{sec:appendix_theory}

\input{decomposition_proofs}

\subsection{Proof of \cref{thm:conjugate_reversal}}
\begin{proof}
Note that because the density $p(z \mid  x)$ must integrate to 1, we can view the function $H( \theta_1,  \theta_2)$ as a normalization factor since
$$H(\theta_{1},\theta_{2}) = \frac{1}{\int_{-\infty}^{\infty} \exp\{\theta_1^T  x -  \theta_2^T  A( x) dx\}} \text{.}$$
Therefore, to compute the marginal density, we have
\begin{align*}
    p( z |  \theta_1, \theta_2,\theta_3) &= \int_{-\infty}^{\infty} h(z) H(\theta_{1},\theta_{2})  \exp\{ (T(z)+\theta_1)^T  x -  ( \theta_2+ \theta_3)^T A( x)\} dx\\
    &= h(z)\frac{H( \theta_1,  \theta_2)}{H( \theta_1 +  T( z),  \theta_2 + \theta_3)}. 
\end{align*} 
Similarly, the computation of the conditional density is straightforward
\begin{align*}
p( x | z, \theta_1, \theta_2, \theta_3) &=\frac{h(z) H(\theta_{1},\theta_{2})  \exp\{ (T(z)+\theta_1)^T  x -  ( \theta_2+ \theta_3)^T A( x)\}}{h(z)\frac{H( \theta_1,  \theta_2)}{H( \theta_1 +  T( z),  \theta_2 + \theta_3)}} \\
&= H(\theta_{1} + T(z),\theta_{2} + \theta_{3})\exp\{ x^T  (\theta_{1} + T(z)) -  (\theta_2 + \theta_3)^T  A(x)\} \\
&= p( x \mid  \theta_1 +  T(z),  \theta_2 +  \theta_3). 
\end{align*}
This completes the proof.
\end{proof}

\subsection{Proof of \cref{thm:normal_regression}}
\begin{proof} This follows as a standard application of OLS properties. First, write
\begin{align*}
	\widehat{\beta}(M) = & (X_{M}^T X_{M})^{-1}X_{M}^T g(Y)  = (X_{M}^T X_{M})^{-1}X_{M}^T [\mu+ \epsilon - \tfrac{1}{\tau}Z ] \\
	= & \beta^{\star}(M) + (X_{M}^T X_{M})^{-1}X_{M}^T [\epsilon - \tfrac{1}{\tau}Z ]. 
\end{align*}
From Section~\ref{sec:list_decomp}, we know that $(\epsilon -\frac{1}{\tau}Z) \sim N(0, (1+\tau^{-2}) \Sigma)$. 
Therefore $$\widehat{\beta}(M) \sim N\left(\beta^{\star}(M),  (1+\tau^{-2})(X_{M}^{T}X_{M})^{-1}X_{M}^{T} \Sigma X_{M} (X_{M}^{T}X_{M})^{-1} \right).$$ The coverage statement then follows straightforwardly from the definition of a CI and the properties of the multivariate Gaussian distribution. 
\end{proof}

\subsection{Proof of \cref{thm:trendfilter_pointwise_method}}
\begin{proof}
	It follows directly from Theorem~\ref{thm:normal_regression} that
	\[\widehat{\beta}(A) \sim  N\left(\beta^{\star}(A),  (1+\tau^{-2})(A^{T}A)^{-1}A^{T} \Sigma A (A^{T}A)^{-1} \right).\]
	To conclude, simply multiply by $a(x_i)^{T}$ and apply standard properties of the Gaussian distribution. 
\end{proof}

\subsection{Extension to unknown variance in Gaussian linear regression} \label{sec:unknown_variance}
In the case where $p$ is fixed, $n \rightarrow \infty$, and $\Sigma = \sigma^{2}I_{n}$ we can extend the argument in Section~\ref{sec:linreg} to accommodate unknown variance. In this setting, if we take $\hat{\sigma}^{2}$ to be the variance estimated from the residuals by fitting the \emph{full} model with all $p$ covariates, it is a consistent estimate of $\sigma^{2}$. See section~5.3 of \cite{selective_inference_asymptotics} for further details.

We can then modify the definition of $f(Y)$ and $g(Y)$ as follows:
\begin{equation} \label{eqn:update_blur_ds}
f(Y) = Y + \hat{\sigma}\tau Z, \text{~  and ~ } g(Y) = Y - \frac{\hat{\sigma}}{\tau}Z. \end{equation}
where $Z\sim N(0,I_{n})$. We can then modify \cref{thm:normal_regression} as follows:

\begin{corollary} 
Let $\widehat{\beta}(M)$ be defined as in~\eqref{eqn:beta-hat} and $\beta^{\star}(M)$ be defined as in~\eqref{eqn:beta-star}. If $f(Y)$ and $g(Y)$ are defined as in~\eqref{eqn:update_blur_ds} and $\Sigma = \sigma^{2}I_{n}$, then the following holds:
\[\widehat{\beta}(M) \overset{d}{\to}  N\left(\beta^{\star}(M),  (1+\tau^{-2}) \sigma^{2} (X_{M}^{T}X_{M})^{-1} \right).\]
Furthermore, if $\hat{\sigma}$ is a consistent estimate of $\sigma$, we can form an asymptotic $1-\alpha$ CI for the $k$th element of $\beta^{\star}(M)$ as
\[ \widehat{\beta}^{k}(M) \pm \hat{\sigma} z_{\alpha/{2}} \sqrt{(1+\tau^{-2}) \left[(X_{M}^{T}X_{M})^{-1}\right]_{kk}}.\]
\end{corollary}
\begin{proof} By assumption, we know $\hat{\sigma}^{2} \overset{p}{\to} \sigma^{2}$.  By Slutsky's theorem, $\frac{Y}{\hat{\sigma}} \overset{d}{\to} N(\frac{\mu}{\sigma},I_{n})$. Draw an independent $Z \sim N(0,I_{n})$ and define $f_{0}(Y) = \frac{Y}{\hat{\sigma}} + \tau Z$ and $g_{0}(Y) = \frac{Y}{\hat{\sigma}} - \frac{1}{\tau}Z$. Applying the continuous mapping theorem to the joint density of $(\frac{Y}{\hat{\sigma}}, Z)$ implies that $$(f_{0}(Y),g_{0}(Y))^{T} \overset{d}{\to} N \left( \frac{\mu}{\sigma}, \begin{pmatrix}
(1+ \tau^{2}) I_{n} & 0  \\
0 & (1 + \frac{1}{\tau^{2}})I_{n}
\end{pmatrix} \right).$$ Applying Slutsky's theorem once more and noting that $f(Y) = \hat{\sigma}f_{0}(Y)$ and $g(Y) = \hat{\sigma}g_{0}(Y)$ gives us
$$(f(Y), g(Y))^{T} \overset{d}{\to} N \left( \mu ,\begin{pmatrix}
\sigma^{2}(1+ \tau^{2}) I_{n} & 0  \\
0 & \sigma^{2}(1 + \frac{1}{\tau^{2}} ) I_{n}
\end{pmatrix} \right).$$ 
Applying the continuous mapping theorem once more to the definition of $\hat{\beta}(M)$ and repeating the arguments in \cref{thm:normal_regression} gives us that \[\widehat{\beta}(M) \overset{d}{\to}  N\left(\beta^{\star}(M),  (1+\tau^{-2}) \sigma^{2} (X_{M}^{T}X_{M})^{-1} \right).\]
\end{proof}
Note that this formulation is only valid in the case when the analyst bases their selection event off of the asymptotic distribution of $f(Y)$ as opposed to its finite sample realizations. Its applicability in real-world data analysis is therefore somewhat heuristic. 

\subsection{Technical exposition of QMLE procedures} \label{sec:appendix_QMLE_details}
Recall the setting described in \cref{sec:qmle}. To recap, after observing $f(Y)$ and $X$ to select a model $M \subseteq [p]$, we conduct inference using a working model of the density for $g(y_{i}) | f(Y), X$ which we denote $p(g(y_{i}) | \beta, f(Y),X_{M})$. Denote the true mean $\mu_{i} = \mathbb{E}[g(y_{i}) | f(Y),X]$ and true variance $\text{Var}\left[g(y_{i}) | f(Y),X \right] = \sigma^{2}_{i}$ with $\mu = (\mu_{1},...,\mu_{n})^{T}$.

The working model implicitly defines a quasi-likelihood function $$L_{n}(\beta) := \sum_{i=1}^{n} \log p(g(y_{i}) | \beta, f(Y), X_{M}).$$ We assume that: (i) the support of $p$ does not depend on $\beta$, (ii) the quasi-likelihood is twice differentiable with respect to $\beta$, and (iii) integration and differentiation with respect to $\beta$ may be interchanged. Under these assumptions, we make note of the corresponding quantities of interest:
$s_{n}(\beta) = \frac{\partial L_{n} }{ \partial \beta} \text{, } H_{n}(\beta) = -\frac{\partial^{2} L_{n} }{ \partial \beta \beta^{T}}  \text{, and } V_{n}(\beta) = \text{Var} \left( s_{n}(\beta) \right).$

Under the assumption that $\mathbb{E}(H_{n}(\beta))$ is positive definite, the target parameter $\beta_{n}^{\star}(M)$ is defined as the root of $E(s_{n}(\beta))$. We define $\hat{\beta}_{n}(M)$ to be an estimator which maximizes the observed score functions $s_{n}(\beta)$.
\begin{definition}[Asymptotic existence]
If $P\left(s_{n}(\hat{\beta}_{n}(M)) = 0 \text{, } H_{n}(\hat{\beta}_{n}(M)) \text{ is p.d.} \right) \rightarrow 1 $ as $n \rightarrow \infty$, then $\hat{\beta}_{n}(M)$ asymptotically exists. 
\end{definition}
Under the assumption that $\hat{\beta}_{n}(M)$ asymptotically exists, we study its convergence to $\beta_{n}^{\star}(M)$. Note that there is no guarantee that $\beta^{\star}_{n}(M)$ itself converges to a single value --- it is a \emph{sequence} that may not converge to anything in general, as it implicitly depends on the sequence of covariates that is observed for each data point. This is a consequence of the fixed-design nature of the regression problem. For random and independent $x_{i}$, we can guarantee $\beta^{\star}(M) := \beta^{\star}_{n}(M)$ will be a single unique value, as described in \cite{huber1967behavior}.

We make the following assumptions about the sequence of target parameters $\beta^{\star}_{n}(M)$.
\begin{itemize}
    \item \textbf{(D)} Divergence: $\lambda_{\text{min}} \{ V_{n}(\beta_{n}^{\star}(M)) \} \rightarrow \infty$.
    \item \textbf{(B)} Bounded from above and below: There exist a $c$, $C$, and $T$ such that
    \begin{enumerate}
        \item $\mathbb{P} \left( \lambda_{\text{min}} \{ V_{n}^{-1/2}(\beta_{n}^{\star}(M)) H_{n}(\beta_{n}^{\star}(M)) V_{n}^{-T/2}(\beta_{n}^{\star}(M)) \} \ge c > 0 \right) \rightarrow 1 $,
        \item $\mathbb{P} \left(  \lambda_{\text{max}} \{ V_{n}^{-1/2}(\beta_{n}^{\star}(M)) H_{n}(\beta_{n}^{\star}(M)) V_{n}^{-T/2}(\beta_{n}^{\star}(M)) \} \le C < \infty\right) \rightarrow 1 $.
    \end{enumerate}
    \item \textbf{(S)} Smoothness: For all $\delta >0$,
    $$ \max_{\beta : \norm{V_{n}^{T/2}(\beta_{n}^{\star}(M))  (\beta - \beta^{\star}_{n}(M)) } \le \delta} \norm{H_{n}^{-1/2}(\beta^{\star}_{n}(M)) H_{n}(\beta) H_{n}^{-T/2}(\beta^{\star}_{n}(M)) - I}_{2} \overset{p}{\rightarrow} 0.$$
    \item \textbf{(N)} Lindeberg-Feller condition for asymptotic normality holds for any $\epsilon > 0$:
    $$\sum_{i=1}^{n} \mathbb{E}[a_{i}^{T} V_{n}^{-1}(\beta_{n}^{\star}(M))a_{i} \mathds{1} \{ a_{i}^{T} V_{n}(\beta_{n}^{\star}(M)) a_{i}\ge \epsilon^{2} \} ] \rightarrow 0,$$
    where $a_{i}= s_{i} (\beta^{\star}_{n}(M)) - s_{i-1} (\beta^{\star}_{n}(M)) $.
\end{itemize}

\begin{theorem}[{\citet[Theorem 4]{fahrmeir_mle}}] \label{thm:glm_clt_est}
Suppose $p(g(y_{i})|\beta,X_{M},f(Y))$ is such that (D), (B1) and (S) hold. If $\beta^{\star}_{n}(M)$ exists, then $\hat{\beta}_{n}(M)$ exists asymptotically and 
$\hat{\beta}_{n}(M) \overset{p}{\rightarrow}   \beta^{\star}_{n}(M).$
If (B2) and (N) also hold, then 
$V_{n}^{-1/2} (\beta^{\star}_{n}(M) )  H_{n} (\beta^{\star}_{n}(M)) \left( \hat{\beta}_{n}(M) - \beta^{\star}_{n}(M) \right) \overset{d}{\rightarrow}  N(0,I).$
\end{theorem}
%\cref{thm:glm_clt_est} shows us that $\hat{\beta}_{n}$ behaves asymptotically like $ N\left(\beta^{\star}_{n}(M), H_{n}^{-1}\left(\beta^{\star}_{n}(M) \right)V_{n}\left(\beta^{\star}_{n}(M) \right) H_{n}^{-1}\left(\beta^{\star}_{n}(M) \right)\right)$. 

In cases of correct specification,  $H_{n}\left(\beta^{\star}_{n}(M) \right) = V_{n}\left(\beta^{\star}_{n}(M) \right)$ which recovers the usual formula for the asymptotic distribution of $\hat{\beta}$. 

All that remains is finding plug-in estimators for $H_{n}$ and $V_{n}$.
To ease notation, let $\hat{H}_{n} := H_{n} (\hat{\beta}_{n}(M))$ which is consistent due to the continuous mapping theorem. There is no known consistent estimator for $V_{n} (\beta^{\star}_{n}(M) )$, but we can find an overestimate by using $\hat{V}_{n} := s_{n}(\hat{\beta}_{n}(M))s_{n}(\hat{\beta}_{n}(M))^{T}$.
This is motivated by the fact that 
\( \mathbb{E}[s_{n}(\beta^{\star}_{n}(M))s_{n}(\beta^{\star}_{n}(M))^{T}] = \text{Var}\left(s_{n}(\beta^{\star}_{n}(M))\right) + \mathbb{E}[s_{n}(\beta^{\star}_{n}(M))] \mathbb{E} [s_{n}(\beta^{\star}_{n}(M))]^{T} .
\)
The last term is positive semidefinite and will generally be small compared to the first term, since the score function should be small if the working model is chosen well. This results in asymptotically \emph{conservative} confidence intervals when used as a plug-in estimator. 
\begin{theorem} \label{thm:var_emp_fahrmeir}
Assume the conditions of \cref{thm:glm_clt_est} hold. We can then form an asymptotically conservative $1- \alpha$ CI for the $k$th element of $\beta_{n}^{\star}(M)$ as 
$ [\hat{\beta}_{n}(M)]_{k} \pm z_{\alpha/2} \sqrt{ [\hat{H}_{n}^{-1}\hat{V}_{n}\hat{H}_{n}^{-1}]_{kk}  }$.
\end{theorem}

\begin{proof}\cref{thm:glm_clt_est} demonstrates that $\hat{\beta}_{n}(M) \overset{p}{\rightarrow}  \beta^{\star}_{n}(M)$. Due to the assumption that the quasi-likelihood function is twice continuously differentiable, $H_{n}$ and $s_{n}$ are both continuous functions with respect to $\beta$, so by the continuous mapping theorem, $\hat{H}_{n} \overset{p}{\rightarrow}  H_{n}(\beta^{\star}_{n}(M))$ and $s_{n}(\hat{\beta}_{n}(M))  \overset{p}{\rightarrow}  s_{n}(\beta^{\star}_{n}(M))$. In the rest of the argument, we drop the explicit dependence on the selected model ($M$) to avoid clunky notation and therefore let $\hat{\beta}_{n} := \hat{\beta}_{n}(M)$ an $\beta^{\star}_{n} := \beta^{\star}_{n}(M)$ in the arguments that follow. Now, note that
\begin{align*}
    \mathbb{E}[s_{n}(\beta^{\star}_{n})s_{n}(\beta^{\star}_{n})^{T}] &= \text{Var}\left(s_{n}(\beta^{\star}_{n})\right) + \mathbb{E}[s_{n}(\beta^{\star}_{n})] \mathbb{E} [s_{n}(\beta^{\star}_{n})]]^{T} 
\end{align*} 
where the last term is positive semidefinite since it is a Gram matrix. Recall that if two matrices $A,B$ are positive semidefinite then $ABA$ is positive semidefnite and $A^{-1}$ is positive semidefinite when $A$ is invertible. Since $H_{n}\left(\beta^{\star}_{n}\right)$ is also positive semidefinite by assumption, we know that 
$$H^{-1}_{n}\left(\beta^{\star}_{n}\right) \left( \mathbb{E}[s_{n}(\beta^{\star}_{n})s_{n}(\beta^{\star}_{n})^{T}] -V_{n} \left(\beta^{\star}_{n} \right)\right) H^{-1}_{n}\left(\beta^{\star}_{n}\right) $$ 
is positive semidefnite.

Denote the Cholesky decomposition of $H^{-1}_{n}\left(\beta^{\star}_{n}\right) \mathbb{E}[s_{n}(\beta^{\star}_{n})s_{n}(\beta^{\star}_{n})^{T}]  H^{-1}_{n}\left(\beta^{\star}_{n}\right) $ as $LL^{T}$  and the Cholesky decomposition of $H^{-1}_{n}\left(\beta^{\star}_{n}\right) V_{n} \left(\beta^{\star}_{n} \right) H^{-1}_{n}\left(\beta^{\star}_{n}\right) $ as $MM^{T}$. From the above, we know that $LL^{T} - MM^{T}$ is positive semidefinite which implies that
$$L_{kk} - M_{kk} =\sqrt{ [H^{-1}_{n}\left(\beta^{\star}_{n}\right) \mathbb{E}[s_{n}(\beta^{\star}_{n})s_{n}(\beta^{\star}_{n})^{T}]  H^{-1}_{n}\left(\beta^{\star}_{n}\right)]_{kk} } - \sqrt{ [H^{-1}_{n}\left(\beta^{\star}_{n}\right) V_{n} \left(\beta^{\star}_{n} \right) H^{-1}_{n}\left(\beta^{\star}_{n}\right)]_{kk}}$$
must be greater than $0$ since the diagonal entries must be positive. 

By the WLLN,  $\hat{V}_{n}  \overset{p}{\rightarrow} \mathbb{E}[s_{n}(\beta^{\star}_{n})s_{n}(\beta^{\star}_{n})^{T}] $ which combined with the above statement shows us 
\begin{equation}
    \label{eqn:lim_gt_plugin}
    \lim_{n \rightarrow \infty} P\left( \sqrt{ [\hat{H}_{n}^{-1}\hat{V}_{n}\hat{H}_{n}^{-1}]_{kk} } \ge  \sqrt{ [H_{n}^{-1}\left(\beta^{\star}_{n} \right) V_{n} \left(\beta^{\star}_{n} \right)H_{n}^{-1}\left(\beta^{\star}_{n} \right)]_{kk}}\right) = 1.
\end{equation}
Putting everything together, we have that
\begin{align*}
 1- \alpha &= \lim_{n \rightarrow \infty} P \left(  [\beta_{n}^{\star}]_{k} \in \left[ [\hat{\beta}_{n}]_{k} \pm z_{\alpha/2} M_{kk} \right] \right) \\
 &\le \lim_{n \rightarrow \infty} P \left(  [\beta_{n}^{\star}]_{k} \in \left[ [\hat{\beta}_{n}]_{k} \pm z_{\alpha/2}\sqrt{ [\hat{H}_{n}^{-1}\hat{V}_{n}^{1}\hat{H}_{n}^{-1}]_{kk} } \right]\right),
\end{align*}
where the first line follows from \cref{thm:glm_clt_est} and the second line follows from~\eqref{eqn:lim_gt_plugin}. 

\end{proof}
We now demonstrate how to apply \cref{thm:var_emp_fahrmeir} in the context of generalized linear models which we then use for the simulations discussed in \cref{sec:qmle} and Appendices~ \ref{sec:appendix_poisson}~and~\ref{sec:appendix_logistic}. We restrict attention to working models in the \emph{exponential dispersion} family, i.e., with a density
$p\left( g(y_{i}) | f(Y), X \right) = \exp \left( \frac{y_{i} \theta_{i} - b(\theta_{i})}{a(\phi)} + c(g(y_{i}), \phi) \right)$
for some functions $a(\cdot), b(\cdot), c(\cdot,\cdot)$. Denote $m = (m_{1},...,m_{n})^{T}$ to be the mean under the assumption of the working model and $v = (v_{1},...,v_{n})^{T}$ to be the variance under the assumption of the working model. The covariates are linked to the (assumed) random component through a \emph{link function} $h$ where $ \eta_{i} := \beta^{T} \widetilde{x}_{i}$ and $\eta_i = h(m_{i})$. The equations defined in the preceding section simplify to 
\begin{align*}
s_{n}(\beta) &= X_{M}^{T} D V^{-1} \left( g(Y) - m \right),\\
H_{n}(\beta^{\star}_{n}(M))  &= \sum_{i=1}^{n} \widetilde{x}_{i} \left(\frac{\partial m_i}{ \partial \eta_i }\right)^{2} v_{i}^{-1} \widetilde{x}_{i}^{T}  = X_{M}^{T} D^{2} V^{-1} X_{M},\\
V_{n}(\beta^{\star}_{n}(M)) &= \sum_{i=1}^{n}\widetilde{x}_{i} \left( \frac{\partial m_{i}}{ \partial \eta_{i}} \right)^{2} \frac{\sigma^{2}_{i}}{v_{i}^{2}}\widetilde{x}_{i}^{T} = X_{M}^{T} D^{2} \Sigma V^{-2}X_{M} ,
\end{align*}
where $D,V,\Sigma$ are diagonal matrices with $D_{ii} = \frac{\partial m_{i}}{ \partial \eta_{i}}$,  $V_{ii} = v_{i}$, and $\Sigma_{ii} = \sigma^{2}_{i}$. Furthermore, we can rewrite the plug-in estimator described above as $$\hat{V}_{n} := s_{n}(\hat{\beta}_{n}(M))s_{n}(\hat{\beta}_{n}(M))^{T} = \sum_{i=1}^{n}\widetilde{x}_{i}\left( g(y_{i}) - m_{i}\right)^{2} v_{i}^{-2} \left(\frac{\partial m_{i}}{ \partial \eta_{i}}\right)^{2} \widetilde{x}_{i}^{T}.$$ 
We now apply these equations to three prototypical examples of GLMs (Gaussian regression, Poisson regression, logistic regression).
\paragraph{Example 1: Gaussian linear regression} Assume that we have split $Y$ into two pieces $f(Y)$ and $g(Y)$ using one of the methodologies in \cref{sec:list_decomp} but the post-selective likelihood function cannot be conveniently modeled. As a simplifying assumption, an analyst decides to model $g(Y)$ using linear regression with a selected model $M \subseteq [p]$ and an assumption of Gaussian errors but wants to conduct inference in a way that is robust to misspecification. If the analyst models the data as homoscedastic with common variance $\sigma_{0}^{2}$, then the score equation simplifies as
\(s_{n}(\beta) = X_{M}^{T}  \left( g(Y) - X_{M} \beta \right),\)
leading to the standard estimator of $\hat{\beta}_{n}(M) = (X_{M}^{T}X_{M})^{-1}X_{M}^{T}g(Y)$. We have our plug-in variance estimator:
\begin{align*}
    \hat{H}_{n}^{-1}\hat{V}_{n} \hat{H}_{n}^{-1} &= \left( \sigma_{0}^{2} X_M^{T}X_{M} \right)^{-1}
\left(X_M^{T}\text{diag}\left(\frac{g(Y) - X_{M}\hat{\beta}}{\sigma_{0}^{4}} \right) X_{M} \right) \left(\sigma_{0}^{2} X_M^{T} X_{M} \right)^{-1} \\
&= (X_{M}^{T}X_{M})^{-1} X_{M}^{T} \text{diag}\left(g(Y) - X_{M}\hat{\beta}\right) X_{M} (X_{M}^{T}X_{M})^{-1},
\end{align*}
which is exactly the sandwich estimator of \cite{huber1967behavior}. However, in this context, the estimator is an overestimate for $\beta^{\star}_{n}(M)$ and leads to conservative rather than exact CIs. 
\paragraph{Example 2: Poisson regression} Considering $y_{i} \sim \mathrm{Pois}(\mu)$, we fission the data as in \cref{sec:list_decomp} with $f(Y_{i}) \sim \text{Bin}(Y_{i},p)$ with $g(y_{i})  = y_{i} - f(y_{i})$. We then attempt to model $g(y_{i})$ using a GLM with log link function using selected covariates $x_{i \cdot M}$ and an offset of $\log(1-p)$ to account for the randomization. 

We assume that the true distribution of the data is indeed Poisson but the mean $\mathbb{E}(y_{i} |  x_{i})$ may not be a linear function of the chosen covariates $\tilde{x}_{i}$. The constructed CIs will then cover the target parameters $\beta^{\star}_{n}(M)$ that minimizes the KL divergence between the true and the modeled distribution:
\begin{align*}
    \beta^{\star}_{n}(M) := \argmin_\beta \mathrm{D}_{KL}\left(\prod_{i=1}^n q \left( (1-p) \mu_{i} \right) || \prod_{i=1}^n q\left(\exp\{ \log(1-p) + \beta^T \tilde{x}_{i}\}\right)\right),
\end{align*}
where $q(\cdot)$ is the distribution of Poisson distribution as a function of its mean.

\iffalse
The projected $\beta^{\star}_{n}(M)$ is equivalent to the solution found that maximizes the expected quasi-likelihood function:
\begin{align*}
    \beta^{\star}_{n}(M) \equiv \argmax_\beta \sum_{i=1}^n \mathbb{E}\left[\beta^T \tilde{x}_{i} y_i - \log(y_i!) - \exp\{\beta^T \tilde{x}_{i}\}\right],
\end{align*}
and the solution satisfies:
\begin{align*}
    \sum_{i=1}^n x_{ik}  \mu_{i}  = \sum_{i=1}^n x_{ik} \exp\{ \beta_{n}^{\star}(M)^T \tilde{x}_{i}\},
\end{align*}
for all selected features $k \in M$.
\fi

For this problem, we have that $\frac{\partial m_{i}}{ \partial \eta_{i}} = m_{i} = v_{i} = \exp\left( \beta^{T} \tilde{x}_{i}\right)$. This leads to the plug-in estimator for variance to be 
\begin{align} \label{eqn:plugin_poisson}
    \hat{H}_{n}^{-1}\hat{V}_{n} \hat{H}_{n}^{-1} &= \left(X_M^{T} \hat{V} X_{M} \right)^{-1}
\left(X_M^{T}\hat{D}X_{M} \right) \left(X_M^{T} \hat{V} X_{M} \right)^{-1},
\end{align}
where $\hat{D}$ and $\hat{V}$ are diagonal matrices with $\hat{D}_{ii} = g(y_{i}) - (1-p)\exp \left( \hat{\beta}^{T}\tilde{x}_{i}\right)$ and $\hat{V}_{ii} = (1-p)\exp\left( \hat{\beta}^{T} \tilde{x}_{i}\right)$. 

\paragraph{Example 3: Logistic regression}
Assume that $y_{i} \sim \mathrm{Ber}(\mu_{i})$ and we fission the data into $f(y_{i}) \sim \mathrm{Ber}(\mu_i + p - 2p\mu_i)$ and $g(y_{i}) | f(y_{i}) \sim \mathrm{Ber}\left(\frac{\mu_i}{\mu_i + (1-\mu_i) [p/(1-p)]^{2f(Y_i) - 1}}\right)$ as described in \cref{sec:list_decomp}. Modelling this likelihood is challenging because the function is non-convex, so it may be easier to try and use a logistic regression with selected covariates $\widetilde{x}_{i}$ to model this probability even though the post-selective distribution does not have this functional form. 
%Using masking for Bernoulli distribution, we draw $Z \sim \mathrm{Ber}(p)$ where the ``flip probability'' $p \in (0,1)$ is a tuning parameter (where $p \neq 0.5$; otherwise we cannot recover the original distribution). Suppose $Y \sim \mathrm{Ber}(\theta)$, then $f(Y) = Y(1 - Z) + (1 - Y)Z$ has marginal distribution $\mathrm{Ber}(\theta + p - 2p\theta)$; and $g(Y) = Y$ has conditional distribution as $\mathrm{Ber}\left(\frac{\theta}{\theta + (1-\theta) [p/(1-p)]^{2f(Y) - 1}}\right)$ given $f(Y)$. To construct CIs for potentially significant features, we can perform a two-step algorithm:
As before, we can still interpret the fitted parameters from a logistic regression as the projection onto the working model. 
\iffalse
The target parameter $\beta^{\star}_{n}(M)$ satisfies:
\begin{align*}
    \sum_{i=1}^n x_{ik} \mathbb{E}[g(y_i) \mid X, f(Y)] = \sum_{i=1}^n x_{ik} \left(1 + \exp\{\beta^{\star}_{n}(M)^{T} x_{ik}\}  \right)^{-1},
\end{align*}
for all selected features $k$. 
\fi 
For this problem, we have that $\frac{\partial m_{i}}{ \partial \eta_{i}} = m_{i}(1-m_{i}) = v_{i}$. This leads to the plug-in estimator for variance to be the same as in~\eqref{eqn:plugin_poisson} except with $\hat{D}_{ii} =g(y_{i}) - \hat{m}_{i} $ and $\hat{V}_{ii} = \hat{m}_{i}(1-\hat{m}_{i})$ with $\hat{m}_{i} = \frac{1}{1 + \exp \left( - \hat{\beta}^{T} \tilde{x}_{i}\right)}$.

\section{Additional comparisons between data splitting and data fission} \label{sec:appendix_splitting_fission}
We can give an analogous statement about how this data fission procedure relates to data splitting for fixed-design linear regression as the one made in \cref{sec:split_fission}. For data splitting,  we pick $a \in \{\frac{1}{n},..., \frac{n-1}{n},1\} $ and set aside the first $an$ observations for selection and the remaining $(1-a)n$ observations for inference. We then select some model $M \subseteq [p]$ during the selection stage. This restricts us to a smaller model design matrix $\XinfM \in \R^{(1-a)n \times |M| }$ for inference. Also denote $\Sigma_{\text{inf}}$ the covariance matrix for the error term of these $(1-a)n$ observations. Then,
%\[ \sqrt{n} \left( \widehat{\beta}_{\mathrm{split}} - \beta \right) \rightarrow N\left(0,(X_{M,\text{inf}}^{T} X_{M,\text{inf}})^{-1}X_{M,\text{inf}}^{T} \Sigma_{\text{inf}} X_{M,\text{inf}} (X_{M,\text{inf}}^{T} X_{M,\text{inf}})^{-1} \right).\]
$$ \widehat{\beta}_{\mathrm{split}} \sim N\left(\beta^{*}(M) ,(\XinfMT  \XinfM)^{-1}\XinfMT \Sigma_{\text{inf}} \XinfM (\XinfMT  \XinfM)^{-1} \right).$$
For data fission, let us denote the model chosen during the selection stage as  $M^{I} \subseteq [p]$ and the corresponding model design matrix $X_{M^{I}} \in \R^{n \times |M^{I}|}$. Then, 
$$\widehat{\beta}_{\mathrm{fission}} \sim N\left(\beta^{*}(M^{I}),(1+\tau^{-2})(\XinfMstar^{T} \XinfMstar)^{-1}\XinfMstar^{T} \Sigma \XinfMstar (\XinfMstar^{T} \XinfMstar)^{-1} \right).$$
%\[ \sqrt{n} \left( \widehat{\beta}_{\mathrm{fission}} - \beta \right) \rightarrow N\left(0,(1+\tau^{-2})(\XinfMstar^{T} \XinfMstar)^{-1}\XinfMstar^{T} \Sigma \XinfMstar (\XinfMstar^{T} \XinfMstar)^{-1} \right).\]
Both quantities are unbiased, but determining whether $\widehat{\beta}_{\mathrm{fission}}$ or $\widehat{\beta}_{\mathrm{split}}$ has lower variance depends on too many unknown quantities to say definitively without additional assumptions. First, there is no guarantee that $M$ and $M^{I}$ will be the same --- i.e. different models may be selected during the first stage for each of these two procedures. Second, nothing is assumed about the  design matrix $X$ so it is hard to say much about how $(X^{T}X)^{-1}$ changes if restricted to only a subset of rows/columns. However, heuristically, the estimator increases variance when compared with data splitting via randomization through the $1 + \tau^{-2}$ term but decreases variance by increasing the number of rows within the model design matrix. We use an idealized example below to draw more parallels. 
\begin{example} \label{example1} Assume that we have
% \begin{enumerate}
% 	\item 
	orthogonal covariates ($x_{i}^Tx_{j} =0$ for all $i \ne j$)
% 	\item 
	and known homoscedastic noise (for all $i$, $\sigma_{i}^{2}=\sigma^{2}$, for some known $\sigma^2$).
% 	\item 
	Also assume that both procedures select the same model $M = M^I$.
% 	\item As in the above, $\Sigma$ is known.
% \end{enumerate}
%These assumptions simplify things because we now have that 
Then $\widehat{\beta}_{\mathrm{fission}} 
\sim N(\beta^{*}(M), (1+\tau^{-2} )  \sigma^{2} (X_{M}^{T}X_{M})  ^{-1})$ and 
$\widehat{\beta}_{\mathrm{split}} 
\sim N(\beta^{*}(M),  \sigma^{2} (\XinfMT \XinfM)  ^{-1}).$
Moreover the matrices are easily invertible since $X^{T}X$ is a diagonal matrix. Therefore, we have:
\[ \Var(\beta_{j}^{\mathrm{fission}})= \frac{(1+\tau^{-2})\sigma^{2}}{\sum_{i=1}^{n} X_{ij}^{2}}, \text{ ~ and ~ } \Var(\beta_{j}^{\mathrm{split}})= \frac{\sigma^{2}}{\sum_{i=an+1}^{n} X_{ij}^{2}}.\]
For a fixed $a$, the parameter $\tau$ that equates these two variances is $\tau = \left( \frac{\sum_{i=1}^{an} X_{ij}^{2}}{\sum_{i=an+1}^{n} X_{ij}^{2}} \right)^{1/2}.$
\end{example}
Although idealized, this example draws attention to a key weakness that data splitting has in fixed-design linear regression. Data fission is always able to smoothly tradeoff the information between the selection and inference datasets through the parameter $\tau$ but data splitting is limited in the ways that information can be divided since it requires that the analyst allocate points discretely. 

\section{Background on trend filtering} \label{sec:background_trendfilter}
We include a summary of trend filtering when the order $k>1$. As before, we consider the problem of estimating the underlying smooth trend of a time series $y_t \in \R$ with $t = 1, \ldots, n$ with structural equation,
\begin{equation} \label{eqn:trendfilter}
    y_{t} = f_{0}(t) + \epsilon_{t}.
\end{equation}
Our goal is to estimate the underlying trend $\left(f_{0}(1),\ldots ,f_{0}(n) \right)$. The approach of trend filtering is to fit a piecewise polynomial of degree $k$ to the data with adaptively chosen breakpoints or \emph{knots}. Formally, the $k$-th order trend filtering estimator is defined to be $\widehat{x} = \left( \widehat{x}_{1}, \ldots , \widehat{x}_{n} \right)$, which is the solution to the minimization problem
\begin{equation} \label{eqn:trend_optimizer}
\widehat{x} =  \argmin_{x \in \R^{n}} \frac{1}{2} \norm{Y - x}_{2}^{2} + \lambda \norm{D^{(k+1)}x}_{1}\text{,}
\end{equation}
where $\lambda \ge 0$ is a tuning parameter and $k$ is the order of the piecewise polynomial that is being chosen to fit the data. $D^{(k+1)} \in R^{(n-k)\times n}$ is the $k$-th order difference matrix defined recursively by defining 
\[D^{(1)} = \begin{bmatrix}
	-1 & 1 & 0 & \dots & 0 & 0 \\
	0 & -1 & 1 & \dots & 0 & 0 \\
	\vdots & \\
	0 & 0 & 0 & \dots & -1 & 1\\
\end{bmatrix} \in \R^{(n-k-1)},\]
and $D^{(k+1)} \in \R^{(n-k) \times n}$ as 
\[D^{(k+1)} = D^{(1)}D^{(k)}. \]
From this definition, we can see that setting $k=0$ yields $D^{(k+1)} = D^{(1)}$ which corresponds to fitting a piecewise constant function. This is also known as the 1-dimensional fused lasso problem of \cite{fused_lasso}. Choosing $k=1$ corresponds to fitting a piecewise linear function across $t$, choosing $k=2$ corresponds to fitting a piecewise quadratic function across $t$, and so on. 

Although trend filtering is a relatively recently developed tool in nonparametric statistics, it has gained substantial popularity over the last several years due to it converging to the true underlying function at the minimax rate (see \cite{adaptive_piecewise_tibshirani}) in addition to being computationally efficient. The specialized alternation direction method of multipliers (ADMM) algorithm  of \cite{Ramdas2014FastAF} converges at the $\mathcal{O}(n^{1.5})$ rate in the worst case but in practice tends to scale extremely well for large datasets. Most other nonparametric estimators that can be computed efficiently such as smoothing splines and uniform-knot regression splines are not locally adaptive and therefore do not converge at the minimax rate. Other methods that are locally adaptive such as the locally adaptive regression spline of \cite{locally_adaptive_regression_splines} can be shown to converge at the minimax rate but are not computationally tractable. For a more complete overview of these competing methods and the various ways in which they may tradeoff between their theoretical properties, the interested reader may wish to consult \cite{10.1093/mnras/staa106}.

\section{Supplemental simulation results} 

\subsection{Simulation results for interactive hypothesis testing}\label{sec:appendix_interactive_hyp_testing}
In this section, we provide a more detailed description of the simulation setup in \cref{sec:interactive_testing}, additional commentary on the simulation results, and additional demonstrations for non-Gaussian datasets.  

\paragraph{Simulation setup} The covariates $x_{i} \in \R^{2}$ are arranged on a $50 \times 50$ grid in the area $[-100,100] \times [-100,100]$. We let the set of non-nulls be arranged in a circle in the center of the grid as shown in \cref{fig:STAR_experiments_normal} and set $\mu_i = 2$ for each non-null and $\mu_i = 0$ for each true null. The data is generated as $y_{i} \sim N(\mu_{i}, 1)$. We form a rejection set using the BH, AdaPT, and STAR procedures, both with and without data fission. For the non-fissioned versions of these procedures using the full dataset, the (one-sided) p-values that are used are calculated as $p_{i}^{\mathrm{full}} = 1-\Phi(y_{i})$. Similarly, for the fissioned versions of these procedures, the (one-sided) p-values are calculated as $p_{i}^{\mathrm{fission}} = 1-\Phi(f(y_{i}) (1+\tau^{2})^{-\frac{1}{2}}  )$. The AdaPT and STAR procedures are designed so that the analyst may explicitly use the covariate information $x_{i}$ in forming their rejection set, with STAR explicitly imposing a structural constraint that the rejection region needs to be spatially convex (with respect to its covariates $x_{i}$). Note that although AdaPT and STAR control the false discovery rate regardless of how the analyst chooses to reject their hypotheses with this information, for the purposes of this simulation, we rely on the algorithmic procedures described in Section~4.2 of \cite{lei2017star} and Section~4 of \cite{adapt_fithian} to form the rejection set. Since the BH procedure is not designed to use side information to form a rejection set, its power is significantly lower than the AdaPT and STAR procedures for this simulation.

An example of a rejection set for each of the six methodologies for a single trial run can be seen in \cref{fig:STAR_experiments_normal}. We note that the increased variance introduced through fission will necessarily decrease the power of these procedures in the selection phase, but the level of degradation is minor for small values of $\tau$. After forming the rejection sets for the fissioned versions of the datasets, we form CIs to cover $\overline{\mu}$ using the above outlined procedure. As a point of comparison, we also compare this methodology to the invalid procedure of ``double dipping'' to form CIs on the non-fissioned dataset. 
%We again reiterate that the BY-corrected CIs designed to cover $\mu_{i}$ only have known theoretical coverage guarantees in the case of the BH procedure. Nonetheless, we compute them for the AdaPT and STAR procedures as well in order to have some baseline of comparison for the fissioned versions of these procedures.

\begin{figure} 
\begin{center}
\begin{subfigure}[t]{0.28\textwidth}
\raisebox{12mm}{\includegraphics[width=\linewidth]{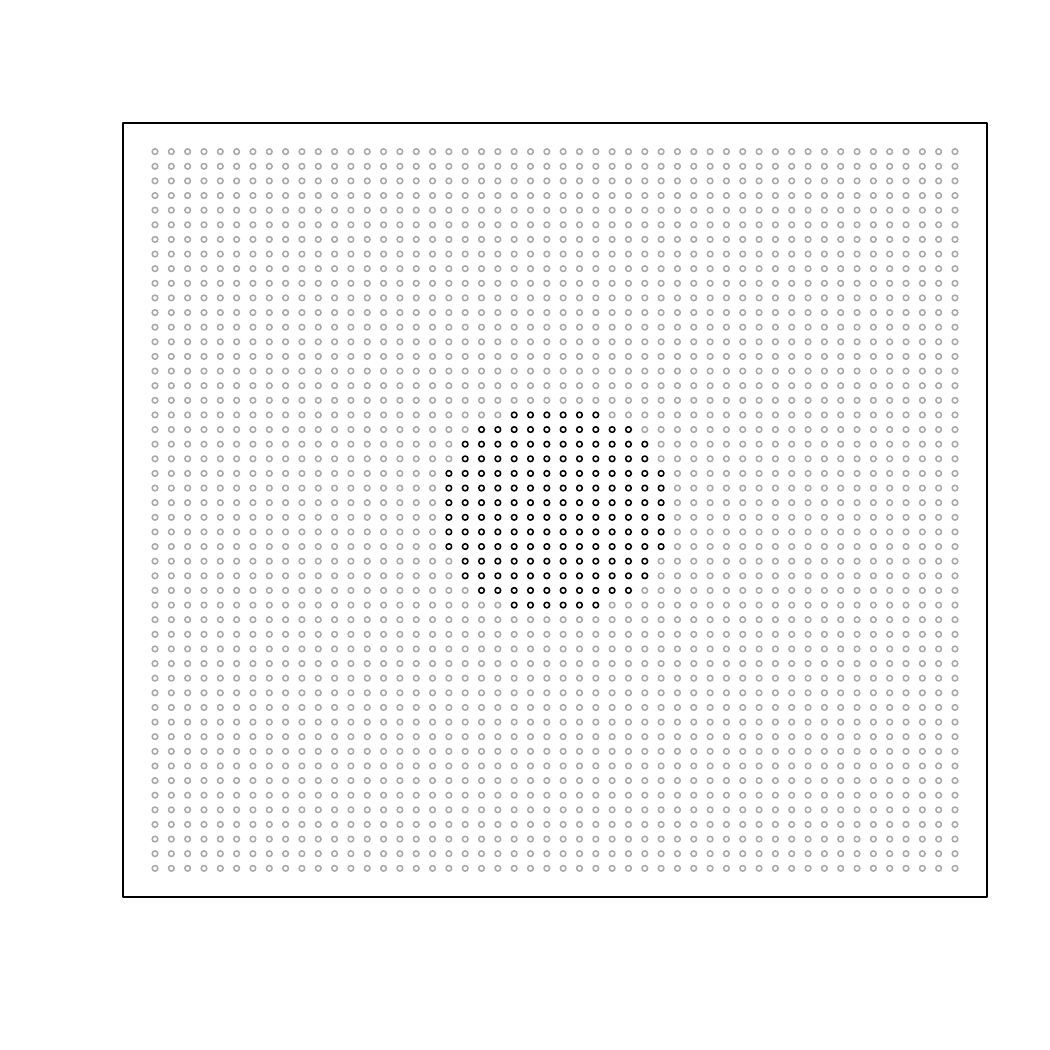}}
\end{subfigure}
\begin{subfigure}[t]{0.7\textwidth}
\includegraphics[width=\linewidth]{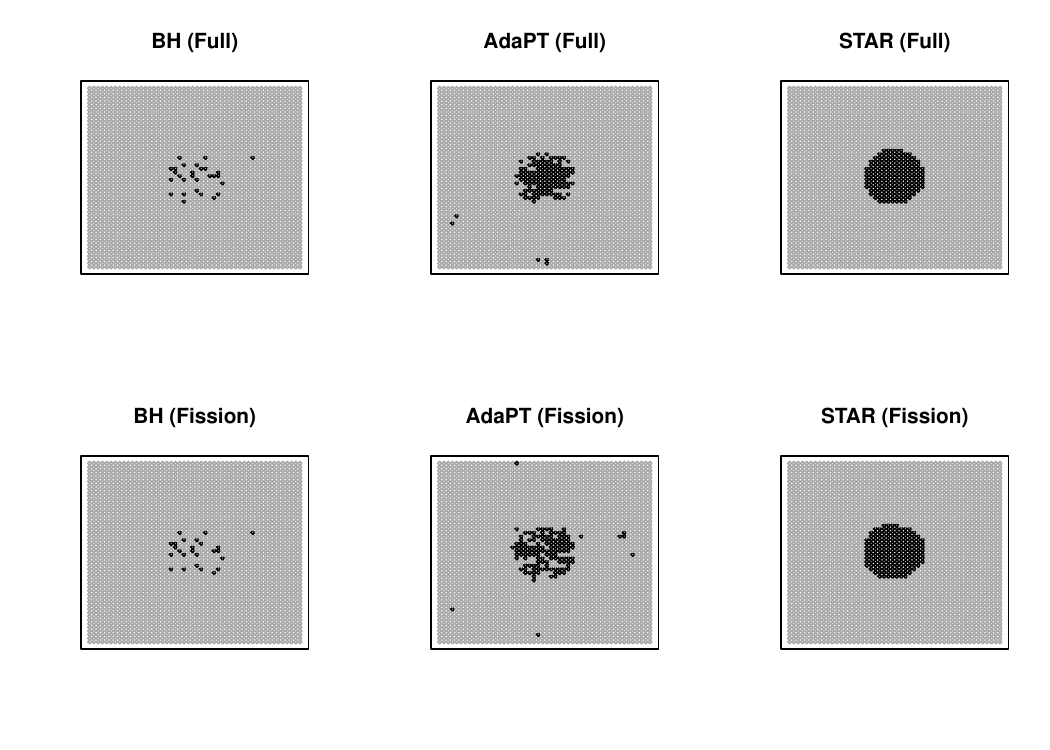}
\end{subfigure}
\caption{Example rejection region for a single trial run. STAR, AdaPT, and BH procedures are used to try and recover the true signal pattern (left) with a target FDR of $0.1$. Fissioning the data (with $\tau=0.1$) to allow for the estimation of signal strength only mildly impairs the ability of these procedures to correctly identify the rejection region.}
\label{STAR_experiments_singletrial}
\end{center} 
\end{figure}

\iffalse
\begin{figure} 
\begin{center}
\includegraphics[width=0.4\linewidth]{figures/true_convex.pdf}
\caption{True underlying region of non-nulls for hypothesis detection simulation. Each point is a signal, with black representing the non-nulls with $\mu_{i} = 2$. Grid size for simulations is $25 \times 25$. The region for detection is convex, making this problem most naturally suited for the STAR algorithm. However, we also compare performance for BH and AdaPT procedures as a point of comparison.}
\end{center} 
\end{figure}
\fi

We repeat this simulation over $250$ trials for a variety of different $\tau$ to observe how information trades off between the selection and inference steps. To measure the effectiveness of our procedure, we track the following metrics for the \emph{selection} stage:
\begin{align*}
    \text{Power}:= \frac{|x_i \in \mathcal{R}: \mu_{i} \ne 0|}{|\{i: \mu_{i} \ne 0 \}|}, \text{ ~ and ~ }
% \end{align}
% \begin{align}
    \text{false discovery proportion } \mathrm{(FDP)}:= \frac{|x_i \in \mathcal{R}: \mu_{i} = 0|}{\max \{|\mathcal{R}|,1\}}.
\end{align*}
For the \emph{inference} stage, we denote the CI created through the data fission procedures that covers $\overline{\mu}$ as $\overline{\text{CI}}$ with lower and upper bounds $\overline{\text{CI}}(1), \overline{\text{CI}}(2)$ and compute the following metrics:
\[\text{Miscoverage} = \mathds{1}\left(\overline{\mu} \not \in \overline{\text{CI}}\right), \text{ ~ and ~ }
\text{CI Length} =  |\overline{\text{CI}}(2) - \overline{\text{CI}}(1)|. \]

%We denote the CIs created through the non-fissioned procedures that cover the individual $\mu_{i}$ as $\text{CI}_{i}$ with lower and upper bounds $\text{CI}_{i}(1)$ and $\text{CI}_{i}(2)$ and compute the following additional metrics \[\mathrm{false coverage rate (FCR)} = \frac{|x_i \in \mathcal{R}: \mu_{i} \not \in CI_{i}|}{\max \{|\mathcal{R}|,1\}}, \text{ ~ and ~ }\text{CI Length} = \frac{1}{\mathcal{R}} \sum_{i \in \mathcal{R}} |\text{CI}_i(2) - \text{CI}_i(1)|. \]

The results are shown in \cref{fig:STAR_experiments_normal} with target coverage level $1- \alpha = 0.8$ and target FDR control at $0.2$. {\color{black} As $\tau$ increases, we see a tradeoff between the power in the selection stage and the CI length in the inference stage. Despite having this general trend, the decrease in CI length is not strictly monotonic with $\tau$ in all cases (e.g. the fissioned BH procedure has CI lengths which actually increase from $\tau \approx 0.5$ to $\tau \approx 1$) because the length of the CIs increases with $\tau$ but decreases with $|\mathcal{R}|$. Since the size of the rejection set increases as power increases, moving $\tau$ causes these two effects to work in opposite directions. Nonetheless, increasing $\tau$ tends to decrease CI length in most instances.}

\paragraph{Poisson Data} Note that in the preceding examples, Gaussian distributed data were used but this procedure can be modified straightforwardly for data that is distributed in any way such that a decomposition rule (e.g. from \cref{sec:appendix_list_decomp}) can be applied. Consider a situation where we now observe $y_{i} \sim \text{Pois}(\mu_{i})$ and during the selection stage, we are attempting to form a rejection set against the null hypothesis $H_{0}: \mu_{i} = 1$ compared to the alternative $H_{1}: \mu_{i} > 1$.

One technical note is that in this case, if we calculate the (one-sided) p-value using the straightforward calculation of $p_{i} = P(Y \le y_{i}) : Y \sim \text{Pois}(1)$, the technique of p-value masking will fail to guarantee error control because $p_{i}$ can be reconstructed from $\text{min}\left(p_{i}, 1-p_{i} \right)$ when the distribution is discrete. To account for this, we add an additional layer of randomization when calculating the $p$-value using \cite{BROCKWELL20071473}. Specifically, we draw $U \sim \text{Unif}(0,1)$ and calculate
$$Y' = F(Y)U + (1-U)F(Y-),$$
where $F(Y)$ denotes the CDF of a $\text{Pois}(1)$ variable and $F(Y-)$ denotes the corresponding left limit. This ensures that $Y'$ is uniformly distributed and stochastically dominated by $F(Y)$ and is therefore a valid $p$-value.

To form CIs at the inference stage, we can manually invert the test statistics but a simpler method described in \cite{ulm_smr} provides a convenient shortcut. In particular, the $1-\alpha$ CI around $\widebar{\mu}$ for the CIs constructed from the fissioned data can be calculated as
$$\frac{1}{2|\mathcal{R}|(1-\tau)} \left(\chi^{2}_{\alpha/2} \left(2\sum_{i} y_{i}\right), \chi^{2}_{1-\alpha/2} \left(2\sum_{i} y_{i} +2\right)\right).$$
Empirical results are shown in \cref{fig:STAR_experiments_poisson}. Counterintuitively, CI length also tends to decrease with $\tau$ up to $\tau \approx 0.6$ , despite the fact that increasing $\tau$ reserves less information for inference. This happens when the effect that increasing $|\mathcal{R}|$ has outweighs the increasing variance from the individual signals through the $1- \tau$ term.
\begin{figure} 
\begin{center}
\begin{subfigure}[t]{1\textwidth}
\centering
\includegraphics[width=0.23\linewidth]{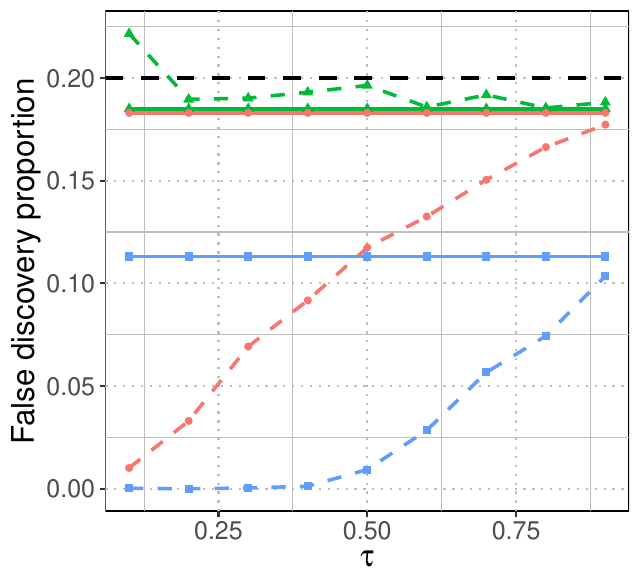}
\includegraphics[width=0.23\linewidth]{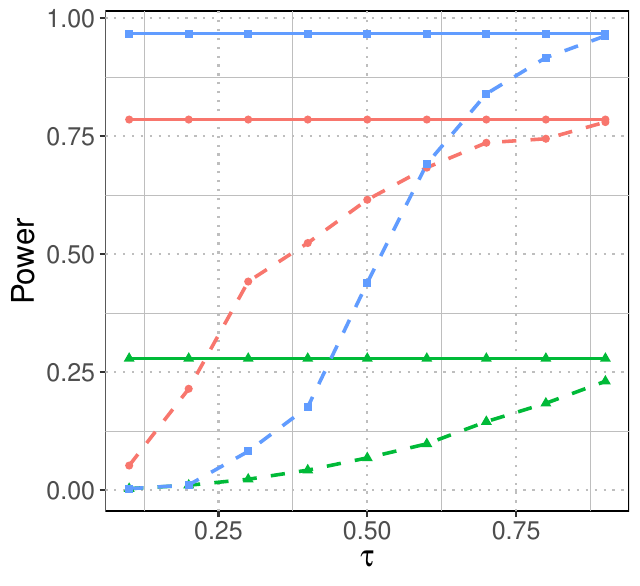}
\includegraphics[width=0.23\linewidth]{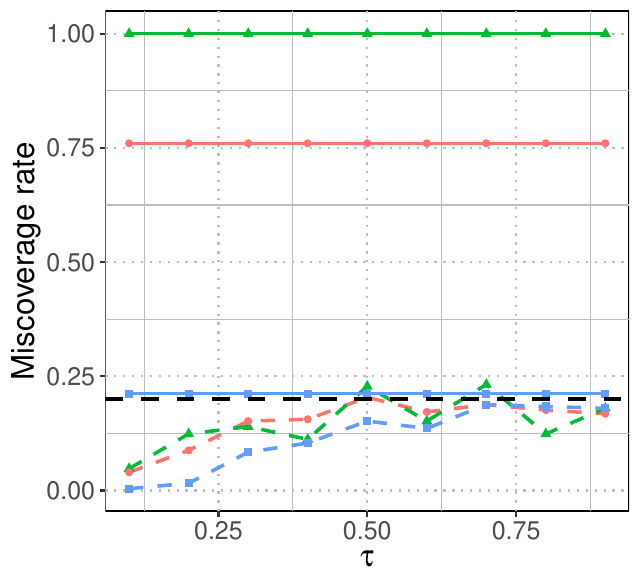}
\includegraphics[width=0.23\linewidth]{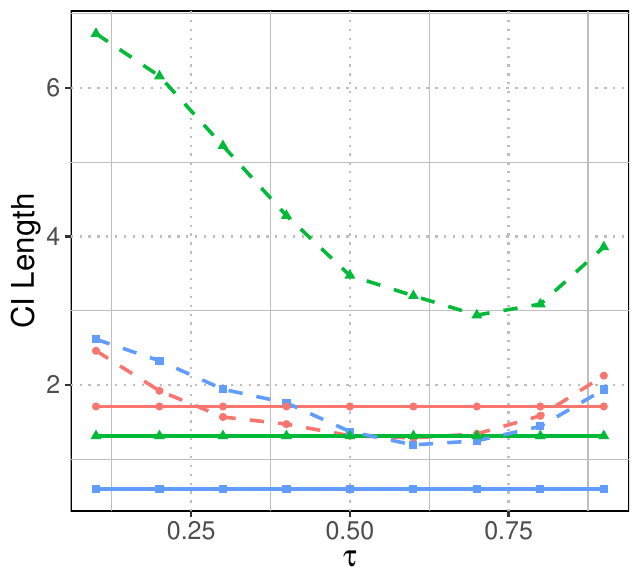}
\end{subfigure}
\vfill
\begin{subfigure}[t]{1\textwidth}
\centering
\includegraphics[width=0.4\linewidth]{figures/legend_interactive.png}
\end{subfigure}
\caption{Numerical results for the \emph{selection stage} for a Poisson signal, averaged over $250$ trials for a $50 \times 50$ grid of hypotheses with target FDR level chosen at $0.2$ and $\tau$ varying over $\{0.1,0.2,0.3,0.5,0.5,0.6,0.7,0.8,0.9\}$. Solid lines denote metrics for the rejection sets formed using the full dataset and dotted lines denote metrics calculated using the rejection sets formed through data fission. All fission methods control the false discovery rate at the desired coverage level but we can see that increasing $\tau$ increases the power of the fissioned procedures as more of the information gets reserved for selection. }
\label{fig:STAR_experiments_poisson}
\end{center} 
\end{figure}

\subsection{Simulation results for fixed-design linear regression}\label{sec:appendix_supplemental_linear}

The empirical results discussed in \ref{sec:linreg} demonstrate the advantage of data fission in fixed-design settings with small sample sizes and a handful of points with high leverage. In cases where the sample size is larger and the distribution of covariates are less irregular, data fission still tends to have slightly smaller confidence intervals but the differences are not as drastic. The following additional simulations help to elucidate this.

\paragraph{Additional metrics} In addition to the metrics described in \cref{sec:linreg}, we compute two additional metrics to aid in comparison between the available methodologies. In particular, we track the averaged proportion of falsely reported CIs among those indicating a non-zero parameter (the false sign rate),
\begin{align}
    \text{FSR} := \frac{|k \in M: \{\beta_k < 0 \text{ and } \text{CI}_k(1) > 0 \}\text{ OR } \{ \beta_k > 0 \text{ and } \text{CI}_k(2) < 0|\}}{\max\{|k \in M: 0 \notin \text{CI}_k|,1\}},
\end{align}
and power of correctly reporting CIs that indicates a nonzero parameter with the correct sign,
\begin{align}
    \text{power}^\text{sign} := \frac{|k \in M: \beta_k > 0 \text{ and } \text{CI}_k(1) > 0|}{\max\{|j \in [p]: \beta_j > 0|,1\}}.
\end{align}
In many of the simulations to follow, power, precision, and CI width are almost identical for data splitting and data fission so computing these metrics offer a complementary set of information to aid in performance evaluation. 

\paragraph{Simulation with independent covariates} We repeat the simulation discussed in \cref{sec:linreg} but now have $p = 100$ features with $x_i \in \mathbb{R}^{100}$. The vector of covariates $x_i$ follows independent standard Gaussians; and $y_i \sim N(\beta^T x_i,\sigma^{2})$ where the parameter $\beta$ is nonzero for 30 features: $(\beta_1, \beta_3, \ldots, \beta_{22}, \beta_{92}, \ldots, \beta_{100}) = S_\Delta \cdot (\underbrace{1, \ldots, 1}_{21}, \underbrace{-1, \ldots, -1}_{9})$ where $S_\Delta$ encodes the signal strength. \cref{fig:linear_indep} shows the result of these simulations averaged over $500$ repetitions. Data splitting and data fission appear to have roughly comparable performance in this case, as the large sample sizes allows for data splitting to trade off data between selection and inference smoothly, although the CI lengths are still slightly smaller for data fission. 

\begin{figure}[H]
\centering
    \begin{subfigure}[t]{1\textwidth}
        \centering
        \includegraphics[width=0.3\linewidth]{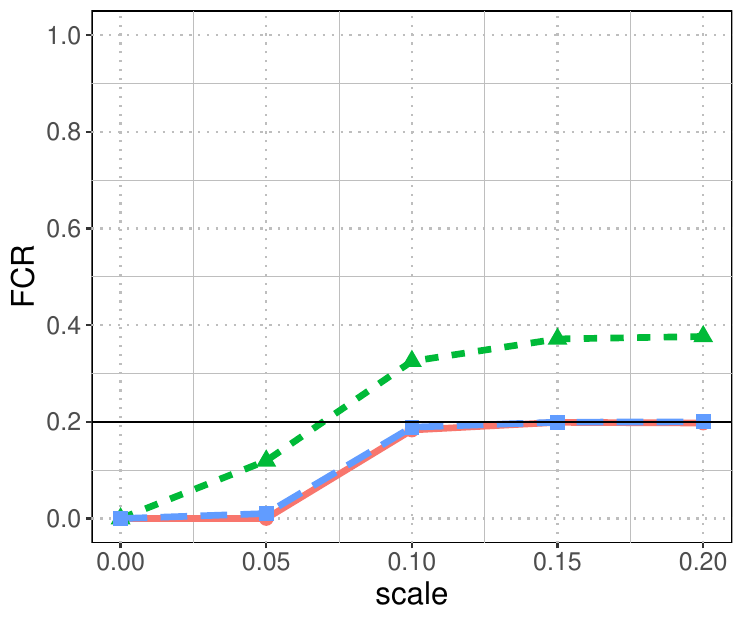}
    \hfill
        \includegraphics[width=0.3\linewidth]{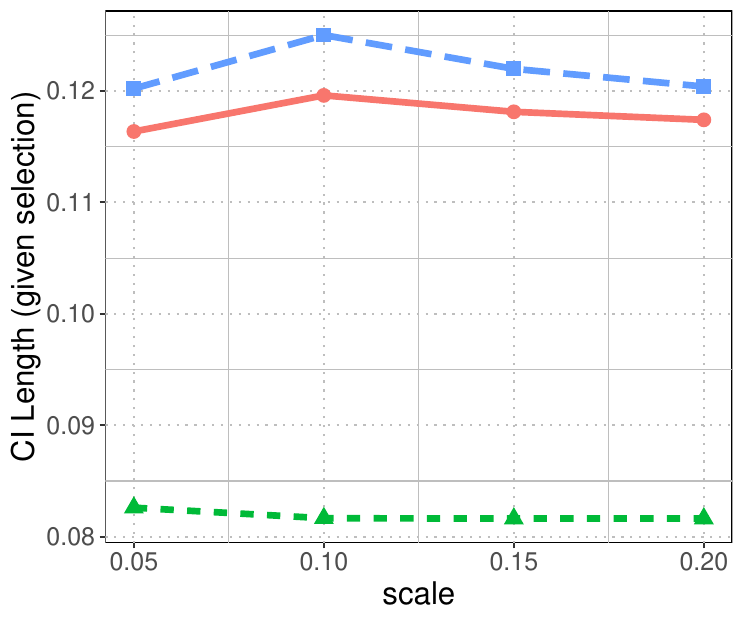}
    \hfill
        \includegraphics[width=0.3\linewidth]{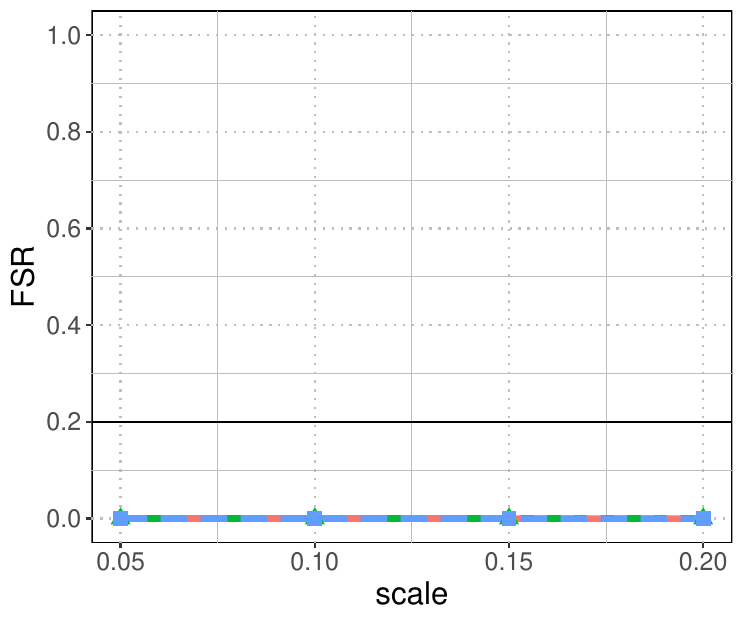}
    \end{subfigure}
\vfill
    \begin{subfigure}[t]{1\textwidth}
        \centering
        \includegraphics[width=0.3\linewidth]{figures/legend_fission_regression.png}
    \end{subfigure}
\vfill
    \begin{subfigure}[t]{1\textwidth}
        \centering
        \includegraphics[width=0.3\linewidth]{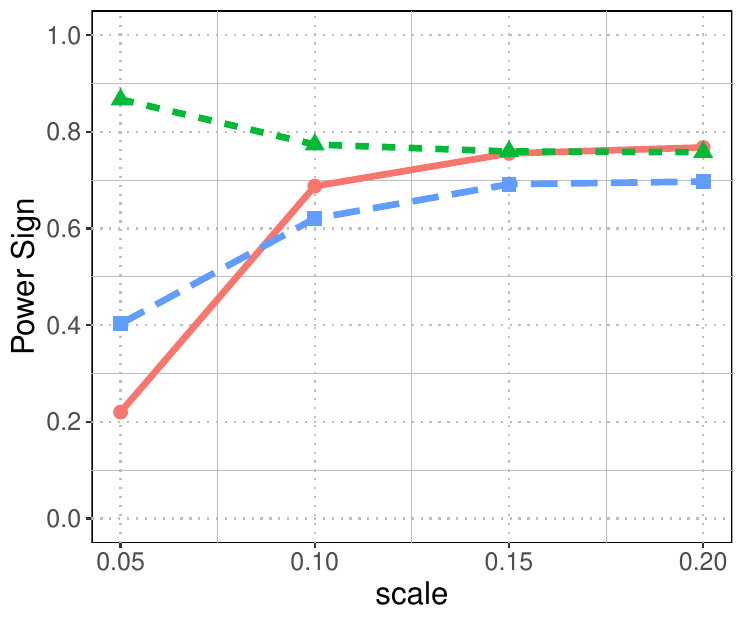}
    \hfill
        \includegraphics[width=0.3\linewidth]{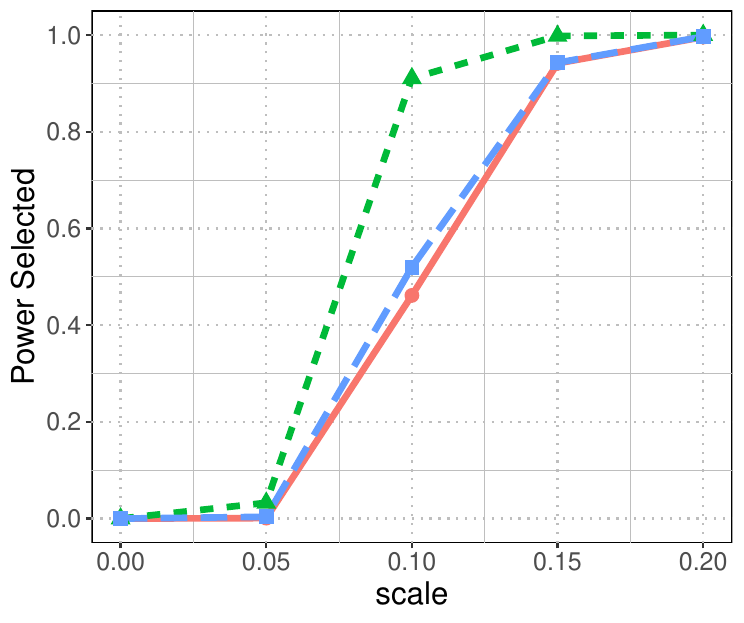}
    \hfill
        \includegraphics[width=0.3\linewidth]{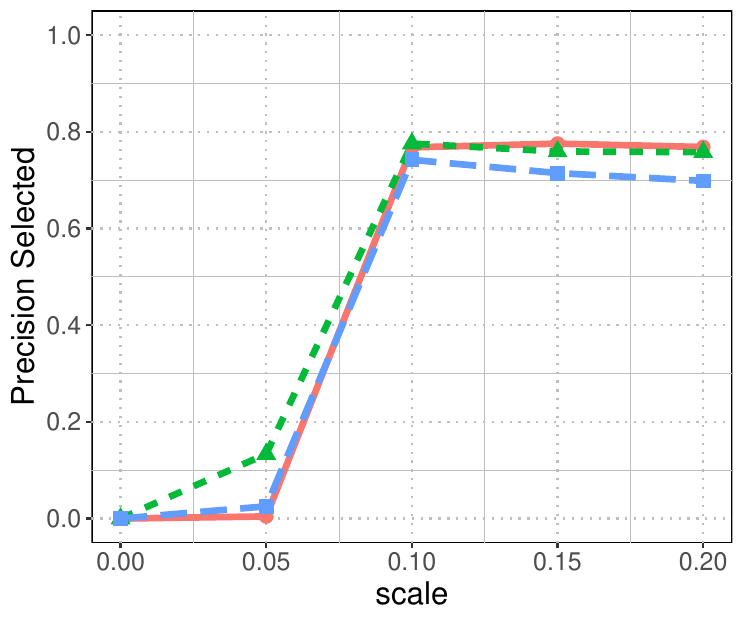}
    \end{subfigure}
    \caption{FCR, length of the CIs, FPR and power for the sign of parameters, and power and precision for the selected features, when varying the signal strength $S_\Delta$ in $\{0, 0.05, 0.1, 0.15, 0.2\}$. The results are averaged over 500 repetitions. The CIs by using the original data twice do not have FCR control guarantee due to selection bias. In contrast, our proposed procedure using the masking idea has valid FCR, without inflating the length of CIs much or reducing the power of selecting non-zero features.}
    \label{fig:linear_indep}
\end{figure}

\paragraph{Simulation with dependent covariates} We repeat the experiment described above but with dependent covariates $X$. We let $X$ be generated from a multivariate Gaussian with zero mean. The covariance matrix is a five-block diagonal matrix, each block a $20 \times 20$ Toeplitz matrix:
\begin{align} \label{eq:rho_mat_dependent}
    \begin{bmatrix}
1 & \rho & \cdots & \rho^{d-2} & \rho^{d-1}\\
\rho & 1 & \rho & \cdots & \rho^{d-2}\\
\vdots & \vdots & \vdots & \vdots & \vdots\\
\rho^{d-2} & \cdots & \rho & 1 & \rho\\
\rho^{d-1}  & \rho^{d-2} & \cdots & \rho & 1
\end{bmatrix},
\end{align}
where $d = 20$. Results are shown in~\cref{fig:i-positive_dependent} for varying $\rho$. We note that in this setting the relative ordering of the three methods is unchanged across evaluation metrics, but the power of both data splitting and data fission decreases when there is negative dependence among the covariates.

\iffalse
%Here, the expectation is $\mathbb{E}_g(Y_i) = \mathbb{E}(\exp\{\beta X_i\} \mid X_i^\text{obs})$ by chain rule, where $\beta X_i$ follows a Gaussian distribution given $X_i^\text{obs}$. Thus, $\mathbb{E}_g(Y_i) = \exp\{\mu + \sigma^2/2\}$, where $\mu = \mathbb{E}_{g_i}(\beta X_i)$ and $\sigma^2 = \text{Var}_{g_i}(\beta X_i)$.
%At the same time, the CIs constructed by profile likelihood or normal approximation may no longer be valid. Rather, one may construct CIs by bootstrap.
\begin{figure}[H]
\centering
    \begin{subfigure}[t]{0.32\textwidth}
        \centering
        \includegraphics[trim=5 20 25 60, clip, width=1\linewidth]{figures/linear_dep_example_masking.png}
    \end{subfigure}
    \hfill
    \begin{subfigure}[t]{0.32\textwidth}
        \centering
        \includegraphics[trim=5 20 25 60, clip, width=1\linewidth]{figures/linear_dep_example_full.png}
    \end{subfigure}
    \hfill
    \begin{subfigure}[t]{0.32\textwidth}
        \centering
        \includegraphics[trim=5 20 25 60, clip, width=1\linewidth]{figures/linear_dep_example_split.png}
    \end{subfigure}
    \caption{An instance with $\rho = -0.25$ of the selected feature at the projected values (blue crosses) and the constructed CIs using the blurred data (left), full data twice (middle), and split data (right). Fission and splitting tends to select few covariates (low power at the selection step).}
    \label{fig:poisson_example2}
\end{figure}
\fi

\begin{figure}[H]
\centering
    \begin{subfigure}[t]{1\textwidth}
        \centering
        \includegraphics[width=0.3\linewidth]{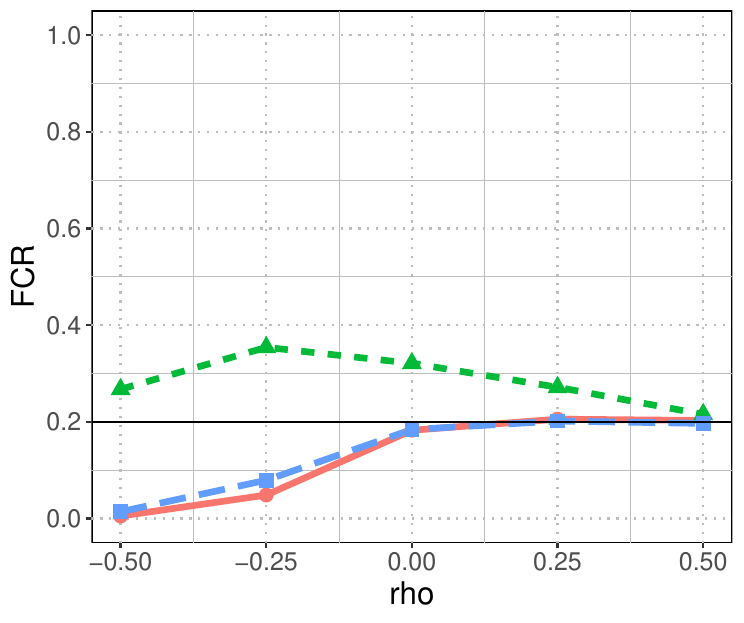}
    \hfill
        \includegraphics[width=0.3\linewidth]{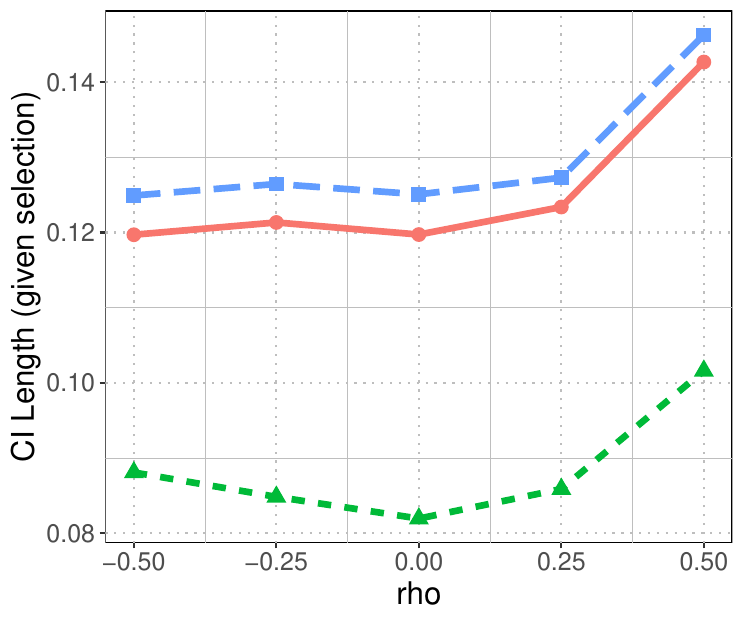}
    \hfill
        \includegraphics[width=0.3\linewidth]{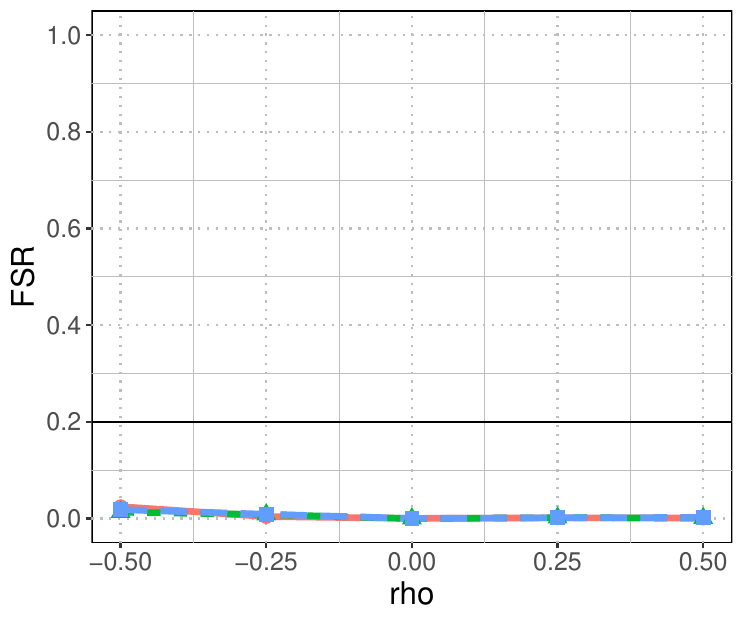}
    \end{subfigure}
\vfill
    \begin{subfigure}[t]{1\textwidth}
        \centering
        \includegraphics[width=0.3\linewidth]{figures/legend_fission_regression.png}
    \end{subfigure}
\vfill
    \begin{subfigure}[t]{1\textwidth}
        \centering
        \includegraphics[width=0.3\linewidth]{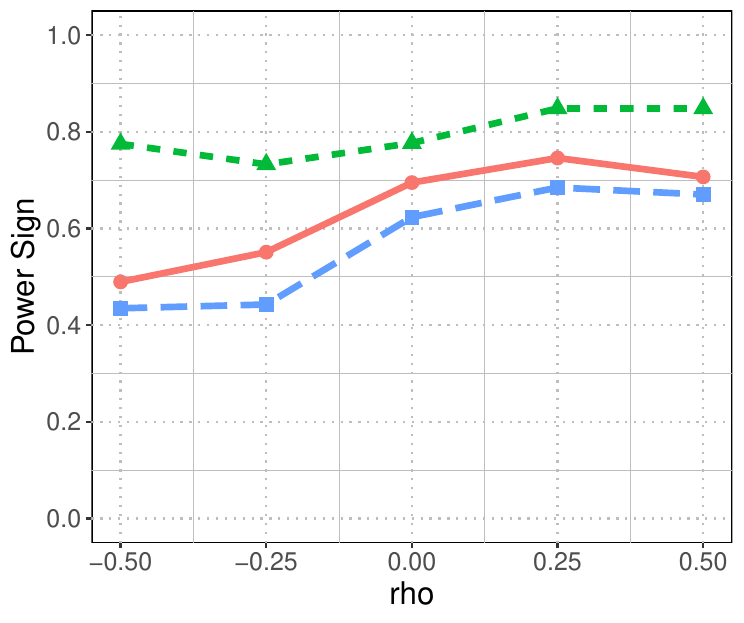}
    \hfill
        \includegraphics[width=0.3\linewidth]{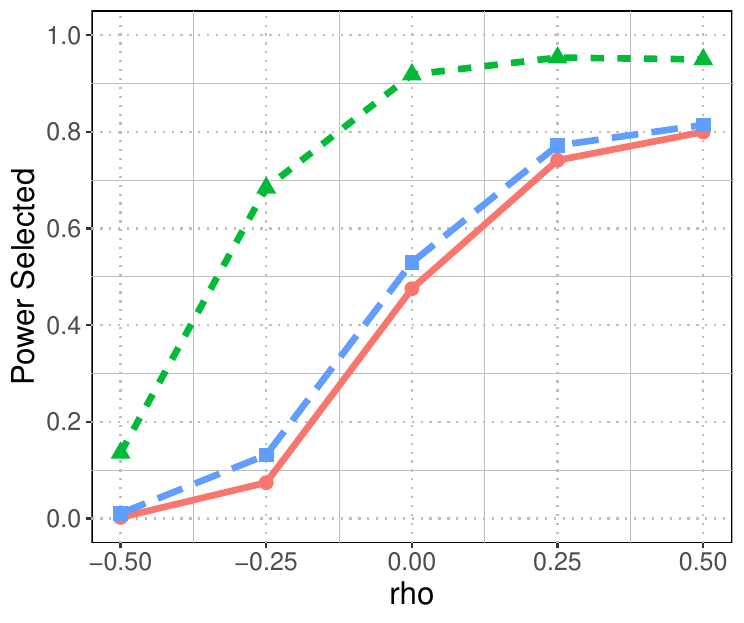}
    \hfill
        \includegraphics[width=0.3\linewidth]{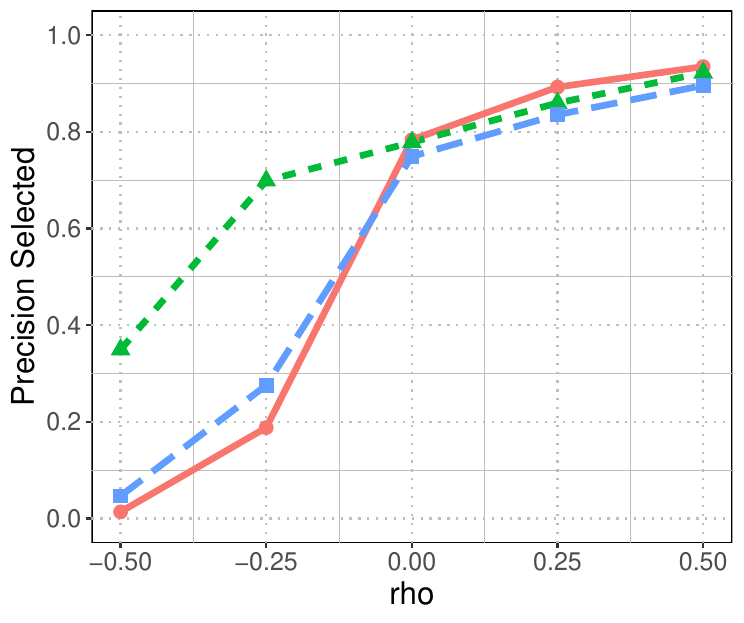}
    \end{subfigure}
    \caption{FCR, length of the CIs, FPR and power for the sign of parameters, and power and precision for the selected features, when varying the correlation parameter $\rho$ in $\{-0.5, -0.25, 0, 0.25, 0.5\}$ (with $\rho = 0$ being independent covariates). It appears that negative correlation leads to lower power for all methods.}
    \label{fig:i-positive_dependent}
\end{figure}

\paragraph{Simulation results with misspecified errors and estimated variance}
{\color{black} We investigate how sensitive the methodology is to a misspecified model for the error distribution. The simulation setup is identical to the previous setting described above, but now the errors are not Gaussian. In particular, we experiment with three different choices for the error term: Laplace, skew normal distribution with scale parameter equal to 1 and shape parameter equal to 5, and a $t$-distribution with $5$ degrees of freedom. In all cases, the error terms are also rescaled to have $0$ mean and unit variance. 

When performing model selection and inference, however, we (incorrectly) use a Gaussian decomposition rule from \cref{sec:list_decomp} to construct $f(Y) = Y + Z$ and $g(Y) = Y - Z$ with $Z \sim N(0,1)$. We also repeat these simulations when the fission is performed using an estimate of the variance (as discussed in \cref{sec:unknown_variance}) in each case. Results are again averaged over $500$ repetitions and reported in \cref{fig:linear_misspecified}. We note that the generated confidence intervals have correct coverage in all cases. The procedure also seems to have very similar performance under misspecification for all error types except Laplace errors, which result in noticeably wider confidence intervals and lower power during the selection stage. These results suggest that data fission is robust to reasonable levels of misspecification in the error term. 
}
\begin{figure}[H]
\centering
    \begin{subfigure}[t]{1\textwidth}
        \centering
        \includegraphics[width=0.3\linewidth]{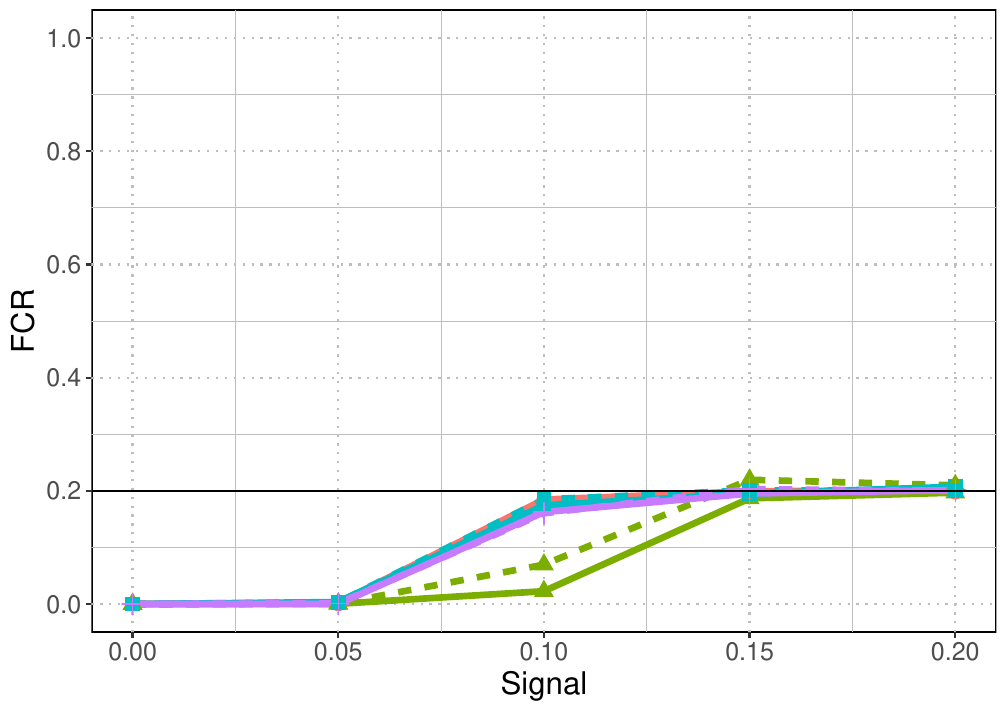}
    \hfill
        \includegraphics[width=0.3\linewidth]{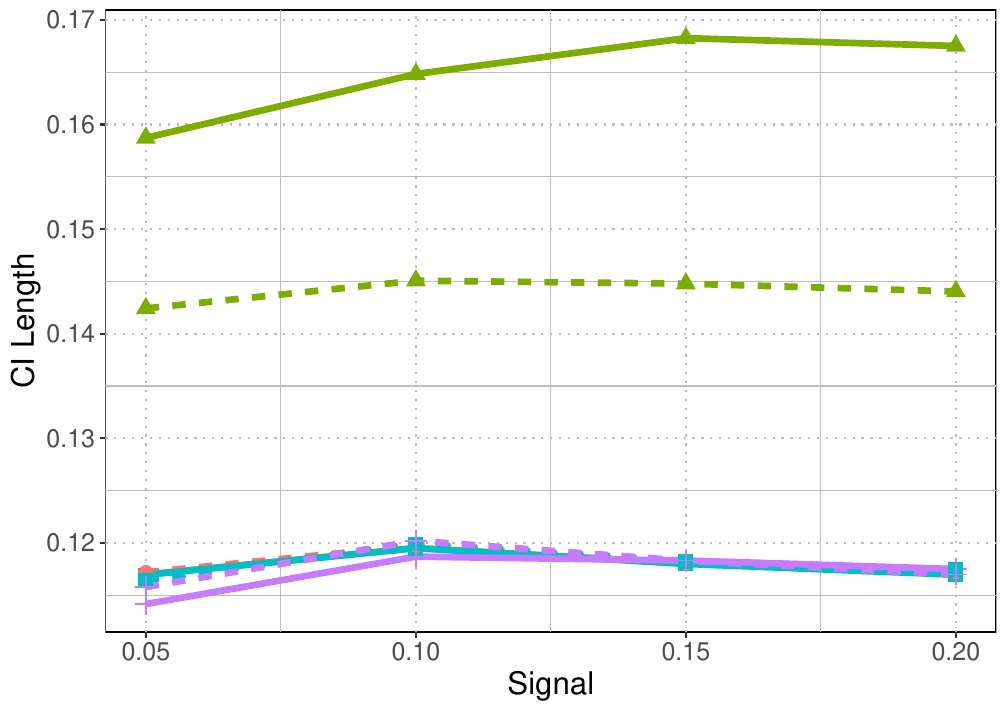}
    \hfill
        \includegraphics[width=0.3\linewidth]{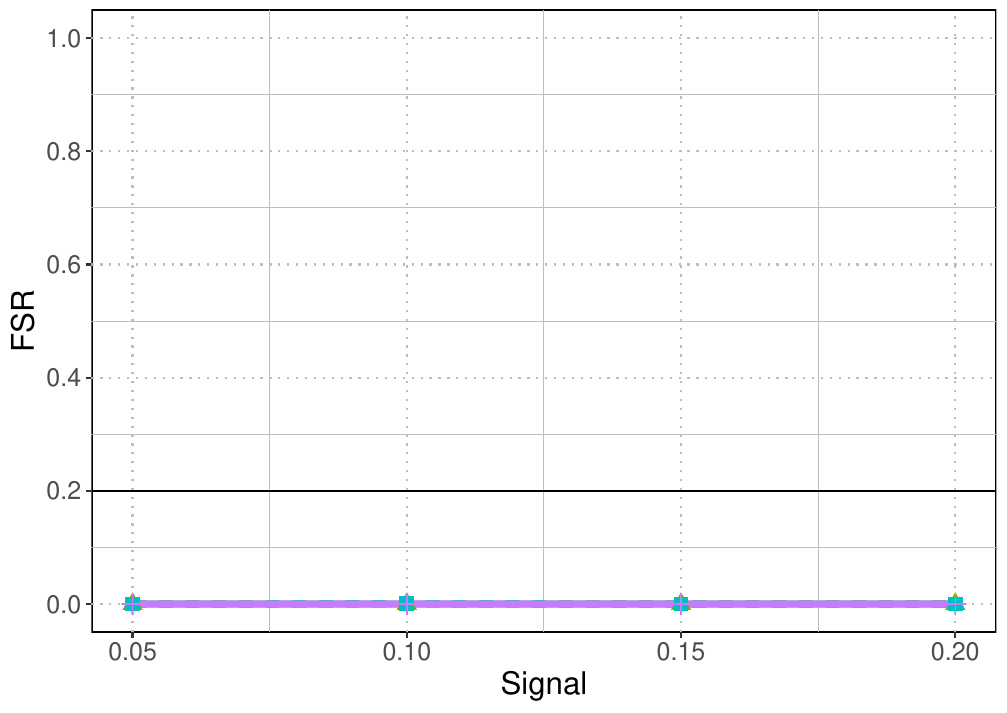}
    \end{subfigure}
\vfill
    \begin{subfigure}[t]{1\textwidth}
        \centering
        \includegraphics[width=0.8\linewidth]{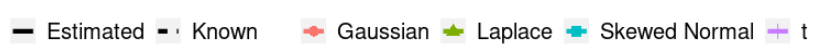}
    \end{subfigure}
\vfill
    \begin{subfigure}[t]{1\textwidth}
        \centering
        \includegraphics[width=0.3\linewidth]{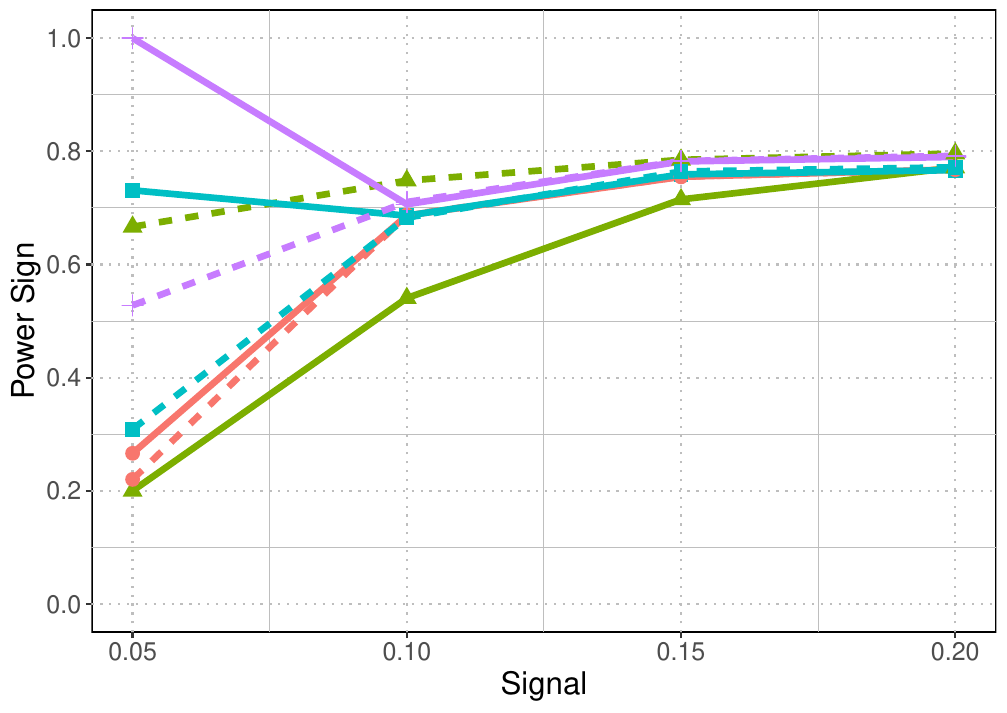}
    \hfill
        \includegraphics[width=0.3\linewidth]{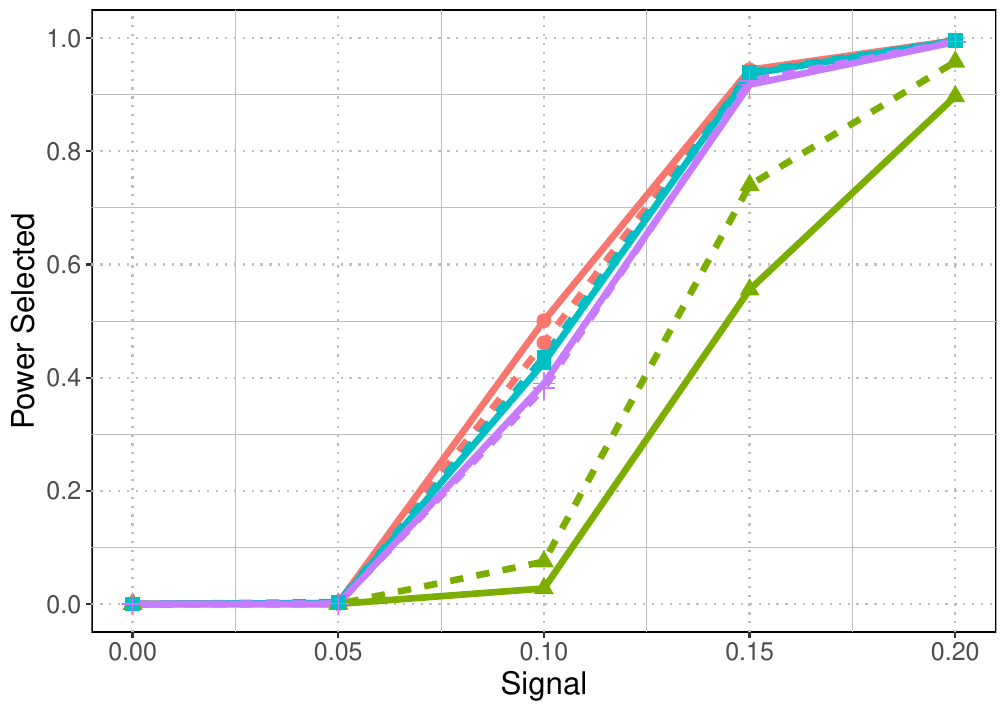}
    \hfill
        \includegraphics[width=0.3\linewidth]{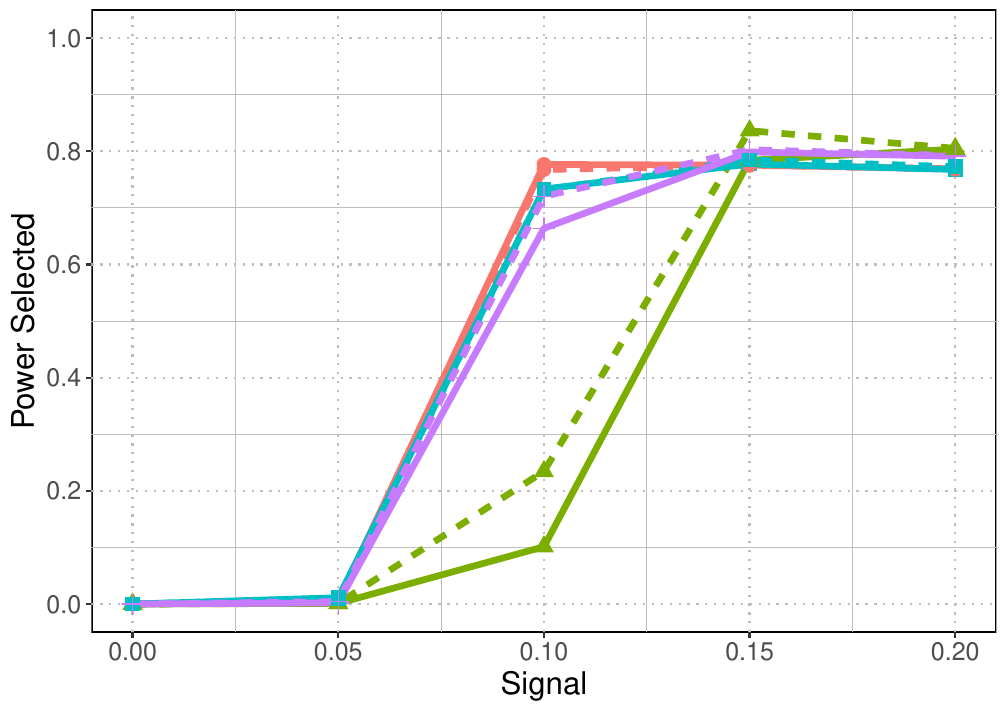}
    \end{subfigure}
    \caption{FCR, length of the CIs, FPR and power for the sign of parameters, and power and precision for the selected features, when varying the signal strength $S_\Delta$ in $\{0, 0.05, 0.1, 0.15, 0.2\}$ and the errors are generated from misspecified models. Under misspecification, the constructed confidence intervals still have adequate coverage. Performance in terms of CI length and power/precision during the selection stage is markedly worse when the error term is Lapalace distributed, but other types of misspecification have minimal impact on the procedure's efficacy.}
    \label{fig:linear_misspecified}
\end{figure}

\subsection{Simulation results for fixed-design Poisson regression}
\label{sec:appendix_poisson}
We also demonstrate how data fission compare with data splitting in the Poisson regression setting when $n$ is large. %Unlike the Gaussian case, the advantages of data fission do not diminish much when used in a large $n$ setting without leverage points. Although we do not have a clear theoretical reason for this result, it is an interesting empirical finding to note as an avenue for future investigation. 
\paragraph{Setup.} Let $y_i$ be the dependent variable and $x_i \in \mathbb{R}^p$ be a vector of $p$ features. Suppose we have $n = 1000$ samples. We have collected $p = 100$ features $X_i \in \{0,1\}^2 \times \mathbb{R}^{98}$, where the first two follow $\mathrm{Ber}(1/2)$ and the rest follow independent Gaussians. Suppose $Y_i$ follow a Poisson distribution with the expected value $\exp\{ \beta^T X_i\}$, where the parameter $ \beta$ is nonzero for 29 features: $(\beta_1, \beta_3, \ldots, \beta_{22}, \beta_{93}, \ldots, \beta_{100}) = S_\Delta \cdot (\underbrace{1, \ldots, 1}_{21}, \underbrace{2, \ldots, 2}_{8})$ and $S_\Delta$ encodes the signal strength.

\paragraph{Simulation results with mutually independent covariates.} As an illustration of the procedure, \cref{fig:poisson_example3} shows an instance of selected features and the constructed CIs for a single trial run. The selected features are marked by blue crosses, which include all of the nonzero coefficients (corresponding to almost 100\%  power for selection) and also many zero coefficients (corresponding to around 50\% precision for selection). The constructed CIs that do not contain the true value are marked in red. 
\iffalse

For example, in the one-dimensional case, the variance is $(\sum_{i=1}^n x_i^2 \mu_i)^{-1},$ where $\mu_i = e^{\beta x_i}$.
\fi

\cref{fig:i-positive3} shows results averaged over 500 trials. Compared with the results using data splitting, the CIs using data fission is significantly tighter. %This is expected since data splitting uses $n/2$ samples following $\mathrm{Poi}(\mu(X_i))$ whereas data fission uses $n$ samples following $\mathrm{Poi}(\mu(X_i)/2)$, which has smaller variance.
Although the precision during the selection step is not high for both methods (around 50\% of selected features do not have true signals), we are able to identify the the true signals by constructing CIs in the second step with FCR control. %Note that the CI width decreases with the signal strength because the variance of the coefficient estimates also decreases their expected value in Poisson regression. 
When interested in the sign indicated by the CIs, we have constructed CIs indicating a positive sign for all of the nonzero (in our case positive) coefficients (reflected as 100\% power for the signs and 0\% FSR).

\begin{figure}[H]
\centering
    \begin{subfigure}[t]{0.32\textwidth}
        \centering
        \includegraphics[trim=5 20 25 60, clip, width=1\linewidth]{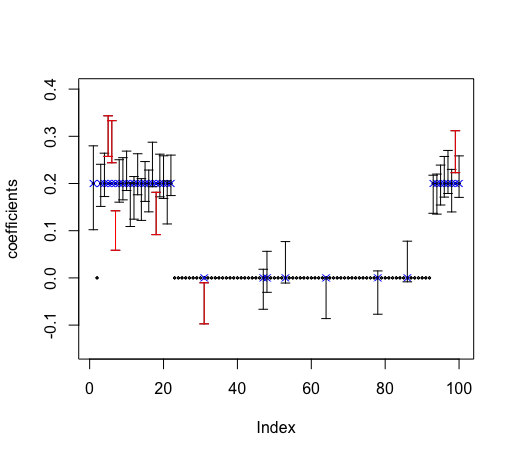}
    \end{subfigure}
    \hfill
    \begin{subfigure}[t]{0.32\textwidth}
        \centering
        \includegraphics[trim=5 20 25 60, clip, width=1\linewidth]{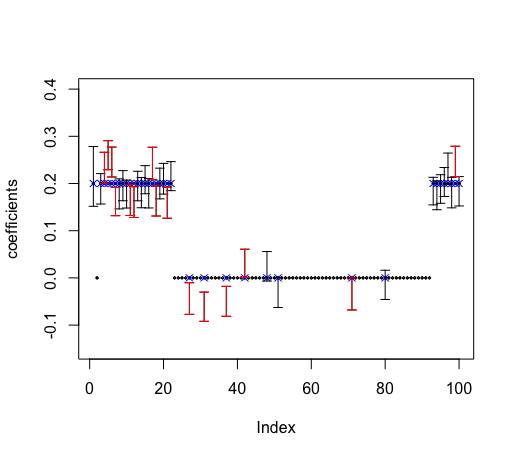}
    \end{subfigure}
    \hfill
    \begin{subfigure}[t]{0.32\textwidth}
        \centering
        \includegraphics[trim=5 20 25 60, clip, width=1\linewidth]{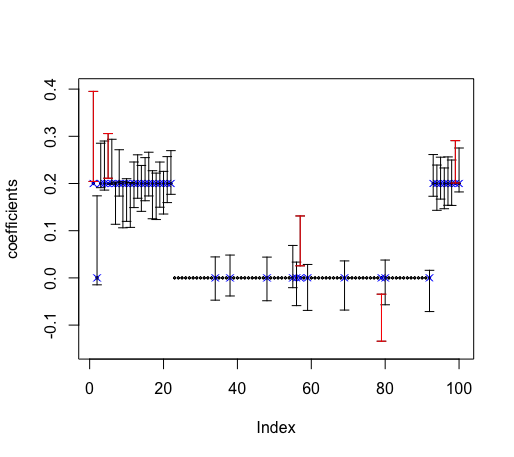}
    \end{subfigure}
    \caption{An instance of the selected feature (blue crosses) and the constructed CIs using fissioned data (left), full data twice (middle), and split data (right). The black dots indicate the actual values for the model parameters and the blue crosses, conditional on selection, represent the target parameters that minimize the KL divergence between the true distribution of the data and the selected model (i.e. $\beta^{\star}_{n}(M)$. CIs which do not cover the parameters correctly are marked in red.}
    \label{fig:poisson_example3}
\end{figure}

\begin{figure}[H]
\centering
    \begin{subfigure}[t]{1\textwidth}
        \centering
        \includegraphics[width=0.3\linewidth]{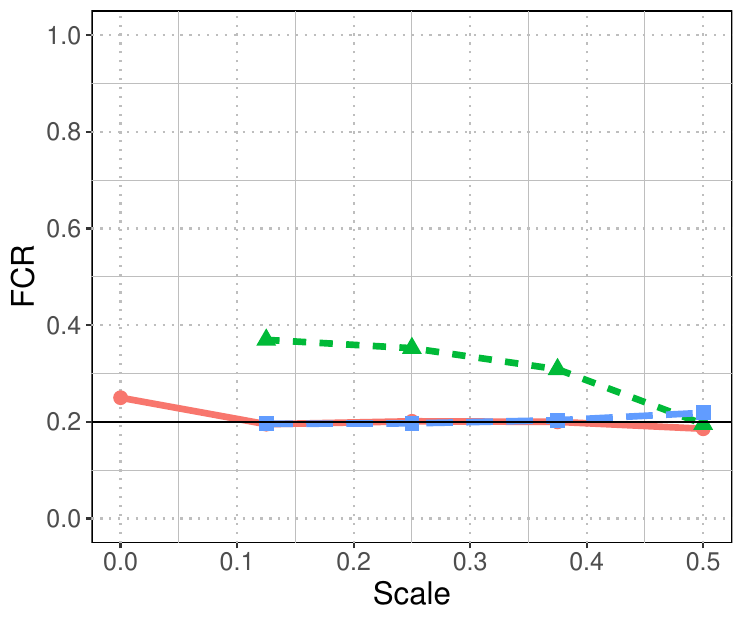}
    \hfill
        \includegraphics[width=0.3\linewidth]{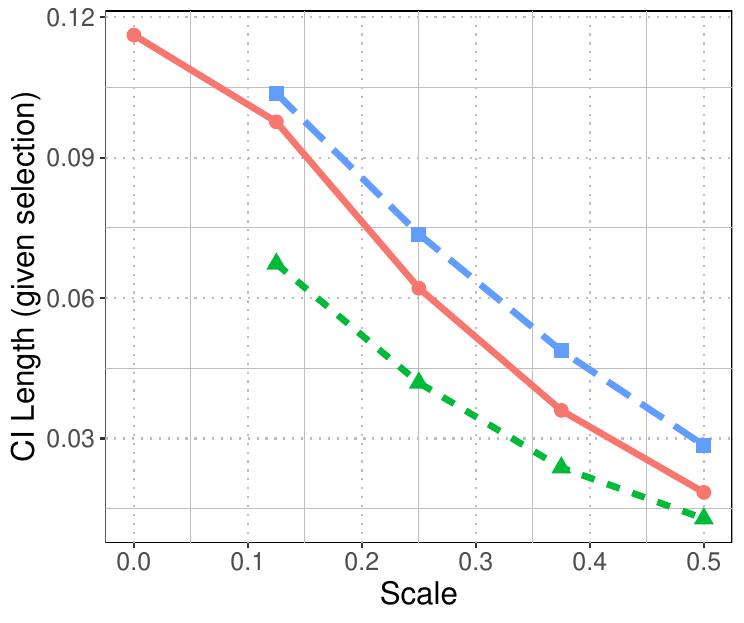}
    \hfill
        \includegraphics[width=0.3\linewidth]{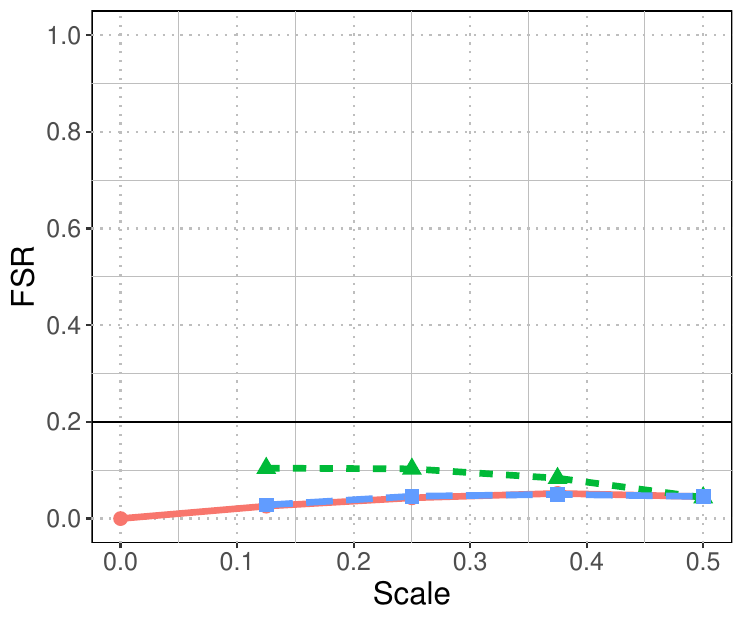}
    \end{subfigure}
\vfill
    \begin{subfigure}[t]{1\textwidth}
        \centering
        \includegraphics[width=0.3\linewidth]{figures/legend_fission_regression.png}
    \end{subfigure}
\vfill
    \begin{subfigure}[t]{1\textwidth}
        \centering
        \includegraphics[width=0.3\linewidth]{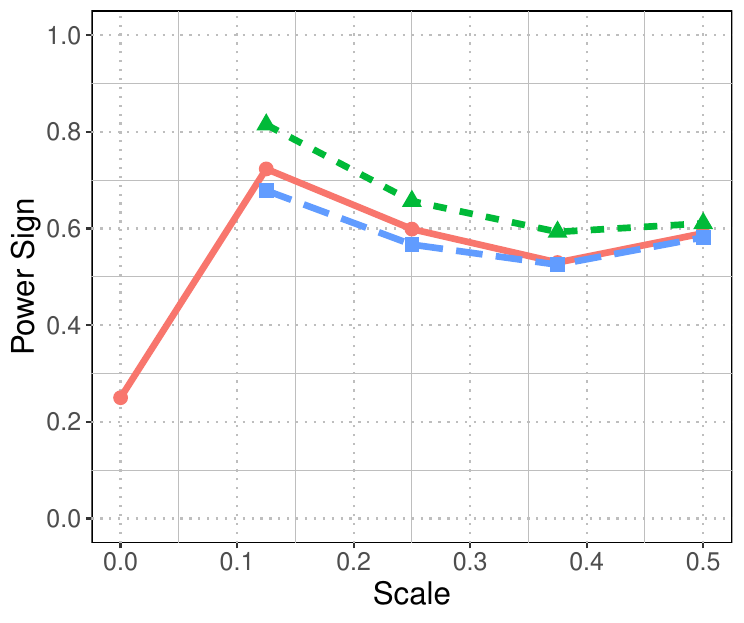}
    \hfill
        \includegraphics[width=0.3\linewidth]{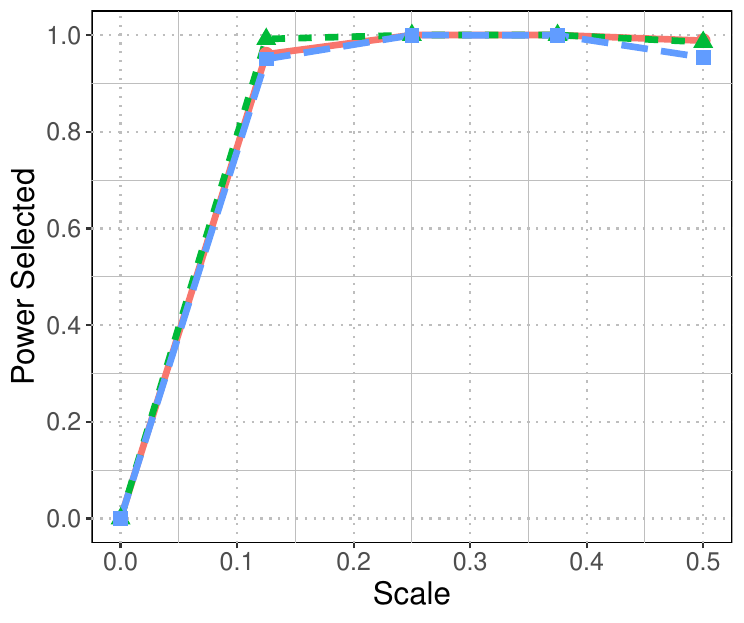}
    \hfill
        \includegraphics[width=0.3\linewidth]{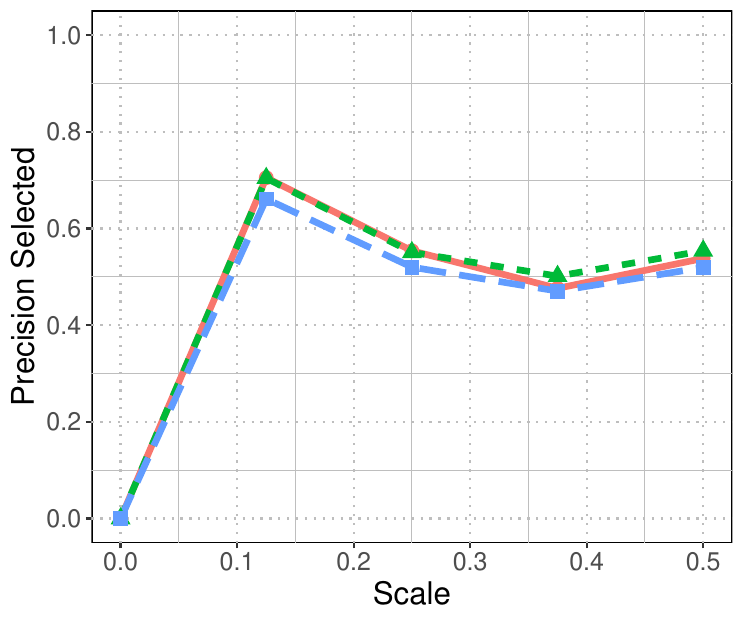}
    \end{subfigure}
    \caption{FCR, length of the CIs, FPR and power for the sign of parameters, and power and precision for the selected features, when varying the signal strength $S_\Delta$ in $\{0, 0.125, 0.25, 0.375, 0.5\}$. The results are averaged over 500 repetitions. The CIs by using the original data twice do not have FCR control guarantee due to selection bias. In contrast, our proposed procedure using fission has valid FCR, without inflating the length of CIs much or reducing the power of selecting non-zero features.}
    \label{fig:i-positive3}
\end{figure}

\iffalse
\begin{figure}[H]
\centering
    \begin{subfigure}[t]{1\textwidth}
        \centering
        \includegraphics[width=0.3\linewidth]{figures/sandwich_regression_poisson_varyRho_FCR.pdf}
    \hfill
        \includegraphics[width=0.3\linewidth]{figures/sandwich_regression_poisson_varyRho_CI_length.pdf}
    \hfill
        \includegraphics[width=0.3\linewidth]{figures/sandwich_regression_poisson_varyRho_FSR.pdf}
    \end{subfigure}
\vfill
    \begin{subfigure}[t]{1\textwidth}
        \centering
        \includegraphics[width=0.7\linewidth]{figures/legend_poisson3.png}
    \end{subfigure}
\vfill
    \begin{subfigure}[t]{1\textwidth}
        \centering
        \includegraphics[width=0.3\linewidth]{figures/sandwich_regression_poisson_varyRho_power_sign.pdf}
    \hfill
        \includegraphics[width=0.3\linewidth]{figures/sandwich_regression_poisson_varyRho_power_selected.pdf}
    \hfill
        \includegraphics[width=0.3\linewidth]{figures/sandwich_regression_poisson_varyRho_FCR.pdf}
    \end{subfigure}
    \caption{FCR, length of the CIs, FPR and power for the sign of parameters, and power and precision for the selected features, when the number of significant features increases to 30. The overall trend seem similar, with slight decrease in the FCR of the signs possibly due to a larger number of CI with signs (more than 30 for $S_\Delta > 0$), and increase in the precision of selected features at the first step.}
    \label{fig:i-positive4}
\end{figure}
\fi

\paragraph{Simulation results with dependent covariates.} We modify the above simulation slightly so that the $X$ are samples from \emph{dependent} rather than independent Gaussians. The covariance matrix is a five-block diagonal matrix, each block a $20 \times 20$ toeplitz matrix:
\begin{align} \label{eq:rho_mat}
    \begin{bmatrix}
1 & \rho & \cdots & \rho^{d-2} & \rho^{d-1}\\
\rho & 1 & \rho & \cdots & \rho^{d-2}\\
\vdots & \vdots & \vdots & \vdots & \vdots\\
\rho^{d-2} & \cdots & \rho & 1 & \rho\\
\rho^{d-1}  & \rho^{d-2} & \cdots & \rho & 1
\end{bmatrix},
\end{align}
where $d = 20$.
Results averaged over $500$ trials are shown in \cref{fig:i-positive_dependent2} for varying $\rho$.
%Here, the expectation is $\mathbb{E}_g(Y_i) = \mathbb{E}(\exp\{\beta X_i\} \mid X_i^\text{obs})$ by chain rule, where $\beta X_i$ follows a Gaussian distribution given $X_i^\text{obs}$. Thus, $\mathbb{E}_g(Y_i) = \exp\{\mu + \sigma^2/2\}$, where $\mu = \mathbb{E}_{g_i}(\beta X_i)$ and $\sigma^2 = \text{Var}_{g_i}(\beta X_i)$.
%At the same time, the CIs constructed by profile likelihood or normal approximation may no longer be valid. Rather, one may construct CIs by bootstrap.

\begin{figure}[H]
\centering
    \begin{subfigure}[t]{1\textwidth}
        \centering
        \includegraphics[width=0.3\linewidth]{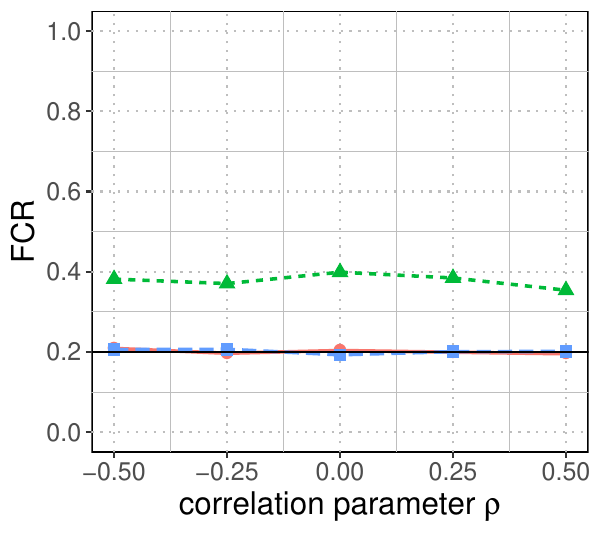}
    \hfill
        \includegraphics[width=0.3\linewidth]{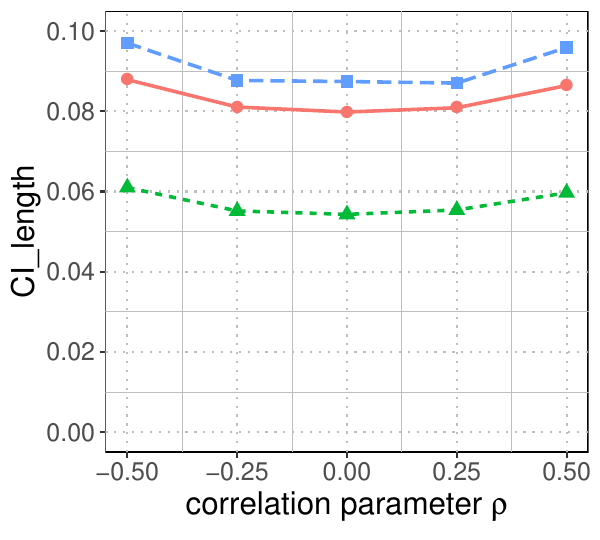}
    \hfill
        \includegraphics[width=0.3\linewidth]{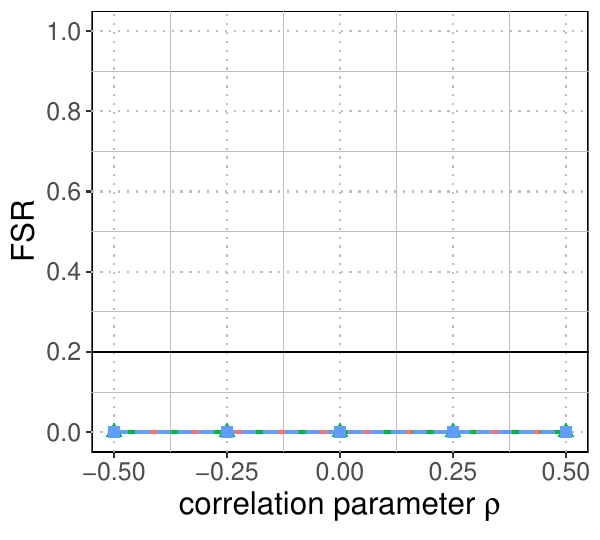}
    \end{subfigure}
\vfill
    \begin{subfigure}[t]{1\textwidth}
        \centering
        \includegraphics[width=0.3\linewidth]{figures/legend_fission_regression.png}
    \end{subfigure}
\vfill
    \begin{subfigure}[t]{1\textwidth}
        \centering
        \includegraphics[width=0.3\linewidth]{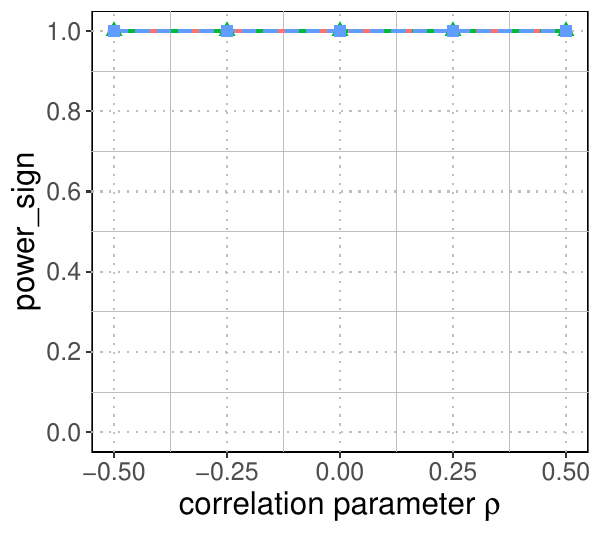}
    \hfill
        \includegraphics[width=0.3\linewidth]{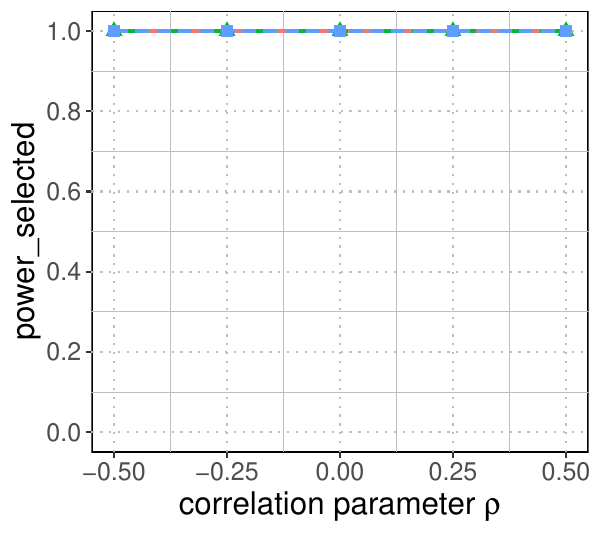}
    \hfill
        \includegraphics[width=0.3\linewidth]{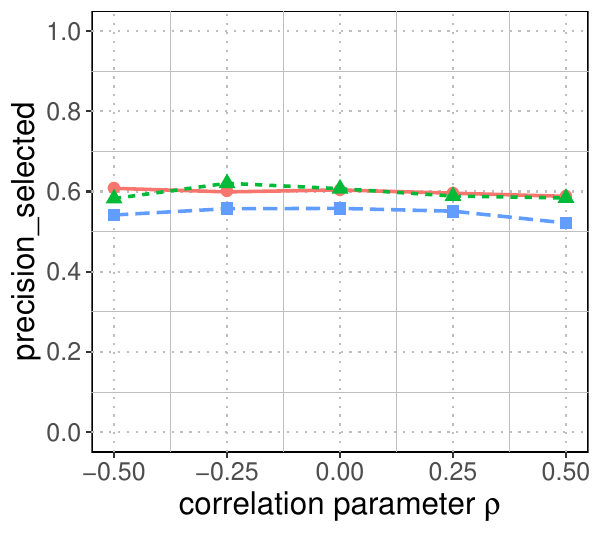}
    \end{subfigure}
    \caption{FCR, length of the CIs, FSR/power for the sign of parameters, and power and precision for the selected features, when varying the correlation parameter $\rho$ in $\{-0.5, -0.25, 0, 0.25, 0.5\}$ (with $\rho = 0$ corresponding to mutually independent covariates). The performance of the three methods is relatively similar under different degrees of dependence.}
    \label{fig:i-positive_dependent2}
\end{figure}
Data fission continues to have tighter confidence intervals than data splitting, though the performances across the three methods are simular for the other metrics. 

\subsection{Simulation results for fixed-design logistic regression}\label{sec:appendix_logistic}
We also explore the empirical performance of data fission in a logistic regression setting. 

\paragraph{Setup.} Let $y_i$ be the dependent variable and $x_i \in \mathbb{R}^p$ be a vector of $p$. We let $n = 1000$ and $p = 100$ with $x_i \in \{0,1\}^2 \times \mathbb{R}^{98}$ generated with a $\mathrm{Ber}(1/2)$ for the first two entries and the remaining following iid Standard Gaussians. Suppose the conditional distribution of $y_i$ given $x_i$ is a Bernoulli distribution with expected value $\left(1 + \exp\{- \beta^T  x_i\}\right)^{-1}$, where the parameter $ \beta$ is nonzero for 30 features: $(\beta_{1}, \beta_3, \ldots, \beta_{22}, \beta_{92}, \ldots, \beta_{100}) = S_\Delta \cdot (\underbrace{1, \ldots, 1}_{21}, \underbrace{2, \ldots, 2}_{9})$ and $S_\Delta$ encodes the signal strength.

\paragraph{Proposed procedure.} Following Section~\ref{sec:list_decomp}, we can use data fission for constructing selective CIs as follows.
\begin{enumerate}
    \item  Draw $Z_i \sim \mathrm{Ber}(p)$ where the ``flip probability'' $p \in (0,1)$ is a tuning parameter; and let $f(Y_i) = Y_i(1 - Z_i) + (1 - Y_i)Z_i$, and $g(Y_i) = Y_i$.
    \item Fit $f(Y_i)$ with a GLM with a logit link function and lasso regularization to select features, denoted as $M \in [p]$ (in our examples, we use \texttt{cv.glmnet} in the \texttt{R} package \texttt{glmnet} and choose the tuning parameter $\lambda$ by the 1 standard deviation rule, which can be found in the value of \texttt{lambda.1se}).
    \item Fit $g(Y_i)$ with another GLM with a logit link function and no regularization using \textit{only} the selected features (we use \texttt{glm} in the \texttt{R} package \texttt{stats}).
    \item Construct CIs for the coefficients trained in the third step, each at level $\alpha$ and with the standard errors estimated as in \cref{thm:var_emp_fahrmeir}. For the experiments shown below, we also apply a finite sample correction as described in \cite{CR2_explanation}.
\end{enumerate}
For greater detail on steps 3 and 4, see the logistic regression discussion in Appendix~\ref{sec:appendix_QMLE_details}. Since the working model is not correctly specified in this instance, the CIs are now are covering the target parameters $\beta_{n}^{\star}(M)$ which minimize the KL divergence between the chosen model and the true distribution of the data.

%With the above decomposition, $f(Y_i)$ has distribution $\mathrm{Ber}(\mu_i + p - 2p\mu_i)$ where $\mu_i = \mathbb{E}(Y_i \mid X_i)$; and $g(Y_i) | f(Y_i)$ is distributed as $\mathrm{Ber}\left(\frac{\mu_i}{\mu_i + (1-\mu_i) [p/(1-p)]^{2f(Y_i) - 1}}\right)$. Although the post-selective distribution is not Using masking for Bernoulli distribution, we draw $Z \sim \mathrm{Ber}(p)$ where the ``flip probability'' $p \in (0,1)$ is a tuning parameter (where $p \neq 0.5$; otherwise we cannot recover the original distribution). Suppose $Y \sim \mathrm{Ber}(\theta)$, then $f(Y) = Y(1 - Z) + (1 - Y)Z$ has marginal distribution $\mathrm{Ber}(\theta + p - 2p\theta)$; and $g(Y) = Y$ has conditional distribution as $\mathrm{Ber}\left(\frac{\theta}{\theta + (1-\theta) [p/(1-p)]^{2f(Y) - 1}}\right)$ given $f(Y)$. To construct CIs for potentially significant features, we can perform a two-step algorithm:

This procedure is illustrated for a single trial run in \cref{fig:i-positive6}. As usual, we compare the CIs constructed using data fission with data splitting and the (invalid) approach of reusing the full dataset for both selection and inference. We see that even though

%Note that in this case, the CIs are now are covering the \emph{projected} parameters $\beta_*$ as opposed to the true underlying parameters $\beta$. 

\begin{figure}[H]
\centering
    \begin{subfigure}[t]{0.4\textwidth}
        \centering
        \includegraphics[trim=5 20 25 60, clip, width=1\linewidth]{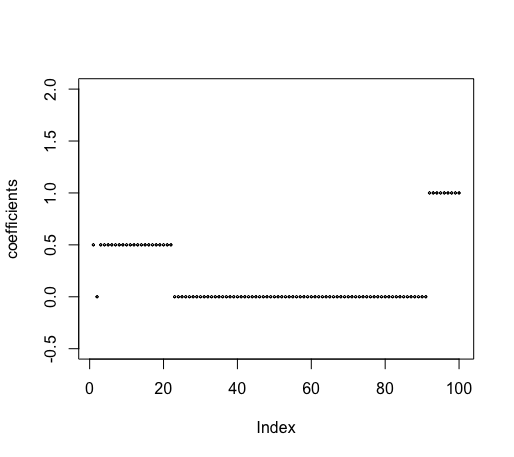}
    \caption{True coefficient values.}
    \end{subfigure}
\hfill
    \begin{subfigure}[t]{0.4\textwidth}
        \includegraphics[trim=5 20 25 60, clip, width=1\linewidth]{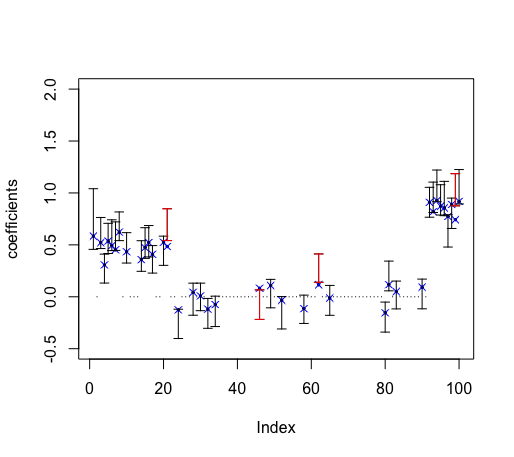}
    \caption{CI for target parameters using data fission.}
    \end{subfigure}
\end{figure}
\begin{figure}[H]
\ContinuedFloat
\centering
    \begin{subfigure}[t]{0.4\textwidth}
        \centering
        \includegraphics[trim=5 20 25 60, clip, width=1\linewidth]{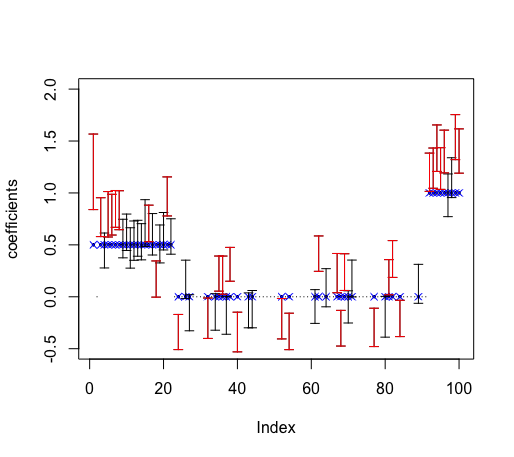}
    \caption{CIs for target parameters using full data twice.}
    \end{subfigure}
\hfill
    \begin{subfigure}[t]{0.4\textwidth}
        \includegraphics[trim=5 20 25 60, clip, width=1\linewidth]{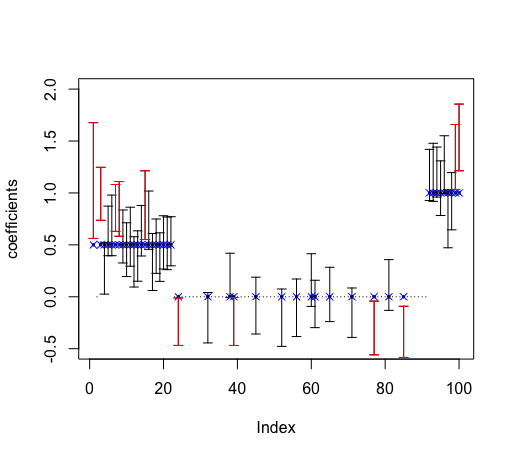}
    \caption{CIs for target parameters using splitting.}
    \end{subfigure}
    \caption{An instance of CIs for an example set of selected features using data fission, data splitting, and the (invalid) approach of using using the full dataset twice for both selection and inference. The upper left-hand graph shows the true coefficients, but the CIs are designed to cover the target parameters ($\beta^{\star}_{n}(M)$, marked by blue crosses) as described above. We note that both data splitting and data fission have valid FCR control (at target level = $0.2$), but the CIs constructed using data fission are much tighter. A drawback of data fission in this case, however, is the target parameters are shifted more from their actual values when compared with data splitting.}
    \label{fig:i-positive6}
\end{figure}

\paragraph{Simulations varying signal stength.} 
We repeat this experiment over $500$ trials while varying signal signal strength from $0$ to $0.5$, with results shown in \cref{fig:i-positive2}. We note that across all observed signal strengths, data fission offers tighter CIs when compared to both data splitting \emph{and} the (invalid) approach that reuses the full dataset for both selection and inference. This fact appears remarkable, but we note the significantly decreased power of data fission in this case. The simpler model, combined with the fact that the model misspecification introduced by data fission can adjust the target parameter, means that the comparison in CI length across methods is not necessarily apples-to-apples. 

%This is a remarkable feat since it is typically the case (e.g. in the Gaussian and Poisson applications of data fission) that larger CIs are the price that is required in order to ensure coverage guarantees. Although we do not have a precise explanation for this phenomenon, it may be due to the misspecification that results from data fission recasting the target parameters into something ``easier'' to estimate. 
\begin{figure}[H]
\centering
    \begin{subfigure}[t]{1\textwidth}
        \centering
        \includegraphics[width=0.3\linewidth]{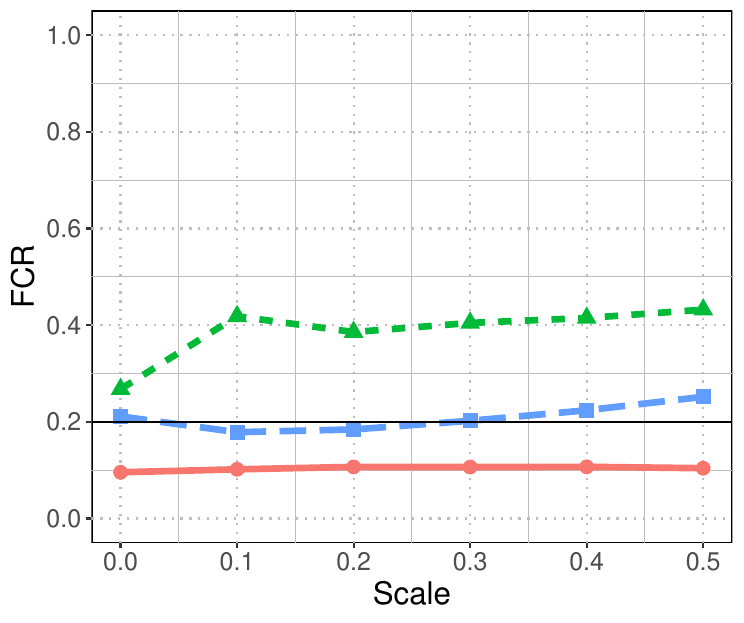}
    \hfill
        \includegraphics[width=0.3\linewidth]{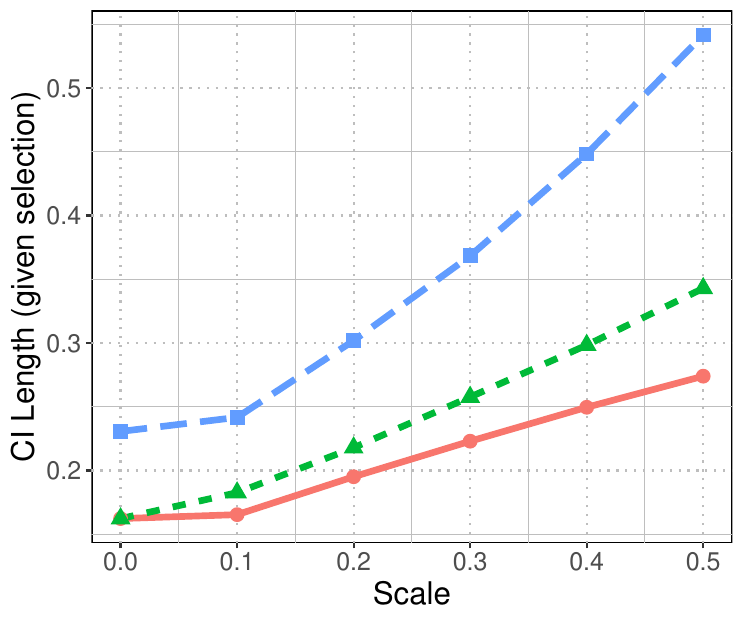}
    \hfill
        \includegraphics[width=0.3\linewidth]{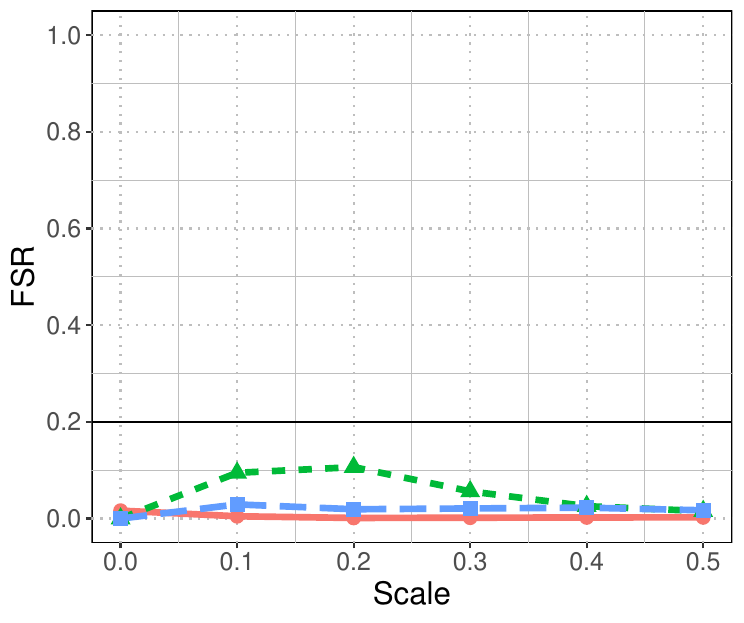}
    \end{subfigure}
\vfill
    \begin{subfigure}[t]{1\textwidth}
        \centering
        \includegraphics[width=0.3\linewidth]{figures/legend_fission_regression.png}
    \end{subfigure}
\vfill
    \begin{subfigure}[t]{1\textwidth}
        \centering
        \includegraphics[width=0.3\linewidth]{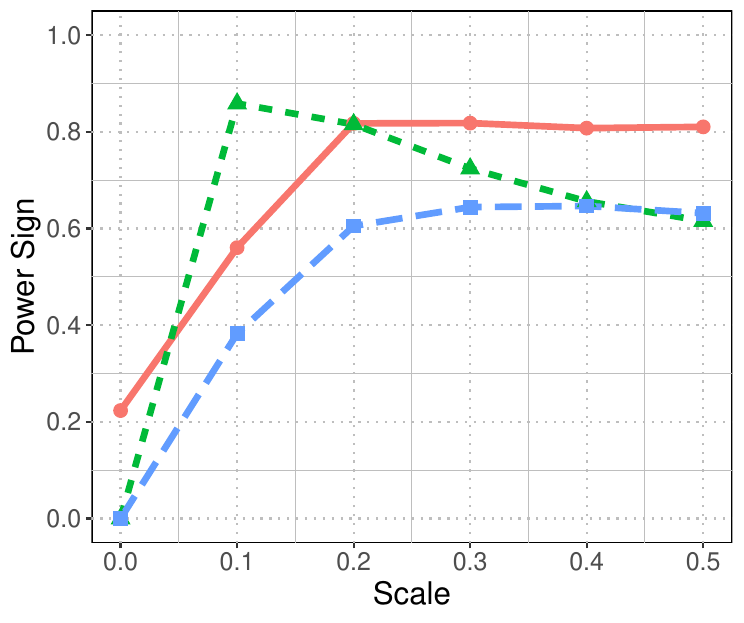}
    \hfill
        \includegraphics[width=0.3\linewidth]{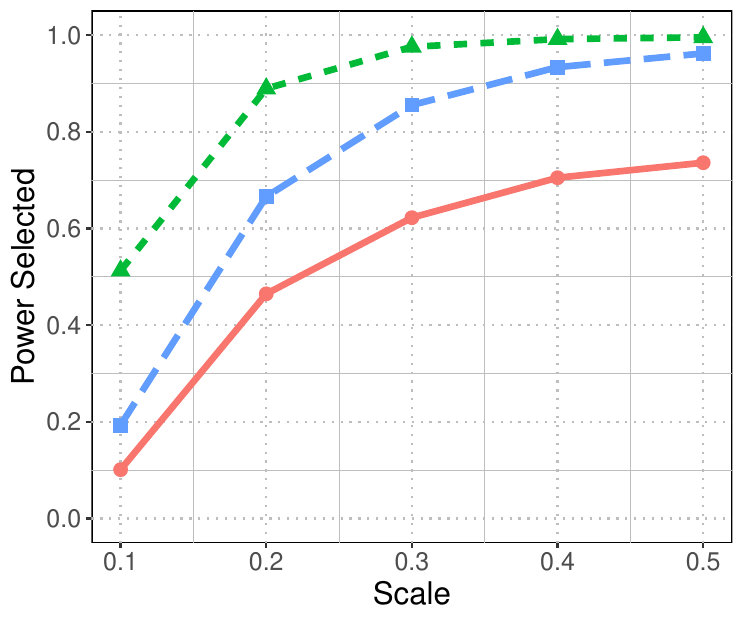}
    \hfill
        \includegraphics[width=0.3\linewidth]{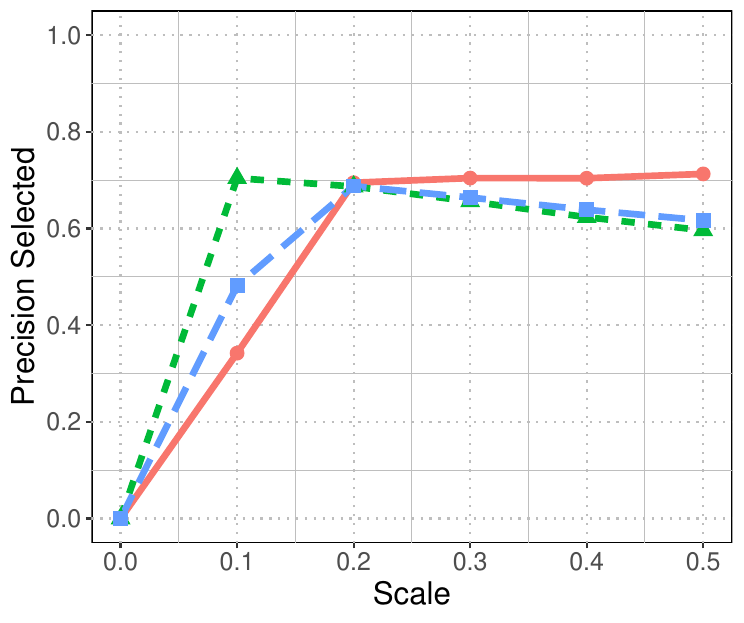}
    \end{subfigure}
\end{figure}
\begin{figure}[H]
\ContinuedFloat
\centering
    \caption{
    FCR, CI width, FSR, power for the sign of parameters, and power and precision for the selected features, when varying the signal strength $S_\Delta$ in $\{0, 0.1, 0.2, 0.3, 0.4, 0.5\}$. The hyperparameter $p$ is chosen as $0.2$ and target FCR for CIs is set at $0.2$.  We note that data fission offers smaller CI widths when compared to both data splitting and the (invalid) approach that reuses the full dataset for both selection and inference.
    }
    \label{fig:i-positive2}
\end{figure}

\paragraph{Simulations varying fission hyperparameter $p$.} We also examine the performance of these experiments as we vary the fission hyperparameter $p$. We note that values of $p$ close to either $0$ or $1$ correspond to having more information reserved for the selection step, while setting $p=0.5$ maximizes the amount of information reserved for inference. In \cref{fig:i-positive9}, we can see this behavior manifested in the ``dips''  in the CI length, power and precision as $p$ ranges from $0.3$ to $0.5$.
\begin{figure}[H]
\centering
    \begin{subfigure}[t]{1\textwidth}
        \centering
        \includegraphics[width=0.3\linewidth]{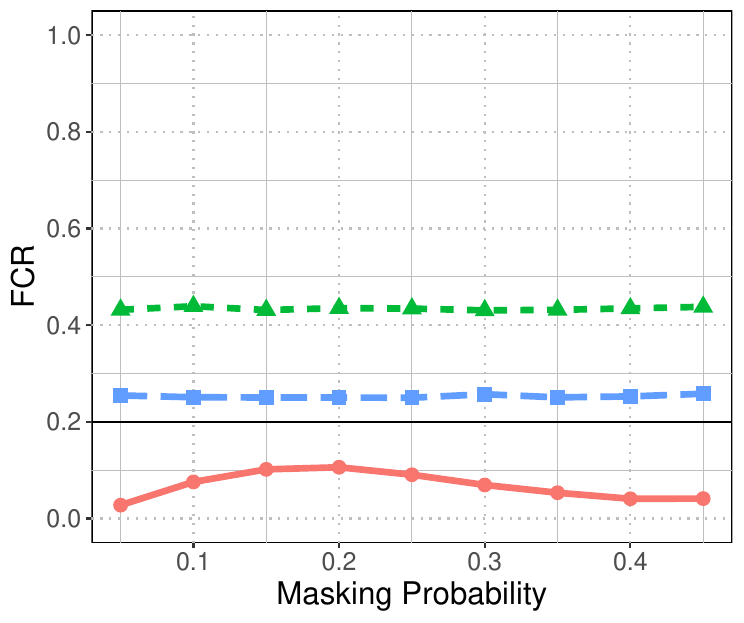}
    \hfill
        \includegraphics[width=0.3\linewidth]{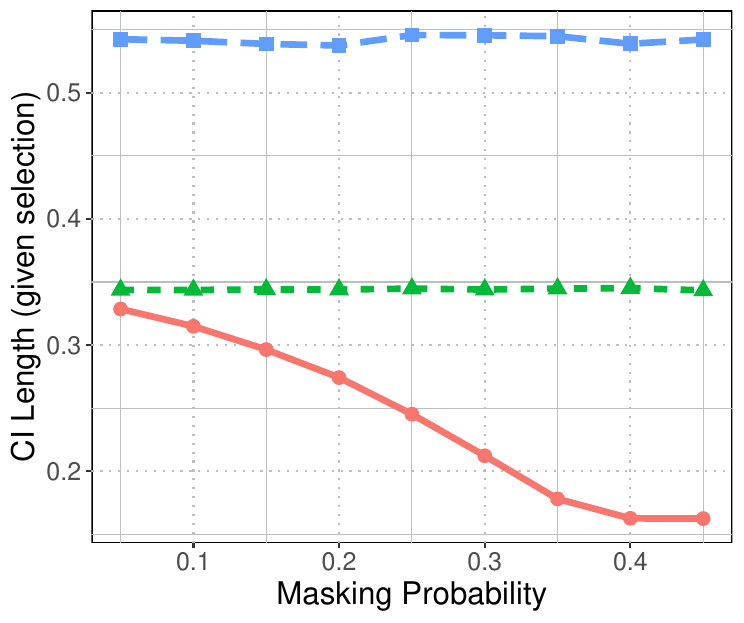}
    \hfill
        \includegraphics[width=0.3\linewidth]{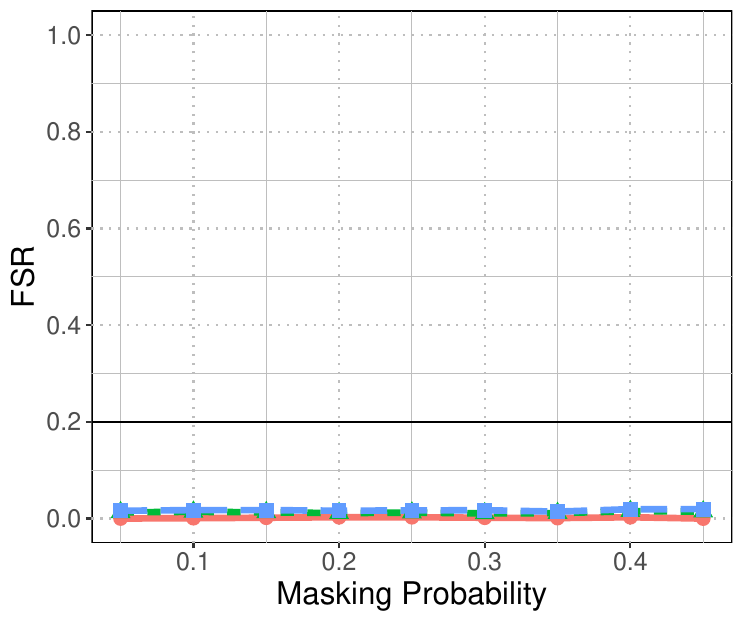}
    \end{subfigure}
\vfill
    \begin{subfigure}[t]{1\textwidth}
        \centering
        \includegraphics[width=0.3\linewidth]{figures/legend_fission_regression.png}
    \end{subfigure}
\vfill
    \begin{subfigure}[t]{1\textwidth}
        \centering
        \includegraphics[width=0.3\linewidth]{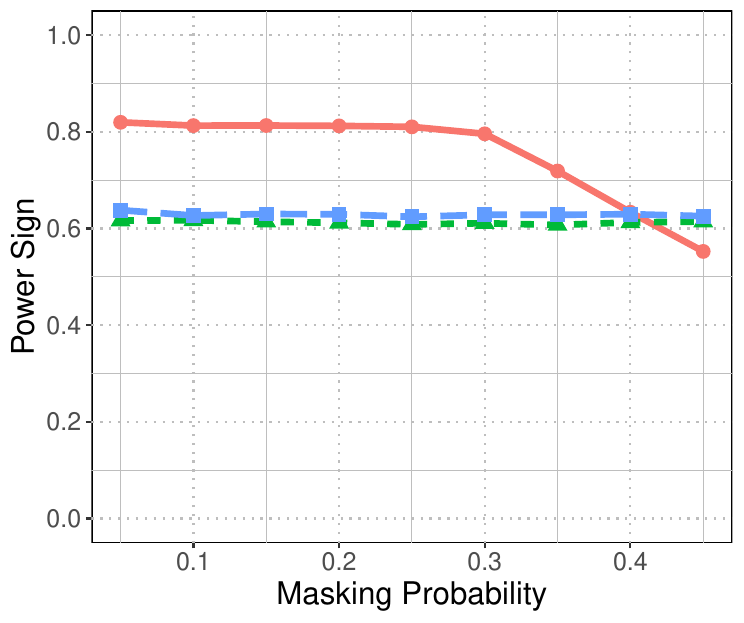}
    \hfill
        \includegraphics[width=0.3\linewidth]{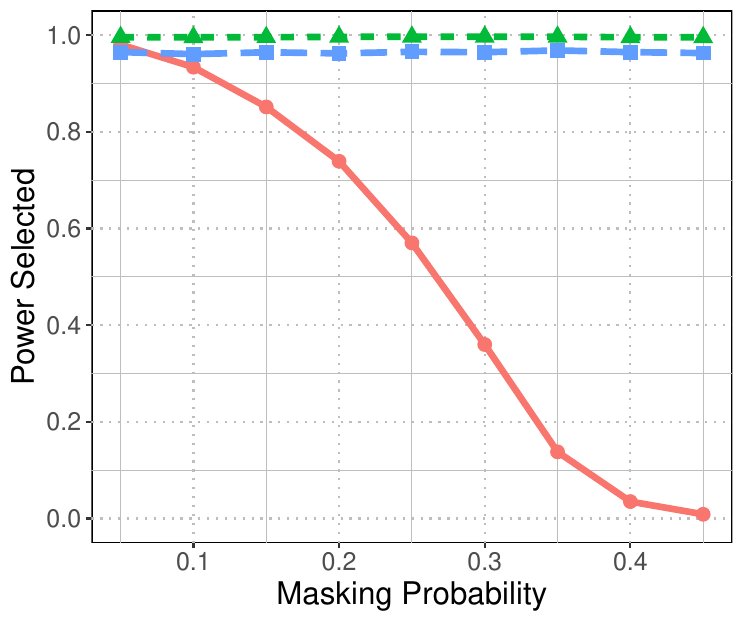}
    \hfill
        \includegraphics[width=0.3\linewidth]{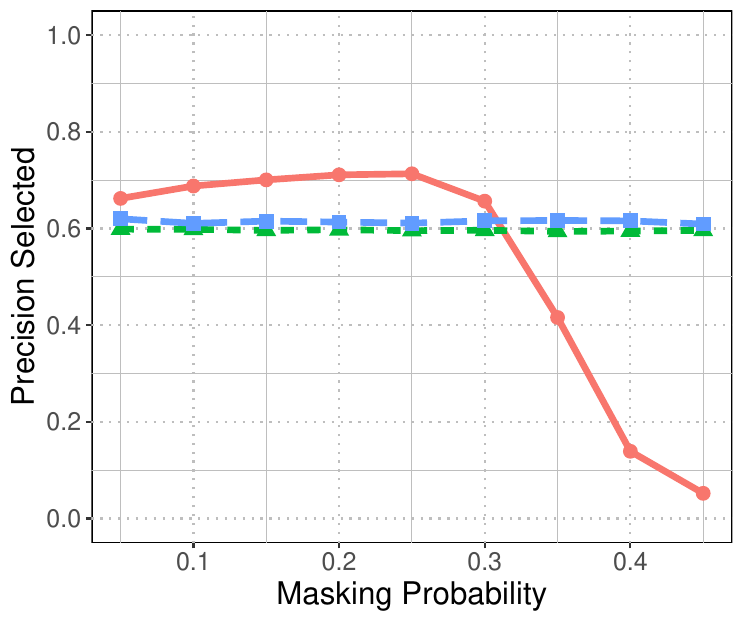}
    \end{subfigure}
    \caption{FCR, CI width, power for the sign of parameters, and power/precision for the selected features, when varying the hyperparameter $p$ in $\{0.05, 0.1, 0.15, 0.2, 0.25, 0.3, 0.35, 0.4, 0.45\}$ with fixed signal strength $S_\Delta=0.5$. The power and precision are higher when $p$ is closer to either $0$ or $1$  rather than $0.5$. Interestingly, however, the CI lengths only slightly increases as $p$ varies away from $0.5$.}
    %with slightly increase in the CI length (due to larger difference between $\mathbb{E}[g(Y_i) \mid X_i, f(Y_i) = 1]$ and $\mathbb{E}[g(Y_i) \mid X_i, f(Y_i) = 0]$ while we didn't include $f(Y_i)$ as a covariate when fitting $g(Y_i)$).
    \label{fig:i-positive9}
\end{figure}

\subsection{Simulation results for trend filtering}\label{appendix:trendfilter_supplemental}
In \cref{sec:trendfilter_empirical}, we note that data fission enables the construction of confidence intervals with proper pointwise and uniform coverage but at the price of wider confidence intervals when compared with using the full dataset twice for selection and inference. We expand on this finding through additional simulations. In particular:
\begin{itemize}
    \item  In Appendices~ \ref{sec:appendix_pointwise} and~\ref{sec:appendix_uniform}, we investigate how the confidence interval lengths compare between these two methods.
    \item In Appendix~\ref{sec:trendfilter_estimated_sigma}, we investigate an extension of the trend filtering methodology that accounts for unknown variance. 
    \item In Appendix~\ref{sec:appendix_alternative_knots}, we investigate alternative ways of selecting knots apart from cross validation such as Stein's unbiased risk estimate (SURE). 
\end{itemize}

%By the repetitive simulations, the averaged length (over time points) of CI seems to have a skewed distribution with a few large values; thus, we show both the mean and median of the CI length over repetitive experiments. As expected, the length of uniform CI is larger than pointwise CI in the previous analysis (with median increases by roughly two times the length; and mean increases by more than twenty times). The length when using the blurred data does not monotonically increase with the noise variance, because the length varies with two quantities: a multiplier $c$ and the estimated noise SD $\widehat \sigma$. While $\widehat \sigma$ increases, the multiplier $c$ decreases because the number of identified knots decreases. 

\subsubsection{Additional results for pointwise CI simulations} \label{sec:appendix_pointwise}
In addition to FCR, we also track the pointwise CI length for the simulations described in \cref{sec:trendfilter_empirical}. We see that when using the full data twice, the average CI length shrinks drastically when the probability of new knots increases while data fission enforces confidence intervals that grow in size as the number of knots increases. In contrast, when the underlying variance increases, both methods increase the length of the confidence intervals. 

\begin{figure}[H]      
\centering
    \begin{subfigure}[t]{0.3\textwidth}
        \includegraphics[width=1\linewidth]{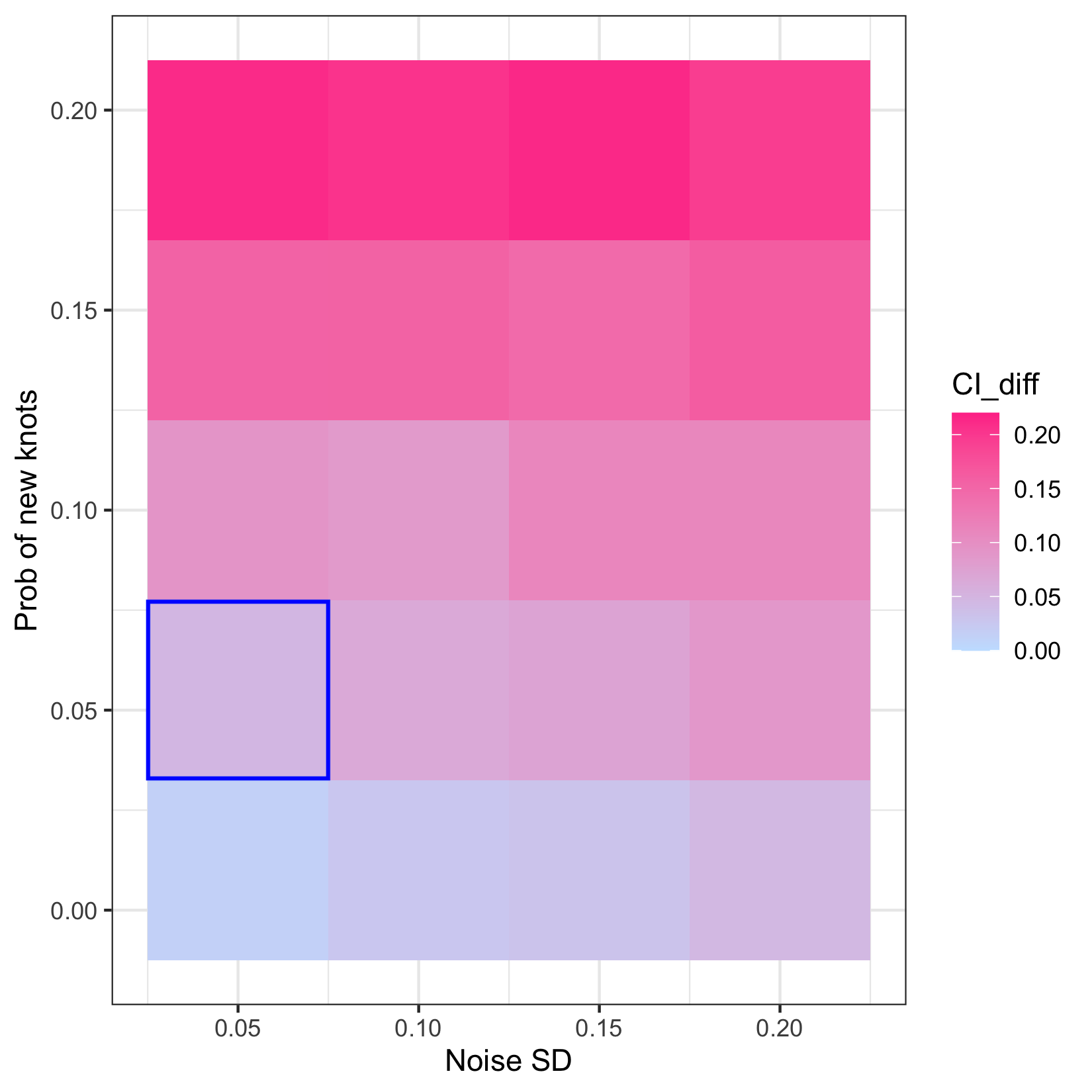}
        \caption{Difference in the averaged CI length when reusing the full dataset, versus data fission.}
    \end{subfigure}
\hfill
    \begin{subfigure}[t]{0.3\textwidth}
        \includegraphics[width=1\linewidth]{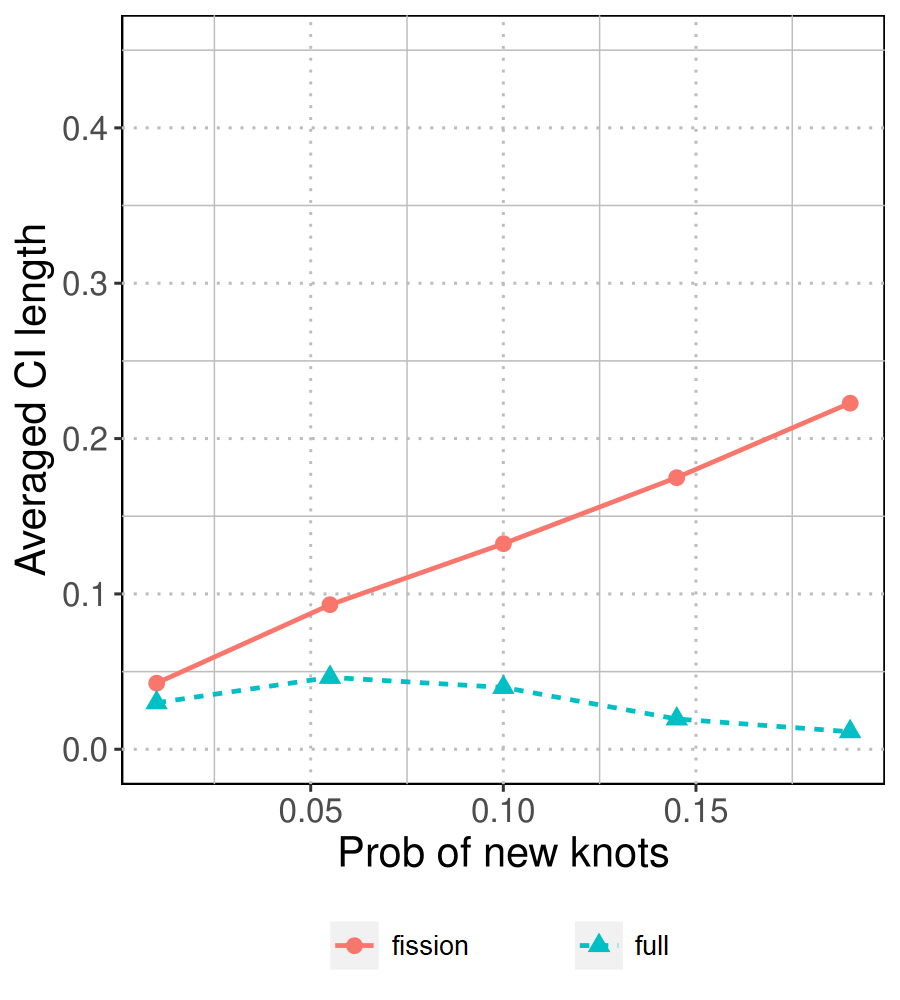}
        \caption{CI lengths when noise SD is 0.05.}
    \end{subfigure}
\hfill
    \begin{subfigure}[t]{0.3\textwidth}
        \includegraphics[width=1\linewidth]{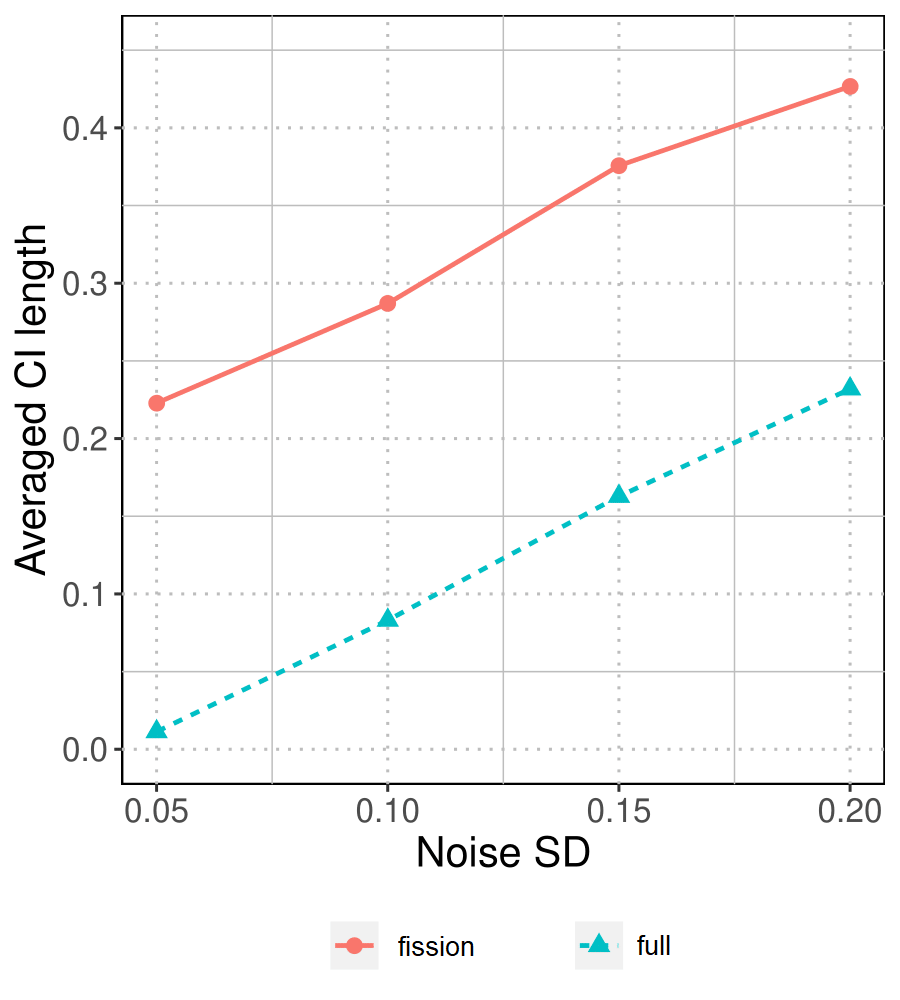}
        \caption{CI lengths when the probability of new knots is $0.2$.}
    \end{subfigure} 
    \caption{The width of \textbf{pointwise CIs} when varying the probability of having knots $p$ in $\{0.01, 0.55, 0.1, 0.145, 0.19\}$ and the noise SD in $\{0.05, 0.1, 0.15, 0.2\}$. The widths for CIs constructed using data fission range from 0.03 to 0.43, and are larger than the CIs constructed when the full dataset is reused for inference. Both methods have increased CI length as the noise SD increases, but when reusing the full dataset, the length is still not adequate to ensure FCR control. Interestingly, only the data fission approach increases the length of the CI as the number of knots increases.}
    \label{fig:trendfilter_CI} 
\end{figure}

\subsubsection{Additional results for uniform CI simulations} \label{sec:appendix_uniform}
In addition to simultaneous type I error control, we also track the uniform CI length for the simulations described in \cref{sec:trendfilter_empirical} in \cref{fig:trendfilter2_CI}. The differences between CIs constructed using data fission and using the full dataset twice are most stark when the noise is small and the probability of new knots is large because these are the circumstances where the trend constructed using the full dataset is likely to overfit.
\begin{figure}[H]
\centering
     \begin{subfigure}[t]{0.32\textwidth}
        \centering
        \includegraphics[width=1\linewidth]{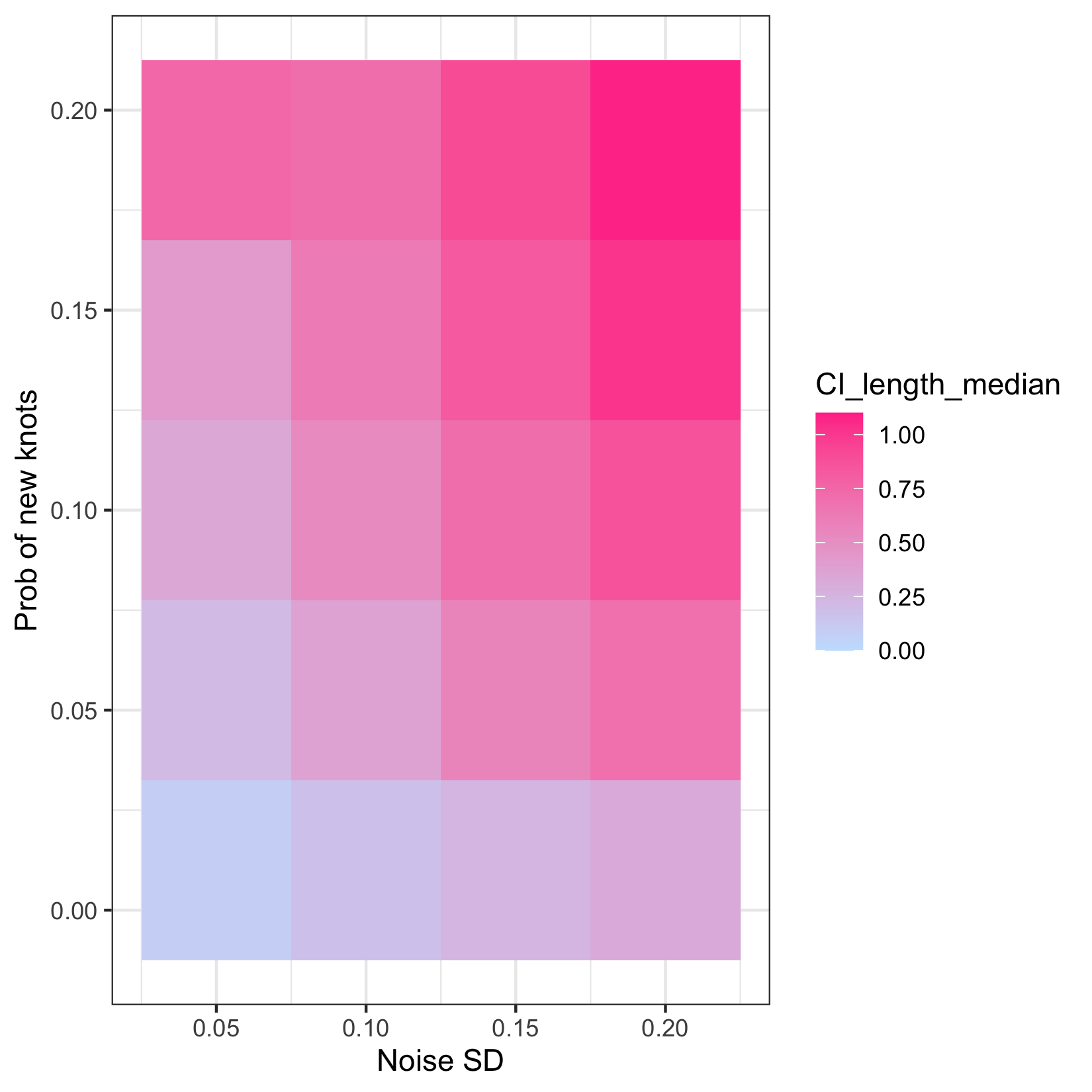}
    \end{subfigure}
\hfill
    \begin{subfigure}[t]{0.32\textwidth}
        \includegraphics[width=1\linewidth]{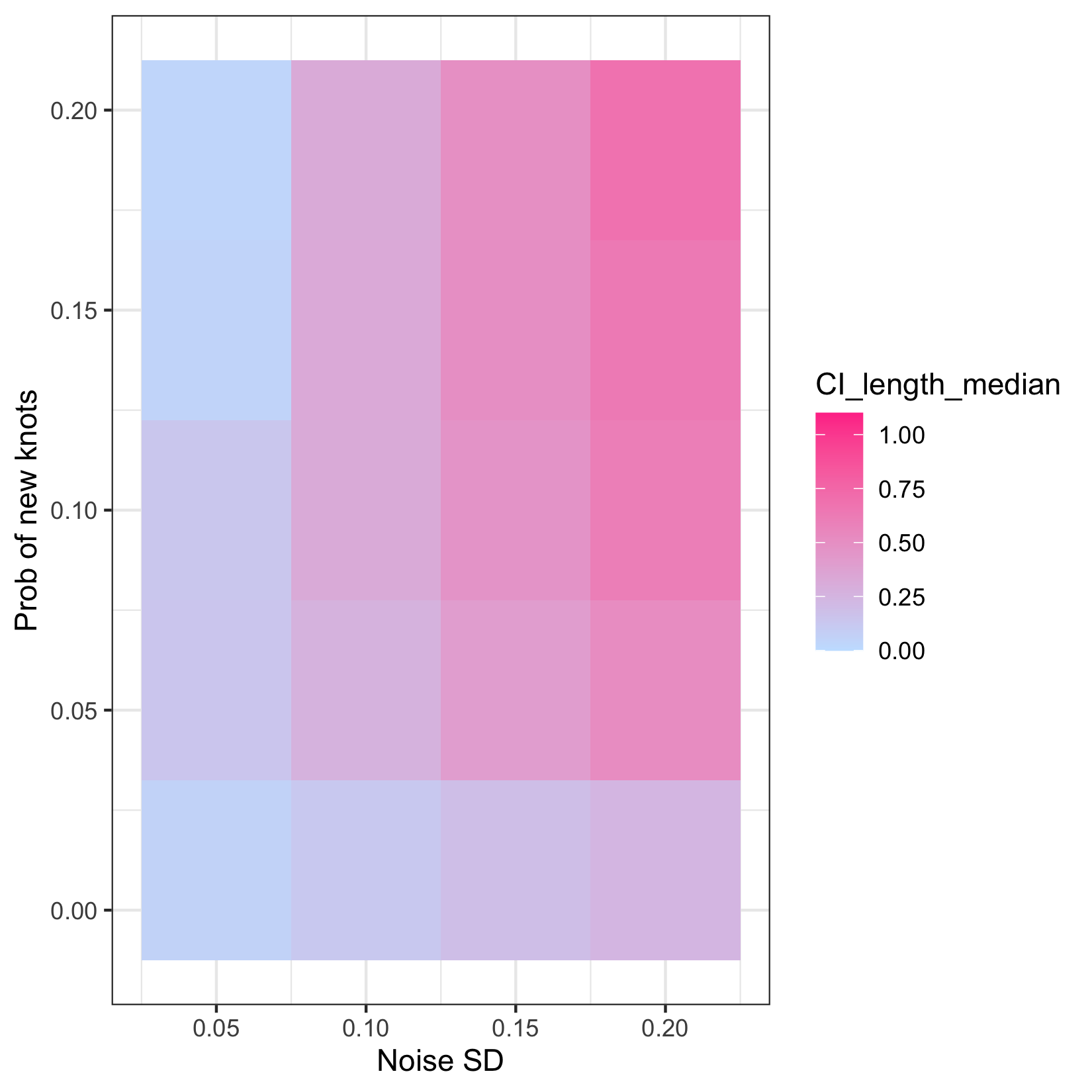}
    \end{subfigure}
\hfill
    \begin{subfigure}[t]{0.32\textwidth}
        \includegraphics[width=1\linewidth]{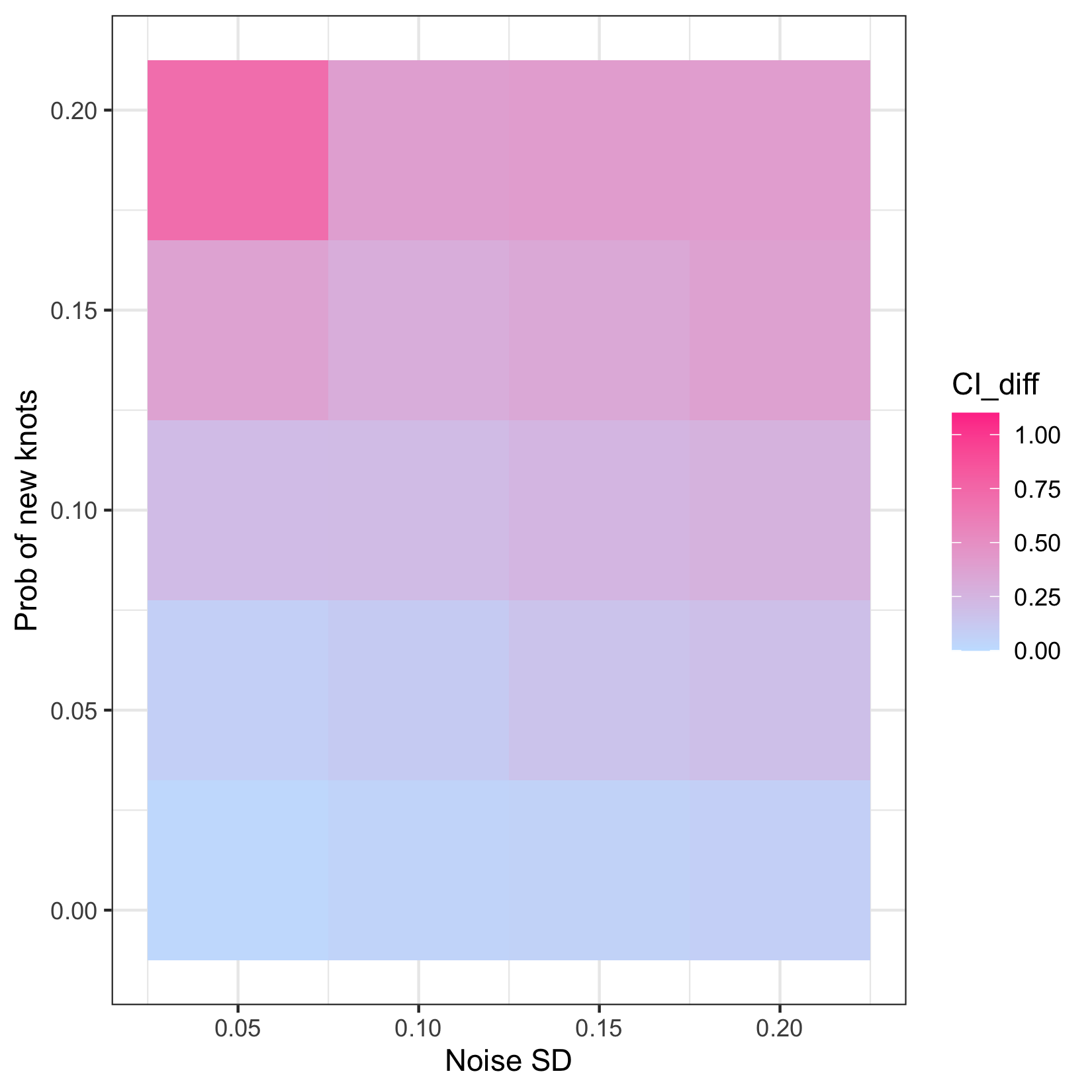}
    \end{subfigure}
\caption{The median CI width using either data fission (left) or full data twice (middle)  for \textbf{uniform CIs} have similar trends with respect to noise SD and probability of new knots. The CI difference (fission CI minus full-data CI) decreases when the noise SD increases (such that the effect of double dipping is smaller) or the probability of new knots decreases. These settings also correspond to instances where reusing the full dataset twice tends not violate FCR or type I error control too badly, since the width of the intervals is almost the same as those created from data fission.}

    \label{fig:trendfilter2_CI}
\end{figure}

\paragraph{Understanding the relationship of $c(\alpha)$ to changing CI lengths}
One obstacle in understanding how the widths of uniform confidence bands change as noise increases and the underlying structural trend changes is that the length is controlled by the multiplier $c(\alpha)$ used in Fact~\ref{trendfilter_uniform_method} which is somewhat opaque. Recall that $c(\alpha)$ is the solution to
\begin{align} \label{eq:multiplier}
	\frac{|\gamma|}{2\pi} e^{-c^{2}/2} + 1 - \Phi(c) = \alpha/2.
\end{align}
To aid in forming an intuition behind the empirical trends noted in \cref{fig:trendfilter2_CI}, we plot some intermediate statistics in \cref{fig:trendfilter_intermediate_statistics}. 

\begin{figure}[H]
\centering
     \begin{subfigure}[t]{0.22\textwidth}
        \centering
        \includegraphics[width=1\linewidth]{figures/regression_trendfilter_varySigmaProb_CI_length_median_masking.png}
    \end{subfigure}
\hfill
    \begin{subfigure}[t]{0.22\textwidth}
        \includegraphics[width=1\linewidth]{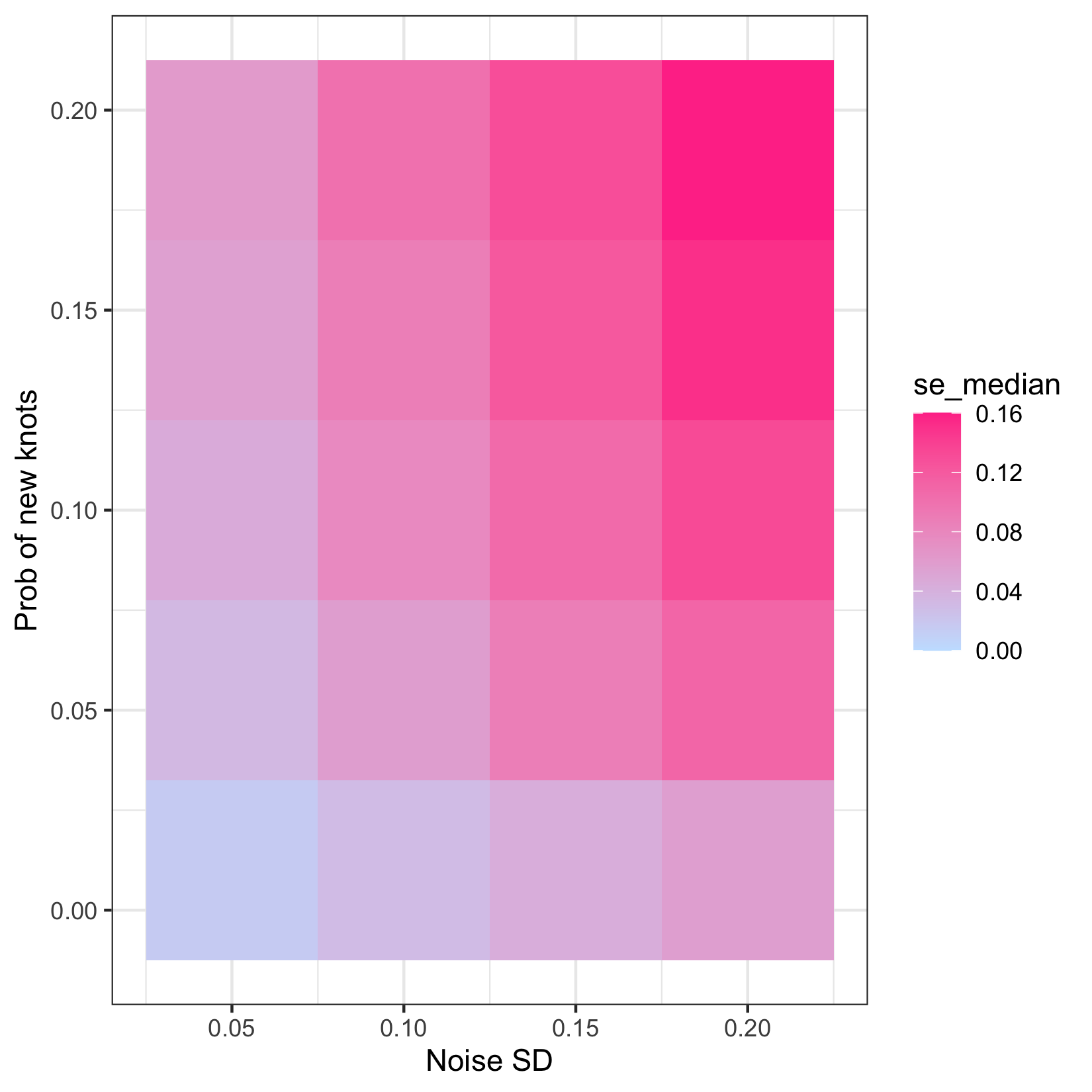}
    \end{subfigure}
\hfill
    \begin{subfigure}[t]{0.22\textwidth}
        \includegraphics[width=1\linewidth]{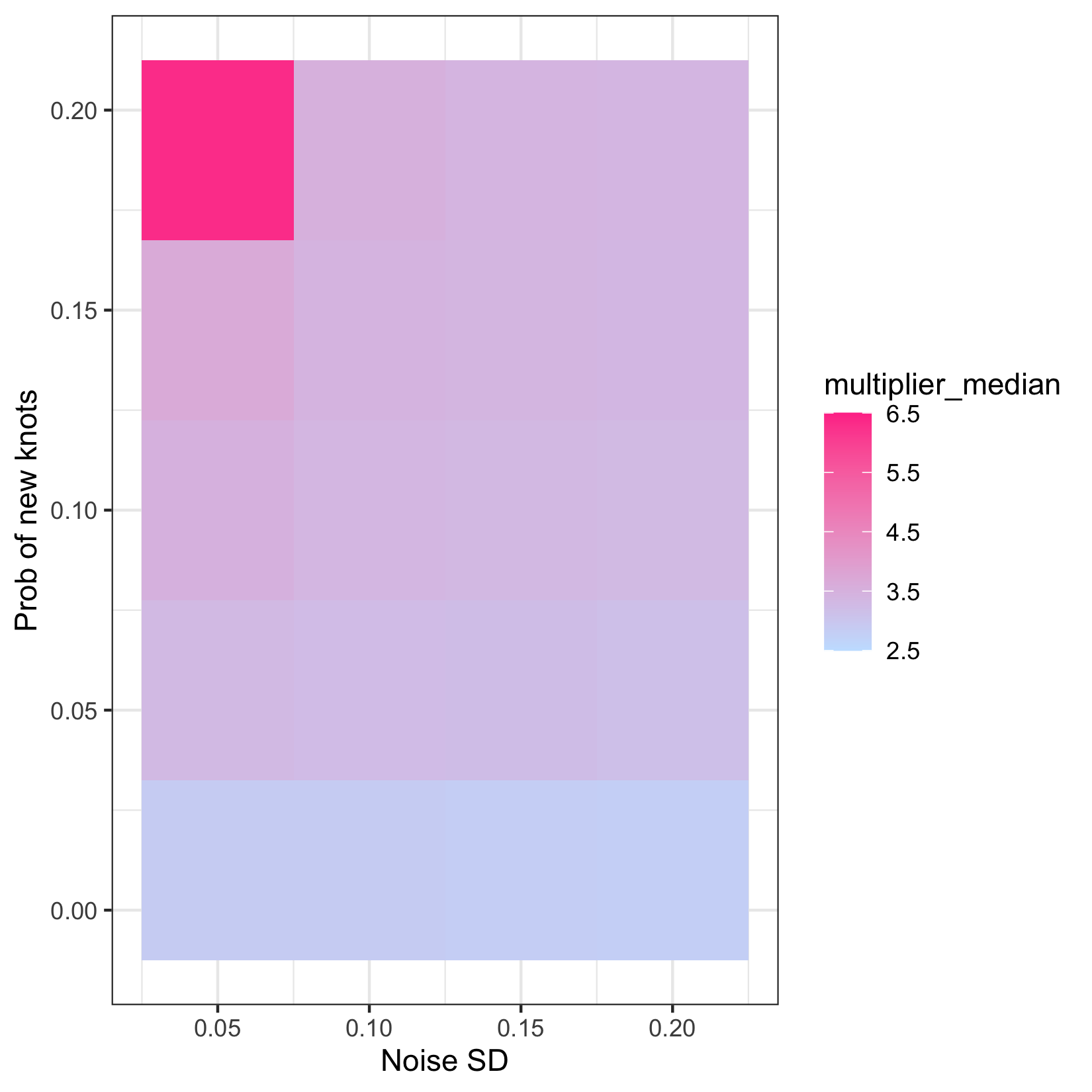}
    \end{subfigure}
\hfill
    \begin{subfigure}[t]{0.22\textwidth}
        \includegraphics[width=1\linewidth]{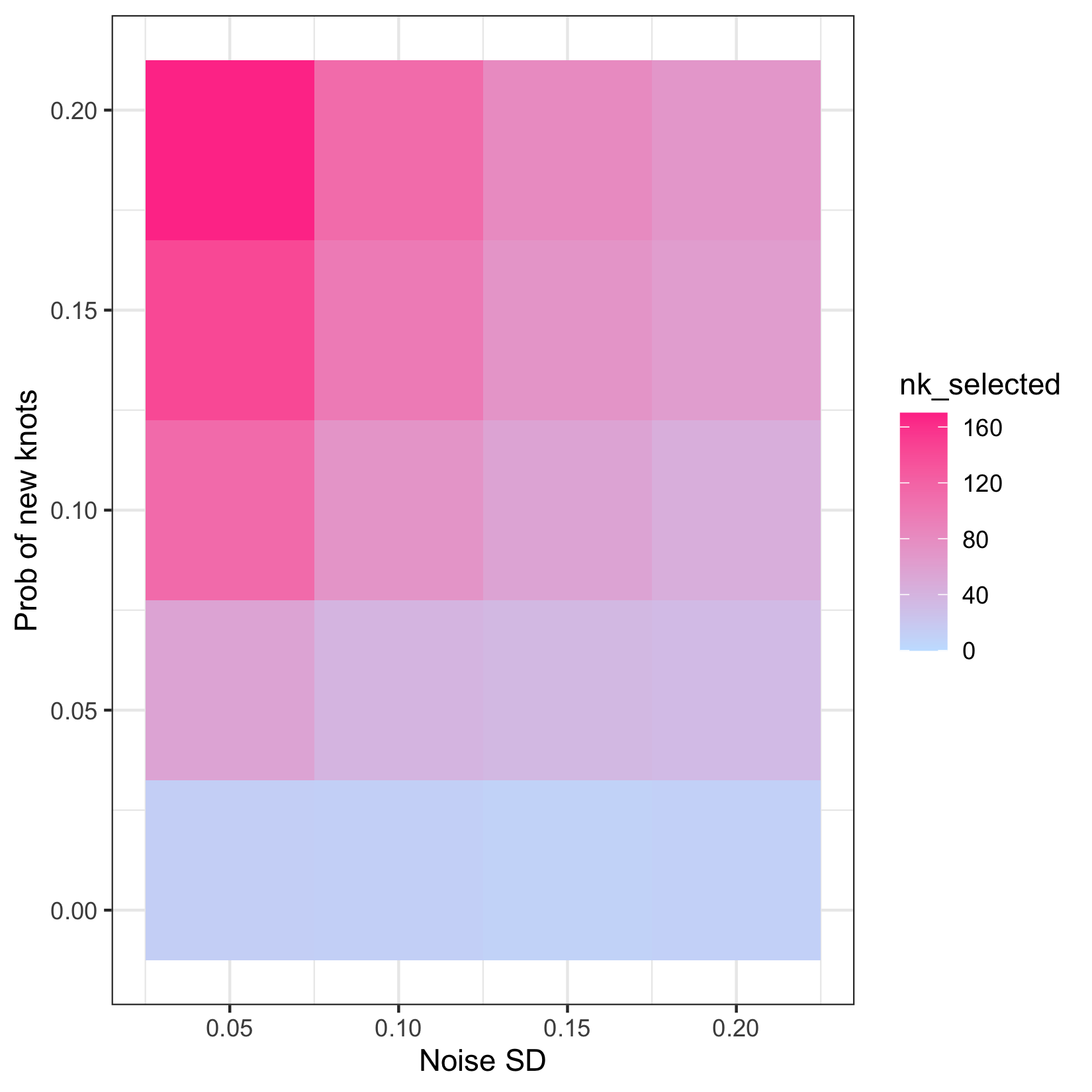}
     \end{subfigure}  
\caption{The CI width using data fission method increases with noise variance and the probability of new knots (first), because $\widehat{\text{SE}}$ increases (second) compared to the change in the multiplier $c(\alpha)$ (third). Both changes can be traced back to change in the number of knots (fourth): $\widehat{\text{SE}}$ decreases and $c(\alpha)$ increases with the number of knots.}
\label{fig:trendfilter_intermediate_statistics}
\end{figure}

The above plots are medians over repetitions since the distribution of CI width is skewed to large values; thus, the mean may not summarize the pattern clearly. The mean of CI width does not have a consistent trend because there are a  few extremely large  CI widths, due to extremely large multipliers $c(\alpha)$. Upon closer inspection, such large value exist when the number of knots is large. In \cref{fig:multiplier}, we show the value of left-hand side of~\eqref{eq:multiplier} for solving $c(\alpha)$ when the number of knots is $197$ (vs $200$ time points). The solution for $c(\alpha)$ is over 400 (at the intersection of the black path and the red level of $\alpha$).

\begin{figure}[H]
\centering
\hspace{1cm}
    \begin{subfigure}[t]{0.32\textwidth}
        \includegraphics[width=1\linewidth]{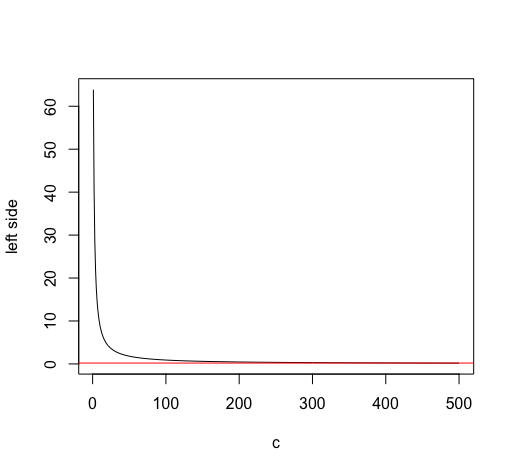}
        \caption{Trajectory of the left-hand side of~\eqref{eq:multiplier} when there are 197 knots.}
    \end{subfigure}
\hfill
    \begin{subfigure}[t]{0.32\textwidth}
        \includegraphics[width=1\linewidth]{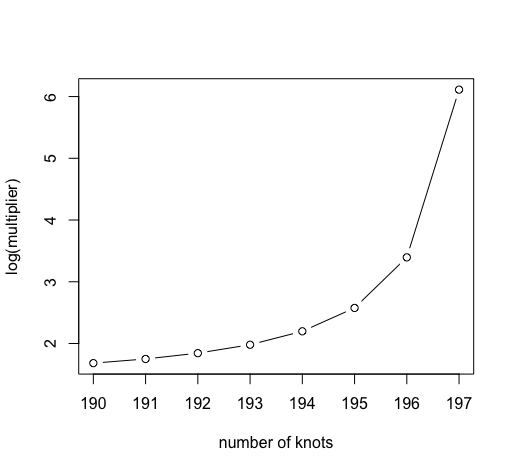}
        \caption{Log of the multiplier when the number of knots increases from 190 to 197.}
    \end{subfigure}
\hspace{1cm}
\caption{The multiplier $c(\alpha)$ increases dramatically when the number of knots is larger than 190.}
\label{fig:multiplier}
\end{figure}

\label{appendix:trendfilter_supplementa_uniforml}
\begin{figure}[H]
\centering
    \begin{subfigure}[t]{0.5\textwidth}
    \vskip 0pt
        \includegraphics[width=1\linewidth]{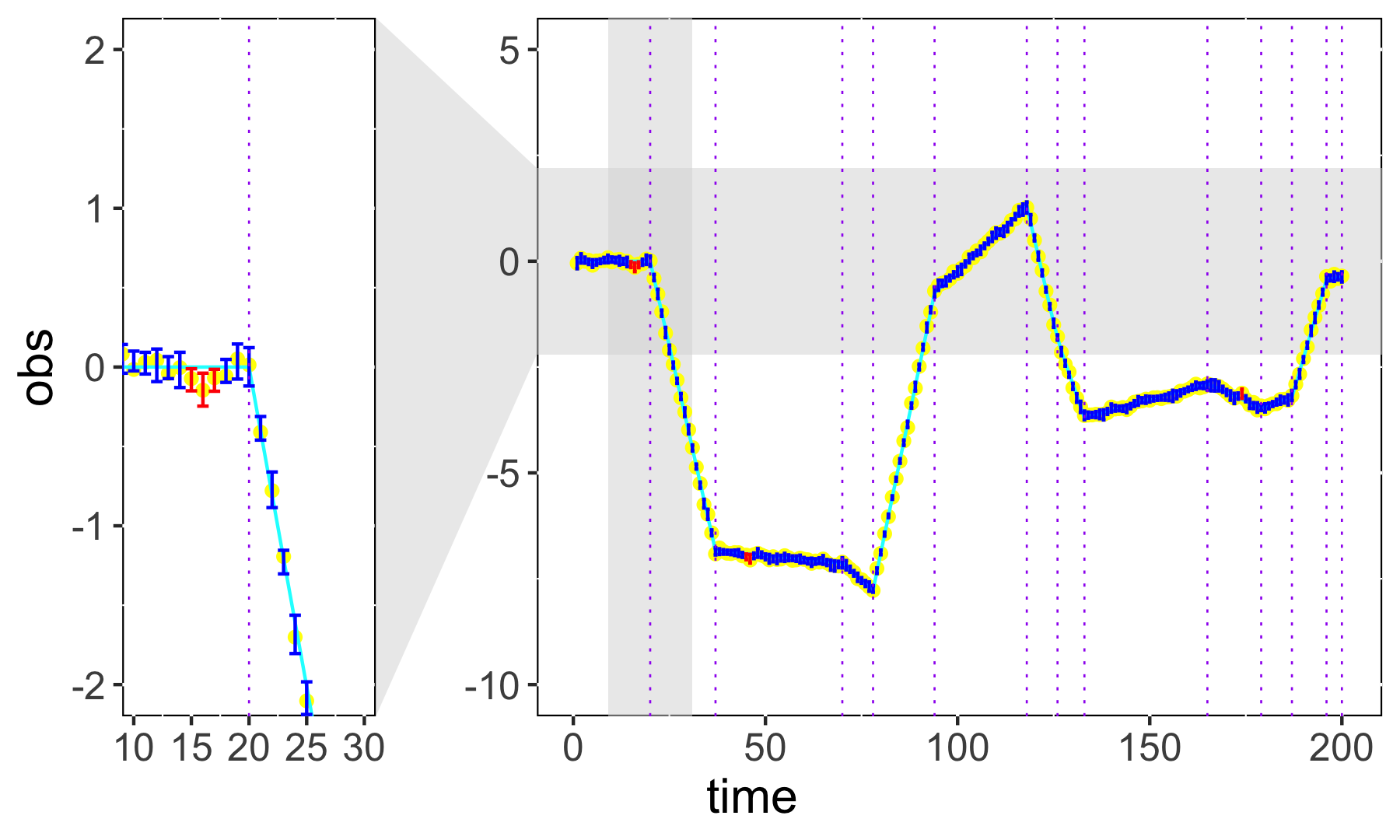}
    \end{subfigure}
    \hfill
    \begin{subfigure}[t]{0.35\textwidth}
    \vskip 0pt
        \includegraphics[ width=1\linewidth]{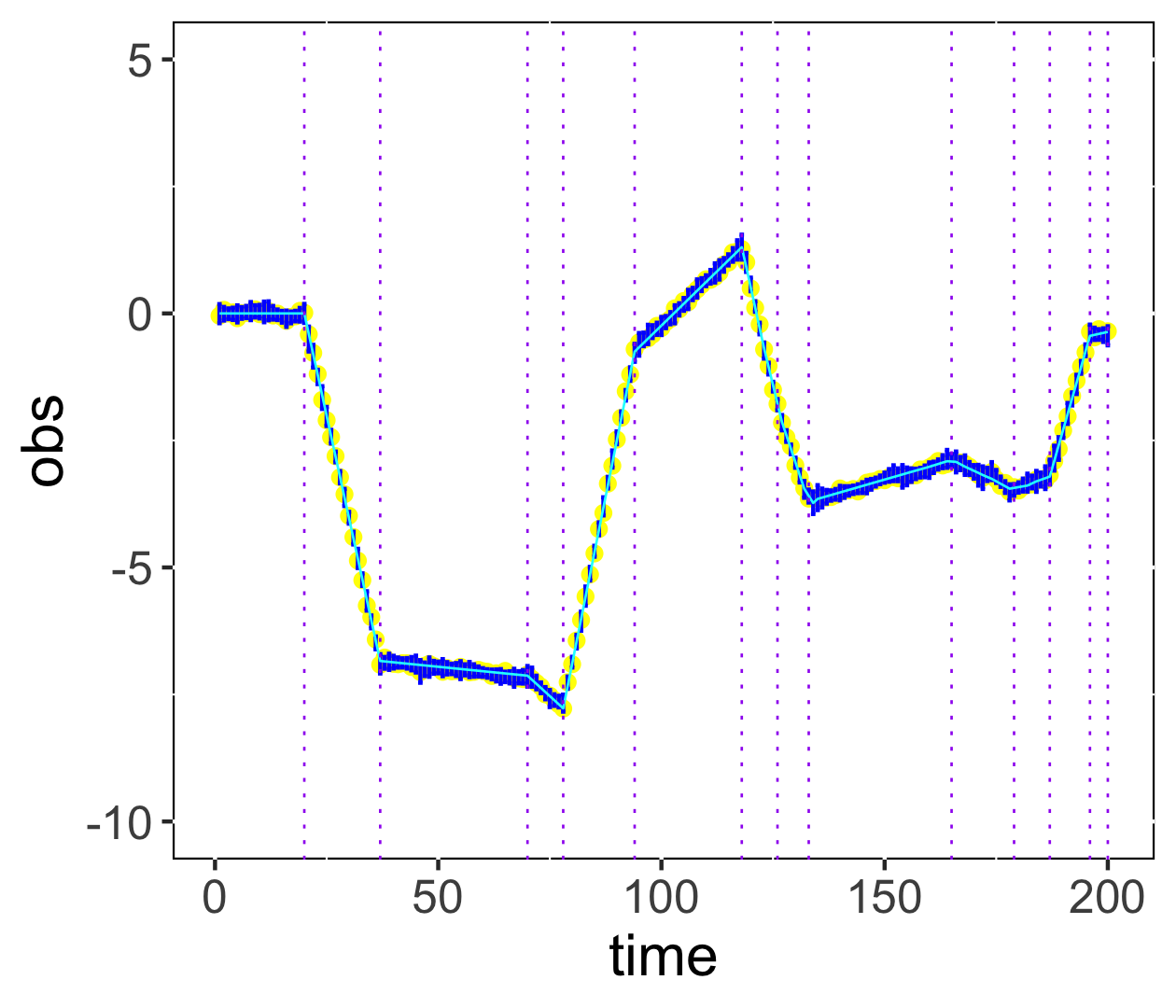}
    \end{subfigure}
\vfill
    \begin{subfigure}[t]{0.5\textwidth}
    \vskip 0pt
        \includegraphics[width=1\linewidth]{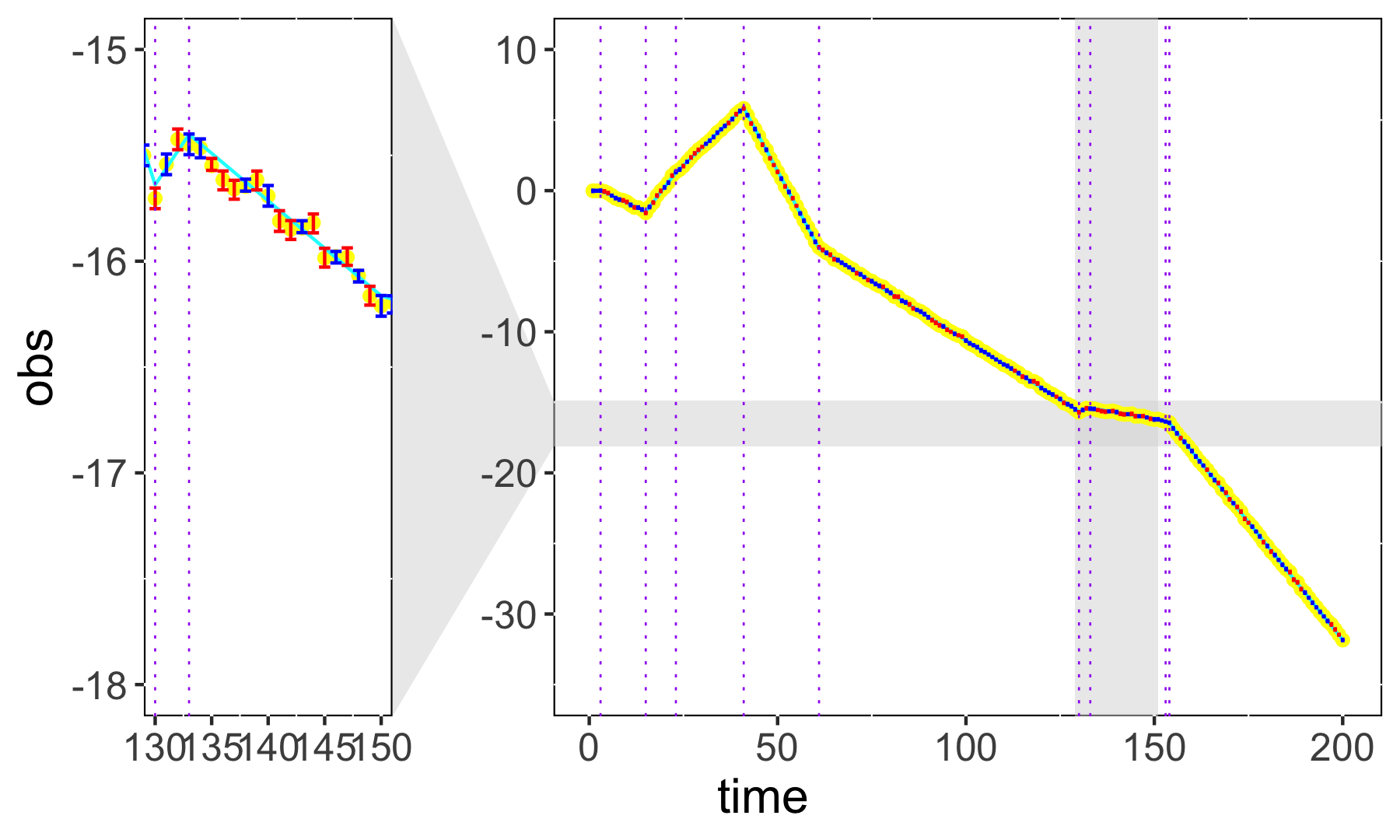}
    \end{subfigure}
    \hfill
    \begin{subfigure}[t]{0.35\textwidth}
    \vskip 0pt
        \includegraphics[ width=1\linewidth]{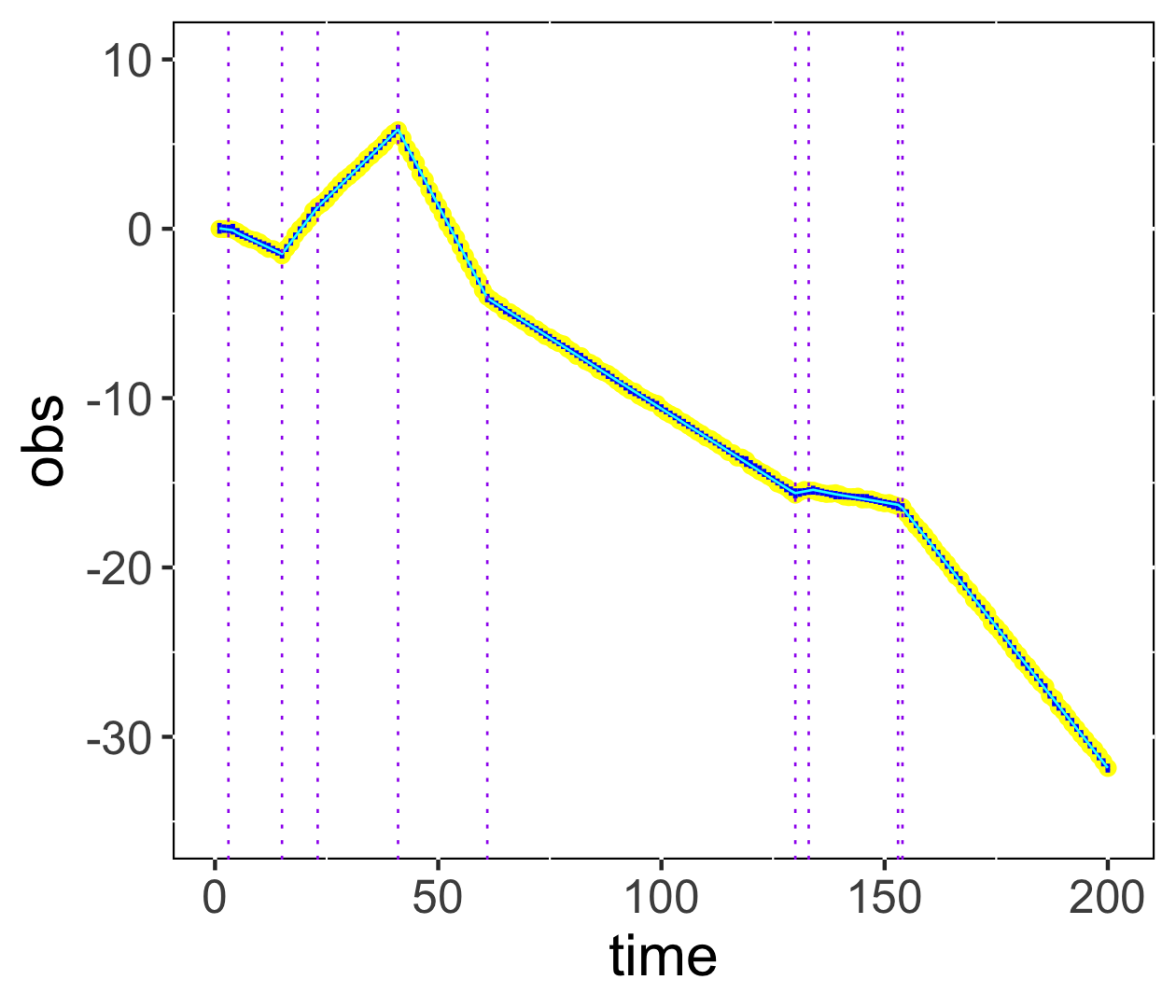}
    \end{subfigure}
\vfill
    \begin{subfigure}[t]{0.5\textwidth}
    \vskip 0pt
        \includegraphics[width=1\linewidth]{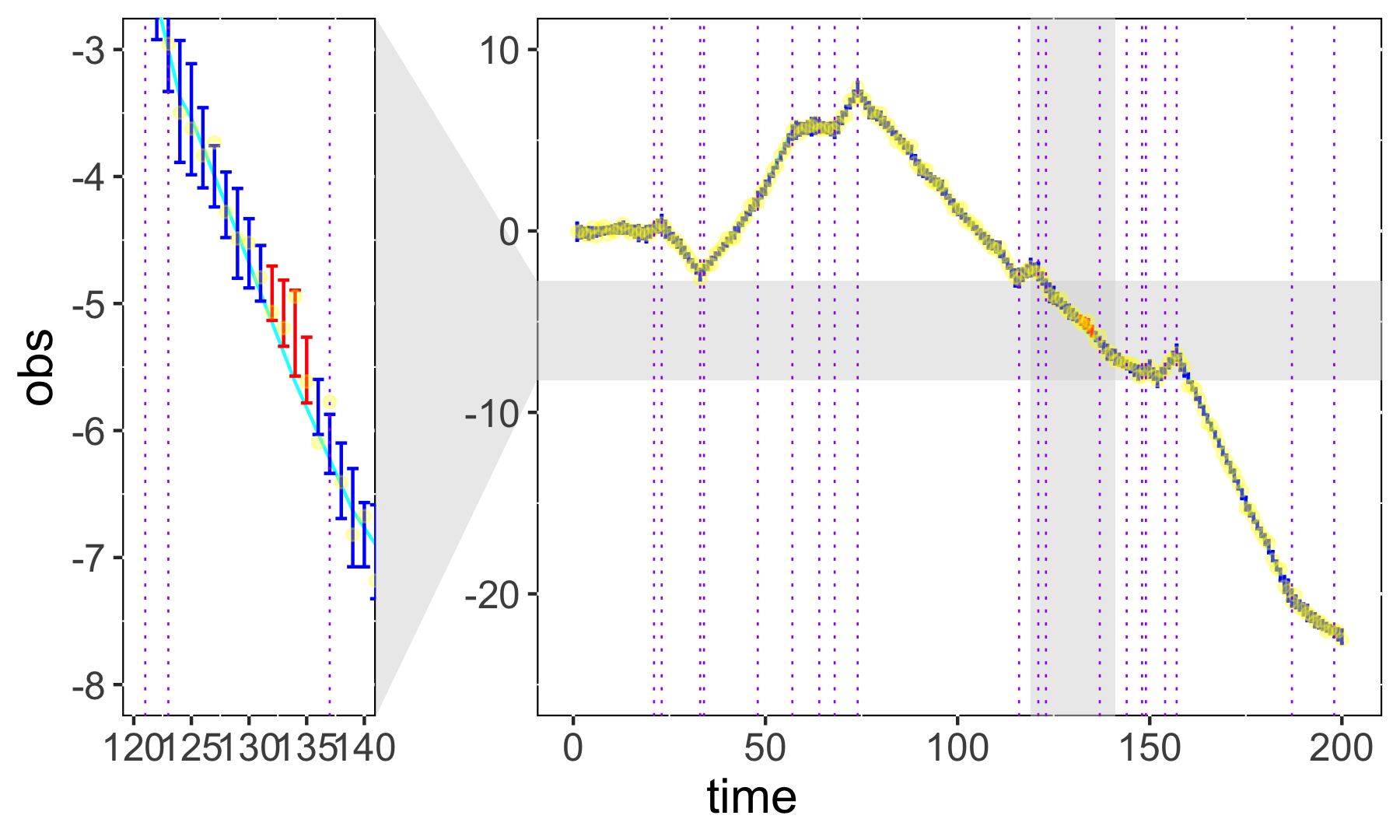}
    \end{subfigure}
    \hfill
    \begin{subfigure}[t]{0.35\textwidth}
    \vskip 0pt
        \includegraphics[ width=1\linewidth]{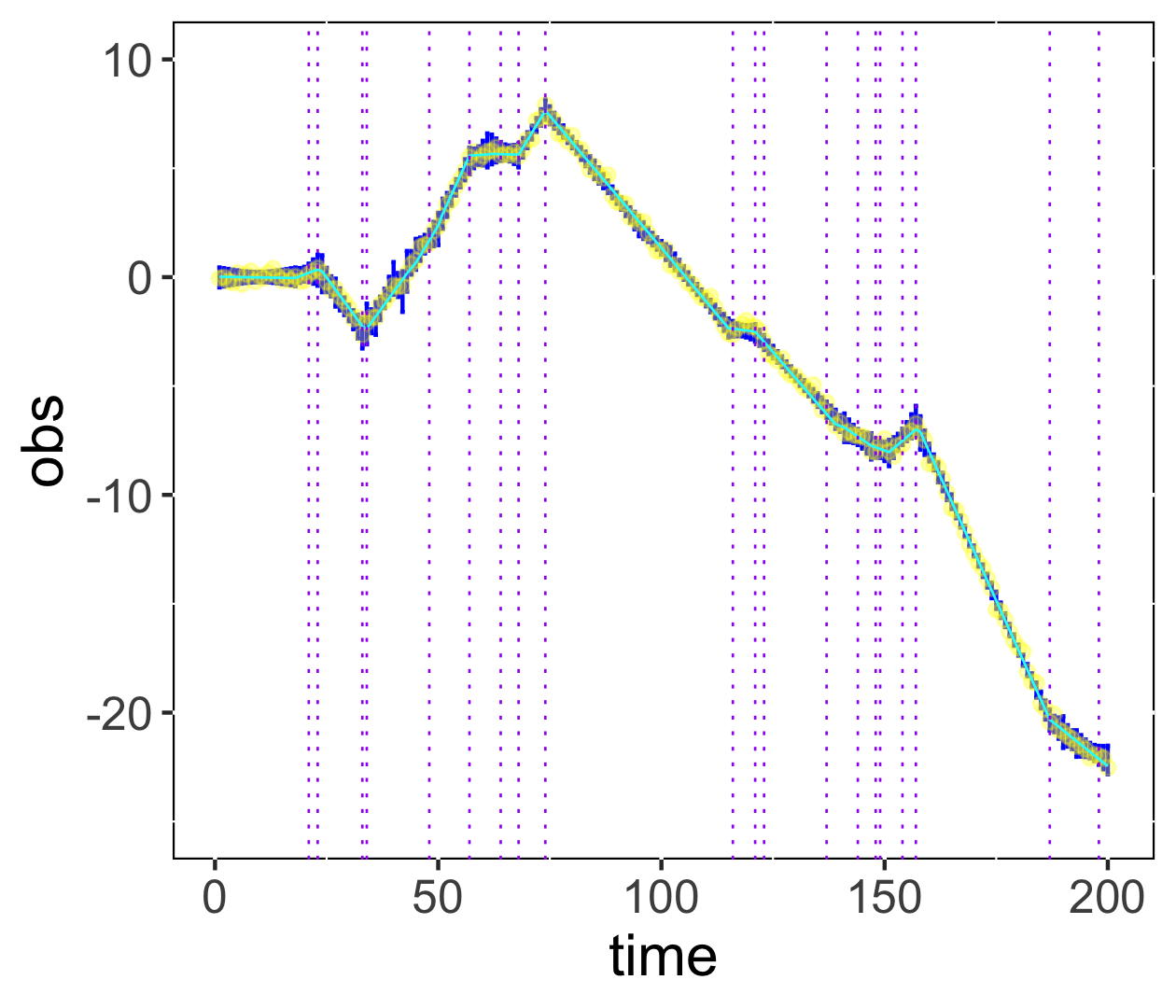}
    \end{subfigure}
\vfill
    \begin{subfigure}[t]{0.5\textwidth}
    \vskip 0pt
        \includegraphics[width=1\linewidth]{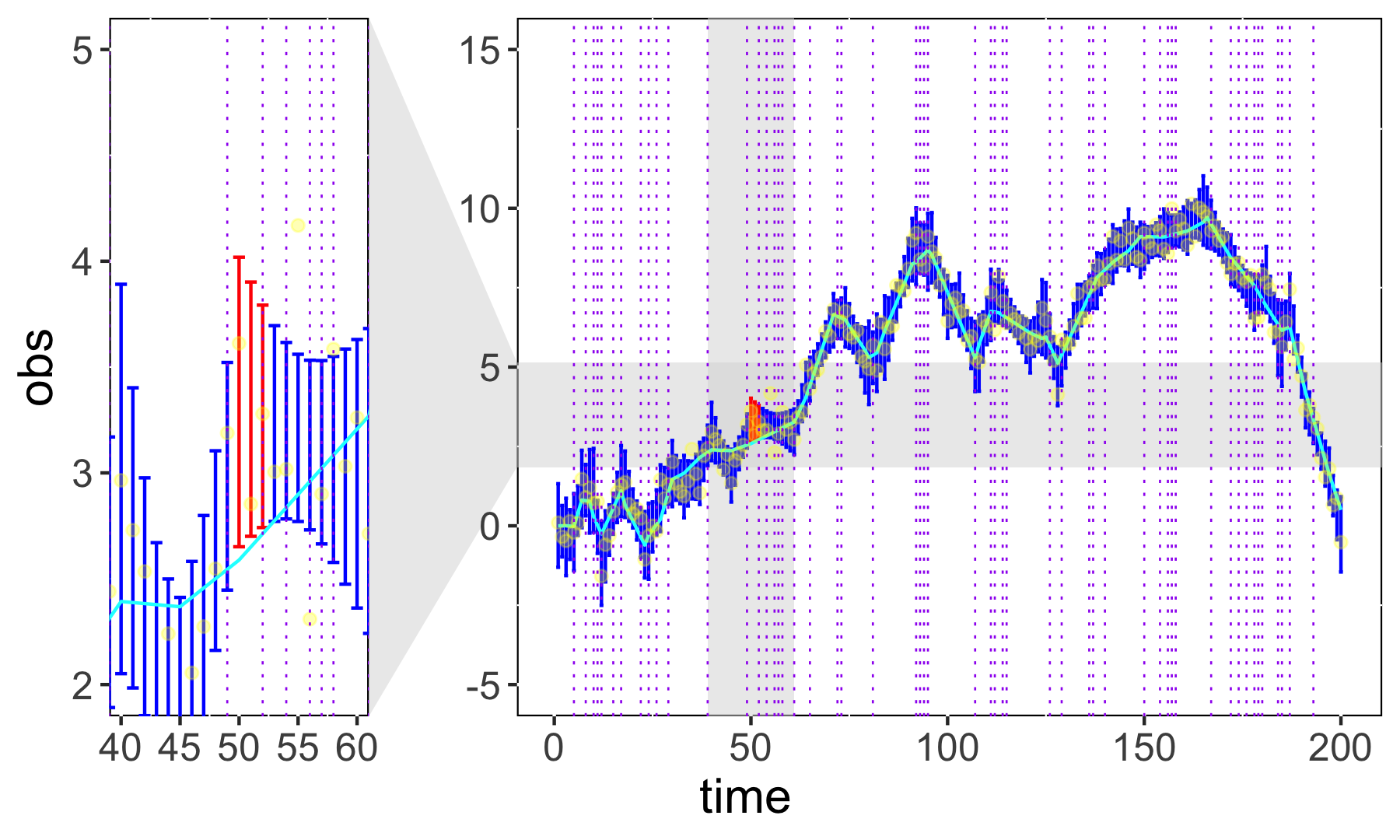}
    \end{subfigure}
    \hfill
    \begin{subfigure}[t]{0.35\textwidth}
    \vskip 0pt
        \includegraphics[ width=1\linewidth]{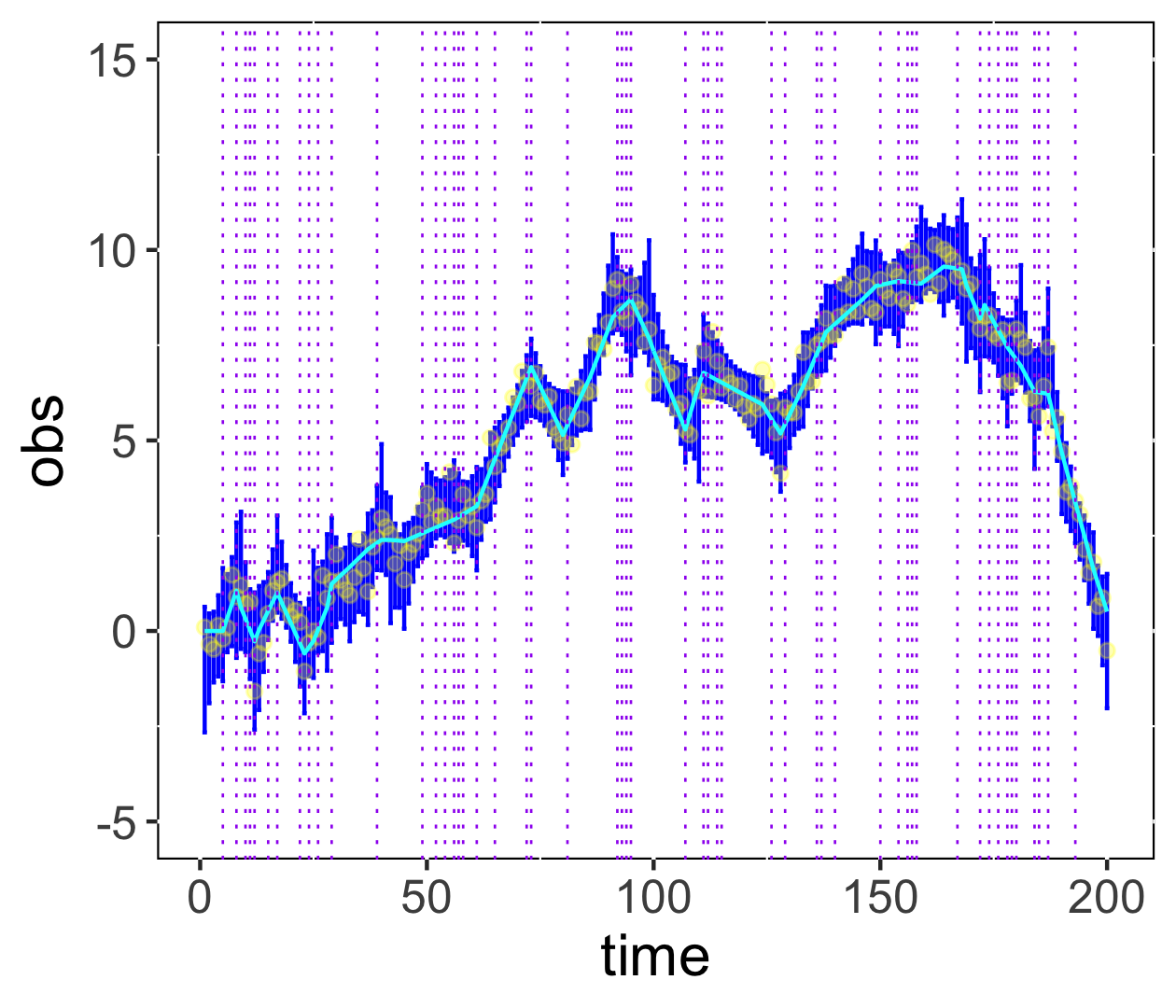}
    \end{subfigure}
    \caption{Four instances of the observed points (in yellow) and the \textit{uniform} CI (in blue if correctly cover the trend, in red if not) using two types of methods: full data twice (left), and data fission (right). The underlying projected mean is marked in cyan, which also mostly overlaps with the original true trend because both methods can find turning points easily under small noise.}
    \label{fig:instance_trendfilter2}
\end{figure}

\paragraph{Example trials for uniform confidence bands}
Similar to the examples shown in \cref{fig:instance_trendfilter}, we visualize the coverage of uniform confidence bands as well over four examples in \cref{fig:instance_trendfilter2}.  The CI using the full data fails to cover the projected mean at many time points because of double dipping, whereas data fission leads to a uniformly-valid CI with type I error control.

\subsubsection{Estimating variance before data fission} \label{sec:trendfilter_estimated_sigma}
The procedure discussed in Section~\ref{sec:trendfilter} uses the knowledge of variance $\sigma^2$, which is usually not known in practice. Alternatively, we can estimate the variance as $\widehat \sigma^2 = \tfrac{1}{2(n-1)}\sum_{i=1}^{n-1} (y_{t+1} - y_t)^2$ before step~1, and simulate $z_t \sim N(0, \widehat \sigma^2)$. We must also modify equation~\eqref{eq:multiplier} in step~4 slightly by replacing it with equation~(5.6) in \cite{koenker2011additive} and instead choose $c(\alpha)$ to be the solution of 
\begin{equation}
 \frac{|\gamma|}{2\pi} (1+c^{2}/v)^{-v/2} + P(t_{v} > c) = \alpha/2
\end{equation}
where $t_{v}$ denotes a $t$-distribution with $v=n-m-1$ degrees of freedom with $m$ being the number of knots. 

Notice that because we use the observed data to estimate variance, it could break the independence between $f(Y)$ and $g(Y)$. However, we notice in simulation that the error control still seems to hold in most cases. First, we evaluate how well the variance can be estimated across different settings in \cref{fig:variance_est}. The proposed methodology tends to overestimate the noise SD in general, with this overestimation increasing as the probability of new knots increases or the slope increases, because the formula to compute $\widehat \sigma^2$ given above treats all the change in adjacent time points as noise.

\begin{figure}[H]
\centering
    \begin{subfigure}[t]{0.3\textwidth}
        \centering
        \includegraphics[width=1\linewidth]{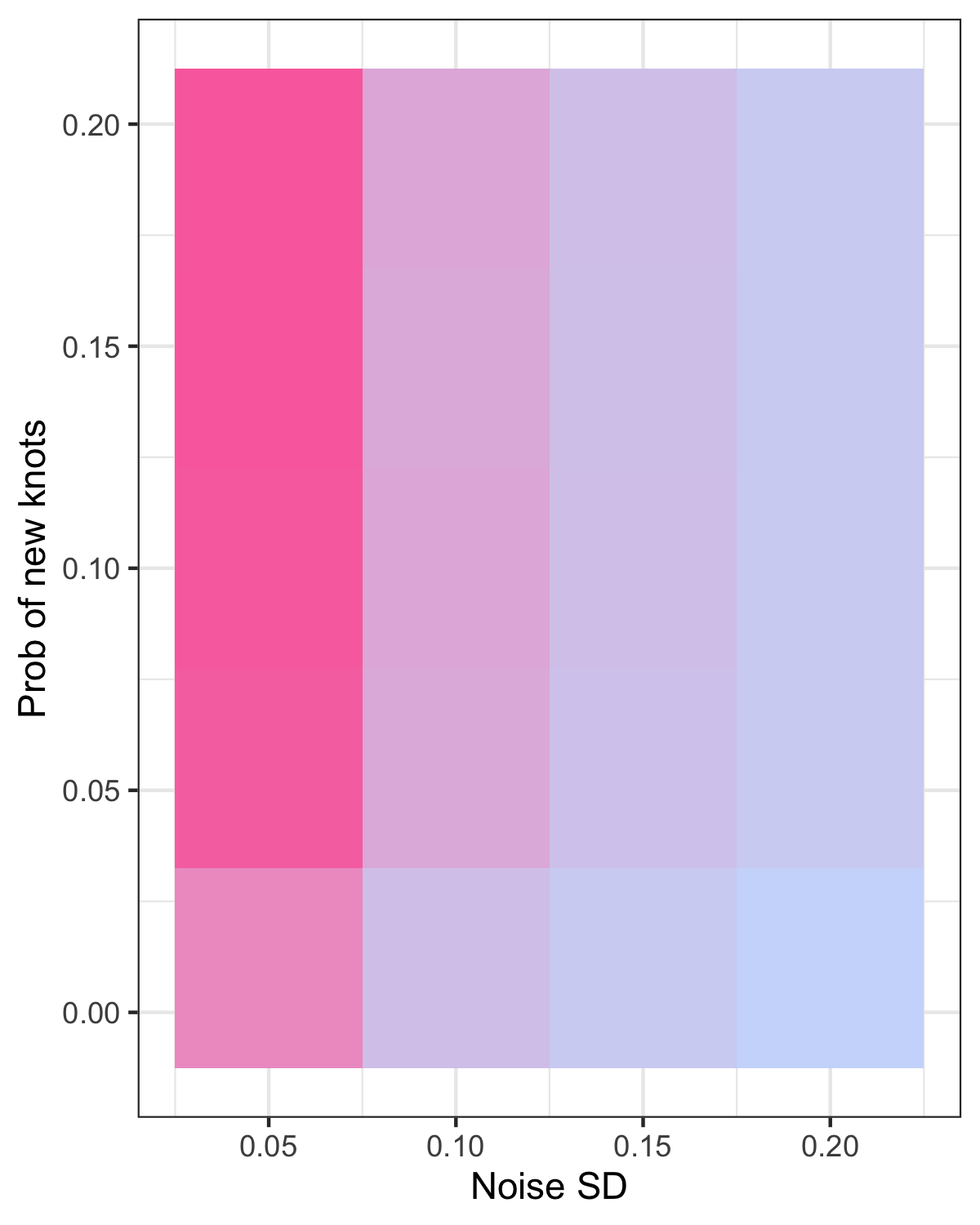}
        \caption{Vary the true variance and prob of new knots when fix the slope range as $0.5$.}
    \end{subfigure}
\hfill
    \begin{subfigure}[t]{0.3\textwidth}
        \includegraphics[width=1\linewidth]{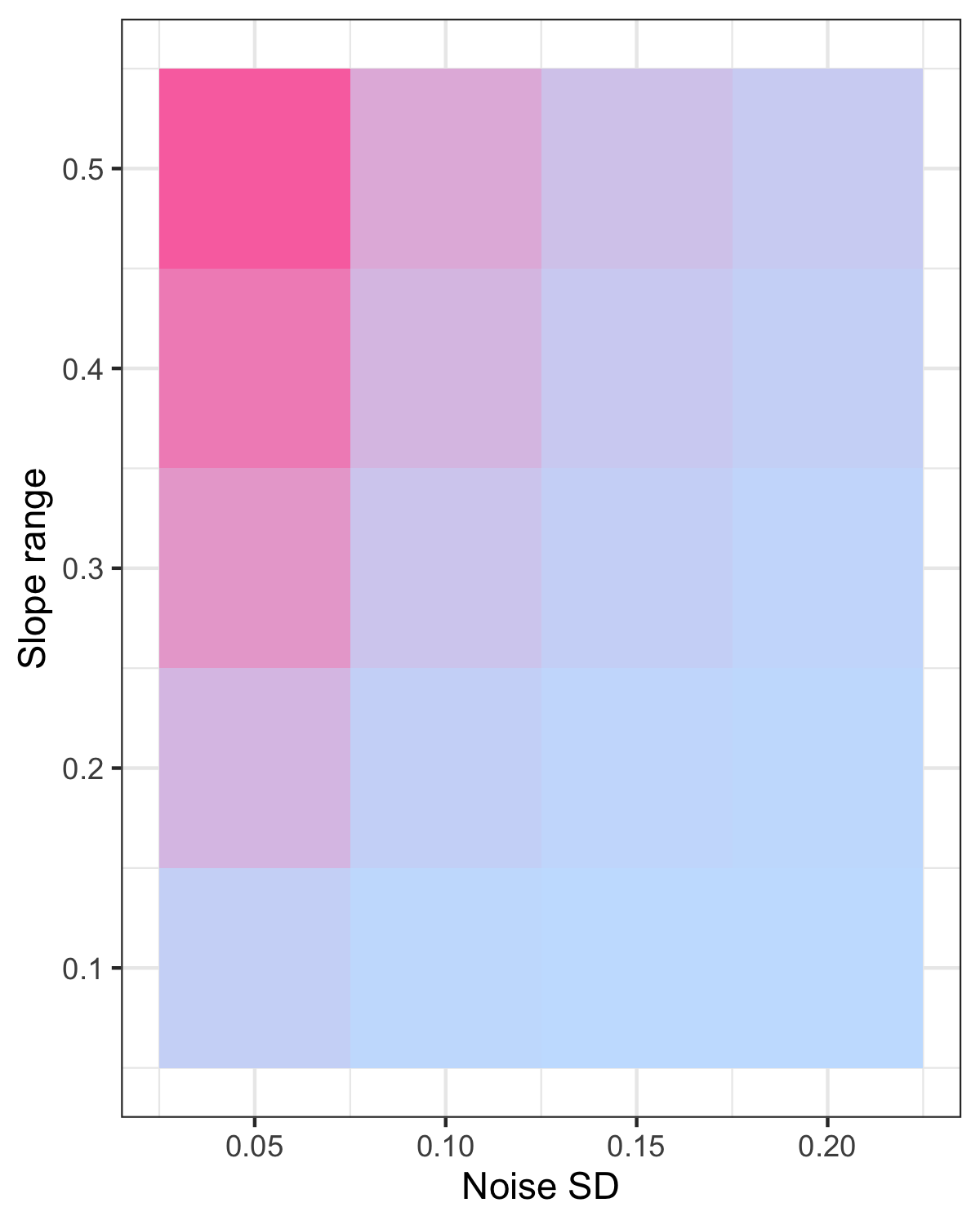}
        \caption{Vary the true variance and slope range when fix the prob of new knots as $0.1$.}
    \end{subfigure}
\hfill
    \begin{subfigure}[t]{0.2\textwidth}
    \hskip 0pt
        \includegraphics[width=1\linewidth]{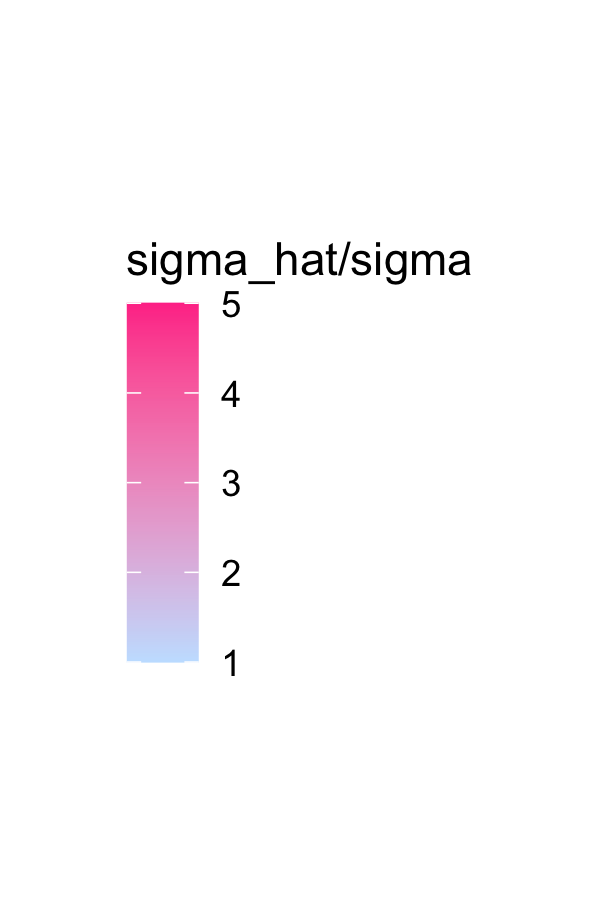}
    \end{subfigure}
\caption{The noise SD is often over estimated with a range in $(0.05, 0.3)$ and the true one varies in $(0.05, 0.2)$. The over-estimation fades when the noise SD increases and the slope decreases (and slightly fades as the probability of new knots decreases).}
\label{fig:variance_est}
\end{figure}

Luckily, a tendency to overestimate the variance leads to conservative CIs, meaning errors are still controlled at the appropriate level in most cases. We see in \cref{fig:var_est_error} that data fission still offers simultaneous type I error control empirically when using uniform confidence bands. This is achieved at the expense of having overly conservative CIs in cases where the underlying structural trend is variable --- either because of many knots points or because of the size of the slope. We investigate how the level of conservatism, as measured by the average widths of the constructed confidence bands, increases as the underlying trend varies more in \cref{fig:var_est_width_diff}.

\begin{figure}[H]
\centering
    \begin{subfigure}[t]{0.22\textwidth}
        \centering
        \includegraphics[width=1\linewidth]{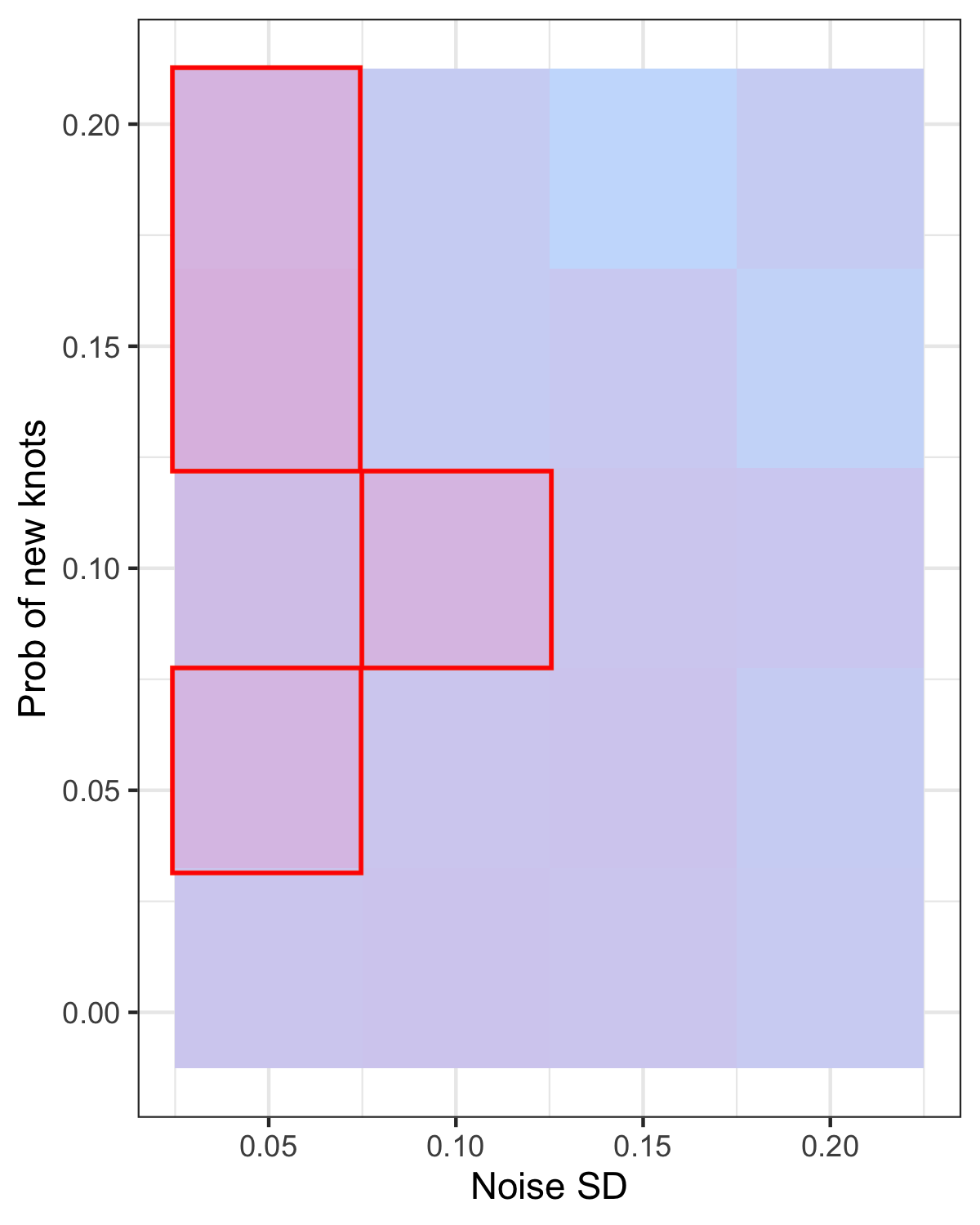}
        \caption{Data fission method when varying prob of new knots.}
    \end{subfigure}
\hfill
     \begin{subfigure}[t]{0.22\textwidth}
        \centering
        \includegraphics[width=1\linewidth]{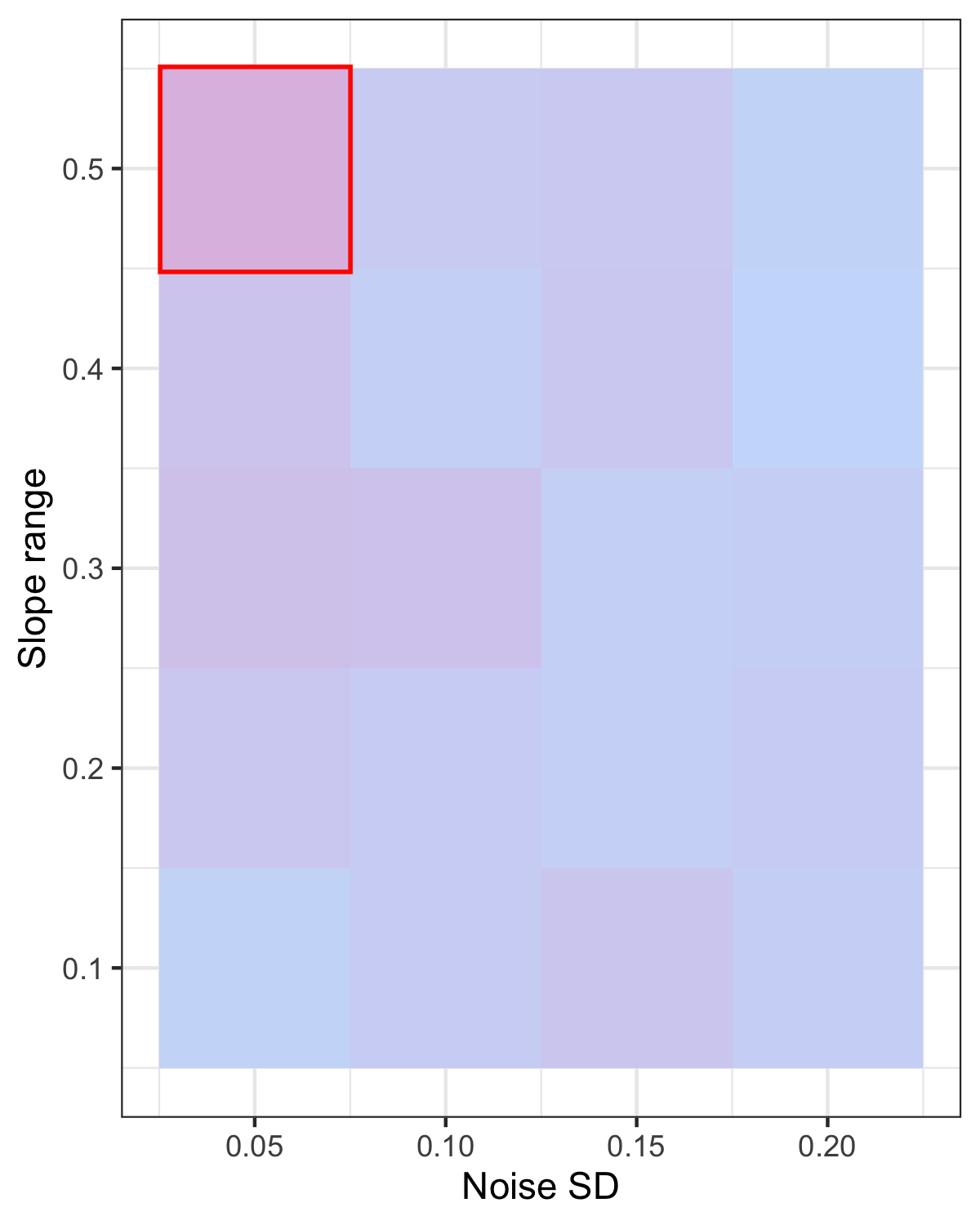}
        \caption{Data fission method when varying slope range.}
    \end{subfigure}
\hfill
    \begin{subfigure}[t]{0.08\textwidth}
    \hskip 0pt
        \includegraphics[width=1\linewidth]{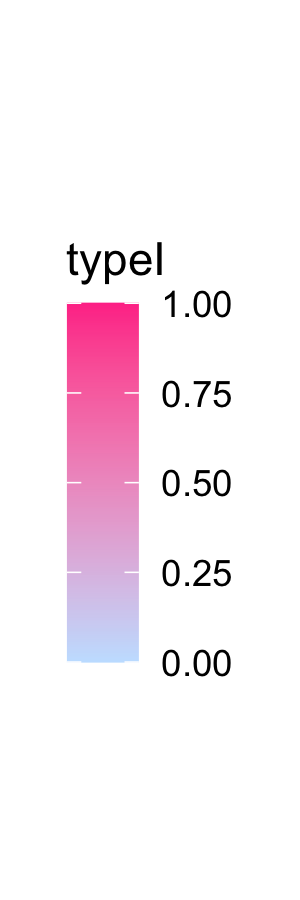}
    \end{subfigure}
\hfill
    \begin{subfigure}[t]{0.22\textwidth}
        \includegraphics[width=1\linewidth]{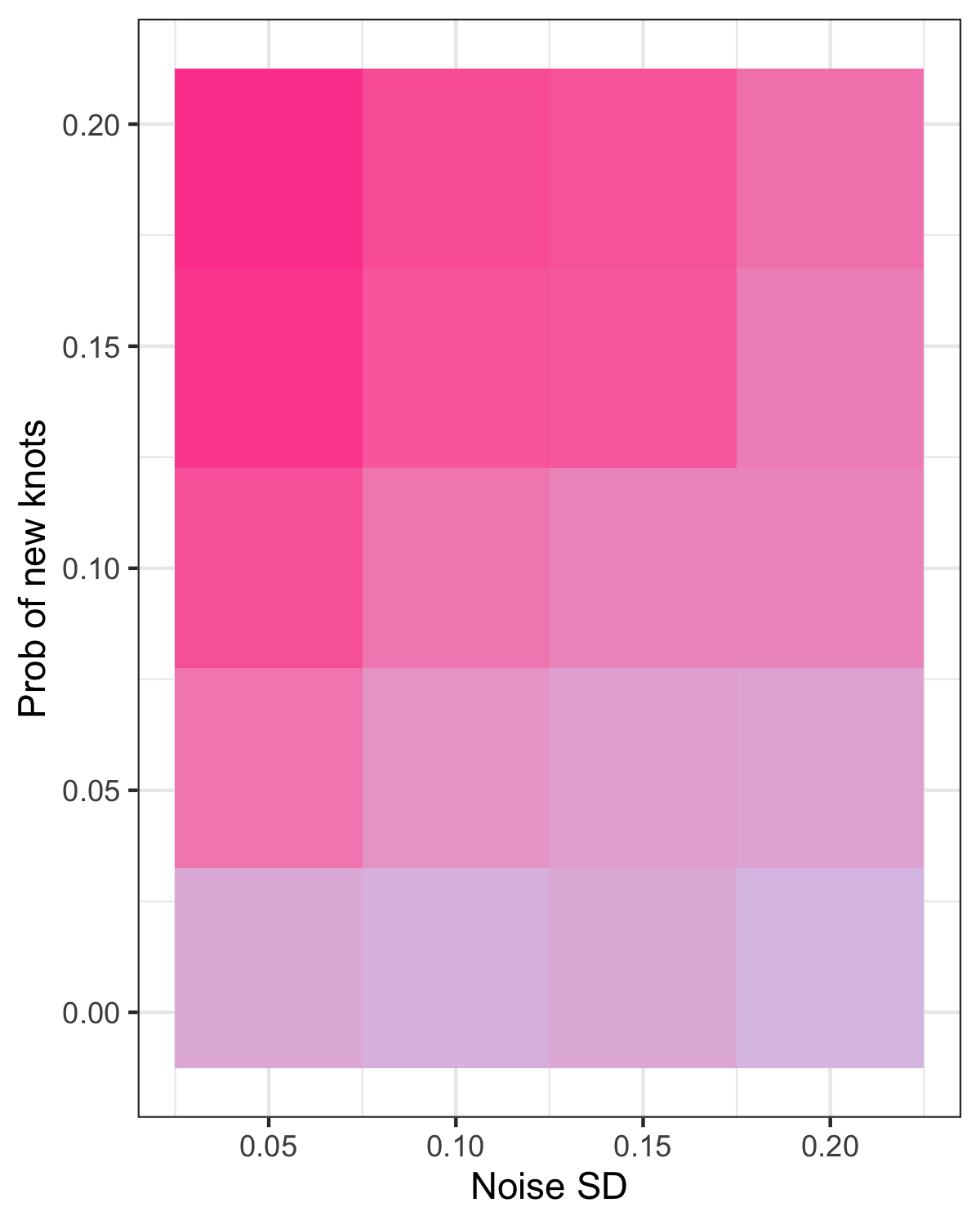}
        \caption{Full data twice when when varying prob of new knots.}
    \end{subfigure}
\hfill
    \begin{subfigure}[t]{0.22\textwidth}
        \includegraphics[width=1\linewidth]{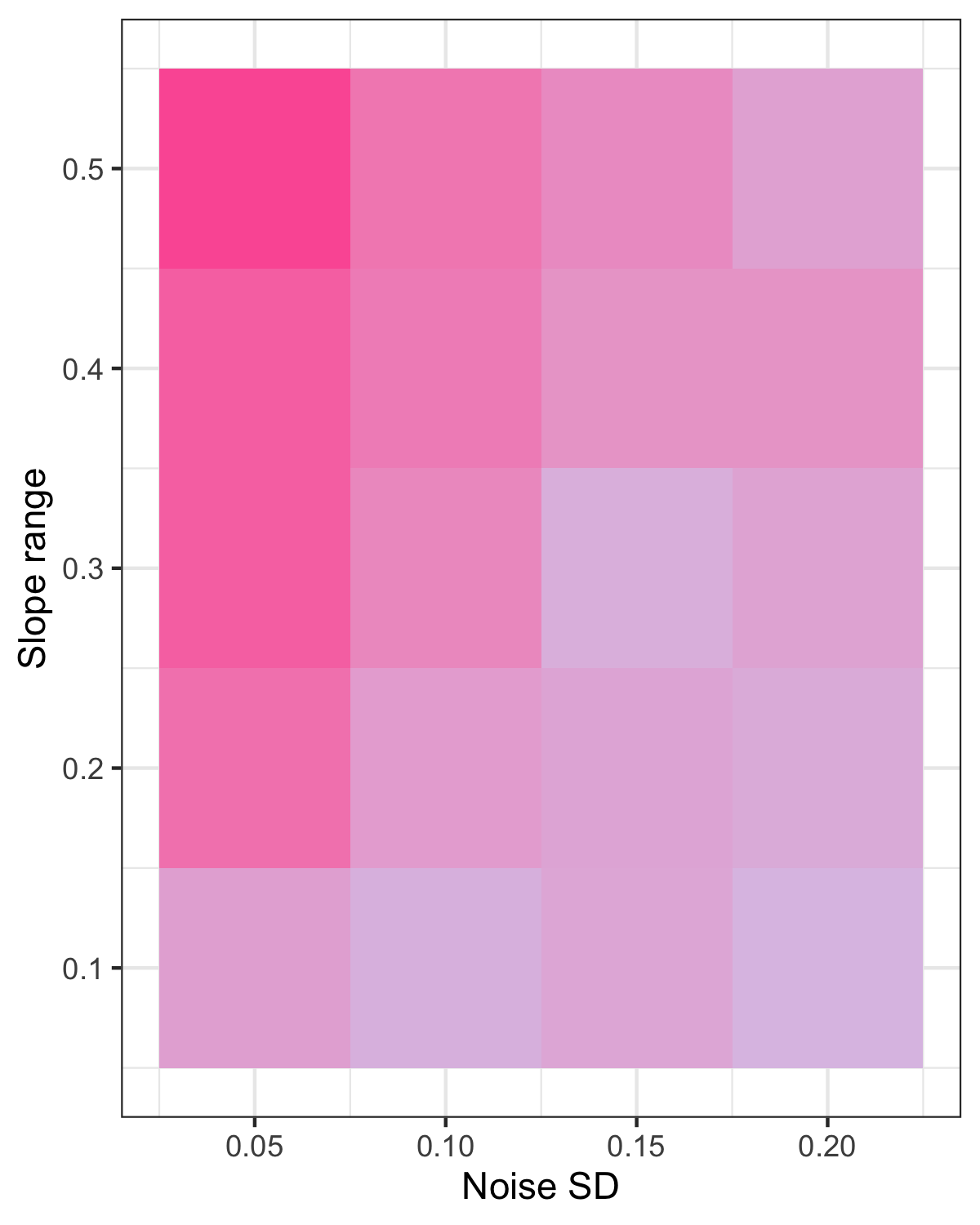}
        \caption{Full data twice when varying slope range.}
    \end{subfigure}
\caption{Simultaneous type I error for the uniform CI constructed using data fission , and full data twice, when varying the probability of new knots, the slope range, and the true noise variance. Data fission method seems to have lower simultaneous type I error than the target level (0.2) in most cases, except cases circled in red (with a max simultaneous type I error of 0.26).}
\label{fig:var_est_error}
\end{figure}

\begin{figure}[H]
\centering
     \begin{subfigure}[t]{0.32\textwidth}
        \centering
        \includegraphics[width=1\linewidth]{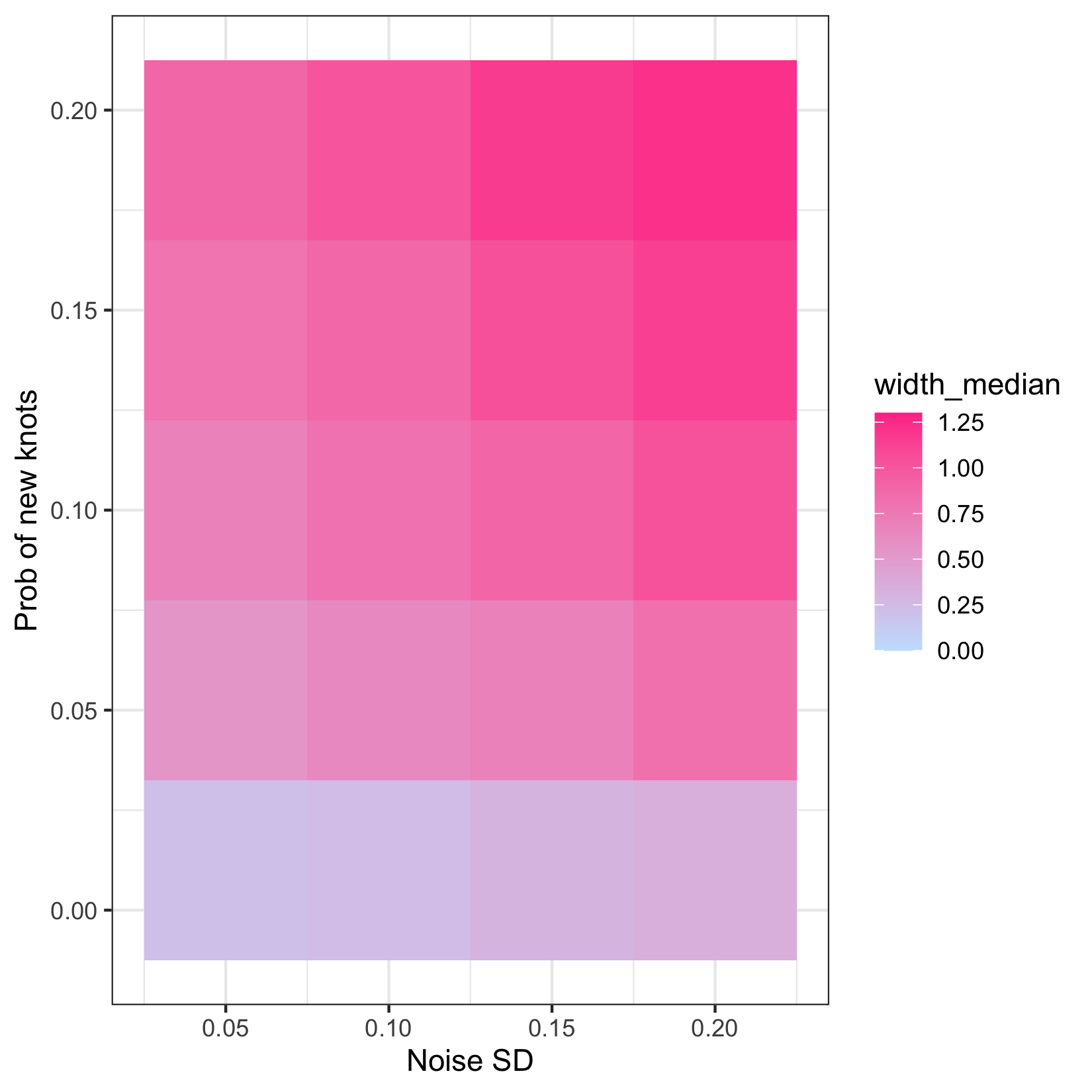}
    \end{subfigure}
\hfill
    \begin{subfigure}[t]{0.32\textwidth}
        \includegraphics[width=1\linewidth]{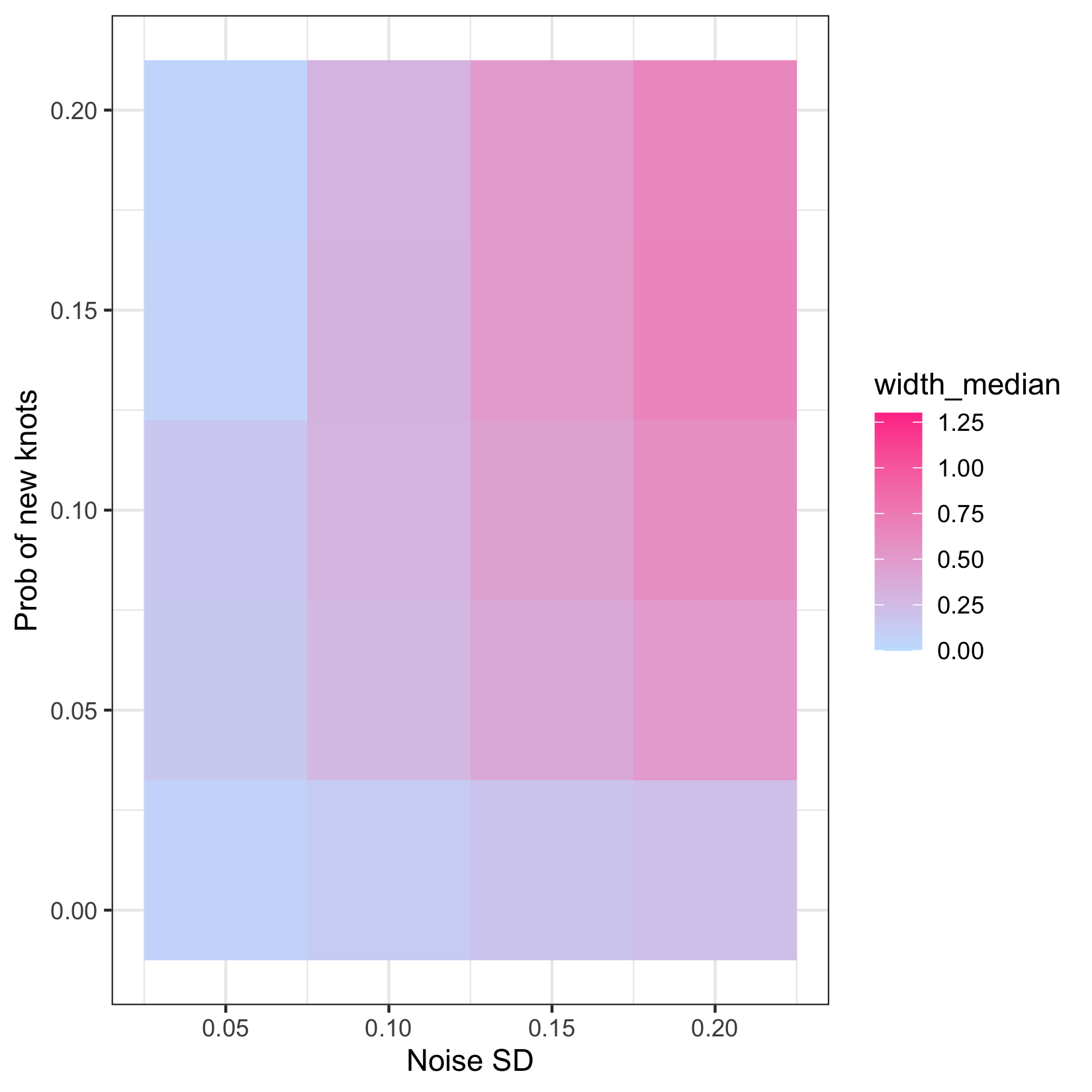}
    \end{subfigure}
\hfill
    \begin{subfigure}[t]{0.32\textwidth}
        \includegraphics[width=1\linewidth]{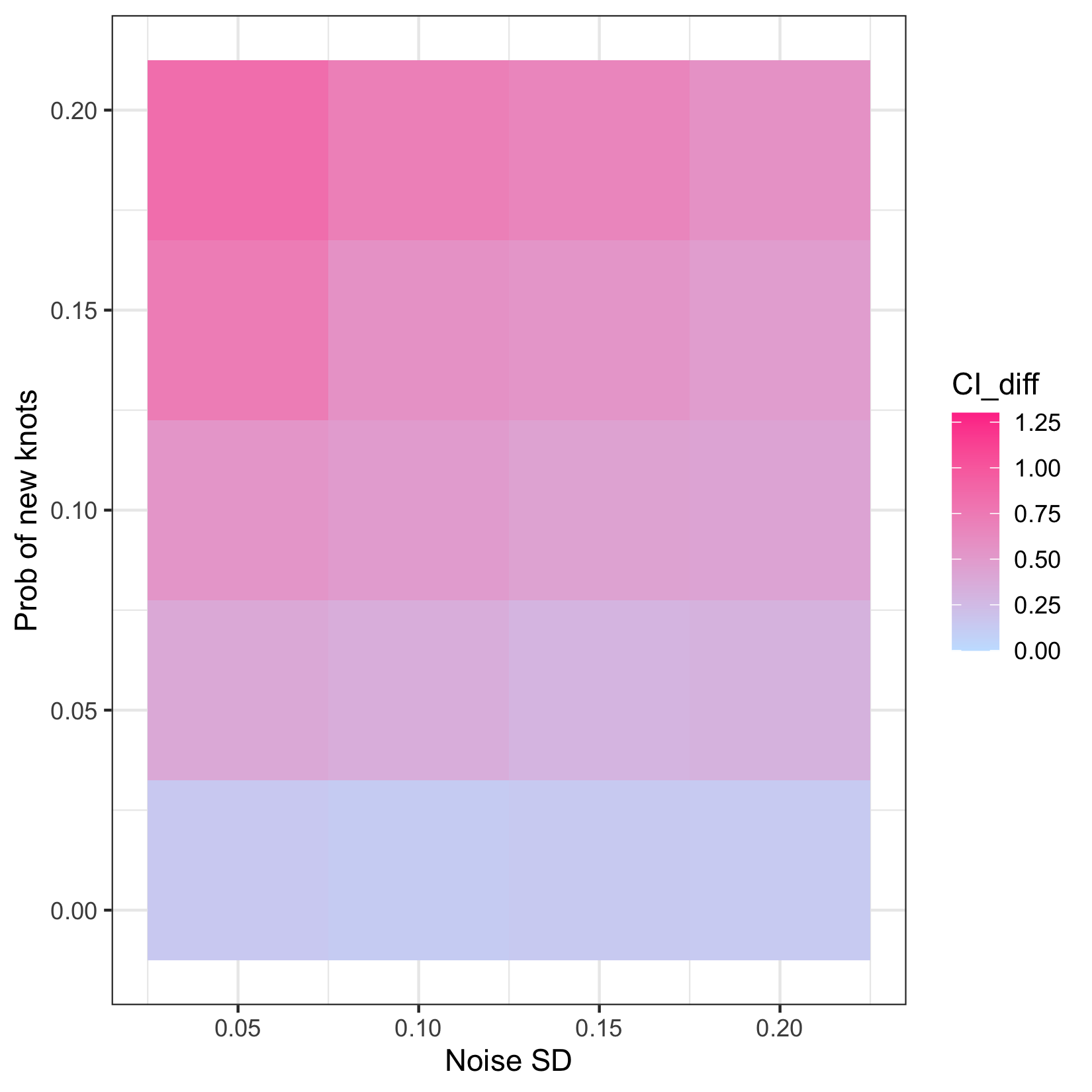}
    \end{subfigure}
\caption{The CI width for \textbf{uniform confidence bands} using either data fission (right) or full data twice (left) have similar trends with respect to noise SD and probability of new knots. The difference in CI length between the two methods decreases when the noise SD increases (such that the effect of double dipping is smaller) or the probability of new knots gets smaller (such that the CI from data fission is tighter).}
\label{fig:var_est_width_diff}
\end{figure}

\subsubsection{Alternative methods for selecting knots} \label{sec:appendix_alternative_knots}
In Section~\ref{sec:trendfilter_empirical}, we selected the knots by choosing the regularization parameter with the smallest cross-validation error, which often leads to selecting more knots than the underlying truth when the full dataset is used for both selection and inference. Although data fission allows for an analyst to guarantee error control under arbitrary selection rules, some selection rules tend to be more robust to the analyst reusing the data in terms of their empirical performance even if statistical guarantees are not available. 

\paragraph{Stein's unbiased risk estimate (SURE)}
An alternative methodology for selecting knots would be to minimize the SURE formula --- which provides an unbiased estimate of mean-squared percentage error in a fixed-design setting. In the context of trend filtering, this can be computed as 
$$\frac{1}{n} \sum_{i=1}^{n} (y_{i} - \hat{\mu}_{i})^{2} + 2 \sigma^{2}\frac{m}{n}$$
where $m$ is the number of knots chosen to fit $\hat{\mu}$. Interestingly, using this formula seems to result in type I error being controlled even when reusing the full data for both selection and inference. Please see \cite{10.1093/mnras/staa106} for further background on how SURE can be used to aid in knot selection. 

\begin{figure}[H]
\centering
    \begin{subfigure}{0.25\textwidth}
        \centering
        \includegraphics[width=1\linewidth]{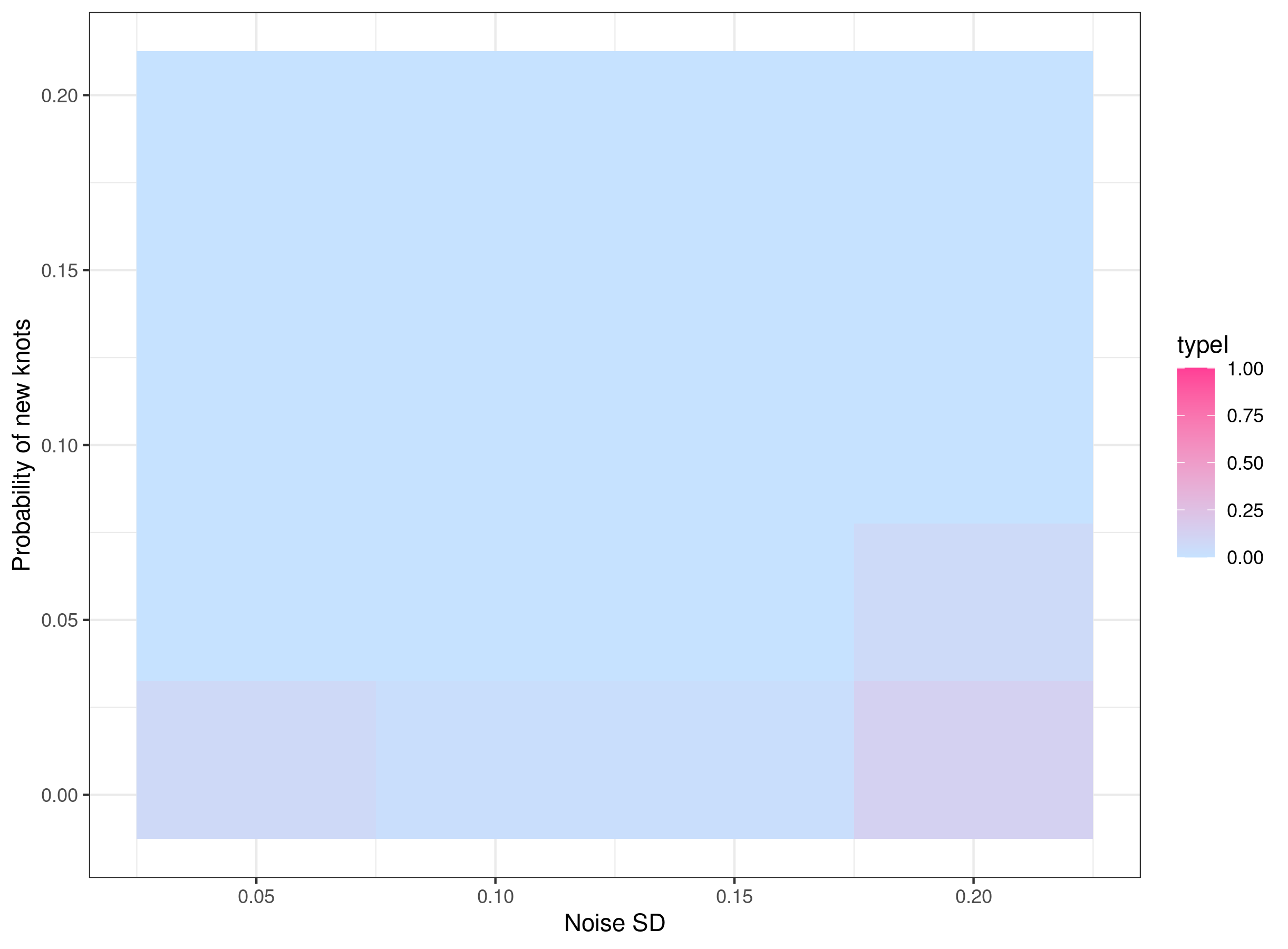}
        \caption{Simultaneous type I error using full data twice.}
    \end{subfigure}
\hfill
    \begin{subfigure}{0.25\textwidth}
        \includegraphics[width=1\linewidth]{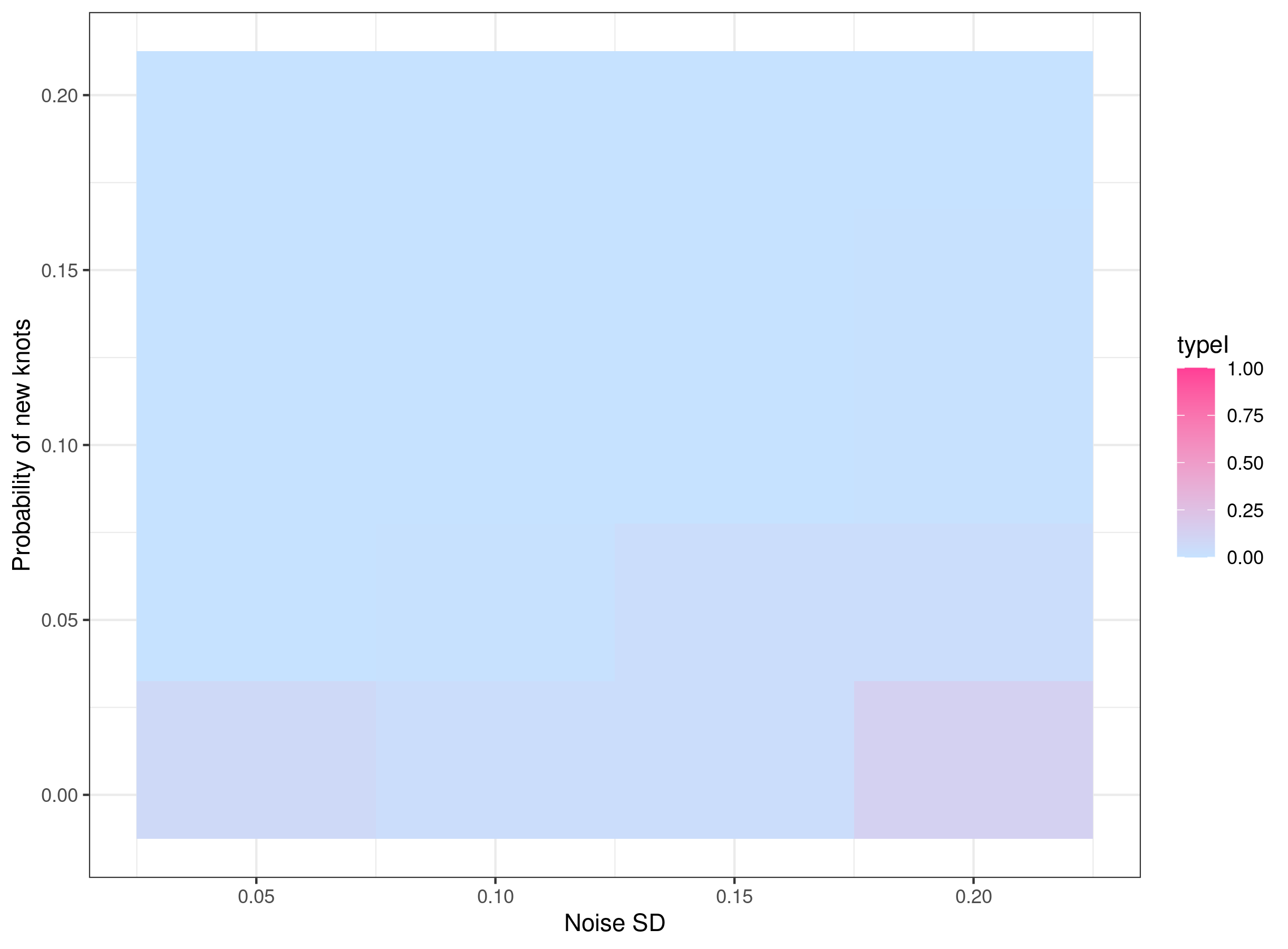}
        \caption{Simultaneous type I error using data fission.}
    \end{subfigure}
\hfill
    \begin{subfigure}{0.1\textwidth}
        \includegraphics[width=1\linewidth]{figures/legend_typeI.png}
    \end{subfigure}
    \caption{Simultaneous type I error for \textbf{uniform CIs} as we vary the  probability of having new knots $q \in \{0.01, 0.55, 0.1, 0.145, 0.19\}$ and the noise SD in $\{0.05, 0.1, 0.15, 0.2\}$. The error control violation when using the full data twice is no longer as stark in the simulation results when using SURE as the selection rule, with the highest simultaneous type I error at 0.2 given target level $\alpha=0.2$.
    }
    \label{fig:trendfilter3}
\end{figure}

\paragraph{1-SD Rule}

\label{sec:1_se}
 To mitigate the possibility of over selection of knots, an alternative approach is to choose the regularization parameter to be the the one with smallest error plus one standard deviation. As we observe in \cref{fig:reg_param}, adding a standard deviation to the selection rule does not change the error much. However, this rule appears to be more robust when the full data is reused for inference---in \cref{fig:trendfilter4}, we can see that simultaneous type I error control is not as seriously violated when using this selection rule. 

\begin{figure}[H]
\centering
\includegraphics[width=0.4\linewidth]{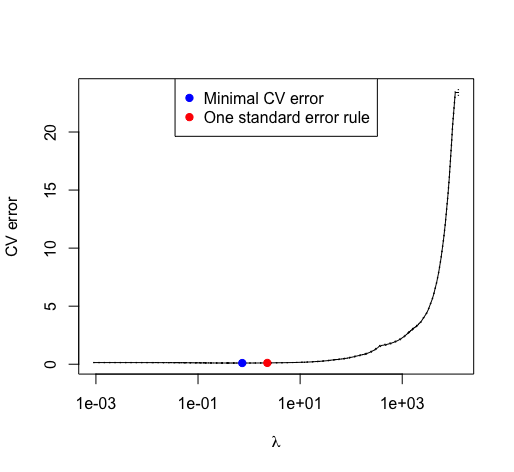}
\caption{In the path of regularization parameter, adding an extra standard deviation does not change the error much.}
\label{fig:reg_param}
\end{figure}

\begin{figure}[H]
\centering
    \begin{subfigure}{0.25\textwidth}
        \centering
        \includegraphics[width=1\linewidth]{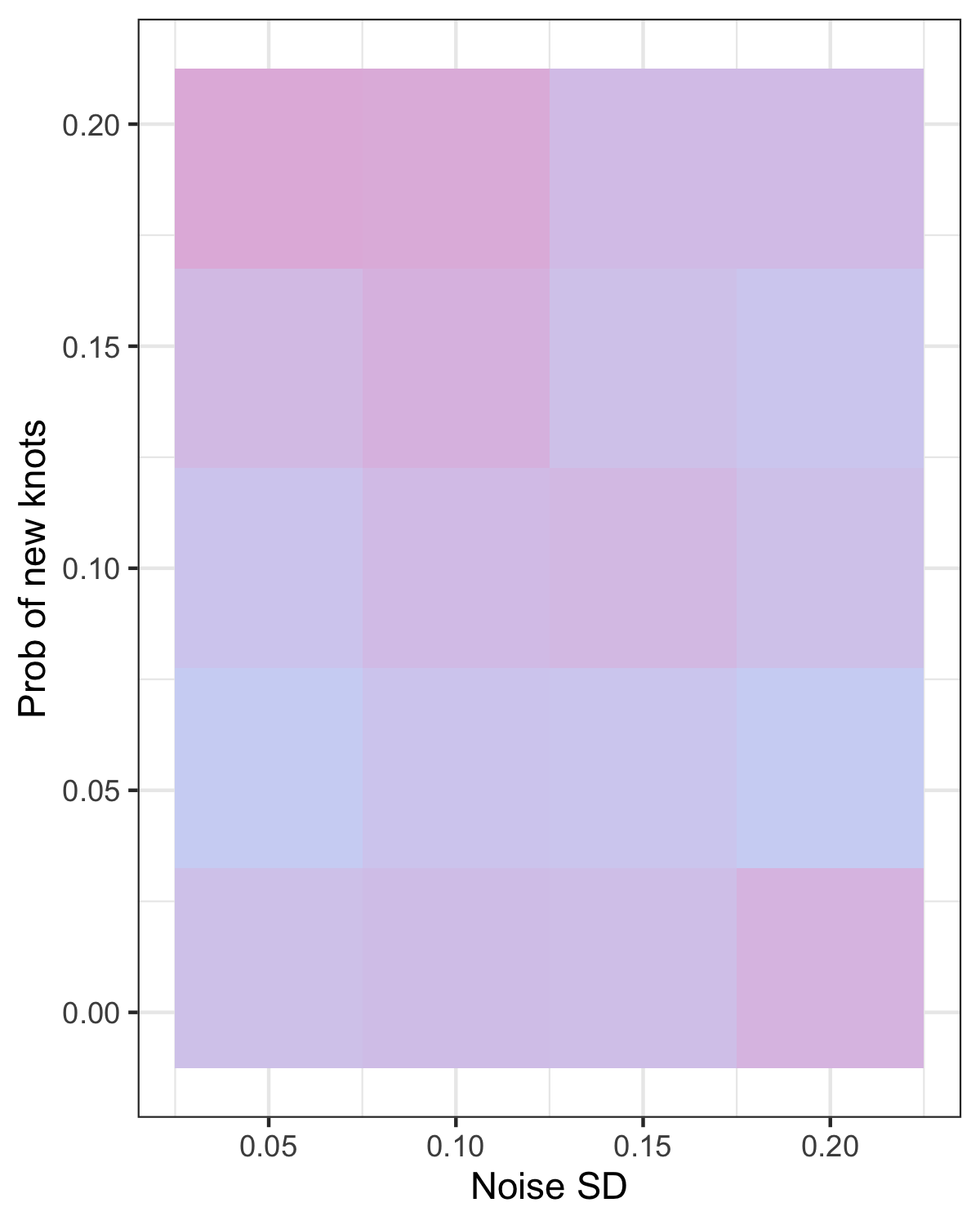}
        \caption{Simultaneous type I error using full data twice.}
    \end{subfigure}
\hfill
    \begin{subfigure}{0.25\textwidth}
        \includegraphics[width=1\linewidth]{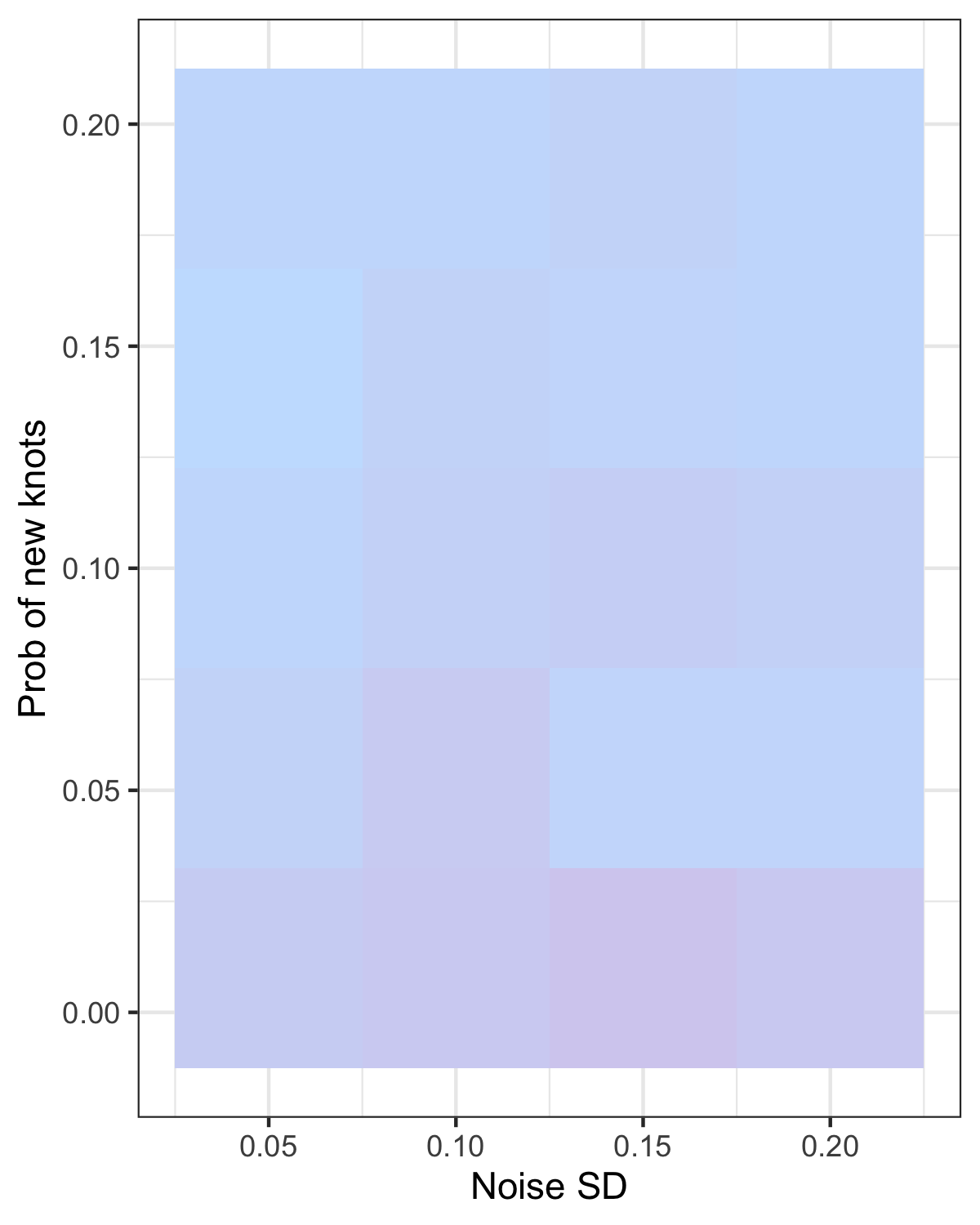}
        \caption{Simultaneous type I error using data fission.}
    \end{subfigure}
\hfill
    \begin{subfigure}{0.1\textwidth}
        \includegraphics[width=1\linewidth]{figures/legend_typeI.png}
    \end{subfigure}
    \caption{Simultaneous type I error for \textbf{uniform CIs} as we vary the  probability of having new knots $q \in \{0.01, 0.55, 0.1, 0.145, 0.19\}$ and the noise SD in $\{0.05, 0.1, 0.15, 0.2\}$. The error control violation when using the full data twice is no longer as stark in the simulation results when using this new selection rule, with the highest simultaneous type I error at 0.3 given target level $\alpha=0.2$.
    }
    \label{fig:trendfilter4}
\end{figure}

\section{Additional empirical results for spectroscopy datasets} \label{sec:spectroscopy_supplement}
We show results in this section for the remaining astronomical objects of interest described in \cref{sec:astro_examples}.
\begin{enumerate}
    \item A \textbf{galaxy}. DR12, Plate= 7140, MJD = 56569, Fiber=68. Located at (RA,Dec, z) \\ $\approx (349.374^{\circ}, 33.617^{\circ},0.138)$. Results shown in \cref{fig:astronomical_galaxy}.
    \item A \textbf{star}. DR12, Plate= 4055, MJD = 55359, Fiber=84. Located at (RA,Dec, z) \\ $\approx (236.834^{\circ}, 0.680^{\circ},0.000)$. Results shown in \cref{fig:astronomical_star}.
\end{enumerate}
As a point of comparison, we use quadratic  rather than linear trend filtering to model these spectra. We note an interesting tradeoff between the degree of the polynomial that is used and the smoothness of the estimated function. Larger degree polynomial result in fewer knots being chosen during the selection step which leads to a smoother-looking confidence band, but at the expense of not capturing some of the more volatile pieces of the data. Smaller degree polynomials results in more knots being chosen during the selection stage which leads to a less smooth band but also a function that tracks the overall volatility of the data more closely. 

\begin{figure} 
\centering
  \includegraphics[width=0.8\linewidth]{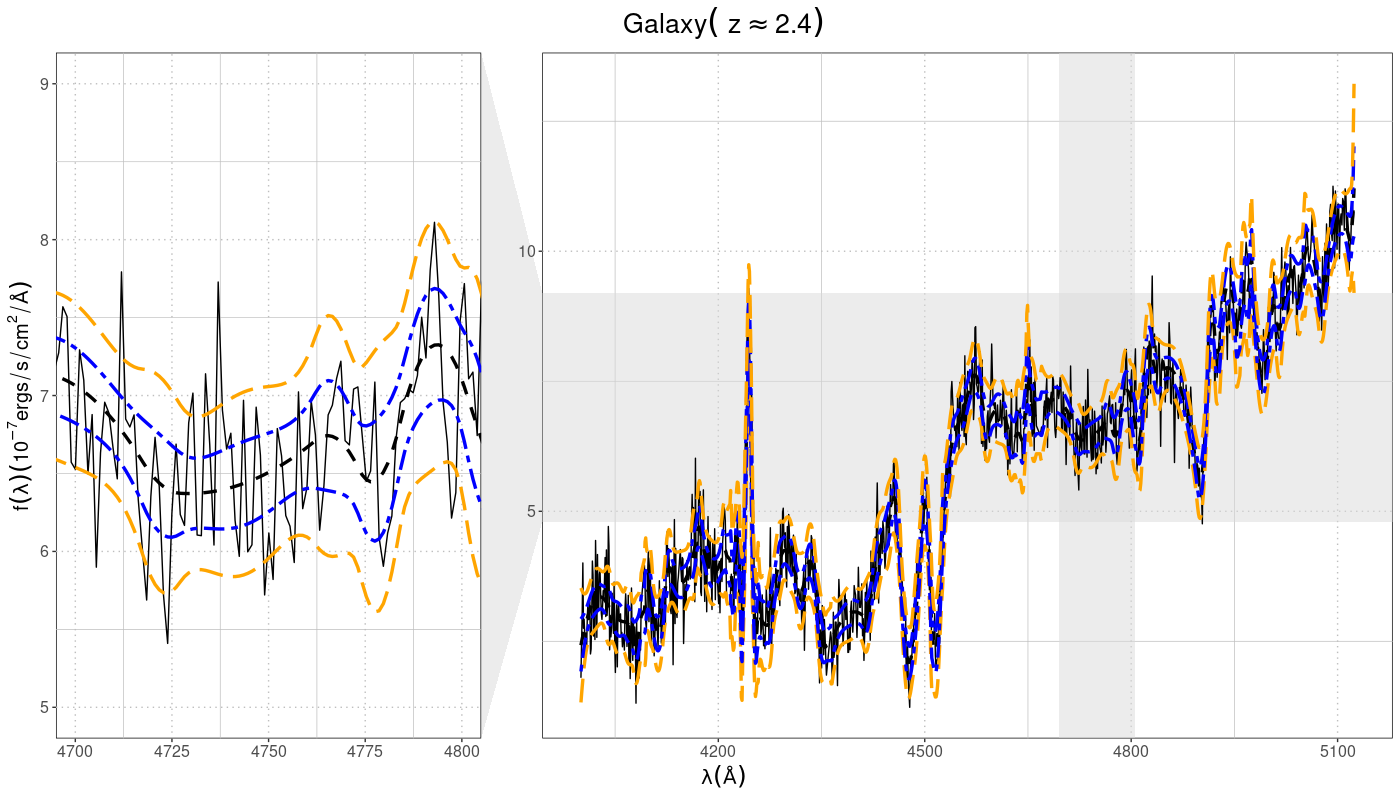}
  \newline
  \includegraphics[width=0.5\linewidth]{figures/legend_astronomical.PNG}
  \caption{Fitted values using quadratic trend filtering as well as uniform and pointwise CIs for a \textbf{galaxy}. DR12, Plate= 7140, MJD = 56569, Fiber=68. Located at (RA,Dec, z)  $\approx (349.374^{\circ}, 33.617^{\circ},0.138)$.} \label{fig:astronomical_galaxy}
\end{figure}

\begin{figure}
\centering
  \includegraphics[width=0.8\linewidth]{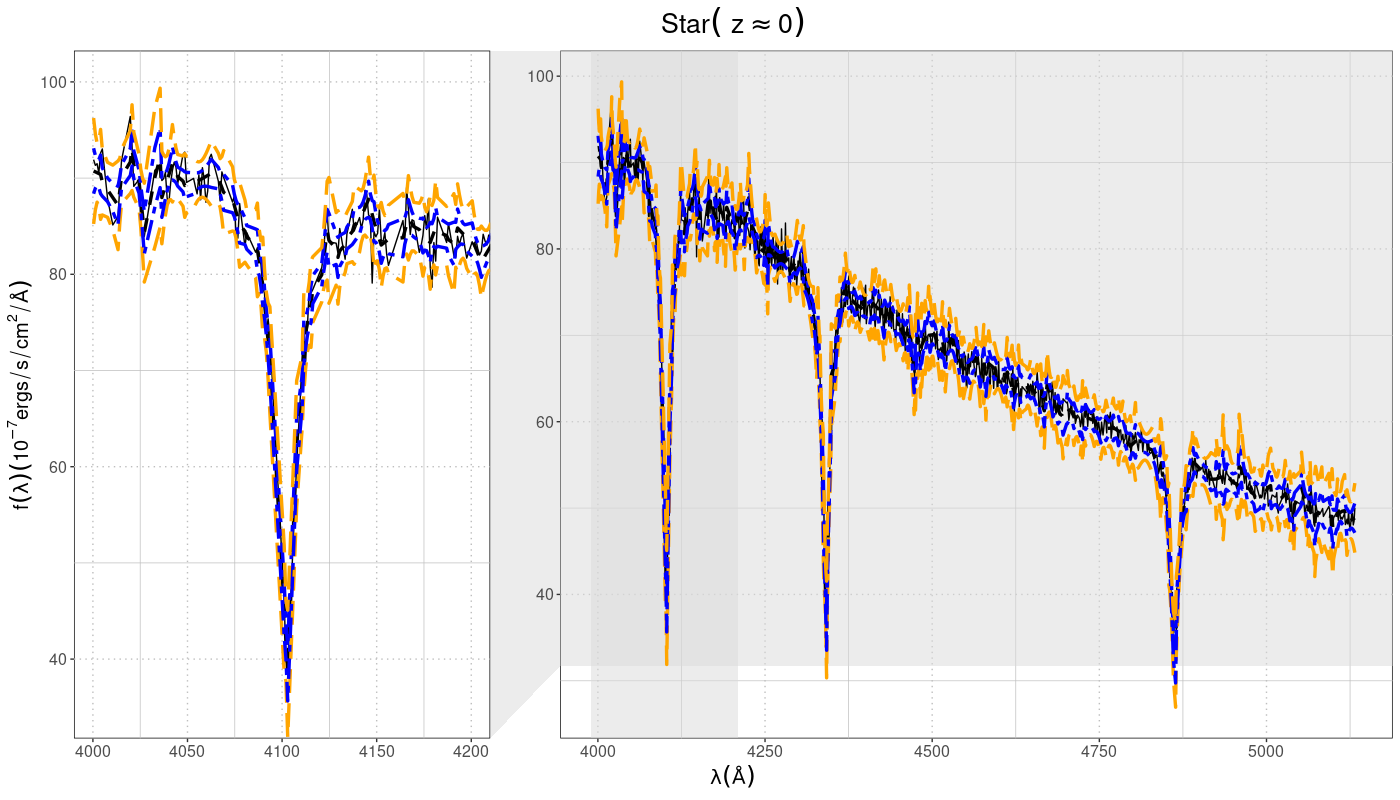}
  \centering
  \includegraphics[width=0.5\linewidth]{figures/legend_astronomical.PNG}
  \caption{Fitted values using quadratic trend filtering as well as uniform and pointwise CIs for a \textbf{star}. DR12, Plate= 4055, MJD = 55359, Fiber=84. Located at (RA,Dec, z)  $\approx (236.834^{\circ}, 0.680^{\circ},0.000)$. } \label{fig:astronomical_star}
\end{figure}

\end{appendices}

%% file: decomposition_proofs.tex
\subsection{Proofs of decomposition rules} \label{sec:decomp_proofs}
This section contains the deferred derivations for the decomposition rules from \cref{sec:appendix_list_decomp}. 
\subsubsection{Gaussian Decomposition}
\begin{proposition}
\textbf{(P1)} Suppose $X \sim N( \mu,  \Sigma)$ is $d$-dimensional ($d \geq 1$). Draw $ Z \sim N(0,  \Sigma)$. Then $f( X) =  X + \tau  Z$, where $\tau \in (0, \infty)$ is a tuning parameter, has  distribution $N( \mu, (1 + \tau^{2})  \Sigma)$; and $g( X) =  X - \tfrac{1}{\tau}  Z$ has distribution $N(\mu,  (1 + \tau^{-2})  \Sigma)$; and $f( X) \independent g( X)$. Larger $\tau$ indicates less informative $f( X)$ (and more informative $g( X) \mid f( X)$). 
\end{proposition}

\begin{proof}
$f(X)$ and $g(X)$ are jointly distributed as Gaussian, so it is sufficient to show the means and variances match and they are uncorrelated. 
$$\text{Cov} \left( f(X),g(X) \right) = \text{Var} \left(X\right)  - \text{Var}(Z) = 0$$
$$\mathbb{E} \left( f(X) \right) =  \mu, \text{ and } \mathbb{E} \left( g(X) \right) =  \mu$$
$$\text{Var} \left( f(X) \right) =  (1 + \tau^{2})  \Sigma, \text{ and } \text{Var} \left( g(X) \right) =  (1 + \tau^{-2}) \Sigma$$
This complete the proof.
\end{proof}

\begin{fact}[Conditioning on Gaussians] \label{fact:gaussian_comparisons}
If $X \sim N(\mu,\Sigma)$ is partitioned as
$
X=\left[\begin{array}{l}
X_{1} \\
X_{2}
\end{array}\right]
$
with $
\mu=\left[\begin{array}{l}
\mu_{1} \\
\mu_{2}
\end{array}\right]$ and $\Sigma =\left[
\begin{array}{ll}
\Sigma_{11} & \Sigma_{12} \\
\Sigma_{21} & \Sigma_{22} \\
\end{array}\right]$, 
then the distribution of $X_{1}| X_{2}=\mathbf{a}$ is distributed as $N(\mu_{1}+\Sigma_{12} \Sigma_{22}^{-1}\left(a-\mu_{2}\right) , \Sigma_{11}-\Sigma_{12} \Sigma^{-1} \Sigma_{21})$.
\end{fact}
\begin{proposition}
\textbf{(P2 CP)}  Suppose $X \sim N( \mu,  \Sigma)$ is $d$-dimensional ($d \geq 1$). Draw $ Z \sim N( X, \tau  \Sigma)$, where $\tau \in (0, \infty)$ is a tuning parameter. Then $f( X) =  Z$ has marginal distribution $N( \mu, (1 + \tau)  \Sigma)$; and $g( X) =  X$ has conditional distribution $N(\frac{\tau}{\tau + 1} ( \mu + f( X)/\tau), \frac{\tau}{\tau + 1}  \Sigma)$.
\end{proposition}

\begin{proof} Rewrite $ Z = X + W$ with $W\sim N(0,\tau\Sigma)$ to see $Z \sim N(\mu,(1+\tau)\Sigma)$. They are then jointly distributed as:
$$\begin{pmatrix}
X \\
Z
\end{pmatrix} \sim N \left( \mu , \begin{pmatrix}
\Sigma & \Sigma   \\
\Sigma  & (1+\tau) \Sigma 
\end{pmatrix}  \right).$$ 
Conclude by applying \cref{fact:gaussian_comparisons}.

\end{proof}

\begin{proposition}
\textbf{(P2)}  Suppose $X \sim N( \mu,  \Sigma)$ is $d$-dimensional ($d \geq 1$). Draw $Z \sim N(0,\Sigma_{0})$ and let $f(X) = X-Z$ with $g(X) = X+Z$ as before. For notational convenience, let $\Sigma_{1} = \Sigma + \Sigma_{0}$ and $\Sigma_{2} = \Sigma - \Sigma_{0}$ Then $f(X) \sim N(\mu, \Sigma + \Sigma_{0})$ and $g(X) | f(X) \sim N\left( \mu + \Sigma_{2}\Sigma_{1}^{-1} \left( f(X) - \mu \right), \Sigma_{1} - \Sigma_{2}\Sigma_{1}^{-1}\Sigma_{2} \right)$.
\end{proposition}

\begin{proof} $f(X)$ and $g(X)$ are jointly distributed as 
$$\begin{pmatrix}
f(X) \\
g(X)
\end{pmatrix} \sim N \left( \mu , \begin{pmatrix}
\Sigma + \Sigma_{0} &  \Sigma - \Sigma_{0}   \\
\Sigma - \Sigma_{0}   & \Sigma + \Sigma_{0} 
\end{pmatrix}  \right).$$ 
Conclude by applying \cref{fact:gaussian_comparisons}.
\end{proof}
\subsubsection{Gamma Decomposition}
\begin{proposition}
\textbf{(P2 CP)}  Generally, suppose $X \sim \mathrm{Gamma}(\alpha, \beta)$. Draw $ Z = (Z_1, \ldots, Z_B)$ where each element is i.i.d.\ $Z_i \sim \mathrm{Poi}(X)$ and $B \in \{1, 2, \ldots\}$ is a tuning parameter. Then $f(X) =  Z$, where each element is marginally distributed as a negative binomial $NB(\alpha, \tfrac{\beta}{\beta+1})$. $g(X) = X$ has conditional distribution $\mathrm{Gamma}(\alpha + \sum_{i=1}^B f_i(X), \beta + B)$.
\end{proposition}
\begin{proof}
Poisson being conjugate prior to a gamma distribution is a standard result. We have that
\begin{align*}
P\left( g(X) | Z_{1},...,Z_{B}\right) & \propto X^{\alpha-1} \prod_{i=1}^{n} e^{-\beta X_{i}}  X^{Z_i} e^{-X} \\
& \propto X^{\alpha+\sum_{i=1}^{n} Z_{i}-1} e^{-(\beta +B)X },
\end{align*}
which is proportional to the density of $\mathrm{Gamma}(\alpha + \sum_{i=1}^{n}f_{i}(X), \beta + B)$. Integrating out $g(X)$, we can recover the marginals. Re-paramaterize $\beta = \frac{p}{1-p}$ and then compute:
\begin{align*}
P(f(X) = z) & =\int_0^{\infty} \frac{x^z}{z !} e^{-x} \times\left(\frac{p}{1-p}\right)^\alpha x^{\alpha-1} \frac{e^{-x \frac{p}{1-p}}}{\Gamma(\alpha)} \mathrm{d} x \\
& =\left(\frac{p}{1-p}\right)^\alpha \frac{1}{z ! \Gamma(\alpha)} \int_0^{\infty} x^{\alpha+z-1} e^{-x \frac{p+1-p}{1-p}} \mathrm{~d} x \\
& =\left(\frac{p}{1-p}\right)^\alpha \frac{1}{z ! \Gamma(\alpha)} \Gamma(\alpha+z)(1-p)^{z+\alpha} \underbrace{\int_0^{\infty} \frac{ (1-p)^{-(z+\alpha)} }{\Gamma(z+ \alpha)}  x^{\alpha+z-1} e^{-x \frac{1}{1-p}} \mathrm{d} x}_{\text{Integrates to 1}} \\
& =\frac{\Gamma(\alpha+z)}{z ! \Gamma(\alpha)}(1-p)^z p^\alpha\\
& ={\alpha + z - 1 \choose \alpha }(1-p)^z p^\alpha,\\
\end{align*}
which is precisely the density of a $NB(\alpha, \tfrac{\beta}{\beta+1})$ variable.
\end{proof}

\begin{proposition}
\textbf{(P2 CP)} We can draw $Z \sim \mathrm{Poi}(\tau X)$, where $\tau \in (0,\infty)$ is a tuning parameter. Then $f(X) = Z$, marginally distributed as $\mathrm{NB}(\alpha, \tfrac{\beta}{\beta+\tau})$. $g(X) = X$ has conditional distribution $ \mathrm{Gamma}(\alpha + f(X), \beta + \tau)$. $f(X)$ is most informative when $\tau$ is comparable with $\theta$, and less informative when $\tau$ approaches zero or infinity.
\end{proposition}
\begin{proof} We have that
\begin{align*}
P\left( g(X) | f(X) = Z\right) & \propto X^{\alpha-1} e^{-\beta X}  (\tau X)^{Z} e^{-\tau X} \\
& \propto X^{\alpha+Z-1} e^{-(\beta +\tau)X },
\end{align*}
which is proportional to the density of $\mathrm{Gamma}(\alpha + f(X), \beta + \tau)$. Conclude by recognizing the Negative Binomial as a Poisson-gamma mixture. Re-parameterize $\beta = \frac{p}{1-p}$ and repeating the arguments before to conclude that the marginal distribution is the density of a $NB(\alpha, \tfrac{\beta}{\beta+\tau})$ variable.

\end{proof}

\begin{corollary}
\textbf{Exponential (P2 CP)} Suppose $X \sim \mathrm{Exp}(\theta)$. Draw $ Z = (Z_1, \ldots, Z_B)$ where  each element is i.i.d.\ $Z_i \sim \mathrm{Poi}(X)$ and $B \in \{1, 2, \ldots\}$ is a tuning parameter. Then $f(X) =  Z$, where each element is marginally distributed as $\mathrm{Geo}(\tfrac{\theta}{\theta+1})$. $g(X) = X$ has conditional distribution $\mathrm{Gamma}(1 + \sum_{i=1}^B f_i(X), \theta + B)$. Larger $B$ indicates more informative $f(X)$ (and less informative $g(X) \mid f(X)$).
\end{corollary}
\begin{proof}
Apply the above decomposition rules, setting $\alpha = 1$ and $\beta=\theta$.
\end{proof}

\begin{corollary}
\textbf{(P2 CP)} Alternatively, we can draw $Z \sim \mathrm{Poi}(\tau X)$, where $\tau \in (0,\infty)$ is a tuning parameter. Then $f(X) = Z$, marginally distributed as $\mathrm{Geo}(\tfrac{\theta}{\theta+\tau})$. $g(X) = X$ has conditional distribution $\mathrm{Gamma}(1 + f(X), \theta + \tau)$. Here, $f(X)$ is most informative when $\tau$ is comparable with $\theta$, and less informative when $\tau$ approaches $0$ or $\infty$. %(although it appears that the left information monotonically decreases in $\tau$).
\end{corollary}
\begin{proof}
Apply the above decomposition rules, setting $\alpha = 1$ and $\beta=\theta$.
\end{proof}

\subsubsection{Beta Decomposition}

\begin{proposition}
\textbf{(P2 CP)} Suppose $X \sim \mathrm{Beta}(\theta,1)$. Draw $Z \sim \mathrm{Bin}(B, X)$, where $B \in \{1, 2, \ldots\}$ is a tuning parameter. Then $f(X) = Z$ has marginal distribution as a discrete uniform in $\{0, 1, \ldots, B\}$ when $\theta = 1$, and stochastically larger (smaller) than a discrete uniform when $\theta$ is larger (smaller) than one (the PMF of $Z$ is $p_\theta(z) = \frac{\theta \Gamma(z + \theta) B!}{\Gamma(B + 1 + \theta) z!}$). $g(X) = X$ has conditional distribution $\mathrm{Beta}(\theta + f(X), B - f(X) + 1)$.
\end{proposition}
\begin{proof} We have that
\begin{align*} 
P\left( X | f(X) = Z\right) & \propto x^{z}(1-x)^{B-z} x^{\theta -1},
\end{align*}
which is precisely the conditional distribution $\mathrm{Beta}(\theta + f(X), B - f(X) + 1)$. For the marginal distribution, it is a standard result that it is distributed as a beta-binomial distribution, which has corresponding density of:
$${B \choose z} \frac{ \beta(z+\theta,n- z + 1)}{\beta(\theta,1)} = \frac{\theta \Gamma(z + \theta) B!}{\Gamma(B + 1 + \theta) z!},$$
yielding the desired result. 
\end{proof}
\begin{proposition}
\textbf{(P2 CP)} Similarly, if $X \sim \mathrm{Beta}(1, \theta)$, we can draw $Z \sim \mathrm{Bin}(B, 1 - X)$. Then, and the resulting $g(X) | f(X) \sim \mathrm{Beta}(B - f(X) + 1,\theta + f(X))$ and $f(X)$ has the same marginal distribution as above.
\end{proposition}
\begin{proof} We have that
\begin{align*} 
P\left( X | f(X) = Z\right) & \propto (1-x)^{z}x^{B-z} (1-x)^{\theta -1},
\end{align*}
which is precisely the conditional distribution $\mathrm{Beta}(B - f(X) + 1,\theta + f(X))$. The computation for the marginal density is the same as above as the beta function is symmetric. 
\end{proof}

\begin{proposition}
\textbf{Multivariate case: Dirichlet (P2 CP)} Suppose $ X \sim \mathrm{Dir}(\theta, 1, \ldots, 1)$, where $(\theta, 1, \ldots, 1)$ is a $d$-dimensional vector with $k \geq 2$. Draw $ Z \sim \mathrm{Multinom}(B,  X)$, where $B \in \{1, 2, \ldots\}$ is a tuning parameter. Then $f( X) =  Z$ has marginal distribution as a discrete uniform in its support $\{z_i \in \{0, \ldots, B\} \text{ for } i \in [d]: \sum_{i=1}^d z_i = B\}$ when $\theta = 1$, and for other $\theta$, the PMF is $p_\theta( z) = \frac{B! \Gamma(z_1 + \theta) \Gamma(d - 1 + \theta)}{z_1! \Gamma(\theta) \Gamma(B + d - 1 + \theta)}$. $g(X) = X$ has conditional distribution $\mathrm{Dir}(\theta + f_1( X), 1 + f_2( X), \ldots, 1+ f_k( X))$. Larger $B$ indicates more informative $f( X)$ (and less informative $g( X) \mid f( X)$).
        
In the general case, where $ X \sim \mathrm{Dir}(\theta_1, \theta_2, \ldots, \theta_d)$, we can use the same construction. Then $f( X) =  Z$ has marginal distribution 
\[
p_{ \theta}( z) = \frac{ \Gamma\left(\sum_{i=1}^d \theta_i \right)B!}{\Gamma\left(B + \sum_{i=1}^d \theta_i\right)}  \prod_{i=1}^d \frac{ \Gamma(\theta_i + z_i)  }{\Gamma(\theta_i)z_{i}!},
\]
and $g(X) = X$ has conditional distribution $\mathrm{Dir}(\theta_1 + f_1( X), \ldots, \theta_k + f_d( X))$.

\end{proposition}

\begin{proof}
Apply \cref{thm:conjugate_reversal}. It is a standard result that the Dirichlet is conjugate prior to the multinomial distribution, yielding the stated conditional distributions. The compound distribution in this case is a Dirichlet-multinomial and the stated density function is again standard.
\end{proof}

\subsubsection{Binomial Decomposition}
\begin{proposition}
 \textbf{(P2)} Suppose $X \sim \mathrm{Bin}(n, \theta)$. Draw $Z \sim \mathrm{Bin}(X, p)$ where $p \in (0,1)$ is a tuning parameter. Then $f(X) = Z$ has marginal distribution $\mathrm{Bin}(n, p\theta)$; and $g(X) = X - Z$ has conditional distribution as $$p_{\theta}(g(X) = y | f(X) = z) = \frac{(n-z)!}{y!(n-z-y)!}\left[\frac{(1-p) \theta}{1-\theta}\right]^y\left[\frac{1-\theta}{1-p \theta}\right]^{n-z}.$$ Larger $p$ indicates more informative $f(X)$. 
     %(and less informative $g(X)$). \\
     Note that the decomposition of Binomial is not trivially applicable to Bernoulli distribution since $X = 1$ with probability one if $Z = 1$. 
\end{proposition}
\begin{proof}
For the binomial decomposition, let $X = \sum_{i=1}^{n} X_{i}$ with $X_{i} \sim \text{Ber}(\theta)$. Then $Z = \sum_{j=1}^{X} Z_{i} = \sum_{j=1}^{n} Z_{i}X_{i}$. When written this way, clearly $Z \sim \text{Binom}\left(n,p\theta \right)$.

For the conditional distribution, we have that:
\begin{align*}
\mathbb{P}(g(X) = y  \mid Z=z) & =\frac{\mathbb{P}(g(X) = y  \text { and } Z=z)}{\mathbb{P}(Z=z)} \\
& =\frac{\mathbb{P}(X=z+y \text { and } Z=z)}{\mathbb{P}(Z=z)} \\
& =\frac{\mathbb{P}(Z=z \mid g(X)=z+y) \mathbb{P}(g(X)=z+y)}{\mathbb{P}(Z=z)} \\
& =\frac{\left(\begin{array}{c}
z+y \\
z
\end{array}\right) p^z(1-p)^y\left(\begin{array}{c}
n \\
z+y
\end{array}\right) \theta^{z+y}(1-\theta)^{n-z-y}}{\left(\begin{array}{l}
n \\
z
\end{array}\right)(p \theta)^z(1-p \theta)^{n-z}} \\
& =\frac{(n-z)!}{y!(n-z-y)!}\left[\frac{(1-p) \theta}{1-\theta}\right]^y\left[\frac{1-\theta}{1-p \theta}\right]^{n-z}.
\end{align*}
\end{proof}
\subsubsection{Bernoulli Decomposition}

\begin{proposition}
\textbf{(P2)} Suppose $X \sim \mathrm{Ber}(\theta)$. Draw $Z \sim \mathrm{Ber}(p)$ where $p \in (0,1)$ is a tuning parameter. Then
     $f(X) = X(1 - Z) + (1 - X)Z$ has marginal distribution $\mathrm{Ber}(\theta + p - 2p\theta)$; and $g(X) = X$ has conditional distribution (given $f(X)$) as $\mathrm{Ber}\left(\frac{\theta}{\theta + (1-\theta) [p/(1-p)]^{2f(X) - 1}}\right)$.  
    %  Note that modeling $\log\left(\frac{\theta + p - 2p\theta}{1 - \theta - p + 2p\theta}\right)$ as $u\beta$ (for covariates $u$) does not imply $\log\left(\frac{\theta}{1 - \theta}\right)$ is $u\beta'$.
     Smaller $p$ indicates more information in $f(X)$.
\end{proposition}

\begin{proof}
We have that $\mathbb{E}[f(X)] = \mathbb{E}\left[X(1 - Z) + (1 - X)Z\right] = \theta(1-p) + p(1-\theta) = \theta + p - 2p\theta $. Therefore, $f(X) \sim \text{Ber}(\theta+p-2p\theta)$.

For the conditional distribution, observe that
\begin{align*}
\mathbb{P} \left[g(X) = 1 | f(X) =1\right] = \frac{\mathbb{P} \left[g(X) = 1,f(X) =1\right]  }{ \mathbb{P} \left[f(X) =1\right]} = \frac{(1-p)\theta }{ \theta + p - 2p\theta}, \\  
\end{align*}
\begin{align*}
\mathbb{P} \left[g(X) = 1 | f(X) =0\right] = \frac{\mathbb{P} \left[g(X) = 1,f(X) =0\right]  }{ \mathbb{P} \left[f(X) =0\right]} = \frac{p(1-\theta) }{1 -  \theta - p + 2p\theta}, \\  
\end{align*}
and then recognize that this aligns with the values given by $\mathbb{P} \left[g(X) = 1 | f(X) \right] = \frac{\theta}{\theta + (1-\theta) [p/(1-p)]^{2f(X) - 1}}$.
\end{proof}
\subsubsection{Categorical Decomposition}
\begin{proposition}
\textbf{(P2)} Suppose $X \sim \mathrm{Cat}\left(\theta_{1},...,\theta_{d}\right)$. Draw $Z \sim \mathrm{Ber}(p)$ where $p \in (0,1)$ is a tuning parameter. Also draw $D \sim \mathrm{Cat}\left(\frac{1}{d},...,\frac{1}{d}\right)$. Let $f(X) = X(1-Z) + DZ$ and $g(X) = X$. Then $f(X) \sim \mathrm{Cat}\left(\phi_{1},...,\phi_{d}\right)$ with $\phi_{i} = (1-p)\theta_{i} + \frac{p}{d}$. Furthermore $g(X) | f(X)$ has distribution 
    \[p_{\theta} \left(g(X) = s | f(X) = t \right) =   
    \begin{cases}
        \frac{(1-p+ \frac{p}{d})\theta_s}{(1-p)\theta_{s} +p/d } & \text{if $s=t$} \\
         \frac{\theta_s \frac{p}{d}}{(1-p)\theta_{t} +p/d } & \text{if $s \ne t$.} \\
    \end{cases}\]
\end{proposition}

\begin{proof}
First, assume $s\ne t$,
\begin{align*}
\mathbb{P} \left(g(X) = s | f(X) = t \right) &=  \frac{ \mathbb{P} \left(f(X) = t | g(X) = s \right) \mathbb{P} \left(g(X) =s\right)    }{ \mathbb{P} \left( f(X) = t, g(X) = s \right) + \mathbb{P} \left( f(X) = t, g(X) \ne s \right)  }\\
&=  \frac{\theta_s \frac{p}{d}   }{ (1-p)\theta_{t} +p/d  }.\\
\end{align*}
Similarly, if $s= t$,
\begin{align*}
\mathbb{P} \left(g(X) = s | f(X) = s \right) &=  \frac{ \mathbb{P} \left(f(X) = s | g(X) = s \right) \mathbb{P} \left(g(X) =s\right)    }{ \mathbb{P} \left( f(X) = s, g(X) = s \right) + \mathbb{P} \left( f(X) = s, g(X) \ne s \right)  }\\
&= \frac{(1-p+ \frac{p}{d})\theta_s}{(1-p)\theta_{s} +p/d }.\\
\end{align*}

\end{proof}

\subsubsection{Poisson Decomposition}
\begin{proposition} \textbf{(P1)}  Suppose $X \sim \mathrm{Poi}(\mu)$. Draw $Z \sim \mathrm{Bin}(X, p)$ where $p \in (0,1)$ is a tuning parameter. Then
     $f(X) = Z$ has marginal distribution $\mathrm{Poi}(p\mu)$; and $g(X) = X - Z$ is independent of $f(X)$ and distributes as $\mathrm{Poi}((1 - p)\mu)$. Larger $p$ indicates more informative $f(X)$. 
     
     Alternatively, draw $Z \sim \mathrm{Poi}(p)$. Construct $f(X) = X+Z \sim \mathrm{Poi}(\mu +p)$. Letting $X=g(X)$ gives $g(X) | f(X) \sim \mathrm{Bin}\left(f(X),\frac{\mu}{\mu + p} \right)$. Larger $p$ corresponds to a less informative $f(X)$ (and more informative $g(X)|f(X)$).
\end{proposition}
\begin{proof}
Both of these decompositions are standard. See \cite{last_penrose_2017}.
\end{proof}

\subsubsection{Negative Binomial Decomposition}

\begin{proposition}
\textbf{(P2)} Suppose $X \sim \mathrm{NB}(r, \theta)$. Draw $Z \sim \mathrm{Bin}(X, p)$ where $p \in (0,1)$ is a tuning parameter. Then $f(X) = Z$ has marginal distribution $\mathrm{NB}(r, \frac{\theta}{\theta + p - p\theta})$; and $g(X) = X - Z$ has conditional distribution $\mathrm{NB}(r + Z, \theta+p - p \theta)$. Larger $p$ indicates more informative $f(X)$.
    %(and less informative $g(X)$)
    A special case of Negative Binomial is geometric distribution where $r = 1$.
\end{proposition}

\begin{proof} For the marginal distribution, observe that
\begin{align*}
\mathbb{P} \left(Z = k \right) &= \sum_{m=k}^{\infty} \mathbb{P} \left( X = m  \right) \mathbb{P} \left( Z = k | X = m \right)  \\
&=\sum_{m=k}^{\infty} { m +r -1 \choose m } (1- \theta)^{m} \theta^{r} { m \choose k } p^{k} (1-p)^{m-k} \\
&=\sum_{m=k}^{\infty} \frac{ (m+r-1)!}{(r-1)!m!} \frac{ m!}{(m-l)!k!} (1- \theta)^{m} (1-p)^{m}\theta^{r} p^{k} (1-p)^{-k} \\
&=\sum_{m=k}^{\infty}  { m +r -1 \choose m-k } \left((1-\theta)(1-p) \right)^{m} \theta^{r} (1-p)^{-k} {r+ k -1 \choose k}.
\end{align*}
Letting $\alpha = m-k $, we can rearrange terms as 
\begin{align*}
\underbrace{\sum_{\alpha=0}^{\infty}{ \alpha + r + k -1 \choose r +k -1} \left[ 1 - (1-\theta)(1-p)\right]^{r+k} }_{=1} \left[(1-\theta)(1-p) \right]^{\alpha} \left[(1-\theta)(1-p)  \right]^{k}\left[ 1 - (1-\theta)(1-p)\right]^{-r-k} \\ 
\theta^{r} p^{k} (1-p)^{-k} { r + k -1 \choose k}.
\end{align*}
Combining terms gives us:
\begin{align*}
\mathbb{P} \left(Z = k \right) &= [(1-\theta)(1-p)]^{k}  [1-(1-\theta)(1-p)]^{-r-k} \theta^{r} p^{k} (1-p)^{-k} { r + k -1 \choose r - 1} \\
&= { r + k -1 \choose r - 1}  \left[ \frac{\theta}{\theta+p-\theta p} \right]^{r}\left[ 1-\frac{\theta}{\theta+p-\theta p} \right]^{k}. \\
\end{align*}
Thus, $Z \sim \text{NB} \left(r, \frac{\theta}{\theta + p - p\theta} \right)$. 

For the conditional distribution, first we compute the joint density.
\begin{align*}
\mathbb{P} \left(Z = k | X-Z = m \right) &= \mathbb{P} \left(X = k = m \right)\mathbb{P} \left(Z=k | X = k+m\right) \\
&= {k + m + r -1 \choose k + m } (1-\theta)^{k+m} \theta^{r} { k + m \choose k } p^{k} (1-p)^{m} \\
&= {k + m + r -1 \choose m }  {r+k - 1\choose k }  \left[(1-\theta)(1-p)\right]^{k+m} \theta^{r}p^{k} (1-p)^{-k}.
\end{align*} 
Now, substituting in the marginal distribution from above will recover the conditional density. 

\begin{align*}
\mathbb{P} \left(X- Z = m | Z = k \right) &= \frac{\mathbb{P} \left(Z = k | X-Z = m \right)}{ \mathbb{P} \left(Z =k \right) } \\
&= {k + m + r -1 \choose k + m }\left[(1-\theta)(1-p)\right]^{m}\left[1-(1-\theta)(1-p)\right]^{r+k},
\end{align*} 
which is the density of $\text{NB} \left(r + Z, \theta + p - \theta p \right).$

\end{proof}